%% file: main-uniform-ocqa.tex
\documentclass[sigconf]{acmart}
\usepackage[vlined,boxed]{algorithm2e}
\usepackage{amsmath}
\usepackage{amsfonts}	
\usepackage{epsfig}
\usepackage{graphics}
\usepackage{color}
\usepackage{comment}
\usepackage{paralist}
\usepackage{mathtools}
\usepackage{multirow}
\usepackage{subfigure}
\usepackage{eepic}
\usepackage{stmaryrd}
\usepackage{textcomp}
\usepackage{balance}
\usepackage[inline]{enumitem}
\usepackage{ulem}
\usepackage{array}
\usepackage{arydshln}
\usepackage{bigdelim}
\usepackage{makecell}

\begin{document}

\input{macros.tex}

\newtheorem{claim}[theorem]{Claim}
\newtheorem{fact}[theorem]{Fact}
\newtheorem{observation}{Observation}
\newtheorem{remark}{Remark}
\newtheorem{apptheorem}{Theorem}[section]
\newtheorem{appcorollary}[apptheorem]{Corollary}
\newtheorem{appproposition}[apptheorem]{Proposition}
\newtheorem{applemma}[apptheorem]{Lemma}
\newtheorem{appclaim}[apptheorem]{Claim}
\newtheorem{appfact}[apptheorem]{Fact}

\newtheorem{manualtheoreminner}{Theorem}
\newenvironment{manualtheorem}[1]{%
	\renewcommand\themanualtheoreminner{#1}%
	\manualtheoreminner
}{\endmanualtheoreminner}

\newtheorem{manualpropositioninner}{Proposition}
\newenvironment{manualproposition}[1]{%
	\renewcommand\themanualpropositioninner{#1}%
	\manualpropositioninner
}{\endmanualpropositioninner}

\newtheorem{manuallemmainner}{Lemma}
\newenvironment{manuallemma}[1]{%
	\renewcommand\themanuallemmainner{#1}%
	\manuallemmainner
}{\endmanuallemmainner}

\fancyhead{}

\title{Uniform Operational Consistent Query Answering}

\author{Marco Calautti}
\affiliation{%
	\institution{University of Trento}
	\country{}
}
\email{marco.calautti@unitn.it}

\author{Ester Livshits}
\affiliation{%
	\institution{University of Edinburgh}
	\country{}
}
\email{ester.livshits@ed.ac.uk}

\author{Andreas Pieris}
\affiliation{%
	\institution{University of Edinburgh \&}
	\country{}
}
\affiliation{%
	\institution{University of Cyprus}
	\country{}
}
\email{apieris@inf.ed.ac.uk}

\author{Markus Schneider}
\affiliation{%
	\institution{University of Edinburgh}
	\country{}
}
\email{m.schneider@ed.ac.uk}

\begin{abstract} 
	Operational consistent query answering (CQA) is a recent framework for CQA, based on revised definitions of repairs and consistent answers, which opens up the possibility of efficient approximations with explicit error guarantees.
	The main idea is to iteratively apply operations (e.g., fact deletions), starting from an inconsistent database, until we reach a database that is consistent w.r.t.~the given set of constraints. This gives us the flexibility of choosing the probability with which we apply an operation, which in turn allows us to calculate the probability of an operational repair, and thus, the probability with which a consistent answer is entailed.	
	A natural way of assigning probabilities to operations is by targeting the uniform probability distribution over a reasonable space such as the set of operational repairs, the set of sequences of operations that lead to an operational repair, and the set of available operations at a certain step of the repairing process. This leads to what we generally call uniform operational CQA.
	The goal of this work is to perform a data complexity analysis of both exact and approximate uniform operational CQA, focusing on functional dependencies (and subclasses thereof), and conjunctive queries. The main outcome of our analysis (among other positive and negative results), is that uniform operational CQA pushes the efficiency boundaries further by ensuring the existence of efficient approximation schemes in scenarios that go beyond the simple case of primary keys, which seems to be the limit of the classical approach to CQA.
	%
	%
	%
\end{abstract}

\maketitle

\input{introduction.tex}

\input{preliminaries.tex}
\input{operational-cqa.tex}
\input{uniform.tex}
\input{uniform-repairs.tex}
\input{uniform-sequences.tex}
\input{uniform-operations.tex}
\input{conclusion.tex}

\bibliographystyle{ACM-Reference-Format}

\bibliography{references}


\newpage
\appendix
\input{app-uniform-generators.tex}
\input{app-uniform-repairs.tex}
\input{app-uniform-sequences.tex}
\input{app-uniform-operations.tex}
\input{app-singleton-operations.tex}

\end{document}

%% file: macros.tex
\newcommand{\OMIT}[1]{}

\DeclarePairedDelimiter\ceil{\lceil}{\rceil}
\DeclarePairedDelimiter\floor{\lfloor}{\rfloor}

\allowdisplaybreaks

\newcommand{\oprh}[3]{\mathsf{ORep}_{#3}(#1,#2)}
\newcommand{\probhom}[2]{\probrep{#1}{#2}}
\newcommand{\crsh}[3]{\mathsf{CRS}_{#3}(#1,#2)}

\newcommand{\op}{\mathit{op}}
\newcommand{\PS}{\mathcal{P}}
\newcommand{\viol}[2]{\mathsf{V}(#1,#2)}
\newcommand{\rs}[2]{\mathsf{RS}(#1,#2)}
\newcommand{\rsone}[2]{\mathsf{RS}^1(#1,#2)}
\newcommand{\crs}[2]{\mathsf{CRS}(#1,#2)}
\newcommand{\crss}[3]{\mathsf{CRS}_{#3}(#1,#2)}
\newcommand{\cancrs}[2]{\mathsf{CanCRS}(#1,#2)}
\newcommand{\cancrss}[3]{\mathsf{CanCRS}_{#3}(#1,#2)}
\newcommand{\opr}[2]{\mathsf{ORep}(#1,#2)}
\newcommand{\copr}[2]{\mathsf{CORep}(#1,#2)}
\newcommand{\ops}[3]{\mathsf{Ops}_{#1}(#2,#3)}
\newcommand{\opsone}[3]{\mathsf{Ops}^1_{#1}(#2,#3)}
\newcommand{\abs}[1]{\mathsf{abs}_{>0}(#1)}
\renewcommand{\abs}[1]{\mathsf{RL}(#1)}
\newcommand{\insP}{\ins{P}}
\newcommand\sem[1]{{[\![ #1 ]\!]}}
\newcommand{\probrep}[2]{\mathsf{P}_{#1}(#2)}
\newcommand{\oca}[2]{\mathsf{OCA}_{#2}(#1)}
\newcommand{\ocqa}[1]{\mathsf{OCQA}(#1)}
\newcommand{\rrelfreq}[1]{\mathsf{RRFreq}(#1)}
\newcommand{\rrelfreqone}[1]{\mathsf{RRFreq}^1(#1)}
\newcommand{\orfreq}[2]{\mathsf{rrfreq}_{#1}(#2)}
\newcommand{\orfreqone}[2]{\mathsf{rrfreq}^1_{#1}(#2)}
\newcommand{\srelfreq}[1]{\mathsf{SRFreq}(#1)}
\newcommand{\srelfreqone}[1]{\mathsf{SRFreq}^1(#1)}
\newcommand{\srfreq}[2]{\mathsf{srfreq}_{#1}(#2)}
\newcommand{\srfreqone}[2]{\mathsf{srfreq}^1_{#1}(#2)}
\newcommand{\ur}{\mathsf{ur}}
\newcommand{\us}{\mathsf{us}}
\newcommand{\uo}{\mathsf{uo}}

\newcommand{\IS}{\mathsf{IS}}
\newcommand{\ISZ}{\mathsf{IS}_{\neq \emptyset}}
\newcommand{\ISC}{\mathsf{IS}^{\mathsf{con}}}
\newcommand{\CC}{\mathsf{CC}}
\newcommand{\cg}[2]{\mathsf{CG}(#1,#2)}

\newcommand{\crsone}[2]{\mathsf{CRS}^1(#1,#2)}
\newcommand{\coprone}[2]{\mathsf{CORep}^1(#1,#2)}

\newcommand{\mi}[1]{\mathit{#1}}
\newcommand{\ins}[1]{\mathbf{#1}}
\newcommand{\adom}[1]{\mathsf{dom}(#1)}
\renewcommand{\paragraph}[1]{\textbf{#1}}
\newcommand{\ra}{\rightarrow}
\newcommand{\fr}[1]{\mathsf{fr}(#1)}
\newcommand{\dep}{\Sigma}
\newcommand{\sch}[1]{\mathsf{sch}(#1)}
\newcommand{\ar}[1]{\mathsf{ar}(#1)}
\newcommand{\body}[1]{\mathsf{body}(#1)}
\newcommand{\head}[1]{\mathsf{head}(#1)}
\newcommand{\guard}[1]{\mathsf{guard}(#1)}
\newcommand{\class}[1]{\mathbb{#1}}
\newcommand{\pos}[2]{\mathsf{pos}(#1,#2)}
\newcommand{\app}[2]{\langle #1,#2 \rangle}
\newcommand{\crel}[1]{\prec_{#1}}

\newcommand{\ccrel}[1]{\prec_{#1}^+}

\newcommand{\tcrel}[1]{\prec_{#1}^{\star}}
\newcommand{\rctaa}{\class{CT}_{\forall \forall}^{\mathsf{res}}}
\newcommand{\rctaapr}{\mathsf{CT}_{\forall \forall}^{\mathsf{res}}}
\newcommand{\rctae}{\class{CT}_{\forall \exists}^{\mathsf{res}}}
\newcommand{\rctaepr}{\mathsf{CT}_{\forall \exists}^{\mathsf{res}}}
\newcommand{\base}[1]{\mathsf{base}(#1)}
\newcommand{\eqt}[1]{\mathsf{eqtype}(#1)}
\newcommand{\result}[1]{\mathsf{result}(#1)}
\newcommand{\chase}[2]{\mathsf{ochase}(#1,#2)}
\newcommand{\pred}[1]{\mathsf{pr}(#1)}
\newcommand{\origin}[1]{\mathsf{org}(#1)}
\newcommand{\eq}[1]{\mathsf{eq}(#1)}
\newcommand{\dept}[1]{\mathsf{depth}(#1)}

\newcommand{\comp}[2]{\mathsf{comp}_{#2}(#1)}

\newcommand{\rep}[2]{\mathsf{rep}_{#2}(#1)}
\newcommand{\repp}[2]{\mathsf{rep}_{#2}\left(#1\right)}
\newcommand{\rfreq}[2]{\mathsf{rfreq}_{#2}(#1)}
\newcommand{\homs}[3]{\mathsf{hom}_{#2,#3}(#1)}
\newcommand{\prob}[1]{\mathsf{#1}}
\newcommand{\key}[1]{\mathsf{key}(#1)}
\newcommand{\keyval}[2]{\mathsf{key}_{#1}(#2)}
\newcommand{\block}[2]{\mathsf{block}_{#2}(#1)}
\newcommand{\sblock}[2]{\mathsf{sblock}_{#2}(#1)}

\newcommand{\rt}[1]{\mathsf{root}(#1)}
\newcommand{\child}[1]{\mathsf{child}(#1)}

\newcommand{\var}[1]{\mathsf{var}(#1)}
\newcommand{\const}[1]{\mathsf{const}(#1)}
\newcommand{\pvar}[2]{\mathsf{pvar}_{#2}(#1)}

\newcommand{\att}[1]{\mathsf{att}(#1)}
\newcommand{\card}[1]{\sharp #1}

\newcommand{\pr}{\mathsf{Pr}}
\newcommand{\prsp}{\mathsf{PS}}

\newcommand{\sign}[1]{\mathsf{sign}(#1)}
\newcommand{\litval}[2]{\mathsf{lval}_{#2}(#1)}
\newcommand{\angletup}[1]{\langle #1 \rangle}

\def\qed{\hfill{\qedboxempty}      
  \ifdim\lastskip<\medskipamount \removelastskip\penalty55\medskip\fi}

\def\qedboxempty{\vbox{\hrule\hbox{\vrule\kern3pt
                 \vbox{\kern3pt\kern3pt}\kern3pt\vrule}\hrule}}

\def\qedfull{\hfill{\qedboxfull}   
  \ifdim\lastskip<\medskipamount \removelastskip\penalty55\medskip\fi}

\def\qedboxfull{\vrule height 4pt width 4pt depth 0pt}

\newcommand{\markfull}{\qedboxfull}
\newcommand{\markempty}{\qed}

%% file: introduction.tex
\section{Introduction}\label{sec:introduction}

Consistent query answering (CQA) is an elegant framework, introduced in the late 1990s by Arenas, Bertossi, and Chomicki~\cite{ArBC99}, that allows us to compute conceptually meaningful answers to queries posed over inconsistent databases, that is, databases that do not conform to their specifications.
%
The key elements underlying CQA are (i) the notion of {\em (database) repair} of an inconsistent database $D$, that is, a consistent database whose difference with $D$ is somehow minimal, and (ii) the notion of query answering based on {\em consistent answers}, that is, answers that are entailed by every repair.
%
Since deciding whether a candidate answer is a consistent answer is most commonly intractable in data complexity (in fact, even for primary keys and conjunctive queries, the problem is coNP-hard~\cite{ChMa05}), there was a great effort on drawing the tractability boundary for CQA; see, e.g.,~\cite{FuFM05,FuMi07,GePW15,KoSu14,KoWi15,KoWi21}.
Much of this effort led to interesting dichotomy results that precisely characterize when CQA is tractable/intractable in data complexity.
However, the tractable fragments do not cover many relevant scenarios that go beyond primary keys.

As extensively argued in~\cite{CaLP18}, the goal of a practically applicable CQA approach should be efficient approximate query answering with explicit error guarantees rather than exact query answering. In the realm of the CQA approach described above, one could try to devise efficient probabilistic algorithms with bounded one- or two-sided error. However, it is unlikely that such algorithms exist since, even for very simple scenarios (e.g., primary keys and conjunctive queries), placing the problem in tractable randomized complexity classes such as RP or BPP would imply that the polynomial hierarchy collapses~\cite{KaLi80}.
Another promising idea is to replace the rather strict notion of consistent answers with the more refined notion of relative frequency, that is, the percentage of repairs that entail an answer, and then try to approximate it via a fully polynomial-time randomized approximation scheme (FPRAS); computing it exactly is, unsurprisingly, $\sharp ${\rm P}-hard~\cite{MaWi13}. Indeed, for primary keys and conjunctive queries, one can approximate the relative frequency via an FPRAS; this is implicit in~\cite{DaSu07}, and it has been made explicit in~\cite{CaCP19}. Moreover, a recent experimental evaluation revealed that approximate CQA in the presence of primary keys and conjunctive queries is not unrealistic in practice~\cite{CaCP21}.
However, it seems that the simple case of primary keys is the limit of this approach. We have strong indications that in the case of FDs the problem of computing the relative frequency does not admit an FPRAS, while in the case of keys it is a highly non-trivial problem~\cite{corr}.


The above limitations of the classical CQA approach led the authors of~\cite{CaLP18} to propose a new framework for CQA, based on revised definitions of repairs and consistent answers, which opens up the possibility of efficient approximations with error guarantees. The main idea underlying this new framework is to replace the declarative approach to repairs with an {\em operational} one that explains the process of constructing a repair.
In other words, we can iteratively apply operations (e.g., fact deletions), starting from an inconsistent database, until we reach a database that is consistent w.r.t.~the given set of constraints. This gives us the flexibility of choosing the probability with which we apply an operation, which in turn allows us to calculate the probability of an operational repair, and thus, the probability with which an answer is entailed.

Probabilities can be naturally assigned to operations in many scenarios leading to inconsistencies. This is illustrated by the following example from~\cite{CaLP18}.
%
	Consider a data integration scenario that results in a database containing the facts ${\rm Emp}(1, {\rm Alice})$ and ${\rm Emp}(1,{\rm Tom})$ that violate the constraint that the first attribute of the relation name ${\rm Emp}$ (the id) is a key.
	Suppose we have a level of
	trust in each of the sources; say we believe that each is 50\% reliable. With probability $0.5 \cdot 0.5 = 0.25$ we do not trust either tuple and apply the operation that removes both facts. With probability $(1-0.25)/2=0.375$ we remove either ${\rm Emp}(1,{\rm Alice})$ or ${\rm Emp}(1,{\rm Tom})$.
	%
	%

The preliminary data complexity analysis of operational CQA performed in~\cite{CaLP18} revealed that computing the probability of a candidate answer is $\sharp ${\rm P}-hard and inapproximable, even for primary keys and conjunctive queries.
%
%
%
However, these negative results should not be seen as the end of the story, but rather as the beginning since operational CQA gives us the flexibility to choose the probabilities assigned to operations.
Indeed, the main question left open by~\cite{CaLP18} is the following: how can we choose the probabilities assigned to operations so that the existence of an FPRAS is guaranteed?

A natural way of choosing those probabilities is to follow the uniform probability distribution over a reasonable space. The obvious candidates for such a space are (i) the set of operational repairs, (ii) the set of sequences of operations that lead to a repair (note that multiple such sequences can lead to the same repair), and (iii) the set of available operations at a certain step of the repairing process. This leads to the so-called {\em uniform operational CQA}.
The obvious question is how the complexity of exact and approximate operational CQA is affected if we assign probabilities to operations according to the above refined ways. In particular, we would like to understand whether uniform operational CQA allows us to go beyond the relatively simple case of primary keys.

Our goal is to perform a complexity analysis of uniform operational CQA, and provide answers to the above central questions.
Our main findings can be summarized as follows:

\begin{enumerate}
	\item Exact uniform operational CQA remains $\sharp ${\rm P}-hard, even in the case of primary keys and conjunctive queries.
	
	\item Uniform operational CQA admits an FPRAS if we focus on primary keys and conjunctive queries.
	
	\item In the case of arbitrary keys and FDs, although assigning probabilities to operations based on uniform repairs and sequences (approaches (i) and (ii) discussed above) does not lead (or it remains open whether it leads) to the approximability of operational CQA, the approach of uniform operations renders the problem approximable. The latter is a significant result since it goes beyond the simple case of primary keys.
\end{enumerate}

\OMIT{
A database is inconsistent if it does not conform to its specifications given in the form of integrity constraints. There is a consensus that inconsistency is a real-life phenomenon that arises due to many reasons such as integration of conflicting sources. With the aim of obtaining conceptually meaningful answers to queries posed over inconsistent databases, Arenas, Bertossi, and Chomicki introduced in the late 1990s the notion of Consistent Query Answering (CQA)~\cite{ArBC99}. The key elements underlying CQA are (i) the notion of {\em (database) repair} of an inconsistent database $D$, that is, a consistent database whose difference with $D$ is somehow minimal, and (ii) the notion of query answering based on {\em certain answers}, that is, answers that are entailed by every repair. A simple example, taken from~\cite{CaCP19}, that illustrates the above notions follows:

\begin{example}\label{exa:cqa}
	Consider the relational schema consisting of a single relation name $\text{\rm Employee}(\text{\rm id}, \text{\rm name}, \text{\rm dept})$ that comes with the constraint that the attribute {\text{\rm id}} functionally determines \text{\rm name} and \text{\rm dept}.
	Consider also the database $D$ consisting of the tuples:
	(1, \text{\rm Bob}, \text{\rm HR}), (1, \text{\rm Bob}, \text{\rm IT}), (2, \text{\rm Alice}, \text{\rm IT}), (2, \text{\rm Tim}, \text{\rm IT}).
	It is easy to see that $D$ is inconsistent since we are uncertain about Bob's department, and the name of the employee with id $2$. To devise a repair, we need to keep one tuple from each conflicting pair, which leads to a maximal subset of $D$ that is consistent.
	Observe now that the query that asks whether employees $1$ and $2$ work in the same department is true only in two out of four repairs, and thus, not entailed. \hfill\markfull
\end{example}

\noindent
\paragraph{Counting Repairs Entailing a Query.}
A key task in this context is to count the number of repairs of an inconsistent database $D$ w.r.t.~a set $\dep$ of constraints that entail a given query $Q$; for clarity, we base our discussion on Boolean queries.
Depending on the shape of the constraints and the query, the data complexity of the above problem can be tractable, i.e., in \text{\rm FP} (the counting analogue of \textsc{PTime}), or intractable, i.e., $\sharp \text{\rm P}$-complete (with $\sharp \text{\rm P}$ being the counting analogue of \text{\rm NP}).
In other words, given a set $\dep$ of constraints and a query $Q$, the problem $\sharp \prob{Repairs}(\dep,Q)$ that takes as input a database $D$, and asks for the number of repairs of $D$ w.r.t.~$\dep$ that entail $Q$, can be tractable or intractable depending on the shape of $\dep$ and $Q$. This leads to the natural question whether we can establish a complete classification, i.e., for every $\dep$ and $Q$, classify $\sharp \prob{Repairs}(\dep,Q)$ as tractable or intractable by simply inspecting $\dep$ and $Q$.

This is a highly non-trivial question for which Maslowski and Wijsen gave an affirmative answer providing that we concentrate on primary keys, i.e., at most one key constraint per relation name, and self-join-free conjunctive queries (SJFCQs), i.e., CQs that cannot mention a relation name more than once~\cite{MaWi13}. 
More precisely, they have established the following dichotomy result: given a set $\dep$ of primary keys, and an SJFCQ $Q$, $\sharp \prob{Repairs}(\dep,Q)$ is either in \text{\rm FP} or $\sharp$\text{\rm P}-complete, and we can determine in polynomial time, by simply analyzing $\dep$ and $Q$, which complexity statement holds.
An analogous dichotomy for arbitrary CQs with self-joins was established by the same authors in~\cite{MaWi14} under the assumption that the primary keys are simple, i.e., they consist of a single attribute. The question whether such a dichotomy result exists for arbitrary primary keys and CQs with self-joins remains a challenging open problem.


Although the picture is rather well-understood for primary keys and SJFCQs, once we go beyond primary keys we know very little concerning the existence of a complete data complexity classification as the one described above.
In particular, the dichotomy result by Maslowski and Wijsen does not apply when we consider the more general class of functional dependencies (FDs), i.e., constraints of the form $R : X \ra Y$, where $X,Y$ are subsets of the set of attributes of $R$, stating that the attributes of $X$ functionally determine the attributes of $Y$.
This brings us to the following question:

\smallskip

\noindent {\em \textbf{Question 1:} Can we lift the dichotomy result for primary keys and SJFCQs to the more general case of functional dependencies?}

\smallskip

The closest known result to the complexity classification asked by Question 1 is for the problem $\sharp \prob{Repairs}(\dep)$, where $\dep$ is a set of FDs, that takes as input a database $D$, and asks for the number of repairs of $D$ w.r.t.~$\dep$ (without considering a query). In particular, we know from~\cite{LiKW21} that whenever $\dep$ has a so-called left-hand side (LHS, for short) chain (up to equivalence), $\sharp \prob{Repairs}(\dep)$ is in \text{\rm FP}; otherwise, it is $\sharp\text{\rm P}$-complete. We also know that checking whether $\dep$ has an LHS chain (up to equivalence) is feasible in polynomial time. Let us recall that a set $\dep$ of FDs has a LHS chain if, for every two FDs $R : X_1 \ra Y_1$ and $R : X_2 \ra Y_2$ of $\dep$, $X_1 \subseteq X_2$ or $X_2 \subseteq X_1$.


\noindent
\paragraph{Approximate Counting.} Another key task in the context of database repairing is to classify $\sharp \prob{Repairs}(\dep,Q)$, for a set of constraints $\dep$ and a query $Q$, as approximable, that is, the target value can be efficiently approximated with error guarantees via a fully polynomial-time randomized approximation scheme (FPRAS), or as inapproximable. Of course, whenever $\sharp \prob{Repairs}(\dep,Q)$ is tractable, then it is trivially approximable. Thus, the interesting task is to classify the intractable cases as approximable or inapproximable.

For a set $\dep$ of primary keys, and a CQ $Q$ (even with self-joins), $\sharp \prob{Repairs}(\dep,Q)$ is always approximable; this is implicit in~\cite{DaSu07}, and it has been made explicit in~\cite{CaCP19}. However, for FDs this is not the case. Depending on the syntactic shape of $\dep$ and $Q$, $\sharp \prob{Repairs}(\dep,Q)$ can be approximable or not; these are actually results of the present work. This leads to the following question:




\noindent
\paragraph{Summary of Contributions.} 
Concerning Question (1), we lift the dichotomy of~\cite{MaWi13} for primary keys and SJFCQs to the general case of FDs (Theorem~\ref{the:fds-dichotomy}). To this end, we build on the dichotomy for the problem of counting repairs (without a query) from~\cite{LiKW21}, which allows us to concentrate on FDs with an LHS chain (up to equivalence) since for all the other cases we can inherit the $\sharp \text{P}$-hardness.
Therefore, our main technical task was actually to lift the dichotomy for primary keys and SJFCQs from~\cite{MaWi13} to the case of FDs with an LHS chain (up to equivalence). Although the proof of this result borrows several ideas from the proof of~\cite{MaWi13}, the task of lifting the result to FDs with an LHS chain (up to equivalence) was a non-trivial one. This is due to the significantly more complex structure of database repairs under FDs with an LHS chain compared to those under primary keys; further details are given in Section~\ref{sec:lhs-chain-fds}.

\OMIT{
Concerning Question (2), although we do not establish a complete classification, we provide results that, apart from being interesting in their own right, are crucial steps towards a complete classification (Theorem~\ref{the:apx-main-result}). After discussing the difficulty underlying a proper 
dichotomy (it will resolve the challenging open problem of whether counting maximal matchings in a bipartite graph is approximable), we show that, for every set $\dep$ of FDs with an LHS chain (up to equivalence) and a CQ $Q$ (even with self-joins), $\sharp \prob{Repairs}(\dep,Q)$ admits an FPRAS. On the other hand, we show that there is a very simple set $\dep$ of FDs such that, for every SJFCQ $Q$, $\sharp \prob{Repairs}(\dep,Q)$ does not admit an FPRAS (under a standard complexity assumption). 
}
}

%% file: preliminaries.tex
\section{Preliminaries}\label{sec:preliminaries}

We recall the basics on relational databases, functional dependencies, and conjunctive queries.  
In the rest of the paper, we assume the disjoint countably infinite sets $\ins{C}$ and $\ins{V}$ of {\em constants} and {\em variables}, respectively. For $n > 0$, let $[n]$ be the set $\{1,\ldots,n\}$.

\medskip

\noindent\paragraph{Relational Databases.}
A {\em (relational) schema} $\ins{S}$ is a finite set of relation names with associated arity; we write $R/n$ to denote that $R$ has arity $n > 0$. Each relation name $R/n$ is associated with a tuple of distinct attribute names $(A_1,\ldots,A_n)$; we write $\att{R}$ for the set $\{A_1,\ldots,A_n\}$ of attributes . 
%
A {\em fact} over $\ins{S}$ is an expression of the form $R(c_1,\ldots,c_n)$, where $R/n \in \ins{S}$, and $c_i \in \ins{C}$ for each $i \in [n]$.
A {\em database} $D$ over $\ins{S}$ is a finite set of facts over $\ins{S}$. 
The {\em active domain} of $D$, denoted $\adom{D}$, is the set of constants occurring in $D$.
For a fact $f = R(c_1,\ldots,c_n)$, with $(A_1,\ldots,A_n)$ being the tuple of attribute names of $R$, we write $f[A_i]$ for the constant $c_i$.
%

\medskip

\noindent
\paragraph{Functional Dependencies.}
A {\em functional dependency} (FD) $\phi$ over a schema $\ins{S}$ is an expression of the form $R : X \ra Y$, where $R/n \in \ins{S}$ and $X,Y \subseteq \att{R}$. 
When $X$ or $Y$ are singletons, we avoid the curly brackets, and simply write the attribute name.
We call $\phi$ a {\em key} if $X \cup Y = \att{R}$. 
%
%
Given a set $\dep$ of keys over $\ins{S}$, we say that $\dep$ is a set of {\em primary keys} if,
for each $R \in \ins{S}$, there exists at most one key in $\dep$ of the form $R : X \ra Y$.
A database $D$ satisfies an FD $\phi = R : X \ra Y$, denoted $D \models \phi$, if, for every two facts $R(\bar c_1),R(\bar c_2) \in D$ the following holds: $R(\bar c_1)[A]=R(\bar c_2)[A]$ for every $A \in X$ implies $R(\bar c_1)[B]=R(\bar c_2)[B]$ for every $B \in Y$. 
%
We say that $D$ is {\em consistent} w.r.t.~a set $\dep$ of FDs, written $D \models \dep$, if $D \models \phi$ for every $\phi \in \dep$; otherwise, we say that $D$ is {\em inconsistent} w.r.t.~$\dep$.
%

\medskip

\noindent
\paragraph{Conjunctive Queries.}
A {\em (relational) atom} $\alpha$ over a schema $\ins{S}$ is an expression of the form $R(t_1,\ldots,t_n)$, where $R/n \in \ins{S}$, and $t_i \in \ins{C} \cup \ins{V}$ for each $i \in [n]$.
%
%
A {\em conjunctive query} (CQ) $Q$ over $\ins{S}$ is an expression of the form $\textrm{Ans}(\bar x)\ \text{:-}\ R_1(\bar y_1), \ldots, R_n(\bar y_n)$, where $R_i(\bar y_i)$, for $ i \in [n]$, is an atom over $\ins{S}$, $\bar x$ are the {\em answer variables} of $Q$, and each variable in $\bar x$ is mentioned in $\bar y_i$ for some $i \in [n]$. We may write $Q(\bar x)$ to indicate that $\bar x$ are the answer variables of $Q$. When $\bar x$ is empty, $Q$ is called {\em Boolean}. 
%
%
The semantics of CQs is given via homomorphisms. Let $\var{Q}$ and $\const{Q}$ be the set of variables and constants in $Q$, respectively. A {\em homomorphism} from a CQ $Q$ of the form $\textrm{Ans}(\bar x)\ \text{:-}\ R_1(\bar y_1), \ldots, R_n(\bar y_n)$ to a database $D$ is a function $h : \var{Q} \cup \const{Q} \ra \adom{D}$, which is the identity over $\ins{C}$, such that $R_i(h(\bar y_i)) \in D$ for each $i \in [n]$.
A tuple $\bar c \in \adom{D}^{|\bar x|}$ is an {\em answer to $Q$ over $D$} if there is a homomorphism $h$ from $Q$ to $D$ with $h(\bar x) = \bar c$. Let $Q(D)$ be the answers to $Q$ over $D$. For Boolean CQs, we write $D \models Q$, and say that $D$ {\em entails} $Q$, if $() \in Q(D)$.


\OMIT{
\smallskip
\noindent
\paragraph{Database Repairs.}
Given a database $D$ and a set $\dep$ of FDs, both over a schema $\ins{S}$, a {\em repair} of $D$ w.r.t.~$\dep$ is a database $D' \subseteq D$ such that (i) $D' \models \dep$, and (ii) for every $D'' \supsetneq D'$, $D'' \not\models \dep$. In simple words, a repair of $D$ w.r.t.~$\dep$ is a maximal subset of $D$ that is consistent w.r.t.~$\dep$. 
Let $\rep{D}{\dep}$ be the set of repairs of $D$ w.r.t~$\dep$.
Given a CQ $Q(\bar x)$ over $\ins{S}$, and a tuple $\bar t \in \adom{D}^{|\bar x|}$, we write $\rep{D}{\dep,Q,\bar t}$ for $\{D' \in \rep{D}{\dep} \mid \bar t \in Q(D')\}$.
If $Q$ is Boolean, and thus, the only possible answer is the empty tuple $()$, we write $\rep{D}{\dep,Q}$ instead of $\rep{D}{\dep,Q,()}$.
%
For brevity, we write $\card{\rep{D}{\dep}}$ and $\card{\rep{D}{\dep,Q,\bar t}}$ for the cardinality of $\rep{D}{\dep}$ and $\rep{D}{\dep,Q,\bar t}$, respectively.
A useful observation is that $\card{\rep{D}{\dep}} = \card{\rep{D'}{\dep,Q,\bar t}}$, for some easily computable database $D'$ and tuple $\bar t$, providing that $Q$ is self-join-free.

\def\lemmarepairreduction{
	Consider a database $D$, a set $\dep$ of FDs, and an SJFCQ $Q(\bar x)$. We can compute in polynomial time a database $D'$ and a tuple $\bar t \in \adom{D'}^{|\bar x|}$ such that $\card{\rep{D}{\dep}} = \card{\rep{D'}{\dep,Q,\bar t}}$.
}
\begin{lemma}\label{lem:cook-reduction}
\lemmarepairreduction
\end{lemma}



\section{Problem Definition}\label{sec:problem-definition}
%

A key problem in the CQA context is $\sharp \prob{Repairs}$, 
which accepts as input a database $D$, a set $\dep$ of FDs, a CQ $Q(\bar x)$, and a tuple $\bar t \in \adom{D}^{|\bar x|}$, and asks for $\card{\rep{D}{\dep,Q,\bar t}}$.
\OMIT{
\medskip

\begin{center}
\fbox{\begin{tabular}{ll}
{\small PROBLEM} : & $\sharp \prob{Repairs}$
\\
{\small INPUT} : & A database $D$, a set $\dep$ of FDs, a CQ $Q(\bar x)$,
\\
& and a tuple $\bar t \in \adom{D}^{|\bar x|}$
\\
{\small OUTPUT} : &  $\card{\rep{D}{\dep,Q,\bar t}}$
\end{tabular}}
\end{center}

\medskip

\noindent 
}
The ultimate goal is to draw the tractability boundary of $\sharp \prob{Repairs}$ in {\em data complexity}. In other words, for each set $\dep$ of FDs and CQ $Q(\bar x)$, we focus on the problem

\medskip

\begin{center}
	\fbox{\begin{tabular}{ll}
			{\small PROBLEM} : & $\sharp \prob{Repairs}(\dep,Q(\bar x))$
			\\
			{\small INPUT} : & A database $D$, and a tuple $\bar t \in \adom{D}^{|\bar x|}$
			\\
			{\small OUTPUT} : &  $\card{\rep{D}{\dep,Q,\bar t}}$
	\end{tabular}}
\end{center}

\medskip

\noindent and the goal is to classify it as tractable (place it in FP) or as intractable (show that is $\sharp$P-complete). 
%
\OMIT{
In particular, FP (the counting analog of \textsc{PTime}) is defined as
\[
\text{\rm FP}\ =\ \{f \mid f \text{ is polynomial-time computable}\},
\]
whereas $\sharp$P (the counting analog of NP) is defined as
\[
\sharp \text{\rm P}\ =\ \{\mathsf{accept}_M \mid M \text{ is a polynomial-time non-deterministic Turing machine}\},
\]
where $\mathsf{accept}_M$ is the function $\{0,1\}^* \ra \mathbb{N}$ such that $\mathsf{accept}_M(x)$ is the number of accepting computations of $M$ on input $x$.
}
Typically, hardness results for $\sharp$P rely on polynomial-time Turing reductions (a.k.a. Cook reductions). In particular, given two counting functions $f,g : \{0,1\}^* \ra \mathbb{N}$, we say that $f$ is {\em Cook reducible} to $g$ if there is a polynomial-time deterministic transducer $M$, with access to an oracle for $g$, such that, for every $x \in \{0,1\}^*$, $f(x) = M(x)$. 
As discussed in Section~\ref{sec:introduction}, Maslowski and Wijsen~\cite{MaWi13} established the following classification:

\begin{theorem}[\cite{MaWi13}]\label{the:pk-dichotomy}
	For a set $\dep$ of primary keys, and an SJFCQ $Q$:
	\begin{enumerate}
		\item $\sharp \prob{Repairs}(\dep,Q)$ is either in \text{\rm FP} or $\sharp$\text{\rm P}-complete. 
		\item We can decide in polynomial time in $||\dep|| + ||Q||$ whether $\sharp \prob{Repairs}(\dep,Q)$ is in \text{\rm FP} or $\sharp$\text{\rm P}-complete.\footnote{As usual, $||o||$  denotes the size of the encoding of a syntactic object $o$.}
	\end{enumerate}
\end{theorem}

\noindent As we shall see in Section~\ref{sec:exact-counting}, one of the main results of this work is a generalization of the above result to arbitrary FDs. 

\OMIT{
A useful observation, which will be used throughout the paper, is that $\sharp \prob{Repairs}(\dep,Q)$ for an SJFCQ $Q$ is as hard as $\sharp \prob{Repairs}(\dep)$, which takes as input a database $D$, and asks for $\card{\rep{D}{\dep}}$.

\begin{lemma}\label{lem:cook-reduction}
	Consider a database $D$, a set $\dep$ of FDs, and an SJFCQ $Q(\bar x)$. We can compute in polynomial time a database $D' \supseteq D$ and a tuple $\bar t \in \adom{D'}^{|\bar x|}$ such that $\card{\rep{D}{\dep}} = \card{\rep{D'}{\dep,Q,\bar t}}$.
\end{lemma}
}

\OMIT{
{\color{red} A useful observation, which will be used throughout the paper, is that $\sharp \prob{Repairs}(\dep,Q)$ for a SJFCQ $Q$ is as hard as the problem $\sharp \prob{Repairs}(\dep)$, which takes as input a database $D$, and asks for $\card{\rep{D}{\dep}}$, i.e., the number of repairs of $D$ w.r.t.~$\dep$ (without considering any query).
The converse statement is unlikely to be true; for example, in the case of primary keys, we know that $\sharp \prob{Repairs}(\dep,Q)$ is $\sharp \text{\rm P}$-hard, but $\sharp \prob{Repairs}(\dep)$ is in \text{\rm FP}. However, it is true for {\em ground} CQs, i.e., Boolean CQs that mention only constants, that are consistent w.r.t.~$\dep$ and entailed by the given database. Note that a ground CQ, seen as a set of atoms, is actually a database, and thus we can naturally say that is consistent w.r.t.~a set of FDs.

\begin{lemma}\label{lem:cook-reduction}
	Consider a set $\dep$ of FDs, and a CQ $Q(\bar x)$. For every database $D$ the following hold:
	\begin{enumerate}
		\item If $Q$ is self-join-free, then we can compute in polynomial time a database $D' \supseteq D$ and a tuple $\bar t \in \adom{D'}^{|\bar x|}$ such that $\card{\rep{D}{\dep}} = \card{\rep{D'}{\dep,Q,\bar t}}$.
		\item If $Q$ is ground, $Q \models \dep$, and $D \models Q$, then we can compute in polynomial time a database $D' \subseteq D$ such that $\rep{D}{\dep,Q} = \rep{D'}{\dep}$.
	\end{enumerate}
\end{lemma}
}
}

\OMIT{
\begin{lemma}\label{lem:cook-reduction}
	Consider a set $\dep$ of FDs, and a CQ $Q$. Then:
	\begin{enumerate}
	\item $\sharp \prob{Repairs}(\dep)$ is Cook reducible to $\sharp \prob{Repairs}(\dep,Q)$.
	\item If $Q$ is a ground CQ, then $\sharp \prob{Repairs}(\dep,Q)$ is Cook reducible to $\sharp \prob{Repairs}(\dep)$.
	\end{enumerate}
\end{lemma}
}

\OMIT{
\begin{proof}
Given a database $D$, we simply add to $D$ a set of facts $D_Q$ such that $D_Q \models \dep$, $D_Q \models Q$, and $\adom{D} \cap (\adom{D_Q} \setminus \const{Q}) = \emptyset$; the latter condition simply states that $D$ and $D_Q$ share only constants that appear in $Q$. Note that such a set of facts $D_Q$ always exists: simply replace each variable $x$ in $Q$ with a fresh constant $c_x \in \ins{C}$, which leads to the so-called canonical database of $Q$, and then apply the FDs over the canonical database of $Q$ in the obvious way in order to obtain the database $D_Q$ that is consistent w.r.t.~$\dep$. It is easy to verify that $\card{\rep{D}{\dep}} = \card{\rep{D \cup D_Q}{\dep,Q,c(\bar x)}}$, where $c(\bar x)$ is the tuple of constants obtained by replacing each variable $x$ in $\bar x$ with $c_x$.
Thus, we can compute $\card{\rep{D}{\dep}}$ by simply constructing in polynomial time $D_Q$ and $c(x)$, and then computing $\card{\rep{D \cup D_Q}{\dep,Q,c(\bar x)}}$ by using an oracle for the problem $\sharp \prob{Repairs}(\dep,Q)$.
\end{proof}
}


\OMIT{
\begin{theorem}\label{the:pk-dichotomy}
	Consider a set $\dep$ of primary keys, and an SJFCQ $Q$. $\sharp \prob{Repairs}(\dep,Q)$ is either in \text{\rm FP} or $\sharp$\text{\rm P}-hard. We can decide in polynomial time in $||\dep|| + ||Q||$ whether $\sharp \prob{Repairs}(\dep,Q)$ is in \text{\rm FP} or $\sharp$\text{\rm P}-hard.\footnote{As usual, $||o||$ is the size of a syntactic object $o$.}
\end{theorem}
}


Another important task is whenever $\sharp \prob{Repairs}(\dep,Q)$ is intractable to classify it as approximable, i.e., the target value can be efficiently approximated with error guarantees via a {\em fully polynomial-time randomized approximation scheme} (FPRAS, for short), or as inapproximable.
Formally, an FPRAS for $\sharp \prob{Repairs}(\dep,Q(\bar x))$ is a randomized algorithm $\mathsf{A}$ that takes as input a database $D$, a tuple $\bar t \in \adom{D}^{|\bar x|}$, $\epsilon > 0$, and $0 < \delta < 1$, runs in polynomial time in $||D||$, $||\bar t||$, $1/\epsilon$ and $\log(1/\delta)$, and produces a random variable $\mathsf{A}(D,\bar t,\epsilon,\delta)$ such that
\[
\text{\rm Pr}\left(|\mathsf{A}(D,\bar t,\epsilon,\delta) - \card{\rep{D}{\dep,Q,\bar t}}|\ \leq\ \epsilon \cdot \card{\rep{D}{\dep,Q,\bar t}}\right)\ \geq\ 1-\delta.
\]
For primary keys, we know that our problem is always approximable; this is implicit in~\cite{DaSu07}, and it has been made explicit in~\cite{CaCP19}.
%
%
Although for primary keys the picture concerning approximate counting is well-understood, for arbitrary FDs is rather unexplored. We make crucial steps of independent interest towards a complete classification of approximate counting for FDs. We show that (under a reasonable assumption) the existence of an FPRAS 
for FDs is not guaranteed, but it is guaranteed for 
%
FDs with a so-called left-hand side chain (up to equivalence); details on this class are given below.

\medskip
\noindent
\paragraph{Boolean vs. Non-Boolean CQs.}
For the sake of technical clarity, both exact and approximate counting are studied for Boolean CQs, but all the results can be generalized to non-Boolean CQs. 
Thus, in the rest of the paper, by CQ or SJFCQ we refer to a Boolean query.

\OMIT{
\medskip

\noindent
\paragraph{Counting Complexity Classes.} We recall basic counting complexity classes,  that is, classes of counting functions of the form $\{0,1\}^* \ra \mathbb{N}$, that are crucial for our analysis.
We first recall the counting analog of polynomial-time:
\begin{eqnarray*}
\text{FP} &=& \{f \mid f \text{ is computable in polynomial-time}\}.
\end{eqnarray*}
Given a nondeterministic Turing machine (NTM) $M$, let $\mathsf{accept}_M$ be the function $\{0,1\}^* \ra \mathbb{N}$ such that $\mathsf{accept}_M(x)$ is the number of accepting computations of $M$ on input $x$. We can then define
\[
\sharp \text{\rm P}\ =\ \{\mathsf{accept}_M \mid M \text{ is a polynomial-time NTM}\}.
\]
Given an NTM $M$ with output tape, i.e., a nondeterministic transducer (NTT), let $\mathsf{span}_M$ be the function $\{0,1\}^* \ra \mathbb{N}$ such that $\mathsf{span}_M(x)$ is the number of {\em distinct valid} outputs of $M$ on input $x$; the output of a computation is considered to be valid if the machine halts in an accepting state. Then:
\[
\text{SpanL}\ =\ \{\mathsf{span}_M \mid M \text{ is a logarithmic-space NTT}\}.
\]
For the above classes we know that $\text{FP} \subseteq \sharp \text{P}$ and $\text{SpanL} \subseteq \sharp \text{P}$. It is also believed that the above inclusions are strict (under reasonable complexity-theoretic assumptions)~\cite{AlJe93}.

Consider two functions $f,g : \{0,1\}^* \ra \mathbb{N}$. We say that $f$ is {\em Cook reducible} to $g$ if there exists a polynomial-time deterministic transducer $M$, with access to an oracle for $g$, such that, for every $x \in \{0,1\}^*$, $f(x) = M(x)$. In other words, a Cook reduction is a polynomial-time Turing reduction. As usual, our hardness results for the class $\sharp$P rely on Cook reductions.
%
}
}

%% file: operational-cqa.tex
\section{Operational CQA}\label{sec:operational-cqa}

We proceed to recall the recent operational approach to consistent query answering, introduced in~\cite{CaLP18}. Although this new framework can deal with arbitrary integrity constraints (i.e., tuple-generating dependencies, equality-generating dependencies, and denial constraints), for our purposes we need its simplified version that only deals with functional dependencies.

\medskip

\noindent\paragraph{Operations and Violations.}
The notion of operation is the building block of the operational approach. In the original framework, the operations are standard updates $+F$ that add a set $F$ of facts to the database, or $-F$ that remove $F$ from the database. However, since in this work we deal with FDs, we only need to remove facts because the addition of a fact would never resolve a conflict.
The formal definition of the notion of operation follows. As usual, we write $\PS(S)$ for the powerset of a set $S$.
%

\begin{definition}(\textbf{Operation})\label{def:operation}
	For a database $D$ over a schema $\ins{S}$, a {\em $D$-operation} is a function $\op : \PS(D) \ra
	\PS(D)$ such that, for some non-empty set $F \subseteq D$ of facts, for every $D' \in \PS(D)$, $\op(D') = D' \setminus F$. We write $-F$ to refer to this operation. \hfill\markfull
\end{definition}

The operations $-F$ depend on the database $D$ as they are defined over $D$. Since $D$ will be clear from the context, we may refer to them simply as operations, omitting $D$. Also, when $F$ contains a single fact $f$, we write $-f$ instead of the more formal $-\{f\}$.
The main idea of the operational approach to CQA is to iteratively apply operations, starting from an inconsistent database $D$, until we reach a database $D' \subseteq D$ that is consistent w.r.t. the given set $\dep$ of FDs. However, as discussed in~\cite{CaLP18}, we need to ensure that at each step of this repairing process, at least one violation is resolved. To this end, we need to keep track of all the reasons that cause the inconsistency of $D$ w.r.t.~$\dep$. This brings us to the notion of FD violation.

\begin{definition}(\textbf{FD Violation})\label{def:violation}
	For a database $D$ over a schema $\ins{S}$, a {\em $D$-violation} of an FD $\phi = R : X \ra Y$ over $\ins{S}$ is a set $\{f,g\} \subseteq D$ of facts such that $\{f,g\} \not\models \phi$.
	%
	%
	We denote the set of $D$-violations of $\phi$ by $\viol{D}{\phi}$. Furthermore, for a set $\dep$ of FDs, we denote by $\viol{D}{\dep}$ the set $\{(\phi,v) \mid \phi \in \dep \textrm{~~and~~} v \in \viol{D}{\phi}\}$. \hfill\markfull
\end{definition}

Thus, a pair $(\phi,\{f,g\}) \in \viol{D}{\dep}$ means that one of the reasons why the database $D$ is inconsistent w.r.t.~$\dep$ is because it violates $\phi$ due to the facts $f$ and $g$.
As discussed in~\cite{CaLP18}, apart from forcing an operation to be fixing, i.e., to fix at least one violation, we also need to force an operation to remove a set of facts only if it contributes as a whole to a violation. Such operations are called justified.

\begin{definition}(\textbf{Justified Operation})\label{def:justified}
	Let $D$ be a database over a schema $\ins{S}$, and $\dep$ a set of FDs over $\ins{S}$. For a database $D' \subseteq D$, a
	$D$-operation $-F$ is called {\em $(D',\dep)$-justified}
	if there exists $(\phi,\{f,g\}) \in \viol{D'}{\dep}$ such that $F \subseteq \{f,g\}$. \hfill\markfull
\end{definition}

Note that justified operations do not try to minimize the number of atoms that need to be removed. As argued in~\cite{CaLP18}, a set of facts that collectively contributes to a violation should be considered as a justified operation during the iterative repairing process since we do not know a priori which atoms should be deleted, and therefore, we need to explore all the possible scenarios.

\medskip

\noindent\paragraph{Repairing Sequences.}
As said above, the main idea of the operational approach is to iteratively apply justified operations.
This is formalized via the notion of repairing sequence.
Consider a database $D$ and a set $\dep$ of FDs. Given a sequence $s = (\mi{op}_i)_{1 \leq i \leq n}$ of $D$-operations, we define $D_{0}^{s} = D$ and $D_{i}^{s} = \mi{op}_i(D_{i-1}^{s})$
for $i \in [n]$. In other words, $D_{i}^{s}$ is obtained by applying to $D$ the first $i$ operations of $s$. The notion of repairing sequence follows:

\begin{definition}(\textbf{Repairing Sequence})\label{def:rep-sequence}
	Consider a database $D$ and a set $\dep$ of FDs. A sequence of $D$-operations $s = (\op_i)_{1 \leq i \leq n}$ is called {\em $(D,\dep)$-repairing} if, for every $i \in [n]$, $\op_i$ is $(D_{i-1}^{s},\dep)$-justified.
	Let $\rs{D}{\dep}$ be the set of all $(D,\dep)$-repairing sequences. \hfill\markfull
\end{definition}

It is easy to verify that
the length of a $(D,\dep)$-repairing sequence is linear in the size of $D$. It is also clear that the set $\rs{D}{\dep}$ is finite.
For a $(D,\dep)$-repairing sequence $s = (\op_i)_{1 \leq i \leq n}$, we define its {\em result} as the database $s(D) = D_{n}^{s}$, and call it {\em complete} if $s(D) \models \dep$, i.e., it leads to a consistent database. Let $\crs{D}{\dep}$ be the set of all complete $(D,\dep)$-repairing sequences.


\medskip

\noindent\paragraph{Operational Repairs.} 
A {\em candidate (operational) repair} of a database $D$ w.r.t.~a set $\dep$ of FDs is a database $D'$ such that $D' = s(D)$ for some $s \in \crs{D}{\dep}$. Let $\copr{D}{\dep}$ be the set of all candidate repairs of $D$ w.r.t.~$\dep$.
Although every database of $\copr{D}{\dep}$ corresponds to a conceptually meaningful way of repairing the database $D$, we would like to have a mechanism that allows us to choose which candidate repairs should be considered for query answering purposes, and assign likelihoods to those repairs.

The fact that we can operationally repair an inconsistent database via repairing sequences gives us the flexibility of choosing which operations (that is, fact deletions) are more likely than others, which in turn allows us to talk about the probability of a repair, and thus, the probability with which an answer is entailed.
The idea of assigning likelihoods to operations extending sequences can be described as follows: for all possible extensions $s\cdot \op_1,\ldots, s\cdot \op_k$ of a repairing sequence $s$, we assign
probabilities $p_1,\ldots, p_k$ to them so they add up to $1$. This is done by exploiting a {\em tree-shaped Markov chain} that arranges its states (i.e., repairing sequences) in a rooted tree, where (i) the empty sequence of operations, which is by definition repairing, is the root, (ii) the children of each state are its possible extensions, and (iii) the set of states corresponding to complete sequences coincide with the set of leaves.
We write $\varepsilon$ for the empty sequence of operations.
We further write $\ops{s}{D}{\dep}$ for the set of $(D,\dep)$-repairing sequences $\{s' \in \rs{D}{\dep} \mid s' = s \cdot \op \text{ for some } D\text{-operation } \op\}$.

\begin{definition}(\textbf{Repairing Markov Chain})\label{def:repaiting-mc}
	For a database $D$ and a set $\dep$ of FDs, a {\em $(D,\dep)$-repairing Markov chain} is an edge-labeled rooted tree $T = (V,E,\ins{P})$, where $V = \rs{D}{\dep}$, $E \subseteq V \times V$, and $\ins{P} : E \ra [0,1]$, such that:
	\begin{enumerate}
		\item the root is the empty sequence $\varepsilon$,
		\item for a non-leaf node $s \in V$, $\{s' \mid (s,s') \in E\} = \ops{s}{D}{\dep}$,
		\item for a non-leaf node $s \in V$, $\sum_{t \in \{s' \mid (s,s') \in E\}} \ins{P}(s,t) = 1$, and
		\item $\{s \in V \mid s \text{ is a leaf}\} = \crs{D}{\dep}$.
	\end{enumerate}
	A {\em repairing Markov chain generator} w.r.t.~$\dep$ is a function $M_\dep$ assigning to every database $D$ a $(D,\dep)$-repairing Markov chain.\hfill\markfull
\end{definition}

We give a simple example, which will serve as a running example, that illustrates the notion of repairing Markov chain:

\begin{figure}[t]
	\centering
	\includegraphics[width=.45\textwidth]{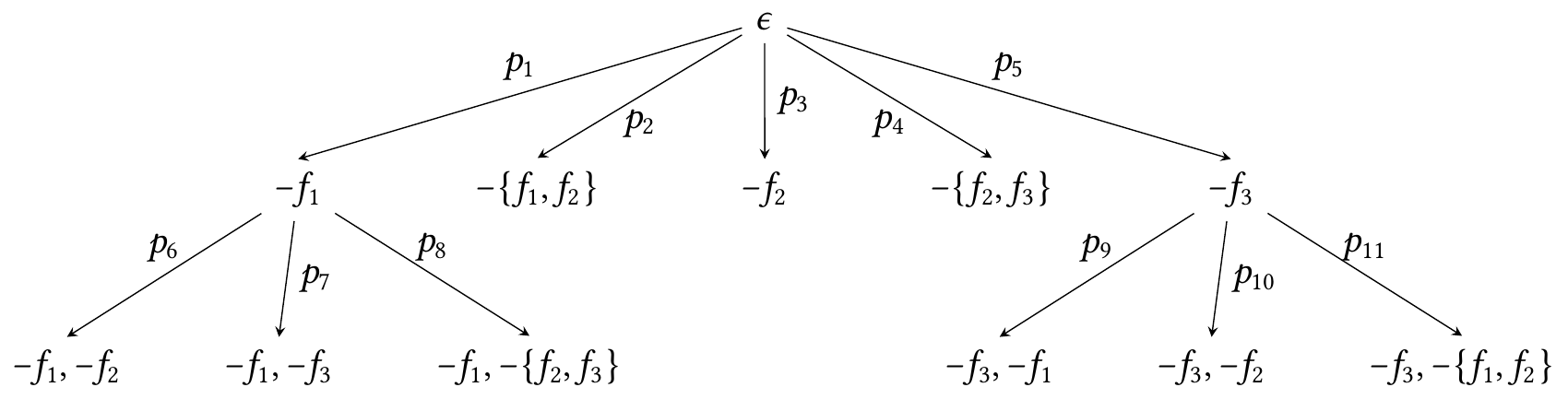}
	\caption{Repairing Markov Chain}
	\label{fig:markov-chain}
\end{figure}

\begin{example}\label{exa:repairing-mc}
	Consider the database $D=\{f_1,f_2,f_3\}$ over the schema $\ins{S} = \{R/3\}$, where $f_1 = R(a_1,b_1,c_1)$, $f_2 = R(a_1,b_2,c_2)$ and $f_3 = R(a_2,b_1,c_2)$. Consider also the set $\dep = \{\phi_1,\phi_2\}$ of FDs over $\ins{S}$, where $\phi_1 = R: A \ra B$ and $\phi_2 = R: C \ra B$, assuming that $(A,B,C)$ is the tuple of attributes of $R$. 
	It is easy to see that $D \not\models \dep$. In particular, we have that $\viol{D}{\dep} = \{(\phi_1,\{f_1,f_2\}), (\phi_2,\{f_2,f_3\})\}$.
	It is easy to verify that for the edge-labeled rooted tree $T = (V,E,\ins{P})$ in Figure~\ref{fig:markov-chain}, $V = \rs{D}{\dep}$, for a non-leaf node $s$ the set of its children is $\ops{s}{D}{\dep}$, and the set of leaves coincides with $\crs{D}{\dep}$. Hence, providing that $p_1+p_2+p_3+p_4+p_5=1$, $p_6+p_7+p_8=1$ and $p_9+p_{10}+p_{11}=1$, $T$ is a $(D,\dep)$-repairing Markov chain. \hfill\markfull
\end{example}

The purpose of a repairing Markov chain generator is to provide a mechanism for defining a family of repairing Markov chains independently of the database. One can design a repairing Markov chain generator $M_\dep$ once, and whenever the database $D$ changes, the desired $(D,\dep)$-repairing Markov chain is simply $M_\dep(D)$.
%


We now recall the notion of operational repair: they are candidate operational repairs obtained via repairing sequences that are {\em reachable} leaves of a repairing Markov chain, i.e., leaves with non-zero probability. The probability of a leaf is coming from the so-called leaf distribution of a repairing Markov chain. Formally, given a database $D$ and a set $\dep$ of FDs, the {\em leaf distribution} of a $(D,\dep)$-repairing Markov chain $T = (V,E,\ins{P})$ is a function $\pi$ that assigns to each leaf $s$ of $T$ a number from $[0,1]$ as follows: assuming that $(s_0,s_1)$, $(s_1,s_2)$, $\ldots$, $(s_{n-1},s_n)$, where $n \geq 0$, $\varepsilon = s_0$ and $s = s_n$, is the unique path in $T$ from $\varepsilon$ to $s$, $\pi(s) = \ins{P}(s_0,s_1) \cdot \ins{P}(s_1,s_2) \cdot \cdots \cdot \ins{P}(s_{n-1},s_n)$.
The set of {\em reachable leaves} of $T$, denoted $\abs{T}$, is the set of leaves of $T$ that have non-zero probability according to the leaf distribution of $T$.

\begin{definition}(\textbf{Operational Repair})\label{def:operational-repair}
	Given a database $D$, a set $\dep$ of FDs, and a repairing Markov chain generator $M_\dep$ w.r.t.~$\dep$, an {\em (operational) repair} of $D$ w.r.t.~$M_\dep$ is a database $D' \in \copr{D}{\dep}$ such that $D' = s(D)$ for some $s \in \abs{M_\dep(D)}$. 	Let $\opr{D}{M_\dep}$ be the set of all operational repairs of $D$ w.r.t.~$M_\dep$. \hfill\markfull
\end{definition}

An operational repair may be obtainable via multiple repairing
sequences that are reachable leaves of the underlying
repairing Markov chain. The probability of a repair $D'$ is calculated by summing up the probabilities of all reachable leaves $s$ so that $D' = s(D)$.


\begin{definition}(\textbf{Operational Semantics})\label{def:operational-semantics}
	Given a database $D$, a set $\dep$ of FDs, and a repairing Markov chain generator $M_\dep$ w.r.t.~$\dep$, the probability of an operational repair $D'$ of $D$ w.r.t.~$M_\dep$ is
	\[
	\probrep{D,M_\dep}{D'}\ =\ \sum\limits_{s \in \abs{M_\dep(D)} \text{ and } D'=s(D)} \pi(s),
	\]
	where $\pi$ is the leaf distribution of $M_\dep(D)$.
	%
	%
	The {\em operational semantics} of $D$ w.r.t.~$M_\dep$ is defined as the set of repair-probability pairs $\sem{D}_{M_\dep} =
	\left\{\left(D',\probrep{D,M_\dep}{D'}\right) \mid D' \in \opr{D}{M_\dep}\right\}$.
	\hfill\markfull
\end{definition}

\medskip

\noindent\paragraph{Operational CQA.} We now have in place all the necessary notions to recall the operational approach to consistent query answering, and define the main problem of interest.
For a database $D$, a set $\dep$ of FDs, a Markov chain generator
$M_\dep$ w.r.t.~$\dep$, a query $Q(\bar x)$, and a tuple $\bar c \in \adom{D}^{|\bar x|}$, the probability of $\bar c$ being an answer to $Q$ over some operational repair of $D$ is defined as
\begin{eqnarray*}
	\probrep{M_\dep,Q}{D,\bar c} &=& \sum\limits_{(D',p) \in \sem{D}_{M_\dep} \textrm{~and~} \bar c \in
		Q(D')} p.
\end{eqnarray*} 
We can now talk about operational consistent answers. 
In particular, the set of {\em operational consistent answers} to $Q$ over $D$ according to $M_\dep$ is defined as the set $\big\{\big(\bar c, \probrep{M_\dep,Q}{D,\bar c}\big) \mid \bar c \in \adom{D}^{|\bar x|}\big\}$.
\OMIT{
\begin{definition}(\textbf{Operational Consistent Answers})\label{def:consistent-answers}
	For a database $D$, a set $\dep$ of FDs, a repairing Markov
	chain generator $M_\dep$ w.r.t.~$\dep$, and a query $Q(\bar x)$, the
	set of {\em operational consistent answers} to $Q$ over $D$ according to $M_\dep$ is defined as the set of tuple-probability pairs
	$\oca{D}{M_\dep,Q}\ =\ \big\{\big(\bar c, \probrep{M_\dep,Q}{D,\bar c}\big) \mid \bar c \in \adom{D}^{|\bar x|}\big\}.$
	\hfill\markfull
\end{definition}
}


\OMIT{
\medskip
\noindent\paragraph{Operational Repairs.} The fact that we can operationally repair an inconsistent database via repairing sequences gives us the flexibility of choosing which operations (i.e., fact deletions) are more likely than others. This allows us to talk about the probability of a repair, which in turn allows us to talk about the probability with which a consistent answer is entailed. The tool from probability theory that allows us to formalize the above discussion is {\em Markov chains}~\cite{KeLa83}.

A Markov chain is essentially an edge-labeled directed graph, where
the nodes are its states and the edges are labeled with a probability
$p \in [0,1]$ so that for each node $s$, the labels of its
outgoing edges sum up to 1. An edge $(s,s')$ with label $p$ says
that with probability $p$, the state changes from
$s$ to $s'$.
Formally, a {\em Markov chain} $M$ over a (finite) state space $S =
\{s_0,\ldots,s_k\}$ is a pair $(s_0,\ins{P})$, where
$s_0$ is the initial state of $M$ and $\ins{P} : S \times S \ra [0,1]$
is a stochastic function, i.e.,  $\sum_{s' \in S}
\ins{P}(s,s') = 1$ for every state $s \in S$.
Since $S$ is finite, the function $\ins{P}$ can be naturally seen as an
$|S| \times |S|$ matrix, called {\em probability transition matrix},
whose $(i,j)$-th cell contains $\ins{P}(s_i,s_j)$. By abuse of notation, whenever $\ins{P}$ is treated as a matrix, we write $\ins{P}(s_i,s_j)$ instead of $\ins{P}(i,j)$.
Starting from the state $s_i$, the probability of reaching $s_j$ after $n$ steps is $\ins{P}^n(s_i,s_j)$, where $\ins{P}^n = \prod_{i=1}^{n} \ins{P}$.
A state $s \in S$ is called {\em absorbing} if $\ins{P}(s,s) = 1$, i.e., $s$ is reachable from itself with probability $1$. The set of {\em reachable absorbing states} of $M$, denoted $\abs{M}$, is the set of absorbing states of $M$ that are reachable from $s_0$ with non-zero probability. \OMIT{i.e., the set
\[
\abs{M} =
\left\{s \in S~~\biggl|
\begin{array}{c}
\ins{P}(s,s) = 1 \text{ and } \\
\text{there exists } n \text{ such that } \ins{P}^n(s_0,s) > 0
\end{array}\right\}.
\]}
Finally, a key notion is the {\em hitting distribution} of $M$, which describes the long-term behaviour of the Markov chain. It is defined as the limit $\lim_{n \ra \infty} \ins{P}^n(s_0)$ if it exists, where
$\ins{P}^n(s_0)$ is the $0$-th row of $\ins{P}^n$.

We now formalize the idea of assigning likelihoods to operations
extending sequences: for all possible extensions $s\cdot \op_1,
\ldots, s\cdot \op_k$ of a repairing sequence $s$, we assign
probabilities $p_1,\ldots, p_k$ to them so they add up to $1$. This is done by exploiting a tree-shaped Markov chain that arranges its states (i.e., repairing sequences) in a tree, where the children of each state are its possible extensions. Furthermore, states corresponding to complete sequences coincide with the absorbing states of the Markov chain.
We write $\varepsilon$ for the empty sequence of operations, which is by definition repairing.
We further write $\ops{s}{D}{\dep}$ for the set of $(D,\dep)$-repairing sequences $\{s' \in \rs{D}{\dep} \mid s' = s \cdot \op \text{ for some } D\text{-operation } \op\}$.

\begin{definition}(\textbf{Repairing Markov Chain})\label{def:repaiting-mc}
	For a database $D$ and a set $\dep$ of FDs, a {\em $(D,\dep)$-repairing Markov chain} is a Markov chain $(\varepsilon, \insP)$, where $\insP : \rs{D}{\dep} \times \rs{D}{\dep} \ra [0,1]$ is such that:
	\begin{enumerate}
		\item For each $s \in \rs{D}{\dep}$, $s$ is complete iff it is absorbing.
		
		\item For each $s,s' \in \rs{D}{\dep}$, if $s \neq s'$ and $\insP(s,s') > 0$, then $s' \in \ops{s}{D}{\dep}$.
	\end{enumerate}
	A {\em repairing Markov chain generator} w.r.t.~$\dep$ is a function $M_\dep$ assigning to every database $D$ a $(D,\dep)$-repairing Markov chain.\hfill\markfull
\end{definition}

The purpose of the repairing Markov chain generator is to provide a generic mechanism for defining a family of repairing Markov chains independently of the input database. In other words, one can design a repairing Markov chain generator $M_\dep$ once, and whenever the database $D$ changes, the desired repairing Markov chain is obtained by simply applying $M_\dep$ on $D$.
We now recall the notion of operational repair: they are obtained by repairing sequences that are reachable absorbing states of a repairing Markov chain.

\begin{definition}(\textbf{Operational Repair})\label{def:operational-repair}
	Given a database $D$, a set $\dep$ of FDs, and a repairing Markov chain generator $M_\dep$ w.r.t.~$\dep$, an {\em (operational) repair} of $D$ w.r.t.~$M_\dep$ is a database $D'$ such that $D' = s(D)$ for some $s \in \abs{M_\dep(D)}$. 	Let $\opr{D}{M_\dep}$ be the set of all operational repairs of $D$ w.r.t.~$M_\dep$. \hfill\markfull
\end{definition}

An operational repair may be obtainable via multiple repairing
sequences that are reachable absorbing states of the underlying
repairing Markov chain. The probability of a repair $D'$ is calculated by summing up the probabilities of all reachable absorbing
states $s$ so that $D' = s(D)$. These probabilities are coming from the hitting distribution, which always exists.
In other words, given a database $D$ and a set $\dep$ of FDs, for every $(D,\dep)$-repairing Markov chain $(\varepsilon, \insP)$, the limit $\lim_{n \ra \infty} \ins{P}^n(s_0)$ always exists.

\begin{definition}(\textbf{Operational Semantics})\label{def:operational-semantics}
	Given a database $D$, a set $\dep$ of FDs, and a repairing Markov chain generator $M_\dep$ w.r.t.~$\dep$, the probability of an operational repair $D'$ of $D$ w.r.t.~$M_\dep$ is
	\[
	\probrep{D,M_\dep}{D'}\ =\ \sum\limits_{s \in \abs{M_\dep(D)} \text{ and } D'=s(D)} \pi(s),
	\]
	where $\pi$ is the hitting distribution of $M_\dep(D)$.
	%
	%
	The {\em operational semantics} of $D$ w.r.t.~$M_\dep$ is defined as the set of repair-probability pairs $\sem{D}_{M_\dep} =
	\left\{\left(D',\probrep{D,M_\dep}{D'}\right) \mid D' \in \opr{D}{M_\dep}\right\}$.
	 \hfill\markfull
\end{definition}

\medskip

\noindent\paragraph{Operational CQA.} We now have in place all the necessary notions to recall the operational approach to consistent query answering, and define the main problem of interest.
For a database $D$, a set $\dep$ of FDs, a Markov chain generator
$M_\dep$ w.r.t.~$\dep$, a query $Q(\bar x)$, and a tuple $\bar c \in \adom{D}^{|\bar x|}$, the probability of $\bar c$ being in the answer to $Q$ over some operational repair of $D$ is defined as
\begin{eqnarray*}
	\probrep{M_\dep,Q}{D,\bar c} &=& \sum\limits_{(D',p) \in \sem{D}_{M_\dep} \textrm{~and~} \bar c \in
			Q(D')} p
\end{eqnarray*} 
The definition of consistent answers follows.

\begin{definition}(\textbf{Operational Consistent Answers})\label{def:consistent-answers}
	For a database $D$, a set $\dep$ of FDs, a repairing Markov
	chain generator $M_\dep$ w.r.t.~$\dep$, and a query $Q(\bar x)$, the
	set of {\em operational consistent answers} to $Q$ over $D$ according to $M_\dep$ is defined as the set of tuple-probability pairs
	$\oca{D}{M_\dep,Q}\ =\ \big\{\big(\bar c, \probrep{M_\dep,Q}{D,\bar c}\big) \mid \bar c \in \adom{D}^{|\bar x|}\big\}.$
	\hfill\markfull
\end{definition}
}


The problem of interest in this context, dubbed $\mathsf{OCQA}$, accepts as input a database $D$, a set $\dep$ of FDs, a repairing Markov chain generator $M_\dep$ w.r.t.~$\dep$, a query $Q(\bar x)$, and a tuple $\bar c \in \adom{D}^{|\bar x|}$, and asks for the probability $\probrep{M_\dep,Q}{D,\bar c}$. We are, in fact, interested in the {\em data complexity} of $\mathsf{OCQA}$, i.e., for a set $\dep$ of FDs, a repairing Markov chain generator $M_\dep$ w.r.t.~$\dep$, and a query $Q(\bar x)$, we focus on

\medskip

\begin{center}
	\fbox{\begin{tabular}{ll}
			{\small PROBLEM} : & $\ocqa{\dep,M_\dep,Q(\bar x)}$
			\\
			{\small INPUT} : & A database $D$,  and a tuple $\bar c \in \adom{D}^{|\bar x|}$.
			\\
			{\small OUTPUT} : &  $\probrep{M_\dep,Q}{D,\bar c}$.
	\end{tabular}}
\end{center}

\medskip


Until now, a repairing Markov chain generator is a general function.
We proceed to discuss the novel idea of uniform operational CQA, which provides concrete ways of defining such a function.

\OMIT{
\medskip

The data complexity of $\mathsf{OCQA}$ has been already studied in~\cite{CaLP18}, showing that it is, in general, intractable. In particular, it has been shown via a reduction from query answering over probabilistic databases, which is known to be $\sharp $P-hard~\cite{DaSu07}, that:

\begin{theorem}[\cite{CaLP18}]\label{the:ocqa-exact}
	There exist a set $\dep$ of primary keys, a repairing Markov chain generator $M_\dep$ w.r.t.~$\dep$, and a CQ $Q$ such that $\ocqa{\dep,M_\dep,Q}$ is $\sharp ${\rm P}-hard.
\end{theorem}

Let us stress that the above result holds even for {\em well-behaved} Markov chain generators $M_\dep$, where, for every database $D$, $M_\dep(D) = (V,E,\insP)$ is computable in polynomial time in the size of $D$, and $\insP$ is also computable in polynomial time in the size of $D$. In simple words, this tells us that the intractability of $\mathsf{OCQA}$ in data complexity is not due to the complexity of the underlying Markov chain generator, but due to reasons that are intrinsic to operational CQA.

With the above intractability result in place, it is natural to ask whether $\ocqa{\dep,M_\dep,Q(\bar x)}$ is approximable, i.e., whether the target probability can be approximated via a {\em fully polynomial-time randomized approximation scheme} (FPRAS, for short).
Formally, an FPRAS for $\ocqa{\dep,M_\dep,Q(\bar x)}$ is a randomized algorithm $\mathsf{A}$ that takes as input a database $D$, a tuple $\bar c \in \adom{D}^{|\bar x|}$, $\epsilon > 0$, and $0 < \delta < 1$, runs in polynomial time in $||D||$, $||\bar c||$,\footnote{As usual, $||o||$ denotes the size of the encoding of a syntactic object $o$.} $1/\epsilon$ and $\log(1/\delta)$, and produces a random variable $\mathsf{A}(D,\bar c,\epsilon,\delta)$ such that
\[
\text{\rm Pr}\left(|\mathsf{A}(D,\bar c,\epsilon,\delta) - \probrep{M_\dep,Q}{D,\bar c}|\ \leq\ \epsilon \cdot \probrep{M_\dep,Q}{D,\bar c}\right)\ \geq\ 1-\delta.
\]
The approximability of the problem in question has been already studied in~\cite{CaLP18}, showing that it does not admit an FPRAS, under the widely accepted complexity assumption that ${\rm RP} \neq {\rm NP}$. Recall that RP is the complexity class of problems that are efficiently solvable via a randomized algorithm with a bounded one-sided error (i.e., the answer may mistakenly be ``no'')~\cite{arora-book}.

\begin{theorem}[\cite{CaLP18}]\label{the:ocqa-apx}
	Unless  ${\rm RP} = {\rm NP}$, there exist a set $\dep$ of primary keys, a Markov chain generator $M_\dep$ w.r.t.~$\dep$, and a CQ $Q$ such that there is no FPRAS for $\ocqa{\dep,M_\dep,Q}$.
\end{theorem}

Let us note that, as for Theorem~\ref{the:ocqa-exact}, the above result holds even for well-behaved Markov chain generators, which means that the reasons for the inapproximability are intrinsic to operational CQA.
}

%% file: uniform.tex
\section{Uniform Operational CQA}\label{sec:uniform}

A natural way of defining a repairing Markov chain is to assign probabilities to operations according to the uniform probability distribution over a reasonable space. The obvious options for such a space are (i) the set of candidate operational repairs, (ii) the set of complete repairing sequences, and (iii) the set of available operations at a certain step of the repairing process.
More precisely, given a set $\dep$ of FDs, it is natural to consider the repairing Markov chain generators $M_{\dep}^{\ur}$ (uniform repairs), $M_{\dep}^{\us}$ (uniform sequences), and $M_{\dep}^{\uo}$ (uniform operations) w.r.t.~$\dep$ such that, for a database $D$:
\begin{enumerate}
	\item[(i)] $\opr{D}{M_{\dep}^{\ur}} = \copr{D}{\dep}$, and for $D' \in \opr{D}{M_{\dep}^{\ur}}$, $\probrep{D,M_{\dep}^{\ur}}{D'} = \frac{1}{|\opr{D}{M_{\dep}^{\ur}}|}$.
	
	\item[(ii)] For every $s \in \crs{D}{\dep}$, assuming that $\pi$ is the leaf distribution of $M_{\dep}^{\us}(D)$, $\pi(s) = \frac{1}{|\crs{D}{\dep}|}$.
	
	\item[(iii)] For every $s,s' \in \rs{D}{\dep}$, assuming that $M_{\dep}^{\uo}(D) = (V,E,\insP)$, $s' \in \ops{s}{D}{\dep}$ implies $\insP(s,s') = \frac{1}{|\ops{s}{D}{\dep}|}$.
\end{enumerate}
%
%
We now explain, by means of an example, how these Markov chain generators are defined; the formal definitions are in Appendix~\ref{appsec:uniform-generators}.

In the rest of the section, let $D$ and $\dep$ be the database and the set of FDs, respectively, from Example~\ref{exa:repairing-mc}. Recall that any $(D,\dep)$-repairing Markov chain looks as the one depicted in Figure~\ref{fig:markov-chain} with $p_1+p_2+p_3+p_4+p_5=1$, $p_6+p_7+p_8=1$ and $p_9+p_{10}+p_{11}=1$.
Thus, the task of understanding how the Markov chain generators $M_{\dep}^{\ur}$, $M_{\dep}^{\us}$ and $M_{\dep}^{\uo}$ should be defined boils down to understanding how the probabilities $p_1,\ldots,p_{11}$ should be calculated by $M_{\dep}^{\ur}$, $M_{\dep}^{\us}$ and $M_{\dep}^{\uo}$ in order to guarantee 
the properties discussed above.
%
%
\OMIT{
Consider the database $D=\{f_1,f_2,f_3\}$ over the schema $\ins{S} = \{R/3\}$, where $f_1 = R(a_1,b_1,c_1)$, $f_2 = R(a_1,b_2,c_2)$ and $f_3 = R(a_2,b_1,c_2)$. Consider also the set $\dep = \{\phi_1,\phi_2\}$ of FDs over $\ins{S}$, where $\phi_1 = R: A \ra B$ and $\phi_2 = R: C \ra B$, assuming that $(A,B,C)$ is the tuple of attributes of $R$. 
It is easy to see that $D \not\models \dep$. In particular, we have that $\viol{D}{\dep} = \{(\phi_1,\{f_1,f_2\}), (\phi_2,\{f_2,f_3\})\}$.
The following tree represents a Markov chain $(\epsilon,\ins{P})$ whose state space is $\rs{D}{\dep}$; an edge $(s,s')$ with label $p$ means that $\ins{P}(s,s')=p$, whereas an edge $(s,s')$ that is not present in the tree means that $\ins{P}(s,s')=0$:

\smallskip

\centerline{\includegraphics[width=.45\textwidth]{example-mc.pdf}}

\smallskip

\noindent It is easy to verify that, for each $s \in \rs{D}{\dep}$, the set of children of $s$ in the above tree is precisely $\ops{s}{D}{\dep}$. One can also verify that the set of leaves coincides with $\crs{D}{\dep}$. 
Thus, any $(D,\dep)$-repairing Markov chain has the shape of $(\epsilon,\ins{P})$ shown above. 
The question is how the probabilities $p_1,\ldots,p_{11}$ should be calculated by the Markov chain generators $M_{\dep}^{\ur}$, $M_{\dep}^{\us}$ and $M_{\dep}^{\uo}$ in order to guarantee that the $(D,\dep)$-repairing Markov chains $M_{\dep}^{\ur}(D)$, $M_{\dep}^{\us}(D)$ and $M_{\dep}^{\uo}(D)$, respectively, enjoy the three desired properties discussed above.
\OMIT{
	, namely:
\begin{enumerate}
	\item for an operational repair $D' \in \{\emptyset,\{f_1\},\{f_2\},\{f_3\},\{f_1,f_3\}\}$, the sum of the products of the probabilities that label the edges of the paths from $\epsilon$ to a sequence $s$ such that $s(D) = D'$ is $\frac{1}{5}$ since there are five operational repairs,
	
	\item for each leaf $s$, the product of the probabilities that label the edges of the path from $\epsilon$ to $s$ is $\frac{1}{9}$ since there are nine leaves, and thus, nine complete $(D,\dep)$-repairing sequences, and
	
	\item for each node $s$, assuming that $s_1,\ldots,s_n$ are the children of $s$, for $n \geq 1$, each edge $(s,s_i)$ is labeled by the probability $\frac{1}{n}$.
\end{enumerate}
}
}
We start by explaining how the probabilities are calculated by $M_{\dep}^{\us}$, which will then help us to explain how the probabilities are calculated by $M_{\dep}^{\ur}$. We finally discuss $M_{\dep}^{\uo}$, which is the simplest one.


\medskip
\noindent \paragraph{Uniform Sequences.}
For a sequence $s \in \rs{D}{\dep}$, let $\crss{D}{\dep}{s}$ be the set of all sequences of $\crs{D}{\dep}$ that have $s$ as a prefix. Thus, $\crss{D}{\dep}{s}$ collects the leaves of the subtree rooted at $s$, with $\crss{D}{\dep}{\epsilon} = \crs{D}{\dep}$ being the set of leaves.
Hence, for $M_\dep^{\us}(D) = (V,E,\ins{P})$ to induce the uniform distribution over the leaves, it suffices, for $s,s' \in \rs{D}{\dep}$ with $s' \in \ops{s}{D}{\dep}$, to let 
\[
\ins{P}(s,s') = \frac{|\crss{D}{\dep}{s'}|}{|\crss{D}{\dep}{s}|}.
\]
Observe that
\[
|\crss{D}{\dep}{\epsilon}|\ =\ 9
\]
\[
|\crss{D}{\dep}{-f_1}|\ =\ |\crss{D}{\dep}{-f_3}|\ =\ 3
\]
\[
|\crss{D}{\dep}{-\{f_1,f_2\}}|\ =\ |\crss{D}{\dep}{-f_2}|\ =\ |\crss{D}{\dep}{-\{f_2,f_3\}}|\ =\ 1.
\]
Hence, $p_1$ $=$ $p_5$ $=$ $\frac{3}{9}$, $p_2 = p_3$ $=$ $p_4$ $=$ $\frac{1}{9}$. Similarly, we obtain that $p_6 = p_7 = p_8 = \frac{1}{3}$, and $p_9 = p_{10} = p_{11} = \frac{1}{3}$.
Thus, $\abs{M_{\dep}^{\us}(D)} = \crs{D}{\dep}$, and $\pi(s)$ $=$ $\frac{1}{9}$, for each $s \in \abs{M_\dep^{\us}(D)}$, with $\pi$ being the leaf distribution of $M_{\dep}^{\us}(D)$, as needed.
\OMIT{
Consequently, 
\[ 
p_1\ =\ p_5\ =\ \frac{3}{9} \quad \text{and} \quad p_2\ =\ p_3\ =\ p_4\ =\ \frac{1}{9}. 
\]
Similarly, we obtain that
\[
p_6\ =\ p_7\ =\ p_8\ =\ \frac{1}{3} \quad \text{and} \quad p_9\ =\ p_{10}\ =\ p_{11}\ =\ \frac{1}{3}.
\]
}

\medskip
\noindent \paragraph{Uniform Repairs.}
Since multiple complete sequences can lead to the same database (e.g., $-f_1,-\{f_2,f_3\}$ and $-f_3,-\{f_1,f_2\}$) we would like to have a mechanism that gives non-zero probability to exactly one such sequence. To this end, for each set of complete sequences that lead to the same consistent database, we identify a representative one.
We say that a $(D,\dep)$-repairing sequence $s \in \crs{D}{\dep}$ is \emph{canonical} if there is no $s' \in  \crs{D}{\dep}$ such that $s(D) = s'(D)$ and $s' \prec s$ for some arbitrary ordering $\prec$ over the set $\rs{D}{\dep}$.
Let $\cancrs{D}{\dep}$ be the set of all sequences of $\crs{D}{\dep}$ that are canonical.
Furthermore, for a sequence $s \in \rs{D}{\dep}$, we write $\cancrss{D}{\dep}{s}$ for the set of all sequences $s'$ of $\cancrs{D}{\dep}$ that have $s$ as a prefix.
Hence, for $s \in \rs{D}{\dep}$, $\cancrss{D}{\dep}{s}$ is the set of canonical leaves of the subtree rooted at $s$, with $\cancrss{D}{\dep}{\epsilon} = \cancrs{D}{\dep}$ being the set of canonical leaves of the tree.
We can now follow the same approach discussed above for uniform sequences with the key difference that only canonical sequences are considered.
In other words, for $M_\dep^{\ur}(D) = (V,E,\ins{P})$ to induce the uniform distribution over the set of operational repairs, it suffices, for nodes $s,s' \in \rs{D}{\dep}$ with $s' \in \ops{s}{D}{\dep}$, to let
\[
\ins{P}(s,s') = \frac{|\cancrss{D}{\dep}{s'}|}{|\cancrss{D}{\dep}{s}|}.
\]
Notice that $\ins{P}(s,s')$ is not defined if the subtree $T_s$ rooted at $s$ has no canonical leaves, i.e., $\cancrss{D}{\dep}{s} = \emptyset$. In this case, none of the leaves of $T_s$ is reachable with non-zero probability, and thus, $\ins{P}(s,s')$ can get an arbitrary probability, e.g., $\frac{1}{|\ops{s}{D}{\dep}|}$.

Let us illustrate the above discussion. Assuming, e.g., that for $s,s' \in \rs{D}{\dep}$, $s \prec s'$ iff $s$ comes before $s'$ in a depth-first traversal of the tree, we have that $\cancrs{D}{\dep}$ consists of  the sequences
\[
\begin{array}{ccccc}
-f_1,-f_2 & -f_1,-f_3 & -f_1,-\{f_2,f_3\} & -f_2 & -\{f_2,f_3\}.
\end{array}
\]
Therefore, we get that
\[
|\cancrss{D}{\dep}{\epsilon}|\ =\ 5 \qquad |\cancrss{D}{\dep}{-f_1}|\ =\ 3
\]
\[
|\cancrss{D}{\dep}{-f_2}|\ =\ |\cancrss{D}{\dep}{-\{f_2,f_3\}}|\ =\ 1
\]
\[
|\cancrss{D}{\dep}{-\{f_1,f_2\}}|\ =\ |\cancrss{D}{\dep}{-f_3}|\ =\ 0.
\]
Hence, $p_1$ $=$ $\frac{3}{5}$, $p_2$ = $p_5$ $=$ $0$, $p_3$ $=$ $p_4$ $=$ $\frac{1}{5}$, $p_6$ $=$ $p_7$ $=$ $p_8$ $=$ $\frac{1}{3}$, and $p_9$ $=$ $p_{10}$ $=$ $p_{11}$ $=$ $\frac{1}{3}$.
Thus, $\abs{M_{\dep}^{\ur}(D)} = \cancrs{D}{\dep}$, and $\pi(s)$ $=$ $\frac{1}{5}$, for each $s \in \abs{M_\dep^{\ur}(D)}$, with $\pi$ being the leaf distribution of $M_{\dep}^{\ur}(D)$. 
Hence, $\opr{D}{M_\dep^{\ur}} = \{\emptyset,\{f_1\},\{f_2\},\{f_3\},\{f_1,f_3\}\}$ with $\probrep{D,M_\dep}{D'} = \frac{1}{5}$, for each $D' \in \opr{D}{M_\dep^{\ur}}$, as needed.

\medskip
\noindent \paragraph{Uniform Operations.} For $M_\dep^{\uo}$ we simply follow our intention. In particular, $M_\dep^{\uo}(D) = (V,E,\ins{P})$ is such that, for nodes $s,s' \in \rs{D}{\dep}$ with $s' \in \ops{s}{D}{\dep}$, $\ins{P}(s,s') = \frac{1}{|\ops{s}{D}{\dep}|}$. Thus, $p_1 = p_2 = p_3 = p_4 = p_5 = \frac{1}{5}$, $p_6 = p_7 = p_8 = \frac{1}{3}$, and $p_9 = p_{10} = p_{11} = \frac{1}{3}$.
Notice that, unlike the Markov chain generators $M_{\dep}^{\ur}$ and $M_{\dep}^{\us}$ discussed above, $M_{\dep}^{\uo}$ is intrinsically ``local'' in the sense that the probabilities assigned to operations at a certain step are completely determined by that step.
As we shall see, the local nature of $M_{\dep}^{\uo}$ has a significant impact on operational CQA when it comes to approximations.

%

\medskip
\noindent \paragraph{Our Main Objective.} 
%
The data complexity of $\mathsf{OCQA}$ for {\em arbitrary} Markov chain generators has been already studied in~\cite{CaLP18}, showing that it is, in general, intractable. In particular:

\begin{theorem}[\cite{CaLP18}]\label{the:ocqa-exact}
	There exist a set $\dep$ of primary keys, a repairing Markov chain generator $M_\dep$ w.r.t.~$\dep$, and a CQ $Q$ such that $\ocqa{\dep,M_\dep,Q}$ is $\sharp ${\rm P}-hard.
\end{theorem}


With the above intractability result in place, the authors of~\cite{CaLP18} asked whether $\ocqa{\dep,M_\dep,Q(\bar x)}$ is approximable, i.e., whether the target probability can be approximated via a {\em fully polynomial-time randomized approximation scheme} (FPRAS, for short).
Formally, an FPRAS for $\ocqa{\dep,M_\dep,Q(\bar x)}$ is a randomized algorithm $\mathsf{A}$ that takes as input a database $D$, a tuple $\bar c \in \adom{D}^{|\bar x|}$, $\epsilon > 0$, and $0 < \delta < 1$, runs in polynomial time in $||D||$, $||\bar c||$,\footnote{As usual, $||o||$ denotes the size of the encoding of a syntactic object $o$.} $1/\epsilon$ and $\log(1/\delta)$, and produces a random variable $\mathsf{A}(D,\bar c,\epsilon,\delta)$ such that
\[
\text{\rm Pr}\left(|\mathsf{A}(D,\bar c,\epsilon,\delta) - \probrep{M_\dep,Q}{D,\bar c}|\ \leq\ \epsilon \cdot \probrep{M_\dep,Q}{D,\bar c}\right)\ \geq\ 1-\delta.
\]
It was shown that
the problem in question 
does not admit an FPRAS, under the widely accepted complexity assumption that ${\rm RP} \neq {\rm NP}$. Recall that RP is the complexity class of problems that are efficiently solvable via a randomized algorithm with a bounded one-sided error (i.e., the answer may mistakenly be ``no'')~\cite{arora-book}.

\begin{theorem}[\cite{CaLP18}]\label{the:ocqa-apx}
	Unless  ${\rm RP} = {\rm NP}$, there exist a set $\dep$ of primary keys, a Markov chain generator $M_\dep$ w.r.t.~$\dep$, and a CQ $Q$ such that there is no FPRAS for $\ocqa{\dep,M_\dep,Q}$.
\end{theorem}


Having the natural Markov chain generators discussed above in place, the question is how the complexity of exact and approximate operational CQA is affected, i.e., how Theorems~\ref{the:ocqa-exact} and~\ref{the:ocqa-apx} are affected if we consider these more refined Markov chain generators instead of an arbitrary one. The goal of this work is to perform such a complexity analysis.
Our main findings are as follows:

\begin{enumerate}
	\item The complexity of exact operational CQA remains $\sharp ${\rm P}-hard, even in the case of primary keys.
	
	\item Operational CQA is approximable, i.e., it admits an FPRAS, if we focus on primary keys.
	
	\item In the case of arbitrary keys and FDs, although the Markov chain generators based on uniform repairs and sequences do not lead (or it remains open whether they lead) to the approximability of operational CQA, the Markov chain generator based on uniform operations renders the problem approximable.\footnote{In the case of FDs, the approximability result holds assuming that only operations that remove a single fact (not a pair of facts) are considered; this is discussed in Section~\ref{sec:uniform-operations}.} The latter should be attributed to the ``local'' nature of the Markov chain generator based on uniform operations. 
\end{enumerate}

The rest of the paper is devoted to discussing the high-level ideas underlying the above results; the formal proofs are in the appendix.

\OMIT{
But first we would like to discuss how our findings demonstrate the advantage of the uniform operational  approach over the classical approach to CQA. Recall that in the classical CQA regime, we are interested in the relative frequency of a tuple of constants, i.e., the percentage of classical repairs (maximal consistent subsets of the input database~\cite{ArBC99}) that entail such a tuple.

\medskip

\noindent \paragraph{Uniform Operational vs. Classical CQA.} The approximability results for primary keys (see item (2)) are in striking contrast with Theorem~\ref{the:ocqa-apx}, and already illustrate the usefulness of the uniform approach to operational CQA. However, they do not demonstrate the advantage of the uniform operational  approach over the classical approach to CQA, since, in the case of primary keys, both frameworks guarantee the existence of an FPRAS for the problem of interest.
The superiority of the uniform operational over the classical approach is demonstrated by the approximability results for keys and FDs when we adopt the Markov chain generator based on uniform operations (see item (3)). Indeed, for keys it is a non-trivial open question whether CQA according to the classical approach admits an FPRAS, whereas in the case of FDs it has been recently shown that the problem does not admit an FPRAS.
}


\OMIT{
	Although exact operational CQA remains $\sharp ${\rm P}-hard even if we consider the Markov chain generators $M_{\dep}^{\ur}$, $M_{\dep}^{\us}$ and $M_{\dep}^{\uo}$ w.r.t.~a set $\dep$ of primary keys, in some cases approximate operational CQA benefits from our uniform Markov chain generators.
	Our main results concerning approximate operational CQA are collected in Table~\ref{}.

	\begin{figure}
		\centering
		\begin{tabular}{r||c|c|c}
			& Primary Keys & Keys & FDs\\\hline\hline
			Uniform Repairs & $\checkmark$ & $\times^\star$ & $\times$ \\\hline
			Uniform Sequences & $\checkmark$ & ? & ? \\\hline
			Uniform Operations & $\checkmark$ & $\checkmark$ & $\checkmark^\star$
		\end{tabular}
		\caption{Each cell corresponds to a class of constraints and a uniform Markov chain generator. The symbol $\checkmark$ (resp., $\times$) means that in the corresponding setting there is (resp., there is no) an FRPAS.} Note that 
		\label{fig:results}
	\end{figure}
	We would say that the main finding of our analysis
}


\OMIT{
\medskip
\noindent\paragraph{Uniform Repairs.} We start with the Markov chain generator based on the uniform probability distribution over the set of operational repairs. Since an operational repair may be obtainable via multiple complete repairing sequences, we need a mechanism that allows us to concentrate only on one sequence that leads to a certain repair, while all the other sequences that lead to the same repair are somehow ignored. This is done via the notion of canonical sequence.
For a database $D$, and a set $\dep$ of FDs, we say that a $(D,\dep)$-repairing sequence $s \in \crs{D}{\dep}$ is \emph{canonical} if there is no $s' \in  \crs{D}{\dep}$ such that $s(D) = s'(D)$ and $s' \prec s$ for some arbitrary ordering $\prec$ over the set $\rs{D}{\dep}$.
Let $\cancrs{D}{\dep}$ be the set of all sequences of $\crs{D}{\dep}$ that are canonical. Furthermore, for a sequence $s \in \rs{D}{\dep}$, we write $\cancrss{D}{\dep}{s}$ for the set of all sequences $s'$ of $\cancrs{D}{\dep}$ that have $s$ as a prefix.
We are now ready to define the desired Markov chain generator.

\begin{definition}(\textbf{Uniform Repairs})\label{def:uniform-repairs}
	Consider a set $\dep$ of FDs. Let $M_{\dep}^{\ur}$ be the function assigning to a database $D$ a pair $(\varepsilon,\insP)$, where $\insP : \rs{D}{\dep} \times \rs{D}{\dep} \ra [0,1]$ is as follows: for $s,s' \in \rs{D}{\dep}$,
	\begin{eqnarray*}
			\insP(s,s')\
			=\ \left\{
			\begin{array}{ll}
				1 & \text{if } s=s' \text{ is complete}\\
				&\\
				\frac{|\cancrss{D}{\dep}{s'}|}{|\cancrss{D}{\dep}{s}|} & 	\text{if } s' \in \ops{s}{D}{\dep} \text{ and } \\
				& \cancrss{D}{\dep}{s'} \neq \emptyset \\
				&\\
				0 & \text{otherwise}. \hspace{28mm}\hfill\markfull
			\end{array} \right. 
	\end{eqnarray*}
\end{definition}



Observe that $\cancrss{D}{\dep}{\varepsilon} = \cancrs{D}{\dep}$. Moreover, for a $(D,\dep)$-repairing sequence $s \in \rs{D}{\dep}$, it holds that
\[
|\cancrss{D}{\dep}{s}|\ =\ \sum_{s' \in \ops{s}{D}{\dep}} |\cancrss{D}{\dep}{s'}|.
\]
These simple facts allow us to show that $M_{\dep}^{\ur}$ captures our intention:

\begin{proposition}\label{pro:uniform-repairs}
	Consider a set $\dep$ of FDs. For every database $D$:
	\begin{enumerate}
		\item $M_{\dep}^{\ur}(D)$ is a $(D,\dep)$-repairing Markov chain such that $\abs{M_{\dep}^{\ur}(D)} = \cancrs{D}{\dep}$.
		\item For every $D' \in \opr{D}{M_{\dep}^{\ur}}$, $\probrep{D,M_{\dep}^{\ur}}{D'} = \frac{1}{|\opr{D}{M_{\dep}^{\ur}}|}$.
	\end{enumerate}
\end{proposition}

\noindent\paragraph{Uniform Sequences.} We now proceed to define the Markov chain generator based on the uniform probability distribution over the set of complete repairing sequences.
It is defined similarly to the Markov chain generator above 
with the difference that we consider arbitrary, not necessarily canonical, complete sequences.
For a database $D$, a set $\dep$ of FDs, and $s \in \rs{D}{\dep}$, let $\crss{D}{\dep}{s}$ be the set of all sequences of $\crs{D}{\dep}$ that have $s$ as a prefix; it is easy to verify that $\crss{D}{\dep}{s} \neq \emptyset$.

\begin{definition}(\textbf{Uniform Sequences})\label{def:uniform-seq}
	Consider a set $\dep$ of FDs. Let $M_{\dep}^{\us}$ be the function assigning to a database $D$ a pair $(\varepsilon,\insP)$, where $\insP : \rs{D}{\dep} \times \rs{D}{\dep} \ra [0,1]$ is as follows: for $s,s' \in \rs{D}{\dep}$,
	\begin{eqnarray*}
		\insP(s,s')\
		=\ \left\{
		\begin{array}{ll}
			1 & \text{if } s=s' \text{ is complete}\\
			&\\
			\frac{|\crss{D}{\dep}{s'}|}{|\crss{D}{\dep}{s}|} & 	\text{if } s' \in \ops{s}{D}{\dep}\\
			&\\
			0 & \text{otherwise}. \hspace{32mm}\hfill\markfull
		\end{array} \right. 
	\end{eqnarray*}
\end{definition}



Observe that $\crss{D}{\dep}{\varepsilon} = \crs{D}{\dep}$. Moreover, for a $(D,\dep)$-repairing sequence $s \in \rs{D}{\dep}$, it holds that
\[
|\crss{D}{\dep}{s}|\ =\ \sum_{s' \in \ops{s}{D}{\dep}} |\crss{D}{\dep}{s'}|.
\]
We can then easily show that $M_{\dep}^{\us}$ captures our intention:

\begin{proposition}\label{pro:uniform-seq}
	Consider a set $\dep$ of FDs. For every database $D$:
	\begin{enumerate}
		\item $M_{\dep}^{\us}(D)$ is a $(D,\dep)$-repairing Markov chain such that $\abs{M_{\dep}^{\us}(D)} = \crs{D}{\dep}$.
		\item For every $s \in \crs{D}{\dep}$, assuming that $\pi$ is the hitting distribution of $M_{\dep}^{\us}(D)$, $\pi(s) = \frac{1}{|\crs{D}{\dep}|}$ .
	\end{enumerate}
\end{proposition}

\noindent\paragraph{Uniform Operations.} 
We finally define the Markov chain generator based on the uniform probability distribution over the set of available operations at a single step of the repairing process.

\begin{definition}(\textbf{Uniform Operations})\label{def:uniform-ops}
	Consider a set $\dep$ of FDs. Let $M_{\dep}^{\uo}$ be the function assigning to a database $D$ a pair $(\varepsilon,\insP)$, where $\insP : \rs{D}{\dep} \times \rs{D}{\dep} \ra [0,1]$ is as follows: for $s,s' \in \rs{D}{\dep}$,
		\begin{eqnarray*}
	\insP(s,s')\
	=\ \left\{
	\begin{array}{ll}
		1 & \text{if } s=s' \text{ is complete}\\
		&\\
		\frac{1}{|\ops{s}{D}{\dep}|} & 	\text{if } s' \in\ops{s}{D}{\dep} \\
		&\\
		0 & \text{otherwise}.  \hspace{32mm}\hfill\markfull
	\end{array} \right.
\end{eqnarray*}
\end{definition}

It is straightforward to see that the function $M_{\dep}^{\uo}$ captures our intention; in fact, this holds by definition:

\begin{proposition}\label{pro:uniform-ops}
	Consider a set $\dep$ of FDs. For every database $D$:
	\begin{enumerate}
		\item $M_{\dep}^{\uo}(D)$ is a $(D,\dep)$-repairing Markov chain such that $\abs{M_{\dep}^{\uo}(D)} = \crs{D}{\dep}$.
		\item For every $s,s' \in \rs{D}{\dep}$, assuming that $M_{\dep}^{\uo}(D) = (\varepsilon, \insP)$, $s' \in \ops{s}{D}{\dep}$ implies $\insP(s,s') = \frac{1}{|\ops{s}{D}{\dep}|}$.
	\end{enumerate}
\end{proposition}

Notice that, unlike the Markov chain generators $M_{\dep}^{\ur}$ and $M_{\dep}^{\us}$ w.r.t.~a set $\dep$ of FDs defined above, the Markov chain generator $M_{\dep}^{\uo}$ is intrinsically ``local'' in the sense that the probabilities assigned to operations at a certain step are completely determined by that step.
As we shall see, the local nature of $M_{\dep}^{\uo}$ has a significant impact on operational CQA when it comes to approximations.

\medskip
\noindent\paragraph{Singleton Removals.} Note that the Markov chain generators $M_{\dep}^{\ur}$, $M_{\dep}^{\us}$ and $M_{\dep}^{\uo}$ w.r.t.~a set $\dep$ of FDs consider all the available operations, i.e., justified operations that either remove a single fact that is an element of a violation, or remove two facts that form a violation.
On the other hand, it is also natural to operationally repair an inconsistent database via a sequence of single fact deletions, which somehow brings operational CQA closer to classical CQA.
%
To this end, given a set $\dep$ of FDs, one can consider the Markov chain generators $M_{\dep}^{\ur,1}$, $M_{\dep}^{\us,1}$ and $M_{\dep}^{\uo,1}$ w.r.t.~$\dep$, which are defined in the same way as the Markov chain generators $M_{\dep}^{\ur}$, $M_{\dep}^{\us}$ and $M_{\dep}^{\uo}$, respectively, with the key difference that only repairing sequences of operations of the form $-f$, where $f$ is fact, are considered; the formal definitions can be found in the appendix.

}

%% file: uniform-repairs.tex
\section{Uniform Repairs}\label{sec:uniform-repairs}

We start our complexity analysis by considering the Markov chain generator based on uniform repairs, and show the following result:

\begin{theorem}\label{the:uniform-repairs}
	\begin{enumerate}
		\item There exist a set $\dep$ of primary keys, and a CQ $Q$ such that $\ocqa{\dep,M_{\dep}^{\ur},Q}$ is $\sharp ${\rm P}-hard.
		
		\item For a set $\dep$ of primary keys, and a CQ $Q$, $\ocqa{\dep,M_{\dep}^{\ur},Q}$ admits an FPRAS.
		
		\item Unless ${\rm RP} = {\rm NP}$, there exist a set $\dep$ of FDs, and a CQ $Q$ such that there is no FPRAS for $\ocqa{\dep,M_{\dep}^{\ur},Q}$.
	\end{enumerate}
\end{theorem}

Notice that the above result does not cover the case of arbitrary keys, which remains an open problem. We can extract, however, from the proof of item (3) that for keys, unless ${\rm RP} = {\rm NP}$, the problem of {\em counting} the number of operational repairs does not admit an FPRAS. We see this as an indication that item (3) holds even in the case of keys.
We now discuss how Theorem~\ref{the:uniform-repairs} is shown.

We start with the simple observation that, for a database $D$, a set $\dep$ of FDs, a CQ $Q(\bar x)$, and a tuple $\bar c \in \adom{D}^{|\bar x|}$,
\begin{eqnarray*}
\probrep{M_{\dep}^{\ur},Q}{D,\bar c} &=& \frac{|\{D' \in \copr{D}{\dep} \mid \bar c \in Q(D')\}|}{|\copr{D}{\dep}|}.
\end{eqnarray*}
\OMIT{
\begin{eqnarray*}
\probrep{M_{\dep}^{\ur},Q}{D,\bar c} &=& \frac{|\{D' \in \opr{D}{M_{\dep}^{\ur}} \mid \bar c \in Q(D')\}|}{|\opr{D}{M_{\dep}^{\ur}}|}\\
&=& \frac{|\{s(D) \mid s \in \crs{D}{\dep} \text{ and } \bar c \in Q(s(D))\}|}{|\{s(D) \mid s \in \crs{D}{\dep}\}|}
\end{eqnarray*}
}
This ratio is the percentage of candidate operational repairs of $D$ w.r.t.~$\dep$ that entail $Q(\bar c)$, which we call the {\em repair relative frequency} of $Q(\bar c)$ w.r.t.~$D$ and $\dep$, and denote $\orfreq{\dep,Q}{D,\bar c}$.
Therefore, we can conveniently restate the problem $\ocqa{\dep,M_{\dep}^{\ur},Q}$ as the problem of computing the repair relative frequency of $Q(\bar c)$ w.r.t.~$D$ and $\dep$, which does not depend on the Markov chain generator $M_{\dep}^{\ur}$:


\medskip

\begin{center}
	\fbox{\begin{tabular}{ll}
			{\small PROBLEM} : & $\rrelfreq{\dep,Q(\bar x)}$
			\\
			{\small INPUT} : & A database $D$,  and a tuple $\bar c \in \adom{D}^{|\bar x|}$.
			\\
			{\small OUTPUT} : &  $\orfreq{\dep,Q}{D,\bar c}$.
	\end{tabular}}
\end{center}

\medskip
\noindent We proceed to discuss how we establish Theorem~\ref{the:uniform-repairs} by directly considering the problem $\rrelfreq{\dep,Q}$ instead of $\ocqa{\dep,M_{\dep}^{\ur},Q}$;
further details can be found in Appendix~\ref{appsec:uniform-repairs}.

\medskip 
\noindent \paragraph{Item (1).} We show that  $\rrelfreq{\dep,Q}$ is $\sharp ${\rm P}-hard for a set $\dep$ consisting of a single key of the form $R : A \ra B$, where $R$ is a binary relation name with $(A,B)$ being its tuple of attributes, and a Boolean CQ $Q$. This is done via a polynomial-time Turing reduction from a graph-theoretic problem called $\sharp H\text{-}\mathsf{Coloring}$, where $H$ is an undirected graph, to $\rrelfreq{\dep,Q}$. The problem $\sharp H\text{-}\mathsf{Coloring}$ takes as input an undirected graph $G$, and asks for the number of homomorphisms from $G$ to $H$. The key of the proof is to carefully choose $H$ so that (i) $ \sharp H\text{-}\mathsf{Coloring}$ is $\sharp ${\rm P}-hard, and (ii) it allows us to devise the desired polynomial-time Turing reduction, i.e., for an undirected graph $G$, we can construct in polynomial time in $||G||$ a database $D_G$ such that the number of homomorphisms from $G$ to $H$ can be computed in polynomial time in $||G||$, assuming that we have access to an oracle for the problem $\rrelfreq{\dep,Q}$, which we can call to compute the number $\orfreq{\dep,Q}{D_G,()}$; we use $()$ to denote the empty tuple.
For choosing $H$, we exploit an interesting dichotomy from~\cite{Dyer00}, which characterizes when $\sharp H\text{-}\mathsf{Coloring}$ is solvable in polynomial time or is $\sharp ${\rm P}-hard, depending on the structure of $H$.

\medskip
\noindent \paragraph{Item (2).} For showing that, for a set $\dep$ of primary keys and a CQ $Q$, $\rrelfreq{\dep,Q}$ admits an FPRAS, we rely on Monte Carlo sampling. 
We first show the existence of an efficient sampler:

\begin{lemma}\label{lem:ur-sampler}
	Given a database $D$, and a set $\dep$ of primary keys, we can sample elements of $\copr{D}{\dep}$ uniformly at random in polynomial time in $||D||$.
\end{lemma}

The above lemma tells us that there exists a randomized algorithm $\mathsf{SampleRep}$ that takes as input $D$ and $\dep$, runs in polynomial time in $||D||$, and produces a random variable $\mathsf{SampleRep}(D,\dep)$ such that $\pr(\mathsf{SampleRep}(D,\dep) = D') = \frac{1}{|\copr{D}{\dep}|}$ for every database $D' \in \copr{D}{\dep}$.
Notice, however, that the efficient sampler provided by Lemma~\ref{lem:ur-sampler} does not immediately imply the existence of an FPRAS for $\rrelfreq{\dep,Q}$ since the number of samples should be proportional to $\frac{1}{\orfreq{\dep,Q}{D,\bar c}}$~\cite{KarpLuby2000}. Hence, to obtain an FPRAS using Monte Carlo sampling, we need show that the repair relative frequency is never ``too small''. 

\begin{lemma}\label{lem:ur-lower-bound}
	Consider a set $\dep$ of primary keys, and a CQ $Q(\bar x)$. For every database $D$, and tuple $\bar c \in \adom{D}^{|\bar x|}$,
	\[
	\orfreq{\dep,Q}{D,\bar c}\ \geq\ \frac{1}{(2 \cdot ||D||)^{||Q||}}
	\] 
	whenever $\orfreq{\dep,Q}{D,\bar c} > 0$.
\end{lemma}

Given a set $\dep$ of primary keys and a CQ $Q$, by exploiting Lemmas~\ref{lem:ur-sampler} and~\ref{lem:ur-lower-bound}, we can easily devise an FPRAS for $\rrelfreq{\dep,Q}$.

\OMIT{
In particular, given a database $D$, we first show that we can sample elements of the set $\{s(D) \mid s \in \crs{D}{\dep}\}$ uniformly at random in polynomial time in the size of $D$. Notice, however, that the existence of such an efficient sampler does not immediately imply the existence of an FPRAS for our problem since the number of samples should be proportional to $\frac{1}{\orfreq{\dep,Q}{D,\bar c}}$~\cite{KarpLuby2000}. Hence, to obtain an FPRAS using Monte Carlo sampling, we need to show that the ratio $\orfreq{\dep,Q}{D,\bar c}$ is not ``too small''. Formally, we show that for every database $D$, there exists a polynomial $\mathsf{pol}(\cdot)$ such that $\orfreq{\dep,Q}{D,\bar c} \geq \frac{1}{\mathsf{pol}(||D||)}$ whenever $\orfreq{\dep,Q}{D,\bar c} > 0$.
}

\medskip
\noindent \paragraph{Item (3).} For showing that there exist a set $\dep$ of FDs and a CQ $Q$ such that, unless ${\rm RP} = {\rm NP}$, there is no FPRAS for $\rrelfreq{\dep,Q}$, we provide a rather involved proof that proceeds in two main steps. We first give an auxiliary lemma that is needed by both steps.

An undirected graph $G$ is called {\em non-trivially connected} if it contains at least two nodes, and is connected. We write $\IS(G)$ for the set that collects all the independent sets of $G$.
Recall that the \emph{conflict graph} of a database $D$ w.r.t.~a set $\dep$ of FDs, denoted $\cg{D}{\dep}$, is an undirected graph whose node set is $D$, and it has an edge between $f$ and $g$ if $\{f,g\} \not\models \dep$.
A database $D$ is {\em non-trivially $\dep$-connected} if $\cg{D}{\dep}$ is non-trivially connected. We then show the following:


\begin{lemma}\label{lem:corepairs-independent-sets}
	Consider a non-trivially $\dep$-connected database $D$, where $\dep$ is a set of FDs. It holds that $|\copr{D}{\dep}| = |\IS(\cg{D}{\dep})|$.
\end{lemma}

Having the above auxiliary lemma in place, we can now describe the two steps of the proof underlying Theorem~\ref{the:uniform-repairs}(3). The first step establishes the following inapproximability result about keys.


\begin{proposition}\label{pro:ur-keys-no-fpras}
	Unless ${\rm RP} = {\rm NP}$, there exists a set $\dep$ of keys over $\{R\}$ such that, given a non-trivially $\dep$-connected database $D$, the problem of computing $|\copr{D}{\dep}|$ does not admit an FPRAS.
\end{proposition}

The above result exploits the fact that, unless ${\rm RP} = {\rm NP}$, the problem of counting the number of independent sets of a non-trivially connected undirected graph of bounded degree does not admit an FPRAS.\footnote{This result is known for arbitrary, not necessarily non-trivially connected graphs~\cite{Sly2010}. Thus, for our purposes, we had to strengthen it to non-trivially connected graphs.} In particular, we show that there exists a set $\dep_K$ of keys over the schema $\ins{S} = \{R/ \Delta+1\}$ such that the following holds: given a non-trivially connected undirected graph $G$ of bounded degree $\Delta$, we can construct in polynomial time in $||G||$ a database $D_G$ over $\ins{S}$ such that $\cg{D_G}{\dep_K}$ is isomorphic to $G$. Thus, by Lemma~\ref{lem:corepairs-independent-sets}, $|\copr{D_G}{\dep_K}| = |\IS(G)|$.
The construction of $D_G$ exploits Vizing's Theorem, which states that a graph of degree $\Delta$ always has a $(\Delta+1)$-edge-coloring, as well as the fact that such an edge-coloring can be constructed in polynomial time as long as $\Delta$ is bounded~\cite{Misra1992}.
Hence, given a database $D$, assuming that the problem of computing the number $|\copr{D}{\dep_K}|$ admits an FPRAS, we can conclude that the problem of counting the number of independent sets of a non-trivially connected undirected graph of bounded degree admits an FPRAS, which, unless ${\rm RP} = {\rm NP}$, leads to a contradiction.
Therefore, Proposition~\ref{pro:ur-keys-no-fpras} follows with $\dep = \dep_K$.
Notice that Proposition~\ref{pro:ur-keys-no-fpras} tells us that for keys, unless ${\rm RP} = {\rm NP}$, the problem of counting the number of operational repairs does not admit an FPRAS. As said above, we see this as an indication that item (3) of Theorem~\ref{the:uniform-repairs} holds even for keys.

We then proceed to show that, unless ${\rm RP} = {\rm NP}$, the existence of an FPRAS for $\rrelfreq{\dep,Q}$, where $\dep$ is a set of FDs and $Q$ a CQ, would contradict Proposition~\ref{pro:ur-keys-no-fpras} . Let $\dep_K$ be the set of keys provided by Proposition~\ref{pro:ur-keys-no-fpras}. We show the following auxiliary result:

\begin{lemma}\label{lem:from-fds-to-keys}
	Assume that $\rrelfreq{\dep,Q}$ admits an FPRAS, for every set $\dep$ of FDs and CQ $Q$. Given a non-trivially $\dep_K$-connected database $D$, the problem of computing $|\copr{D}{\dep_K}|$ admits an FPRAS.
\end{lemma}

To establish the above result, we show that there exists a set $\dep_F$ of FDs such that, for every non-trivially $\dep_K$-connected database $D$, we can construct in polynomial time in $||D||$ a database $D_F$ such that  
$\cg{D_F}{\dep_F}$ consists of a graph $G$ that is isomorphic to $\cg{D}{\dep_K}$, and an additional node that is connected via an edge with every node of $G$. Therefore, by Lemma~\ref{lem:corepairs-independent-sets}, we get that 
\[
|\copr{D_F}{\dep_F}|\ =\ |\copr{D}{\dep_K}| + 1. 
\]
Let us clarify that this is the place where we need the power of FDs; it is unclear how we can devise a set of keys that has the same properties as $\dep_F$.
We then construct an atomic Boolean CQ $Q_F$ with
\[
\orfreq{\dep_F,Q_F}{D_F,()}\ =\ \frac{1}{|\copr{D_F}{\dep_F}|}\ =\ \frac{1}{|\copr{D}{\dep_K}| + 1}; 
\]
we use $()$ to denote the empty tuple. Now, by exploiting the above equality, the fact that $D_F$ can be constructed in polynomial time, and the FPRAS for $\rrelfreq{\dep_F,Q_F}$ (which exists by hypothesis), we can devise an FPRAS for the problem of computing $|\copr{D}{\dep_K}|$ given a non-trivially $\dep_K$-connected database $D$, as claimed.

It is now straightforward to see that from Proposition~\ref{pro:ur-keys-no-fpras} and Lemma~\ref{lem:from-fds-to-keys}, we get that, unless ${\rm RP} = {\rm NP}$, there exist a set $\dep$ of FDs and a CQ $Q$ such that there is no FPRAS for $\rrelfreq{\dep,Q}$.


\OMIT{
For showing that there exist a set $\dep$ of FDs and a CQ $Q$ such that there is no FPRAS for $\rrelfreq{\dep,Q}$ (unless ${\rm RP} = {\rm NP}$), we give a proof by contradiction, which proceeds in two main steps:

We first show that there exists a set $\dep^\star$ of keys such that, given a database $D$, the problem of computing the number $|\{s(D) \mid s \in \crs{D}{\dep^\star}\}|$ does not admit an FPRAS. This inapproximability result exploits the fact that the problem of counting the number of independent sets of a connected graph of degree six does not admit an FPRAS (unless ${\rm RP} = {\rm NP}$).\footnote{This result is known for arbitrary, not necessarily connected graphs~\cite{Sly2010}. Therefore, for the purposes of our proof, we had to strengthen it to connected graphs.} In particular, we show that that there exists a set $\dep_K$ of keys such that the following holds: given a connected graph $G$ of degree six, we can construct in polynomial time in $||G||$ a database $D_G$ with $|\{s(D_G) \mid s \in \crs{D_G}{\dep_K}\}|$ being the number of independent sets of $G$. Hence, assuming that for a database $D$, computing $|\{s(D) \mid s \in \crs{D}{\dep_K}\}|$ admits an FPRAS, we can conclude that the problem of counting the number of independent sets of a connected graph of degree six admits an FPRAS, which leads to a contradiction (unless ${\rm RP} = {\rm NP}$). Thus, the main inapproximability result of this step follows with $\dep^\star = \dep_K$.

The second step shows that, for a set $\dep$ of FDs and a CQ $Q$, the existence of an FPRAS for $\rrelfreq{\dep,Q}$ implies the existence of an FPRAS for the problem of computing $|\{s(D) \mid s \in \crs{D}{\dep^\star}\}|$ given a database $D$, which contradicts the inapproximability result established in the first step.
In particular, we show that there exists a set $\dep_F$ of FDs such that, for every database $D$, we can construct in polynomial time in $||D||$ a database $D_F$ with 
$|\{s(D_F) \mid s \in \crs{D_F}{\dep_F}\}| = |\{s(D) \mid s \in \crs{D}{\dep^\star}\}| + 1$. Note that this is the place in the proof where we need the power of FDs; it is unclear how this can be achieved using only keys.
We then construct a simple Boolean CQ $Q_F$ such that
\[
\orfreq{\dep_F,Q_F}{D_F,()}\ =\ \frac{1}{|\{s(D) \mid s \in \crs{D}{\dep^\star}\}| + 1}; 
\]
we use $()$ to denote the empty tuple. Now, by exploiting the above equality, the fact that $D_F$ can be constructed in polynomial time, and the FPRAS for $\rrelfreq{\dep,Q}$ for any set $\dep$ of FDs and CQ $Q$ (which exists by hypothesis), we can devise an FPRAS for the problem of computing $|\{s(D) \mid s \in \crs{D}{\dep^\star}\}|$ for a database $D$, which leads to a contradiction since such an FPRAS does not exist (by step 1).
}

\OMIT{
\begin{itemize}
\item We note that Theorem~\ref{the:uniform-repairs} remains exactly the same even if we use the Markov chain generator $M_{\dep}^{\ur,1}$, which considers only singe fact deletions. Items (1) and (2) are easily obtainable, but we need to argue a bit further for item (3).

\item We conclude the section by clarifying that whether there exists an FPRAS for $\mathsf{OCQA}$ in the case of keys remains open. We have, nevertheless, an inapproximability result for the problem of counting the number of operational repairs, which can be seen as an indication that the answer to the open question is rather negative. The proof of this results can be extracted from the proof of item (3) of Theorem~\ref{the:uniform-repairs}.
\end{itemize}
}

%% file: uniform-sequences.tex
\section{Uniform Sequences}\label{sec:uniform-sequences}

We now concentrate on the Markov chain generator based on uniform sequences, and establish the following complexity result.

\begin{theorem}\label{the:uniform-sequences}
	\begin{enumerate}
		\item There exist a set $\dep$ of primary keys, and a CQ $Q$ such that $\ocqa{\dep,M_{\dep}^{\us},Q}$ is $\sharp ${\rm P}-hard.
		
		\item For a set $\dep$ of primary keys, and a CQ $Q$, $\ocqa{\dep,M_{\dep}^{\us},Q}$ admits an FPRAS.
	\end{enumerate}
\end{theorem}

Notice that the above result does not cover the cases of arbitrary keys and FDs. Unfortunately, despite our efforts, we have not managed to prove or disprove the existence of an FPRAS for the problem in question. We conjecture that there is no FPRAS even for keys, i.e., unless ${\rm RP} = {\rm NP}$, there exist a set $\dep$ of keys, and a CQ $Q$ such that there is no FPRAS for $\ocqa{\dep,M_{\dep}^{\us},Q}$.
We proceed to discuss how Theorem~\ref{the:uniform-sequences} is shown.

As for Theorem~\ref{the:uniform-repairs}, we can conveniently restate the problem in question as a problem of computing a ``relative frequency'' ratio that does not depend on the Markov chain generator.
In particular, for a database $D$, a set $\dep$ of FDs, a CQ $Q(\bar x)$, and a tuple $\bar c \in \adom{D}^{|\bar x|}$,
\begin{eqnarray*}
	\probrep{M_{\dep}^{\us},Q}{D,\bar c} &=& \frac{|\{s \in \crs{D}{\dep} \mid \bar c \in Q(s(D))\}|}{|\crs{D}{\dep}|}.
\end{eqnarray*}
This ratio is the percentage of complete $(D,\dep)$-repairing sequences that lead to an operational repair that entails $Q(\bar c)$, which we call the {\em sequence relative frequency} of $Q(\bar c)$ w.r.t.~$D$ and $\dep$, and denote $\srfreq{\dep,Q}{D,\bar c}$. 
Thus, we can restate $\ocqa{\dep,M_{\dep}^{\us},Q}$ as the problem of computing the sequence relative frequency of $Q(\bar c)$ w.r.t.~$D$ and $\dep$, which is independent from the Markov chain generator $M_{\dep}^{\us}$:


\medskip

\begin{center}
	\fbox{\begin{tabular}{ll}
			{\small PROBLEM} : & $\srelfreq{\dep,Q(\bar x)}$
			\\
			{\small INPUT} : & A database $D$,  and a tuple $\bar c \in \adom{D}^{|\bar x|}$.
			\\
			{\small OUTPUT} : &  $\srfreq{\dep,Q}{D,\bar c}$.
	\end{tabular}}
\end{center}

\medskip
\noindent We now discuss how we establish Theorem~\ref{the:uniform-sequences} by directly considering the problem $\srelfreq{\dep,Q}$ instead of $\ocqa{\dep,M_{\dep}^{\us},Q}$; further details can be found in Appendix~\ref{appsec:uniform-sequences}.


\medskip
\noindent \paragraph{Item (1).} Let $\dep$ and $Q$ be the singleton set of primary keys and the Boolean CQ, respectively, for which $\rrelfreq{\dep,Q}$ is $\sharp ${\rm P}-hard; $\dep$ and $Q$ are obtained from the proof of Theorem~\ref{the:uniform-repairs}(1).
We show that also $\srelfreq{\dep,Q}$ is $\sharp ${\rm P}-hard via a polynomial-time Turing reduction from $\sharp H\text{-}\mathsf{Coloring}$. Actually, we can exploit the same construction as in the proof of item (1) of Theorem~\ref{the:uniform-repairs}.
\OMIT{
Assuming that, for an undirected graph $G$, $D_G$ is the database that the construction in the proof of item (1) of Theorem~\ref{the:uniform-repairs} builds, we show that $\orfreq{\dep,Q}{D_G,()} = \srfreq{\dep,Q}{D_G,()}$, which implies that the polynomial-time Turing reduction from $\sharp H\text{-}\mathsf{Coloring}$ to $\rrelfreq{\dep,Q}$ is also a polynomial-time Turing reduction from $\sharp H\text{-}\mathsf{Coloring}$ to $\srelfreq{\dep,Q}$.
Therefore, we get that $\srelfreq{\dep,Q}$ is $\sharp ${\rm P}-hard, as needed.
}

\medskip
\noindent \paragraph{Item (2).} 
For showing that, for a set $\dep$ of primary keys and a CQ $Q(\bar x)$, $\srelfreq{\dep,Q}$ admits an FPRAS, we rely again on Monte Carlo sampling.
%
%
We first show that an efficient sampler exists. This relies on a non-trivial technical lemma, which states that, for a database $D$, $|\crs{D}{\dep}|$ can be computed in polynomial time in $||D||$.

\begin{lemma}\label{lem:us-sampler}
	For a database $D$, and a set $\dep$ of primary keys, we can sample elements of $\crs{D}{\dep}$ uniformly at random in polynomial time in $||D||$.
\end{lemma}

\OMIT{
The above lemma essentially states that there exists a randomized algorithm $\mathsf{SampleSeq}$ that takes as input $D$ and $\dep$, runs in polynomial time in $||D||$, and produces a random variable $\mathsf{SampleSeq}(D,\dep)$ such that $\Pr(\mathsf{SampleSeq}(D,\dep) = s) = \frac{1}{|\crs{D}{\dep}|}$ for every sequence $s \in \crs{D}{\dep}$.
}

%
To establish that the problem in question admits an FPRAS based on Monte Carlo sampling, it remains to show the following:

\begin{lemma}\label{lem:us-lower-bound}
	Consider a set $\dep$ of primary keys, and a CQ $Q(\bar x)$. For every database $D$, and tuple $\bar c \in \adom{D}^{|\bar x|}$,
	\[
	\srfreq{\dep,Q}{D,\bar c}\ \geq\ \frac{1}{(2 \cdot ||D||)^{||Q||}}
	\] 
	whenever $\srfreq{\dep,Q}{D,\bar c} > 0$.
\end{lemma}

Given a set $\dep$ of primary keys and a CQ $Q$, by exploiting Lemmas~\ref{lem:us-sampler} and~\ref{lem:us-lower-bound}, we can easily devise an FPRAS for $\srelfreq{\dep,Q}$.

\OMIT{
\begin{itemize}
	\item We note that Theorem~\ref{the:uniform-sequences} remains exactly the same even if we use the Markov chain generator $M_{\dep}^{\us,1}$, which considers only single fact deletions.
	
	\item We conclude the section by clarifying that whether there exists an FPRAS for $\mathsf{OCQA}$ in the case of keys or FDs remains open. We conjecture that the problem does not admit an FPRAS, even in the case of keys (no matter whether we consider arbitrary or singleton deletions).
\end{itemize}
}

%% file: uniform-operations.tex
\section{Uniform Operations}\label{sec:uniform-operations}

We finally consider the Markov chain generator based on uniform operations, and establish the following complexity result.

\begin{theorem}\label{the:uniform-operations}
	\begin{enumerate}
		\item There exist a set $\dep$ of primary keys, and a CQ $Q$ such that $\ocqa{\dep,M_{\dep}^{\uo},Q}$ is $\sharp ${\rm P}-hard.
		
		\item For a set $\dep$ of keys, and a CQ $Q$, $\ocqa{\dep,M_{\dep}^{\uo},Q}$ admits an FPRAS.
	\end{enumerate}
\end{theorem}

Notice that the above result does not cover the case of FDs, which remains an open problem. However, as we explain below, for FDs we can establish an approximability result under the assumption that only operations that remove a single fact (not a pair of facts) are considered. But let us first discuss the proof of Theorem~\ref{the:uniform-operations}. 

Unlike Theorems~\ref{the:uniform-repairs} and~\ref{the:uniform-sequences} presented above, there is no obvious way to conveniently restate the problem of interest as a problem of computing a ``relative frequency'' ratio.
Thus, the proof of Theorem~\ref{the:uniform-operations}, which we discuss next, has to deal with $\ocqa{\dep,M_{\dep}^{\uo},Q}$ for a set $\dep$ of FDs and a CQ $Q$; details are in Appendix~\ref{appsec:uniform-operations}.

\medskip
\noindent \paragraph{Item (1).} As we did for item (1) of Theorem~\ref{the:uniform-sequences}, we reuse the construction underlying the proof of item (1) of Theorem~\ref{the:uniform-repairs}.
\OMIT{
In particular, assuming that $\dep$ and $Q$ are the singleton set of primary keys and the Boolean CQ, respectively, for which $\rrelfreq{\dep,Q}$ is $\sharp ${\rm P}-hard (which are extracted from the proof of Theorem~\ref{the:uniform-repairs}(1)), we show that $\ocqa{\dep,M_{\dep}^{\uo},Q}$ is $\sharp ${\rm P}-hard via a polynomial-time Turing reduction from $\sharp H\text{-}\mathsf{Coloring}$ by reusing the construction in the proof of item (1) of Theorem~\ref{the:uniform-repairs}.
Assuming that, for an undirected graph $G$, $D_G$ is the database that the construction in the proof of item (1) of Theorem~\ref{the:uniform-repairs} builds, we show that $\orfreq{\dep,Q}{D_G,()} = \probrep{M_{\dep}^{\uo},Q}{D_G,()}$, which implies that the polynomial-time Turing reduction from $\sharp H\text{-}\mathsf{Coloring}$ to $\rrelfreq{\dep,Q}$ is also a polynomial-time Turing reduction from $\sharp H\text{-}\mathsf{Coloring}$ to $\ocqa{\dep,M_{\dep}^{\uo},Q}$.
Therefore, we conclude that $\ocqa{\dep,M_{\dep}^{\uo},Q}$ is $\sharp ${\rm P}-hard, as needed.
}

\medskip
\noindent \paragraph{Item (2).} We show that $\ocqa{\dep,M_{\dep}^{\uo},Q}$, where $\dep$ is a set of keys and $Q$ a CQ, admits an FPRAS by relying once again on Monte Carlo sampling. The existence of an efficient sampler follows easily from the definition of the Markov chain generator $M_{\dep}^{\uo}$. In particular:

\begin{lemma}\label{lem:uo-sampler}
	Given a database $D$, and a set $\dep$ of keys, we can sample elements of $\abs{M_{\dep}^{\uo}(D)}$ according to the leaf distribution of $M_{\dep}^{\uo}(D)$ in polynomial time in $||D||$.
\end{lemma}

\OMIT{
In other words, Lemma~\ref{lem:uo-sampler} tells us that there exists a randomized algorithm $\mathsf{SampleOp}$ that takes as input $D$ and $\dep$, runs in polynomial time in $||D||$, and produces a random variable $\mathsf{SampleOp}(D,\dep)$ such that $\Pr(\mathsf{SampleOp}(D,\dep) = s) = \pi(s)$ for every sequence $s \in \abs{M_{\dep}^{\uo}(D)}$, where $\pi$ is the leaf distribution of $M_{\dep}^{\uo}(D)$.
}

The interesting task towards an FPRAS for the problem in question is to show that the target probability is never ``too small''.

\begin{proposition}\label{pro:uo-lower-bound}
	Consider a set $\dep$ of keys, and a CQ $Q(\bar x)$. There is a polynomial $\mathsf{pol}$ such that, for every database $D$, and $\bar c \in \adom{D}^{|\bar x|}$,
	\[
	\probrep{M_{\dep}^{\uo},Q}{D,\bar c}\ \geq\ \frac{1}{\mathsf{pol}(||D||)}
	\] 
	whenever $\probrep{M_{\dep}^{\uo},Q}{D,\bar c} > 0$.
\end{proposition}

We proceed to discuss the main ideas underlying the proof of the above result. For the sake of clarity, we focus on atomic queries, i.e., CQs with only one atom. The generalization to arbitrary CQs can be found in the appendix.
In the sequel, let $\dep$ be a set of keys, $Q(\bar x)$ an atomic query, $D$ a database, and $\bar c$ a tuple of $\adom{D}^{|\bar x|}$.

Clearly, if there is no homomorphism $h$ from $Q$ to $D$ with $h(\bar x) = \bar c$, then $\probrep{M_{\dep}^{\uo},Q}{D,\bar c} = 0$. Assume now that such a homomorphism $h$ exists, and let $f$ be the fact of $D$ obtained after applying $h$ to the single atom of $Q$. It is not difficult to see that
\[
\probrep{M_{\dep}^{\uo},Q}{D,\bar c}\ \geq\ \underbrace{\sum\limits_{D' \in \opr{D}{M_{\dep}^{\uo}} \text{ and } f \in D'} \probrep{D,M_{\dep}^{\uo}}{D'}}_{\Lambda}.
\]
Thus, it suffices to show that there exists a polynomial $\mathsf{pol}$ such that $\Lambda \geq \frac{1}{\mathsf{pol}(||D||)}$.
Let $S_f$ and $S_{\neg f}$ be the sets of sequences of $\abs{M_{\dep}^{\uo}(D)}$ that keep $f$ and remove $f$, respectively, i.e.,
\begin{eqnarray*}
	S_f &=& \{s \in \abs{M_{\dep}^{\uo}(D)} \mid f \in s(D)\}\\
	S_{\neg f} &=& \{s \in \abs{M_{\dep}^{\uo}(D)} \mid f \not\in s(D)\}.
\end{eqnarray*}
With $\pi$ being the leaf distribution of $M_{\dep}^{\uo}(D)$, $\Lambda = \frac{\Lambda_f}{\Lambda_f+\Lambda_{\neg f}}$, where
\[
\Lambda_f\ =\ \sum_{s \in S_f}\pi(s) \qquad \text{and} \qquad \Lambda_{\neg f}\ =\ \sum_{s \in S_{\neg f}}\pi(s).
\]
\OMIT{
and
\[
\Lambda_2\ =\ \sum_{{s\in\abs{M_{\dep}^{\uo}(D)} \text{ and } f\not\in s(D)}}\pi(s).
\]
\[
\frac{\sum_{{s\in\abs{M_{\dep}^{\uo}(D)} \text{ and } f\in s(D)}}\pi(s)}{\underbrace{\sum_{{s\in\abs{M_{\dep}^{\uo}(D)} \text{ and } f\in s(D)}}\pi(s)}_{\Lambda_1}+\underbrace{\sum_{{s\in\abs{M_{\dep}^{\uo}(D)} \text{ and } f\not\in s(D)}}\pi(s)}_{\Lambda_2}}.
\]
}
Therefore, to establish the desired lower bound $\frac{1}{\mathsf{pol}(||D||)}$ for $\Lambda$, it suffices to show that there exists a polynomial $\mathsf{pol}'$ such that $\Lambda_{\neg f} \leq \mathsf{pol}'(||D||) \cdot \Lambda_f$. Indeed, in this case we can conclude that
\[
\Lambda\ =\ \frac{\Lambda_f}{\Lambda_f+\Lambda_{\neg f}}\ \geq\ \frac{\Lambda_f}{\Lambda_f + \mathsf{pol}'(||D||) \cdot \Lambda_f}\ =\ \frac{1}{1 + \mathsf{pol}'(||D||)},
\]
and the claim follows with $\mathsf{pol}(||D||) = 1 + \mathsf{pol}'(||D||)$. The rest of the proof is devoted to showing that a polynomial $\mathsf{pol}'$ such that $\Lambda_{\neg f} \leq \mathsf{pol}'(||D||) \cdot \Lambda_f$ exists.

To get this inequality, we establish a rather involved technical lemma that relates the sequences of $S_{\neg f}$ with the sequences of $S_f$; as usual, we write $\pi$ for the leaf distribution of $M_{\dep}^{\uo}(D)$:

\begin{lemma}\label{lem:relate-sequences}
	There exists a function $\mathsf{F} : S_{\neg f} \ra S_{f}$ such that:
	\begin{enumerate}
		\item There exists a polynomial $\mathsf{pol}''$ such that, for every $s \in S_{\neg f}$, $\pi(s) \leq \mathsf{pol}''(||D||) \cdot \pi(\mathsf{F}(s))$.
		\item For every $s' \in S_{f}$, $|\{s \in S_{\neg f} \mid \mathsf{F}(s)=s'\}| \leq 2 \cdot ||D|| - 1$.
	\end{enumerate}
\end{lemma}

For showing item (1) of Lemma~\ref{lem:relate-sequences}, we transform each sequence $s \in S_{\neg f}$ into a sequence $s' \in S_{f}$, and let $\mathsf{F}(s) = s'$. This is done by first deleting or replacing the operation $\op$ in $s$ that removes $f$. In particular, if $\op = -f$, then we simply delete it; otherwise, if $\op = -\{f,g\}$, then we replace it with the operation $-g$. Notice, however, that there is no guarantee that the sequence $\hat{s}$, obtained after removing $\op$ from $s$, is a complete sequence of $\crs{D}{\dep}$. This is because $s(D)$ might contain facts that are in a conflict with $f$, and thus, by keeping $f$, there is no guarantee that $\hat{s}(D) \models \dep$. We then convert $\hat{s}$ into a complete sequence $s'$ by simply adding at the end of $\hat{s}$ additional operations (in some arbitrary order) that resolve all the conflicts.
Now, to show that $\pi(s) \leq \mathsf{pol}''(||D||) \cdot \pi(s')$, for some polynomial $\mathsf{pol}''$, we rely  on the following two crucial facts:
(1) Although the probabilities of the operations in $s'$ coming after the operation in $s$ that removes $f$ might decrease, we can show that they do not decrease ``too much''.
%
(2) The number of operations that we need to add at the end of $\hat{s}$ in order to get $s'$ depends only on $\dep$ (not on $||D||$). More precisely, by exploiting the fact that $\dep$ consists of keys, we can show that $f$ can be in a conflict with {\em at most} $k \geq 0$ facts of $s(D)$, where $k$ is the number of keys in $\dep$ over the relation name of $f$. This implies that we do not need to add more than $k$ operations at the end of $\hat{s}$.
%
Note that the above facts do {\em not} hold for FDs.
To establish that $\pi(s) \leq \mathsf{pol}''(||D||) \cdot \pi(s')$ using the above facts, we rely on the Cauchy–Schwarz inequality for $n$-dimensional Euclidean spaces.
Finally, once we have $\mathsf{F}$ in place, it is then not difficult to show item (2) via a combinatorial argument.

It is now easy to establish the existence of the polynomial $\mathsf{pol}'$ such that $\Lambda_{\neg f} \leq \mathsf{pol}'(||D||) \cdot \Lambda_f$. Indeed, with $\mathsf{F}$ and $\mathsf{pol}''$ being the function and the polynomial, respectively, provided by Lemma~\ref{lem:relate-sequences},
\begin{eqnarray*}
\Lambda_{\neg f}\ =\ \sum_{s \in S_{\neg f}}\pi(s)\ &\leq&  \sum_{s \in S_{\neg f}} \mathsf{pol}''(||D||) \cdot \pi(\mathsf{F}(s))\\
&\leq& \mathsf{pol}''(||D||) \cdot (2 \cdot ||D|| - 1) \cdot \sum_{s \in S_f} \pi(s)\\
&=& \mathsf{pol}''(||D||) \cdot (2 \cdot ||D|| - 1) \cdot \Lambda_f,
\end{eqnarray*}
and the claim follows with $\mathsf{pol}'(||D||) = \mathsf{pol}''(||D||) \cdot (2 \cdot ||D|| - 1)$.

\medskip
\noindent \paragraph{An FPRAS for FDs.} Recall that Theorem~\ref{the:uniform-operations} does not cover the case of FDs, which remains an open problem. At this point, one may wonder whether Monte Carlo sampling can be used for devising an FPRAS in the case of FDs. Indeed, the efficient sampler provided by Lemma~\ref{lem:uo-sampler} holds even for FDs since the proof of that lemma does not exploit keys in any way, but only the ``local'' nature of the Markov chain generator. However, we do not have a result analogous to Proposition~\ref{pro:uo-lower-bound}, which states that the target probability is never ``too small''.
In fact, there exist a set $\dep$ of FDs, a Boolean atomic query $Q$, and a family of databases $\{D_n\}_{n>0}$ with $|D_n| = n$, such that $0<\probrep{M_{\dep}^{\uo},Q}{D_n,()}\le \frac{1}{2^{n-1}}$; the proof is in the appendix.
Hence, for devising an FPRAS in the case of FDs (if it exists), we need a more sophisticated machinery than the one based on Monte Carlo sampling.
On the other hand, we can establish a result analogous to  Proposition~\ref{pro:uo-lower-bound} for FDs, assuming that only operations that remove a single fact (not a pair of facts) are considered. Given a set $\dep$ of FDs, let $M_{\dep}^{\uo,1}$ be the Markov chain generator defined as $M_{\dep}^{\uo}$, with the difference that only sequences consisting of operations that remove a single fact are considered. We then get the following:

\begin{theorem}\label{the:uniform-operations-singleton}
	For a set $\dep$ of FDs, and a CQ $Q$, $\ocqa{\dep,M_{\dep}^{\uo,1},Q}$ admits an FPRAS.
\end{theorem}

Note that singleton operations do not alter the data complexity of exact operational CQA; we can show that item (1) of Theorem~\ref{the:uniform-operations} continues to hold.
Let us also clarify that focusing on singleton operations does not affect Theorem~\ref{the:uniform-repairs} and Theorem~\ref{the:uniform-sequences}; all the details about these results can be found in Appendix~\ref{appsec:singleton-operations}.

\OMIT{
\begin{lemma}
	There exists a polynomial $\mathsf{pol}$ such that, for every sequence $s \in \abs{M_{\dep}^{\uo}(D)}$ with $f \not\in s(D)$, there exists a sequence $s' \in \abs{M_{\dep}^{\uo}(D)}$ with $f \in s'(D)$ and $\pi(s) \leq \mathsf{pol}(||D||) \cdot \pi(s')$.
\end{lemma}

The above lemma induces a non-injective function $\mathsf{F}$ from the set of sequences of $\abs{M_{\dep}^{\uo}(D)}$ that delete $f$ to the set of sequences of $\abs{M_{\dep}^{\uo}(D)}$ the keep $f$. We can show the following:

\begin{lemma}
	For every sequence $s' \in \abs{M_{\dep}^{\uo}(D)}$ with $f \in s'(D)$, $|\{s \in \abs{M_{\dep}^{\uo}(D)} \mid \mathsf{F}(s)=s'\}| \leq 2\cdot |D| - 1$.
\end{lemma}
}

\OMIT{
\begin{itemize}
	\item We note that item (1) of Theorem~\ref{the:uniform-operations} remains the same even if we use the Markov chain generator $M_{\dep}^{\uo,1}$, which considers only singe fact deletions. On the other hand, item (2) can be stated for arbitrary FDs (not only keys). This is shown by providing a polynomial bound, and exploiting the obvious efficient sampler.
	
	\item Whether there exists an FPRAS in the case of FDs and $M_{\dep}^{\uo}$ remains open. We can show, however, that in this case there is no polynomial bound, and thus, a more sophisticated machinery is needed towards an FPRAS.
\end{itemize}
}

%% file: conclusion.tex
\section{Future Work}\label{sec:conclusion}
%

%
Although we understand pretty well uniform operational CQA, there are still interesting open problems on approximability: (i) the case of keys and uniform repairs (we only have a negative result for the problem of counting repairs), (ii) the case of keys/FDs and uniform sequences, and (iii) the case of FDs and uniform operations (we only have a positive result assuming singleton operations).

%% file: app-uniform-generators.tex
\section{Uniform Operational CQA}\label{appsec:uniform-generators}

We provide the formal definitions of the ``uniform'' Markov chain generators discussed in Section~\ref{sec:uniform}, and show that they indeed capture our intention. In what follows, for a database $D$, a set $\dep$ of FDs, and a sequence $s =\op_1,\ldots,\op_n \in \rs{D}{\dep}$, we write $s_0$ for the empty sequence $\varepsilon$, and $s_i$ for the sequence $\op_1,\ldots,\op_i$, for $i \in [n]$.

\subsection{Uniform Repairs}

We start with the Markov chain generator based on the uniform probability distribution over the set of candidate operational repairs. As discussed in the main body, since multiple complete repairing sequences can lead to the same consistent database, we focus on canonical complete sequences. Recall that, for a database $D$, and a set $\dep$ of FDs, we say that a $(D,\dep)$-repairing sequence $s \in \crs{D}{\dep}$ is canonical if there is no $s' \in  \crs{D}{\dep}$ such that $s(D) = s'(D)$ and $s' \prec s$ for some arbitrary ordering $\prec$ over the set $\rs{D}{\dep}$, and we write $\cancrs{D}{\dep}$ for the set of all sequences of $\crs{D}{\dep}$ that are canonical. Furthermore, for a sequence $s \in \rs{D}{\dep}$, we write $\cancrss{D}{\dep}{s}$ for the set of all sequences $s'$ of $\cancrs{D}{\dep}$ that have $s$ as a prefix, i.e., $s' = s \cdot s''$ for some (possibly empty) sequence $s''$.
We are now ready to define the desired Markov chain generator.

\begin{definition}(\textbf{Uniform Repairs})\label{def:uniform-repairs}
	Consider a set $\dep$ of FDs. Let $M_{\dep}^{\ur}$ be the function assigning to a database $D$ the $(D,\dep)$-repairing Markov chain $(V,E,\ins{P})$, where, for each $(s,s') \in E$,
	\begin{eqnarray*}
		\insP(s,s')\
		=\ \left\{
		\begin{array}{ll}
			\frac{|\cancrss{D}{\dep}{s'}|}{|\cancrss{D}{\dep}{s}|} & 	\text{if } \cancrss{D}{\dep}{s} \neq \emptyset \\
			&\\
			\frac{1}{|\ops{s}{D}{\dep}|} & 	\text{otherwise.} \hspace{28mm}\hfill\markfull
		\end{array} \right. 
	\end{eqnarray*}
\end{definition}

Note that the above Markov chain generator is well-defined since, for each $s \in \rs{D}{\dep}$ that is not complete,
\[
|\cancrss{D}{\dep}{s}|\ =\ \sum_{s' \in \ops{s}{D}{\dep}} |\cancrss{D}{\dep}{s'}|,
\]
and thus, for a non-leaf node $s \in V$, $\sum_{t \in \{s' \mid (s,s') \in E\}} \ins{P}(s,t) = 1$.
We now show that the above definition captures our intention.

\begin{proposition}\label{pro:uniform-repairs}
	Consider a set $\dep$ of FDs. For every database $D$:
	\begin{enumerate}
		\item $\opr{D}{M_{\dep}^{\ur}} = \copr{D}{\dep}$.
		\item For every $D' \in \opr{D}{M_{\dep}^{\ur}}$, $\probrep{D,M_{\dep}^{\ur}}{D'} = \frac{1}{|\opr{D}{M_{\dep}^{\ur}}|}$.
	\end{enumerate}
\end{proposition}

\begin{proof}
	{\bf Item (1).} It suffices to prove that $\abs{M_{\dep}^{\ur}(D)} = \cancrs{D}{\dep}$. Let $M_\dep^{\ur}(D) = (V,E,\ins{P})$, and assume that $\pi$ is its leaf distribution. Recall that for a sequence $s = \op_1,\ldots,\op_n \in \crs{D}{\dep}$, $\pi(s) = \ins{P}(s_0,s_1) \cdots \ins{P}(s_{n-1},s_n)$.
	
	($\supseteq$) Assume first that $s = \op_1,\ldots,\op_n \in \cancrs{D}{\dep}$. This implies that $\cancrss{D}{\dep}{s_i} \neq \emptyset$, for each $i \in \{0,1,\ldots,n\}$. Therefore, $\ins{P}(s_i,s_{i+1}) > 0$, for $i \in \{0,1,\ldots,n\}$, and thus $\pi(s) > 0$. The latter implies that $s \in \abs{M_\dep^{\ur}(D)}$, which in turn shows that $\abs{M_\dep^{\ur}(D)} \supseteq \cancrs{D}{\dep}$, as needed.
	
	($\subseteq$) Assume now that $s = \op_1,\ldots,\op_n \in \abs{M_\dep^{\ur}(D)}$. By contradiction, assume that $s \not \in \cancrs{D}{\dep}$. Since $s \in \abs{M_\dep^{\ur}(D)}$, $s$ must be complete. The fact that $s$ is complete but not canonical implies that there exists $i \in \{0,\ldots,n\}$ such that $\cancrss{D}{\dep}{s_i} = \emptyset$. In particular, let $\ell$ be the smallest integer in $\{0,1,\ldots,n\}$ such that $\cancrss{D}{\dep}{s_{\ell}} = \emptyset$. Clearly, $\ell > 0$, since $\cancrss{D}{\dep}{\epsilon}$ is always non-empty. Thus, by the first rule of the expression defining $\ins{P}$ in Definition~\ref{def:uniform-repairs}, we have that $\ins{P}(s_{\ell-1},s_{\ell}) = 0$. Hence, $\pi(s) = 0$, and thus, $s \not \in \abs{M_\dep^{\ur}(D)}$, which contradicts our hypothesis.
	
	\medskip
	{\bf Item (2).} By the proof of item (1), $\abs{M_\dep^{\ur}(D)} = \cancrs{D}{\dep}$. Hence, we conclude that
	\[
	|\opr{D}{M_\dep^{\ur}}|\ =\ |\cancrs{D}{\dep}|\ =\ |\abs{M_\dep^{\ur}(D)}|.
	\] 
	Therefore, it suffices to show that, for $s \in \cancrs{D}{\dep}$, $\pi(s) = \frac{1}{|\cancrs{D}{\dep}|}$.
	Let $s = \op_1,\ldots,\op_n \in \cancrs{D}{\dep}$. Since $s \in \abs{M_\dep^{\ur}(D)}$, $\pi(s)$ is equal to
	\[
	\frac{|\cancrss{D}{\dep}{s_1}|}{|\cancrss{D}{\dep}{s_0}|} \cdots \frac{|\cancrss{D}{\dep}{s_n}|}{|\cancrss{D}{\dep}{s_{n-1}}|} = 
	\frac{|\cancrss{D}{\dep}{s_n}|}{|\cancrss{D}{\dep}{s_0}|}.
	\]
	Since $\cancrss{D}{\dep}{s_0} = \cancrss{D}{\dep}{\epsilon} = \cancrs{D}{\dep}$, and $\cancrss{D}{\dep}{s_n} = \{s_n\}$, then
	$\pi(s) = \frac{1}{|\cancrs{D}{\dep}|}$, as needed.
\end{proof}

\subsection{Uniform Sequences}
We now proceed to define the Markov chain generator based on the uniform probability distribution over the set of complete repairing sequences.
It is defined similarly to the Markov chain generator above 
with the difference that we consider arbitrary, not necessarily canonical, complete sequences.

\begin{definition}(\textbf{Uniform Sequences})\label{def:uniform-seq}
	Consider a set $\dep$ of FDs. Let $M_{\dep}^{\us}$ be the function assigning to a database $D$ the $(D,\dep)$-repairing Markov chain $(V,E,\ins{P})$, where, for each $(s,s') \in E$,
	\begin{flalign*}
	&& \insP(s,s') = \frac{|\crss{D}{\dep}{s'}|}{|\crss{D}{\dep}{s}|} && \markfull
	\end{flalign*}
\end{definition}

Observe that the above Markov chain generator is well-defined since, for each $s \in \rs{D}{\dep}$ that is not complete,
\[
|\crss{D}{\dep}{s}|\ =\ \sum_{s' \in \ops{s}{D}{\dep}} |\crss{D}{\dep}{s'}|,
\]
and thus, for a non-leaf node $s \in V$, $\sum_{t \in \{s' \mid (s,s') \in E\}} \ins{P}(s,t) = 1$.
We can easily show that $M_{\dep}^{\us}$ captures our intention:

\begin{proposition}\label{pro:uniform-seq}
	Consider a set $\dep$ of FDs. For every database $D$:
	\begin{enumerate}
		\item $\abs{M_{\dep}^{\us}(D)} = \crs{D}{\dep}$.
		\item For every $s \in \crs{D}{\dep}$, assuming that $\pi$ is the leaf distribution of $M_{\dep}^{\us}(D)$, $\pi(s) = \frac{1}{|\crs{D}{\dep}|}$ .
	\end{enumerate}
\end{proposition}
\begin{proof}
	{\bf Item (1).} This item follows from the fact that each $s \in \abs{M_\dep^{\us}(D)}$ is complete by definition, and each $s = \op_1,\ldots,\op_n \in \crs{D}{\dep}$ is such that $\crss{D}{\dep}{s_i} \neq \emptyset$, for $i \in \{0,\ldots,n\}$, and thus $\pi(s) > 0$, where $\pi$ is the leaf distribution of $M_\dep^{\us}(D)$.
	
	{\bf Item (2).} It is shown via a proof similar to the one used above for item (2) of Proposition~\ref{pro:uniform-repairs}.
\end{proof}

\subsection{Uniform Operations}
We finally define the Markov chain generator based on the uniform probability distribution over the set of available operations at a single step of the repairing process.

\begin{definition}(\textbf{Uniform Operations})\label{def:uniform-ops}
	Consider a set $\dep$ of FDs. Let $M_{\dep}^{\uo}$ be the function assigning to a database $D$ the $(D,\dep)$-repairing Markov chain $(V,E,\ins{P})$, where, for each $(s,s') \in E$,
	\begin{flalign*}
		&& \insP(s,s') = \frac{1}{|\ops{s}{D}{\dep}|} && \markfull
	\end{flalign*}
\end{definition}

It is straightforward to see that the function $M_{\dep}^{\uo}$ captures our intention; in fact, the following holds by definition:

\begin{proposition}\label{pro:uniform-ops}
	Consider a set $\dep$ of FDs. For every database $D$:
	\begin{enumerate}
		\item $\abs{M_{\dep}^{\uo}(D)} = \crs{D}{\dep}$.
		\item Assuming that $M_{\dep}^{\uo}(D) = (V,E,\ins{P})$, $(s,s') \in E$ implies $\insP(s,s') = \frac{1}{|\ops{s}{D}{\dep}|}$.
	\end{enumerate}
\end{proposition}

%% file: app-uniform-repairs.tex
\section{Proofs of Section~\ref{sec:uniform-repairs}}\label{appsec:uniform-repairs}
In this section, we prove the main result of Section~\ref{sec:uniform-repairs}, which we recall here for the sake of readability:

\begin{manualtheorem}{\ref{the:uniform-repairs}}
\begin{enumerate}
	\item There exist a set $\dep$ of primary keys, and a CQ $Q$ such that $\ocqa{\dep,M_{\dep}^{\ur},Q}$ is $\sharp ${\rm P}-hard.
	
	\item For a set $\dep$ of primary keys, and a CQ $Q$, $\ocqa{\dep,M_{\dep}^{\ur},Q}$ admits an FPRAS.
	
	\item Unless ${\rm RP} = {\rm NP}$, there exist a set $\dep$ of FDs, and a CQ $Q$ such that there is no FPRAS for $\ocqa{\dep,M_{\dep}^{\ur},Q}$.
\end{enumerate}
\end{manualtheorem}

As discussed in Section~\ref{sec:uniform-repairs}, we actually need to prove the above result for the problem $\rrelfreq{\dep,Q}$.

\subsection{Proof of Item~(1) of Theorem~\ref{the:uniform-repairs}}

Consider the undirected graph $H = (V_H,E_H)$, where $V_H = \{0,1,?\}$ and $E_H = \{\{u,v\} \mid (u,v) \in (V_H \times V_H) \setminus \{(1,1)\}\}$, i.e., the graph:

\medskip

\centerline{\includegraphics{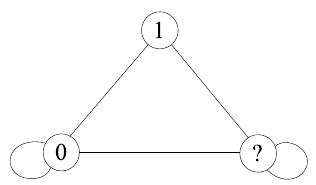}}

\medskip

\noindent Given an undirected graph $G$, a homomorphim from $G$ to $H$ is a mapping $h : V_G \rightarrow V_H$ such that $\{u,v\} \in E_G$ implies $\{h(u),h(v)\} \in E_H$. We write $\mathsf{hom}(G,H)$ for the set of homomorphisms from $G$ to $H$.
The problem $\#H\text{-}\mathsf{Coloring}$ is defined as follows:

\medskip

\begin{center}
	\fbox{\begin{tabular}{ll}
			{\small PROBLEM} : & $\sharp H\text{-}\mathsf{Coloring}$\\
			{\small INPUT} : & An undirected graph $G$.\\
			{\small OUTPUT} : &  The number $|\mathsf{hom}(G,H)|$.
	\end{tabular}}
\end{center}

\medskip

\noindent It is implicit in~\cite{Dyer00} that $\sharp H\text{-}\mathsf{Coloring}$ is $\sharp ${\rm P}-hard. In fact,~\cite{Dyer00} establishes the following dichotomy result: $\sharp \hat{H}\text{-}\mathsf{Coloring}$ is $\sharp ${\rm P}-hard if $\hat{H}$ has a connected component which is neither an isolated node without a loop, nor a complete graph with all loops present, nor a complete bipartite graph without loops; otherwise, it is solvable in polynomial time. Since our fixed graph $H$ above consists of a single connected component which is neither a single node, nor a complete graph with all loops present (the loop $(1,1)$ is missing), nor a bipartite graph, we get that $\#H\text{-}\mathsf{Coloring}$ is indeed $\sharp ${\rm P}-hard.

We proceed to show via a polynomial-time Turing reduction from $\sharp H\text{-}\mathsf{Coloring}$ that $\rrelfreq{\dep,Q}$ is $\sharp ${\rm P}-hard, where $\dep$ and $Q$ are as follows. Let $\ins{S}$ be the schema $\{V/2, E/2, T/1\}$, and let $(A,B)$ be the tuple of attributes of $V$.
The set $\dep$ consists of the single key
\[
V: A \ra B
\]
and the (constant-free) Boolean CQ $Q$ is
\[
\textrm{Ans}()\ \text{:-}\ E(x,y), V(x,z), V(y,z), T(z).
\]

Given an undirected graph $G = (V_G,E_G)$, we define the following database over $\ins{S}$ encoding $G$:
\[
D_G\ =\ \{V(u,0),V(u,1) \mid u \in V_G\}\ \cup\ \{E(u,v) \mid \{u,v\} \in E_G\}\ \cup\ \{T(1)\}.
\]
We then define the algorithm $\mathsf{HOM}$, which accepts as input an undirected graph $G = (V_G,E_G)$, as follows:
\begin{enumerate}
	\item Construct the database $D_G$.
	\item Compute the number $r = \orfreq{\dep,Q}{D_G,()}$.
	\item Output the number $3^{|V_G|} \cdot (1- r)$.
\end{enumerate}
It is clear that $\mathsf{HOM}(G)$ runs in polynomial time in $||G||$ assuming access to an oracle for the problem $\rrelfreq{\dep,Q}$. It remains to show that $|\mathsf{hom}(G,H)| = \mathsf{HOM}(G)$.
Recall that
$$\orfreq{\dep,Q}{D_G,()} = \frac{|\copr{D_G}{\dep,Q}|}{|\copr{D_G}{\dep}|},$$
where $\copr{D_G}{\dep,Q}$ is the set of candidate repairs $D$ of $D_G$ w.r.t.~$\dep$ such that $D \models Q$.
Observe that there are $3^{|V_G|}$ candidate repairs of $D_G$ w.r.t.~$\dep$, i.e., in each such a repair $D$, for each node $u \in V$ of $G$, either $V(u,0) \in D$ and $V(u,1) \not \in D$, or $V(u,0) \not \in D$ and $V(u,1) \in D$, or $V(u,0),V(u,1) \not \in D$.
%
%
Therefore,
\[
\orfreq{\dep,Q}{D_G,()}\ =\ \dfrac{|\copr{D_G}{\dep,Q}|}{3^{|V_G|}}.
\]
Thus, $\mathsf{HOM}(G)$ coincides with
\[
3^{|V_G|} \cdot \left(1 - \dfrac{|\copr{D_G}{\dep,Q}|}{3^{|V_G|}}\right)\ =\ 3^{|V_G|} - |\copr{D_G}{\dep,Q}|.
\]
Since $D_G$ has $3^{|V_G|}$ candidate repairs w.r.t.~$\dep$, we can conclude that $\mathsf{HOM}(G)$ is precisely the cardinality of the set $\copr{D_G}{\dep,\neg Q}$, which collects the candidate repairs $D$ of $D_G$ w.r.t.~$\dep$ such that $D \not \models Q$. We proceed to show that:

\begin{lemma}\label{lem:keys-aux}
	$|\mathsf{hom}(G,H)|\ =\ |\copr{D_G}{\dep,\neg Q}|$.
\end{lemma}

\begin{proof}
	It suffices to show that there exists a bijection from $\mathsf{hom}(G,H)$ to $\copr{D_G}{\dep,\neg Q}$. To this end, we define the mapping $\mu : \mathsf{hom}(G,H) \ra \PS(D_G)$ as follows: for each $h \in \mathsf{hom}(G,H)$,
	\begin{multline*}
		\mu(h)\ =\ \{V(u,\star) \mid u \in V_G \text{ and } h(u) = \star \in \{0,1\}\}\ \cup \\
		\{E(u,v) \mid \{u,v\} \in E_G\}\ \cup\ \{T(1)\}.
	\end{multline*}
	We proceed to show the following three statements:
	\begin{enumerate}
		\item $\mu$ is correct, that is, it is indeed a function from $\mathsf{hom}(G,H)$ to $\copr{D_G}{\dep,\neg Q}$.
		\item $\mu$ is injective.
		\item $\mu$ is surjective.
	\end{enumerate}
	
	\medskip 
	
	\noindent
	\paragraph{The mapping $\mu$ is correct.} Consider an arbitrary homomorphism $h \in \mathsf{hom}(G,H)$. We need to show that there exists a $(D_G,\dep)$-repairing sequence $s_h$ such that $\mu(h) = s_h(D_G)$, $s_h(D_G) \models \dep$ (i.e., $s_h$ is complete), and $Q(s_h(D_G)) = \emptyset$. Let $V_G = \{u_1,\ldots,u_n\}$. Consider the sequence $s_h = \op_1,\ldots,\op_n$ such that, for every $i \in [n]$:
	\[
	\op_{i}\ =\
	\begin{cases}
		-V(u_i,1) & \text{if } h(u_i) = 0 \\
		-V(u_i,0) & \text{if } h(u_i) = 1 \\
		-\{V(u_i,0),V(u_i,1)\} & \text{if } h(u_i) =\ ? \\
	\end{cases}
	\]
	In simple words, the homomorphism $h$ guides the repairing process, i.e., $h(u_i) = 0$ (resp., $h(u_i)=1$) implies $V(u_i,0)$ (resp., $V(u_i,1)$) should be kept, while $h(u_i) =\ ?$ implies none of the atoms $V(u_i,0),V(u_i,1)$ should be kept. It is easy to verify that $s_h$ is indeed a $(D_G,\dep)$-repairing sequence $s_h$ such that $\mu(h) = s_h(D_G)$ and $s_h(D_G) \models \dep$.
	The fact that $Q(s_h(D_G)) = \emptyset$ follows from the fact that, for every edge $\{u,v\} \in E_G$, $\{h(u),h(v)\} \in E_H$ cannot be the self-loop on node 1, since it is not in $H$. This implies that for every $\{u,v\} \in E_G$, it is not possible that the atoms $V(u,1),V(v,1)$ coexist in $s_h(D_G)$, which in turn implies that $Q(s_h(D_G)) = \emptyset$, as needed.

	\medskip

	\noindent
	\paragraph{The mapping $\mu$ is injective.} Assume that there are two distinct homomorphisms $h,h' \in \mathsf{hom}(G,H)$ such that $\mu(h) = \mu(h')$. By the definition of $\mu$, we get that $h(u) = h'(u)$, for every node $u \in V_G$. But this contradicts the fact that $h$ and $h'$ are different homomorphisms of $\mathsf{hom}(G,H)$. Therefore, for every two distinct homomorphisms $h,h' \in \mathsf{hom}(G,H)$, $\mu(h) \neq \mu(h')$, as needed.

	\medskip

	\noindent
	\paragraph{The mapping $\mu$ is surjective.} Consider an arbitrary candidate repair $D \in \copr{D_G}{\dep,\neg Q}$. We need to show that there exists $h \in \mathsf{hom}(G,H)$ such that $\mu(h) = D$. We define the mapping $h_D : V_G \ra V_H$ as follows: for every $u \in V_G$:
	\[
	h_D(u)\ =\
	\begin{cases}
		1 & \text{if } V(u,1) \in D \text{ and } V(u,0) \not\in D \\
		0 & \text{if } V(u,1) \not\in D \text{ and } V(u,0) \in D \\
		? & \text{if } V(u,1) \not\in D \text{ and } V(u,0) \not\in D \\
	\end{cases}
	\]
	It is clear that $h_D$ is well-defined: for every $u \in V_G$, $h_D(u) = x$ and $h_D(u) = y$ implies $x=y$. It is also clear that $\mu(h_D) = D$. It remains to show that $h_D \in \mathsf{hom}(G,H)$. Consider an arbitrary edge $\{u,v\} \in E_G$. By contradiction, assume that $\{h_D(u),h_D(v)\} \not\in E_H$. This implies that $h_D(u) = 1$ and $h_D(v) = 1$. Therefore, $D$ contains both atoms $V(u,1)$ and $V(v,1)$, which in turn implies that $Q(D) \neq \emptyset$, which contradicts the fact that $D \in \copr{D_G}{\dep,\neg Q}$. 
\end{proof}

Since $\mathsf{HOM}(G) = |\copr{D_G}{\dep,\neg Q}|$, Lemma~\ref{lem:keys-aux} implies
\[
\mathsf{HOM}(G)\ =\ |\mathsf{hom}(G,H)|,
\]
which shows that indeed $\mathsf{HOM}$ is a polynomial-time Turing reduction from $\#H\text{-}\mathsf{Coloring}$ to $\rrelfreq{\dep,Q}$.

\subsection{Proof of Item~(2) of Theorem~\ref{the:uniform-repairs}}

We prove that, for a set $\dep$ of primary keys, and a CQ $Q$, the problem $\rrelfreq{\dep,Q}$ admits an FPRAS. Our proof consists of two main steps, which we briefly explain before going into the detailed proofs.

The first step is to show that, given a database $D$, we can sample elements of $\copr{D}{\dep}$ uniformly at random in polynomial time in $||D||$. The existence of such an efficient sampler implies that we can employ Monte Carlo Sampling to obtain a \emph{polynomial-time randomized approximation with additive (or absolute) error} for $\rrelfreq{\dep,Q(\bar x)}$, that is, a randomized algorithm $\mathsf{A}$ that takes a input a database $D$, a tuple $\bar c \in \adom{D}^{|\bar x|}$, $\epsilon>0$, and $0<\delta<1$ runs in polynomial time in $||D||$, $||\bar c||$, $1/\epsilon$ and $\log(1/\delta)$, and produces a random variable $\mathsf{A}(D,\bar c,\epsilon,\delta)$ such that
\[
\pr\left(|\mathsf{A}(D,\bar c,\epsilon,\delta) - \orfreq{\dep,Q}{D,\bar c}|\ \leq\ \epsilon\right)\ \geq\ 1-\delta.
\]
More precisely, $\mathsf{A}(D,\bar c,\epsilon,\delta)$ samples $N = O\left(\frac{\log(\frac{1}{\delta})}{\epsilon^2}\right)$ elements of the set $\copr{D}{\dep}$, and returns the number
$
\frac{S}{N} \cdot |\copr{D}{\dep}|,
$
where $S$ is the number of sampled repairs $D'$ such that $\bar c \in Q(D')$.

\OMIT{
That is, given a database $D$, $\epsilon>0$, and $0<\delta<1$, we can sample $O(\frac{\log(\frac{1}{\delta})}{\epsilon^2})$ elements of $\copr{D}{\dep}$, and return the number:
\[A(D,\bar c,\epsilon,\delta)=\frac{S}{N}\times |\copr{D}{\dep}|\]
where $S$ is the number of sampled repairs $E$ that satisfy $\bar c\in Q(E)$. It is well known that:
\[
\text{\rm Pr}\left(|\mathsf{A}(D,\bar c,\epsilon,\delta) - \orfreq{\dep,Q}{D,\bar c}|\ \leq\ \epsilon\right)\ \geq\ 1-\delta.
\]
}

However, in general, the existence of an efficient sampler does not guarantee the existence of an FPRAS, which bounds the \emph{multiplicative (or relative)} error. 
In order to obtain an FPRAS via Monte Carlo Sampling, the number of samples should be proportional to $\frac{1}{\orfreq{\dep,Q}{D,\bar c}}$~\cite{KarpLuby2000}. 
This brings us to the second step of our proof, where we show that 
the ratio $\orfreq{\dep,Q}{D,\bar c}$ is never ``too small''. Formally, we show that, for every database $D$, it either holds that $\orfreq{\dep,Q}{D,\bar c}=0$ or  $\orfreq{\dep,Q}{D,\bar c}\ge \frac{1}{\mathsf{pol}(||D||)}$ for some polynomial $\mathsf{pol}$. In this case, we can use Monte Carlo Sampling (with a different, yet polynomial, number of samples) to obtain an FPRAS.

\def\nempt{\mathsf{ne}}
\def\empt{\mathsf{e}}

We now proceed to formally show the existence of an efficient sample, and the fact that target ratio is  never ``too small''


\subsubsection*{Step 1: Efficient Sampler}
The formal statement, already given in the main body of the paper, and its proof follow:

\begin{manuallemma}{\ref{lem:ur-sampler}}
Given a database $D$, and a set $\dep$ of primary keys, we can sample elements of $\copr{D}{\dep}$ uniformly at random in polynomial time in $||D||$.
\end{manuallemma}
\begin{proof}
For every relation name $R$ of the underlying schema with a primary key $R: X\rightarrow Y$ in $\dep$, we partition the set of facts of $D$ over $R$ into blocks of facts that agree on the values of all the attributes of $X$. Clearly, two facts that belong to the same block, always jointly violate the key of the corresponding relation; hence, an operational repair will contain, for every block $B$ with $|B|>1$, either a single fact of $B$ or none of the facts of $B$ (hence, there are $|B|+1$ possible options). An operational repair of the first type can be obtained, for example, via a sequence that removes the facts of $B$ one by one until there is one fact left. An operational repair of the second type can be obtained, for example, by removing the facts of $B$ one by one until there are two facts left, and then removing the last two facts together.
For a block $B$ such that $|B|=1$, there is no justified operation that removes the single fact of $B$; hence, this fact will appear in every operational repair. 

We denote all the blocks of $D$ (over all the relations of the schema) that have at least two facts by $B_1,\dots,B_n$.
To sample a repair of $\copr{D}{\dep}$, we select, for each block $B_i$, one of its $|B_{i}|+1$ possible outcomes, with probability $\frac{1}{|B_{i}|+1}$. Then, an operational repair is obtained by taking the union of all the selections, as well as all the facts of $D$ over every relation $R$ of the schema that has no primary key in $\dep$, and all the facts that belong to blocks consisting of a single fact. It is rather straightforward that the probability of obtaining each operational repair is $\frac{1}{|\copr{D}{\dep}|}$, as the number of operational repairs is
\[
|\copr{D}{\dep}|=(|B_1|+1)\times\dots\times (|B_n|+1),
\]
and the claim follows.
\end{proof}

We give a simple example that illustrates the proof of Lemma~\ref{lem:ur-sampler}.

\begin{figure}
    \centering
    \begin{tabular}{c||c|c}
         & $A_1$ & $A_2$ \\\hline\hline
        $f_{1,1}$ & $a_1$ & $b_1$ \\\hline
        $f_{1,2}$ & $a_1$ & $b_2$ \\\hline
        $f_{1,3}$ & $a_1$ & $b_3$ \\\hline
        $f_{2,1}$ & $a_2$ & $b_1$ \\\hline
        $f_{3,1}$ & $a_3$ & $b_1$ \\\hline
        $f_{3,2}$ & $a_3$ & $b_2$ \\
    \end{tabular}
    \caption{A database over $\{R/2\}$  that is inconsistent w.r.t.~the primary key $R:A_1\rightarrow A_2$.}
    \label{fig:example_pkeys}
\end{figure}

\begin{example}\label{example:primary_keys}
Consider the database depicted in Figure~\ref{fig:example_pkeys} over the schema $\{R/2\}$, with $(A_1,A_2)$ being the tuples of attributes of $R$, and the set $\dep=\{R:A_1\rightarrow A_2\}$ consisting of a single key. We write $f_{i,j}$ for $R(a_i,b_j)$. Clearly, for $j \neq k$,
$\{f_{i,j},f_{i,k}\}\not\models\dep$. The database consists of three blocks w.r.t.~$R:A_1\rightarrow A_2$: 
\[\{f_{1,1},f_{1,2},f_{1,3}\} \,\,\,\,\,\,\,\, \{f_{2,1}\} \,\,\,\,\,\,\,\, \{f_{3,1},f_{3,2}\}\]
Since the fact $f_{2,1}$ is not involved in any violations of the constraints, it will appear in every operational repair; however, every operational repair will contain at most one fact of the first block and at most one fact of the third block. The number of operational repairs according to the formula in the proof of Lemma~\ref{lem:ur-sampler} is then $(3+1)\times (2+1)=12$.
(Note that the blocks of size one are not considered in the computation.) Indeed, there are twelve repairs:
\begin{align*}
    &\{f_{2,1}\} \,\,\,\,\,\,\,\, \{f_{1,1},f_{2,1}\} \,\,\,\,\,\,\,\, \{f_{1,2},f_{2,1}\} \,\,\,\,\,\,\,\, \{f_{1,3},f_{2,1}\}\\
    &\{f_{2,1},f_{3,1}\} \,\,\,\,\,\,\,\, \{f_{1,1},f_{2,1},f_{3,1}\} \,\,\,\,\,\,\,\, \{f_{1,2},f_{2,1},f_{3,1}\} \,\,\,\,\,\,\,\, \{f_{1,3},f_{2,1},f_{3,1}\}\\
    &\{f_{2,1},f_{3,2}\} \,\,\,\,\,\,\,\, \{f_{1,1},f_{2,1},f_{3,2}\} \,\,\,\,\,\,\,\, \{f_{1,2},f_{2,1},f_{3,2}\} \,\,\,\,\,\,\,\, \{f_{1,3},f_{2,1},f_{3,2}\}
\end{align*}
The repair $\{f_{1,1},f_{2,1},f_{3,1}\}$, for example, is obtained by keeping the fact $f_{1,1}$ of the first block with probability $\frac{1}{4}$ (as there are three facts in the block, there are four possible options: \emph{(1)} keep $f_{1,1}$, \emph{(2)} keep $f_{1,2}$, \emph{(3)} keep $f_{1,3}$, or \emph{(4)} remove all the facts of the block), and  the fact $f_{3,1}$ of the third block with probability $\frac{1}{3}$. Hence, the probability of selecting this operational repair is $\frac{1}{4}\times \frac{1}{3}=\frac{1}{12}$, and the same holds for any other operational repair. \hfill\markfull
\end{example}

\subsubsection*{Step 2: Polynomial Lower Bound}
Now that we have an efficient sampler for the operational repairs, we proceed to show that there is a polynomial lower bound on $\orfreq{\dep,Q}{D,\bar c}$.
The formal statement, already given in the main body of the paper, and its proof follow:

\begin{manuallemma}{\ref{lem:ur-lower-bound}}
Consider a set $\dep$ of primary keys, and a CQ $Q(\bar x)$. For every database $D$, and tuple $\bar c \in \adom{D}^{|\bar x|}$,
	\[
	\orfreq{\dep,Q}{D,\bar c}\ \geq\ \frac{1}{(2 \cdot ||D||)^{||Q||}}
	\] 
	whenever $\orfreq{\dep,Q}{D,\bar c} > 0$.
\end{manuallemma}
\begin{proof}
By abuse of notation, we treat the CQ $Q$ as the set of atoms on the right-hand side of $\text{:-}$ (hence, $|Q|$ is the number of atoms occurring in $Q$). 
Consider a database $D$, and a tuple $\bar c \in \adom{D}^{|x|}$. 
If there is no homomorphism $h$ from $Q$ to $D$ such that $h(Q)\models\dep$ and $h(\bar x)=\bar c$, then 
it clearly holds that $\orfreq{\dep,Q}{D,\bar c}=0$.

Consider now the case that such a homomorphism $h$ exists. Assuming that $Q=\{R_i(\bar y_i)\mid i \in [n]\}$, let $h(Q) = \{R_i(h(\bar y_i)) \mid i \in [n]\}$. Assume that $|h(Q)|=m$ for some $m\le |Q|$. For every relation name $R$ of the schema with a key $R:X\rightarrow Y$ in $\dep$, we partition the set of facts of $D$ over $R$ into blocks of facts that agree on the values of all the attributes of $X$. For a relation name $R$ with no key in $\dep$, we assume that every fact is a separate block.
Let $B_1,\dots,B_n$ be the blocks of $D$ w.r.t.~$\dep$ (over all the relation names of the schema). We assume, without loss of generality, that the facts of $h(Q)$ belong to the blocks $B_1,\dots,B_m$. Clearly, no two facts of $h(Q)$ belong to the same block; otherwise,  $h(Q) \not\models \dep$, which is a contradiction.
 
 Let $R_{D,\dep,h(Q)}^\nempt$ be the set of repairs $E\in \copr{D}{\dep}$ such that $E\cap B_j\neq\emptyset$ for every $j\in[m]$. Let $R_{D,\dep,h(Q)}^\empt$ be the set of repairs $E\in \copr{D}{\dep}$ such that $E\cap B_j=\emptyset$ for some $j \in [m]$. Clearly,
 \[
 \left|\copr{D}{\dep}\right|\ =\ \left|R_{D,\dep,h(Q)}^\nempt\right| + \left|R_{D,\dep,h(Q)}^\empt\right|.
 \]
 Now, consider a repair $E$ of $R_{D,\dep,h(Q)}^\empt$, and assume that $E$ is disjoint with
 precisely $\ell$ blocks of $\{B_1,\dots,B_m\}$. Assume, without loss of generality, that these are the blocks $B_1,\dots,B_\ell$. We can transform the repair $E$ into a repair $E'\in R_{D,\dep,h(Q)}^\nempt$ by bringing back an arbitrary fact of each block $B_j$ for $j\in[\ell]$. Therefore, the repair $E$ is mapped to $|B_1|\times\dots\times |B_\ell|$ distinct repairs of $R_{D,\dep,h(Q)}^\nempt$.
 
Observe that at most $2^m-1$ repairs $E\in R_{D,\dep,h(Q)}^\empt$ are mapped to the same repair $E'\in R_{D,\dep,h(Q)}^\nempt$. This holds since the repair $E'$ determines, for every block $B_j$ that is not one of $B_1,\dots,B_m$, whether we keep a fact of $B_j$ in the repair and which fact of $B_j$ we keep. For the blocks $B_1,\dots,B_m$, a repair $E$ that is mapped to $E'$ can either contain the same fact as $E'$ contains from this block, or none of the facts of the block. Hence, there are two possibilities for each block of $\{B_1,\dots,B_m\}$ and the total number of possibilities is $2^m$. However, we have to disregard one of these possibilities, as it represents $E'$ itself (where for every block of $\{B_1,\dots,B_m\}$ we keep the same fact as $E'$). We conclude that
 \[
 \left|R_{D,\dep,h(Q)}^\empt\right| \le (2^m-1)\times \left|R_{D,\dep,h(Q)}^\nempt\right|
 \]
 and
 \begin{align*}
     |\copr{D}{\dep}|&= \left|R_{D,\dep,h(Q)}^\empt\right|+  \left|R_{D,\dep,h(Q)}^\nempt\right|\\
     &\le (2^m-1)\times\left|R_{D,\dep,h(Q)}^\nempt\right|+ \left|R_{D,\dep,h(Q)}^\nempt\right|\\
     &=2^m\times\left|R_{D,\dep,h(Q)}^\nempt\right|.
 \end{align*}
 Note that $(2^m-1)\times \left|R_{D,\dep,h(Q)}^\nempt\right|$ is just an upper bound on $\left|R_{D,\dep,h(Q)}^\empt\right|$ because, as said above, each repair $E$ of $R_{D,\dep,h(Q)}^\empt$ is mapped to several distinct repairs of $R_{D,\dep,h(Q)}^\nempt$.
 
 Finally, each repair of $R_{D,\dep,h(Q)}^\nempt$ keeps a fact of every block in $\{B_1,\dots,B_m\}$. Here, we are interested in the repairs that keep all the facts of $h(Q)$, as these repairs $E$ satisfy $\bar c\in Q(E)$. Clearly,
 \[|\{E\in \copr{D}{\dep}\mid h(Q)\subseteq E\}|=\frac{1}{|B_1|\times\dots\times |B_m|} \times \left|R_{D,\dep,h(Q)}^\nempt\right|
 \]
 as all the facts of a block are symmetric.
 Hence, we conclude that:
\begin{align*}
   \frac{|\{E\in \copr{D}{\dep}\mid h(Q)\subseteq E\}|}{|\copr{D}{\dep}|}&\ge \frac{\frac{1}{|B_1|\times\dots\times |B_m|} \times \left|R_{D,\dep,h(Q)}^\nempt\right|}{2^{m}\times \left|R_{D,\dep,h(Q)}^\nempt\right|}\\
   &=\frac{1}{|B_1|\times\dots\times |B_m|\times 2^{m}}\\
   &\ge \frac{1}{|D|^m\times 2^{m}}\ge \frac{1}{|D|^{|Q|}\times 2^{|Q|}}\\
   &=\frac{1}{(2|D|)^{|Q|}}\ge \frac{1}{(2||D||)^{||Q||}}
\end{align*}
and this is clearly a lower bound on $\orfreq{\dep,Q}{D,\bar c}$, as needed.
\end{proof}

Here is a simple example that illustrates the argument given in the proof of Lemma~\ref{lem:ur-lower-bound}.

\begin{example}\label{example:pkeys_ur}
Consider again the database $D$ depicted in Figure~\ref{fig:example_pkeys}, and the set $\dep=\{R:A_1\rightarrow A_2\}$ consisting of a single key.
Let $Q$ be the CQ $\textrm{Ans}(x)\ \text{:-}\ R(a_1,x)$. 
A homomorphism $h$ from $Q$ to $D$ with $h(Q)\models\dep$ and $h(x)=b_1$ is such that $h(Q)=\{R(a_1,b_1)\}$. The fact $R(a_1,b_1)$ belongs to the block $\{f_{1,1},f_{1,2},f_{1,3}\}$. Hence, the set $R_{D,\dep,h(Q)}^\nempt$ consists of the repairs:
\begin{align*}
    &\{f_{1,1},f_{2,1}\} \,\,\,\,\,\,\,\, \{f_{1,2},f_{2,1}\} \,\,\,\,\,\,\,\, \{f_{1,3},f_{2,1}\}\\
    & \{f_{1,1},f_{2,1},f_{3,1}\} \,\,\,\,\,\,\,\, \{f_{1,2},f_{2,1},f_{3,1}\} \,\,\,\,\,\,\,\, \{f_{1,3},f_{2,1},f_{3,1}\}\\
    & \{f_{1,1},f_{2,1},f_{3,2}\} \,\,\,\,\,\,\,\, \{f_{1,2},f_{2,1},f_{3,2}\} \,\,\,\,\,\,\,\, \{f_{1,3},f_{2,1},f_{3,2}\}
\end{align*}
and the set $R_{D,\dep,h(Q)}^\empt$ consists of the repairs:
\begin{align*}
    &\{f_{2,1}\} \,\,\,\,\,\,\,\, \{f_{2,1},f_{3,1}\} \,\,\,\,\,\,\,\, \{f_{2,1},f_{3,2}\}
\end{align*}

According to the mapping defined in the proof of Lemma~\ref{lem:ur-lower-bound}, the repair $\{f_{2,1}\}$ is mapped to the repairs in \[\{\{f_{1,1},f_{2,1}\},\{f_{1,2},f_{2,1}\},\{f_{1,3},f_{2,1}\}\}\] 
that have one additional fact from the block of $R(a_1,b_1)$. Similarly, the repair $\{f_{2,1},f_{3,1}\}$ is mapped to the repairs in 
\[\{\{f_{1,1},f_{2,1},f_{3,1}\},\{f_{1,2},f_{2,1},f_{3,1}\},\{f_{1,3},f_{2,1},f_{3,1}\}\}\] and the repair $\{f_{2,1},f_{3,2}\}$ is mapped to the repairs in \[\{\{f_{1,1},f_{2,1},f_{3,2}\},\{f_{1,2},f_{2,1},f_{3,2}\},\{f_{1,3},f_{2,1},f_{3,2}\}\}.\]
Hence, each repair of $R_{D,\dep,h(Q)}^\empt$ is mapped to precisely three repairs of $R_{D,\dep,h(Q)}^\nempt$, since three is the size of the block of $R(a_1,b_1)$. Moreover, in this case, $2^m-1=1$, and a single repair of $R_{D,\dep,h(Q)}^\empt$ is mapped to every repair of $R_{D,\dep,h(Q)}^\nempt$.

Since all the facts of a block are symmetric with each other, precisely $\frac{1}{3}$ of the repairs in $R_{D,\dep,h(Q)}^\nempt$ contain the fact $R(a_1,b_1)$---three repairs. Thus, it holds that
\[
|\{E\in \copr{D}{\dep}\mid h(Q)\subseteq E\}|\ =\ \frac{1}{3}\times 9=3\]
and
\[|\copr{D}{\dep}|\ =\ 12.\]
We conclude that
\[\frac{|\{E\in \copr{D}{\dep}\mid h(Q)\subseteq E\}|}{|\copr{D}{\dep}|}\ =\ \frac{3}{12}\ =\ \frac{1}{4}.\]
Note that
\[\frac{1}{(2|D|)^{|Q|}}\ =\ \frac{1}{12}\]
is indeed a lower bound on that value, and it is also a lower bound on the ratio $\orfreq{\dep,Q}{D,(b_1)}$ that, in this case, equals $\frac{1}{4}$. \hfill\markfull
\end{example}

\subsection{Proof of Item~(3) of Theorem~\ref{the:uniform-repairs}}
As discussed in the main body of the paper, the proof of item (3) of Theorem~\ref{the:uniform-repairs} proceeds in two main steps, which correspond to Proposition~\ref{pro:ur-keys-no-fpras} and Lemma~\ref{lem:from-fds-to-keys}.
Before giving the formal proofs, we first need to introduce some auxiliary notions and results. In the sequel, we concentrate on undirected graphs without self-loops.

\subsubsection*{Auxiliary Notions and Results}
Consider an undirected graph $G = (V, E)$ and an integer $\Delta \ge 0$. We say that $G$ has {\em degree} $\Delta$ if each node of $V$ participates in at most $\Delta$ edges. Moreover, $G$ is connected if there is a path between every two nodes of $G$. We call $G$ {\em trivially connected} if $|V| \le 1$; otherwise, it is {\em non-trivially connected}. Finally, $\IS(G)$ denotes the set of \emph{all} independent sets of $G$.

For a database $D$ and a set $\dep$ of FDs, the \emph{conflict graph of $D$ w.r.t.\ $\dep$} is the undirected graph $\cg{D}{\dep} = (V,E)$, where $V=D$, and $\{f,g\} \in E$ if $\{f,g\} \not \models \dep$. We call $D$ non-trivially (resp., trivially) $\dep$-connected if $\cg{D}{\dep}$ is non-trivially (resp., trivially) connected.

In order to prove the desired claims, we establish an auxiliary result that relates the number of candidate repairs of an inconsistent database that is non-trivially connected with the number of independent sets of the underlying conflict graph; this is Lemma~\ref{lem:corepairs-independent-sets} in the main body, which we recall and prove here:

\begin{manuallemma}{\ref{lem:corepairs-independent-sets}}
	Consider a non-trivially $\dep$-connected database $D$, where $\dep$ is a set of FDs. It holds that $|\copr{D}{\dep}| = |\IS(\cg{D}{\dep})|$.
\end{manuallemma}
\begin{proof}	
	$(\subseteq)$ Consider a candidate repair $D' \in \copr{D}{\dep}$. By definition, $D'$ is consistent, i.e., there are no two facts $f,g \in D'$ such that $\{f,g\} \not\models \dep$. By definition of the conflict graph of $D$ w.r.t.~$\dep$, we conclude that no two facts $f,g \in D'$ are connected via an edge in $\cg{D}{\dep}$. Hence, $D'$ is an independent set of $\cg{D}{\dep}$.
	
	$(\supseteq)$ Consider now an independent set $D' \in \IS(\cg{D}{\dep})$. Since there are no two facts $f,g \in D'$ that are connected via an edge of $\cg{D}{\dep}$, $D'$ is consistent w.r.t.\ $\dep$. It remains to show that there exists a sequence $s \in \crs{D}{\dep}$ such that $s(D) = D'$; we distinguish the two cases where either $D'\neq \emptyset$ or $D'=\emptyset$.
	
	\medskip
	\noindent \textbf{Case 1.}
	Let us first concentrate on the case where $D' \neq \emptyset$. In order to define the repairing sequence $s \in \crs{D}{\dep}$ such that $s(D)=D'$, we first define a convenient stratification of the facts of $D$. We inductively define the strata $L_0,L_1,\ldots$ as follows: 
	\begin{itemize}
		\item $L_0 = D'$. 
		\item For each $i \geq 1$,
	\begin{align*}
		L_i = \{f\in D\ |\ &f\not\in L_0\cup\ldots\cup L_{i-1}
		\textrm{ and }\\ 
		&\textrm{there is } f'\in L_{i-1}\textrm{ with }\{f,f'\} \not \models \dep\}.
	\end{align*}
	\end{itemize}
	Observe that, since $\cg{D}{\dep}$ is connected, each fact $f \in D$ occurs in some $L_i$, i.e., if $n$ is the smallest integer such that $L_\ell = \emptyset$, for $\ell > n$, we have that $D = \bigcup^n_{i = 0} L_i$.
	
	Let $L_i = \left\{f^i_1,\ldots,f^i_{|L_i|}\right\}$, for each $i \in [n]$. We now construct the desired sequence $s$ as follows. We let the first $|L_n|$ operations be $-f^n_1,\ldots,-f^n_{|L_n|}$. To see that $-f^n_j$ is a $(D^{s}_{j-1},\dep)$-justified operation for every $1\leq j\leq |L_n|$, observe that, by definition of $L_n$, each $f^n_j$ is in a violation with some fact in $L_{n-1}$, which has not been removed yet. The operations $\mathit{op}_{|L_n|+1},\ldots,\mathit{op}_{|L_n|+|L_{n-1}|}$ will be $-f^{n-1}_1,\ldots,-f^{n-1}_{|L_{n-1}|}$, which are all justified because there exists a violation for each fact with some fact from $L_{n-2}$ that has not been removed yet. The next operations of $s$ are defined in the same way for the remaining strata, until the last $|L_1|$ operations, which will be $-f^1_1,\ldots,-f^1_{|L_1|}$. Again, these operations are justified since, by definition, each of the facts $f^1_1,\ldots,f^1_{|L_1|}$ is in conflict with some fact from $L_0=D'$. 
	Summing up, the sequence $s$ is 
	\[ -f^n_1,\ldots,-f^n_{|L_n|},-f^{n-1}_1,\ldots,-f^{n-1}_{|L_{n-1}|},\ldots,-f^1_1,\ldots,-f^1_{|L_1|}.
	\]
	We have that $s \in \rs{D}{\dep}$ and that $s(D) = D' \models \dep$. Hence, $s \in \crs{D}{\dep}$, which implies
	$D'\in\copr{D}{\dep}$, as needed. 
	
	\medskip
	\noindent \textbf{Case 2.}
	The case where $D'=\emptyset$ is treated similarly. We only need to slightly adjust the last operation of the sequence $s$. Fix some fact $f^*\in D$. We stratify the facts of $D$ as in the first case, but we let $L_0=\{f^*\}$. We then define $s$ in the same way as above. Let $L_n$ be again the last non-empty stratum. Since $\cg{D}{\dep}$ is non-trivially connected, $D \not \models \dep$, and thus, we have that $n>0$, i.e., at least stratum $L_1$ is non-empty. 
	Then, we have that the first $|L_n|$ operations are $-f^n_1,\ldots,-f^n_{|L_n|}$. We continue with the remaining strata $L_{n-1}$ to $L_2$ as before. The last $|L_1|$ operations are defined as $-f^1_1,\ldots,-f^1_{|L_1|-1},-\left\{f^1_{|L_1|},f^*\right\}$. Note that, by definition of $L_1$, every fact $f^1_1,\ldots,f^1_{|L_1|-1}$ is in a violation with $f^*$, and thus, their removal is a justified operation. Now, there are only two facts left, $f^1_{|L_1|}$ and $f^*$, which together violate $\dep$ (recall that $L_1$ is non-empty). Therefore, $-\left\{f^1_{|L_1|},f^*\right\}$ is a justified operation, and we have that $s \in \rs{D}{\dep}$ and $s(D) = D' = \emptyset \models \dep$. Hence, $s \in \crs{D}{\dep}$, which in turn implies that $D'\in\copr{D}{\dep}$, as needed.
\end{proof}

We are now ready to proceed with the two main steps of the proof of item (3) of Theorem~\ref{the:uniform-repairs}, which correspond to Proposition~\ref{pro:ur-keys-no-fpras} and Lemma~\ref{lem:from-fds-to-keys}, respectively. Note that both results are essentially dealing with the following counting problem for a set $\dep$ of FDs:

\medskip

\begin{center}
	\fbox{\begin{tabular}{ll}
			{\small PROBLEM} : & $\sharp \mathsf{CORep}^{\mathsf{con}}(\dep)$\\
			{\small INPUT} : & A non-trivially $\dep$-connected database $D$.\\
			{\small OUTPUT} : &  The number $|\copr{D}{\dep}|$.
	\end{tabular}}
\end{center}

\medskip

\subsubsection*{Step 1: An Inapproximability Result About Keys.} The formal statement, already given in the main body, and its proof follow. Note that the statement of Proposition~\ref{pro:ur-keys-no-fpras} given below is more compact than the one given in the main body of the paper since we explicitly use the name of the problem $\sharp \mathsf{CORep}^{\mathsf{con}}(\dep)$.

\OMIT{
\begin{manualproposition}{\ref{pro:ur-keys-no-fpras}}
		Unless ${\rm RP} = {\rm NP}$, there exists a set $\dep$ of keys over $\{R\}$ such that, given a non-trivially $\dep$-connected database $D$, the problem of computing $|\copr{D}{\dep}|$ does not admit an FPRAS.
\end{manualproposition}

\noindent Proposition~\ref{pro:ur-keys-no-fpras} essentially states the following: unless ${\rm RP} = {\rm NP}$, there exists a set $\dep$ of keys over $\{R\}$ such that $\sharp \mathsf{CORep}^{\mathsf{con}}(\dep)$ does not admit an FPRAS. 
To this end, we prove that, unless RP = NP, there is a set $\dep_K$ of keys such that $\sharp \mathsf{CORep}^{\mathsf{con}}(\dep_K)$ does not admit an FPRAS by showing that such an FPRAS could be converted into an FPRAS for counting the independent sets of non-trivially connected graphs of bounded degree. The proof will exploit Lemma~\ref{lem:corepairs-independent-sets}, and another auxiliary lemma stating that, unless ${\rm RP} = {\rm NP}$, counting the independent sets of non-trivially connected graphs with degree $\Delta \ge 6$ does not admit an FPRAS.

\begin{enumerate}
\item We prove that, unless RP = NP, there is a set $\dep_K$ of keys such that $\sharp \mathsf{CORep}^{\mathsf{con}}(\dep_K)$ does not admit an FPRAS by showing that such an FPRAS could be converted into an FPRAS for counting the independent sets of non-trivially connected graphs of bounded degree. The proof will exploit Lemma~\ref{lem:corepairs-independent-sets}, and another auxiliary lemma stating that, unless ${\rm RP} = {\rm NP}$, counting the independent sets of non-trivially connected graphs with degree $\Delta \ge 6$ does not admit an FPRAS.

\item We then prove that, unless RP = NP, there exists a set $\dep_F$ of FDs and a CQ $Q_F$ such that $\rrelfreq{\dep_F,Q_F}$ does not admit an FPRAS by showing that such an FPRAS could be converted into an FPRAS for $\sharp \mathsf{CORep}^{\mathsf{con}}(\dep_K)$, which would contradict This proof also exploits Lemma~\ref{lem:corepairs-independent-sets}.
\end{enumerate}

We start with the first part of the proof, which is Proposition~\ref{pro:ur-keys-no-fpras} from the main body; here we restate it using the notation we introduced for the problem $\#\mathsf{CORep}^{\mathsf{con}}(\dep)$.

}

\begin{manualproposition}{\ref{pro:ur-keys-no-fpras}}
Unless ${\rm RP} = {\rm NP}$, there exists a set $\dep$ of keys over $\{R\}$ such that $\sharp \mathsf{CORep}^{\mathsf{con}}(\dep)$ does not admit an FPRAS.
\end{manualproposition}


Before giving the proof of Proposition~\ref{pro:ur-keys-no-fpras}, we need an auxiliary result about the problem of counting the number of independent sets of undirected graphs.
For an integer $\Delta \geq 0$, we define

\medskip

\begin{center}
	\fbox{\begin{tabular}{ll}
			{\small PROBLEM} : & $\sharp \IS_\Delta$\\
			{\small INPUT} : & An undirected graph $G$ of degree $\Delta$.\\
			{\small OUTPUT} : &  The number $|\IS(G)|$.
	\end{tabular}}
\end{center}

\medskip

We know from~\cite{Sly2010} that the following holds:

\begin{proposition}\label{pro:count-is}
	For every $\Delta \geq 6$, unless ${\rm RP} = {\rm NP}$, $\#\IS_\Delta$ does not admit an FPRAS.
\end{proposition}

Note that the above result states the inapproximability of $\sharp \IS_\Delta$ for arbitrary, not necessarily non-trivially connected graphs. However, for showing Proposition~\ref{pro:ur-keys-no-fpras}, we need the stronger version of Proposition~\ref{pro:count-is} that establishes the inapproximability of $\sharp \IS_\Delta$ even for non-trivially connected graphs.
Let $\sharp \ISC_\Delta$ be the problem defined as $\sharp \IS_\Delta$ with the difference that the input is a non-trivially connected undirected graph. We proceed to show the following:

\begin{lemma}\label{lem:is-con-inapprox}
	For every $\Delta \geq 6$, unless ${\rm RP} = {\rm NP}$, $\#\IS^{\mathsf{con}}_\Delta$ does not admit an FPRAS.
\end{lemma}

\OMIT{
\begin{proof}

	The problem $\#\IS_\Delta$, for some integer $\Delta \ge 0$, takes as input an undirected graph $G$ of degree $\Delta$, and asks for the number $|\IS(G)|$. It is well-known that, unless ${\rm RP} = {\rm NP}$, $\#\IS_\Delta$ does not admit an FPRAS, for all $\Delta \ge 6$.~\cite{Sly2010}
	To prove the claim, we rely on a reduction from a restricted version of $\#\IS_\Delta$.
	In particular, we first show that for each $\Delta \ge 6$, $\#\IS_\Delta$ remains inapproximable  even for non-trivially connected graphs. We denote this restriction of $\#\IS_\Delta$ as $\#\ISC_\Delta$. Although this result can be considered folklore, we could not find an explicit reference for it, and thus we prove it explicitly.
	
	\begin{lemma}\label{lem:is-con-inapprox}
		Unless ${\rm RP} = {\rm NP}$, $\#\IS^{\mathsf{con}}_\Delta$ has no FPRAS, for all $\Delta \ge 6$.
	\end{lemma}
}

	\begin{proof}
		By contradiction, assume that $\#\IS^{\mathsf{con}}_\Delta$ admits an FPRAS, for some $\Delta \ge 6$, i.e., there is a randomized algorithm $\mathsf{A}$ that takes as input a non-trivially connected graph $G = (V, E)$ of degree $\Delta$, $\epsilon > 0$, and $0 < \delta < 1$, runs in polynomial time in $||G||$, $1/\epsilon$, and $\log(1/\delta)$, and produces a random variable $\mathsf{A}(G,\epsilon,\delta)$ such that
		\[\pr\left((1-\epsilon)\cdot |\IS(G)| \leq \mathsf{A}(G,\epsilon,\delta) \leq (1+\epsilon) \cdot |\IS(G)|\right)\ \geq\ 1-\delta.\]
		From this, we can construct an FPRAS $\mathsf{A}'$ for $\#\IS_\Delta$ as follows. Given a graph $G=(V,E)$ of degree $\Delta$, let the connected components, i.e., the maximal connected subgraphs, of $G$ be $(\CC_i)_{1\leq i\leq n}$ with $\CC_i = (V_i,E_i)$.
		Furthermore, assume, w.l.o.g., that $\CC_1,\ldots,\CC_\ell$, are all the trivially connected components of $G$, for some $\ell \le n$.
		Given $G$, $\epsilon > 0$, and $0 < \delta < 1$, $\mathsf{A}'$ is defined as
		\[
		\mathsf{A}'(G,\epsilon,\delta)\ =\  2^\ell \cdot \prod\limits^n_{i=\ell+1} \mathsf{A}\left(\CC_i,\frac{\epsilon}{2n}, \frac{\delta}{2n}\right).
		\]
		Note that a run of $\mathsf{A}$ does not depend on any other run, and thus, the random variables $\mathsf{A}(\CC_i, \frac{\epsilon}{2n}, \frac{\delta}{2n})$ are independent from each other. It is also easy to see that 
		\[
		|\IS(G)|\ =\ 2^\ell \cdot \prod_{i=\ell+1}^{n}|\IS(\CC_i)|. 
		\]
		Therefore, since $\mathsf{A}$ is an FPRAS for $\#\ISC_\Delta$, we have that
		\begin{multline*}
			\pr\left(\left(1-\frac{\epsilon}{2n}\right)^n\cdot |\IS(G)| \leq \mathsf{A'}(G,\epsilon,\delta) \leq \right.\\\left. \left(1+\frac{\epsilon}{2n}\right)^n \cdot |\IS(G)|\right) \ge \left(1-\frac{\delta}{2n}\right)^n.
		\end{multline*}
		Finally, we know (see, e.g.,~\cite{Jerrum1986}) that
		the following inequalities hold: for $0\leq x\leq 1$ and $m\geq 1$,
		\[
		1 - x \leq \left(1-\frac{x}{2m}\right)^m \quad \text{and} \quad \left(1+\frac{x}{2m}\right)^m \leq 1+x.
		\]
		Consequently, 
		\[
			\pr\left((1-\epsilon)\cdot |\IS(G)| \leq \mathsf{A'}(G,\epsilon,\delta) \leq (1+\epsilon) \cdot |\IS(G)|\right) \ge 1 - \delta.
		\]
		Hence, $\mathsf{A}'$ fulfils the probabilistic guarantees required for an FPRAS. To confirm the desired running time of $\mathsf{A'}$, note that there are at most $n=|V|$ connected components of $G$, which can be computed in polynomial time via any textbook algorithm. Thus, since $\mathsf{A}$ is an FPRAS for $\#\ISC_\Delta$, for each each $i \in [n]$, the random variable $\mathsf{A}(\CC_i,\frac{\epsilon}{2n},\frac{\delta}{2n})$ can be computed in polynomial time w.r.t.\ $\CC_i$, $\frac{|V|}{\epsilon}$, and $\frac{|V|}{\delta}$. Since $\mathsf{A}'$ multiplies at most $|V|$ such random variables, $\mathsf{A'}$ is an FPRAS for $\sharp \IS_\Delta$, which contradicts Proposition~\ref{pro:count-is}.
	\end{proof}

	With Lemma~\ref{lem:is-con-inapprox} in place, we can now prove Proposition~\ref{pro:ur-keys-no-fpras}.
	
	\medskip

	\textsc{Proof of Proposition~\ref{pro:ur-keys-no-fpras}.}
	Consider a non-trivially connected undirected graph $G=(V,E)$ with degree $\Delta=6$. Let $\ins{S}$ be the schema consisting of the single relation name $\{R/\Delta+1\}$ with $(A_1,\ldots,A_{\Delta+1})$ being the tuple of attributes of $R$, and $\dep_K = \{\phi_1,\ldots,\phi_{\Delta+1}\}$ a set of keys over $\ins{S}$, where $\phi_i = R : A_i \ra \att{R}$ for each $i \in [\Delta +1]$.
	We show that, unless ${\rm RP} = {\rm NP}$, 
	we can construct a non-trivially $\dep_K$-connected database $D_G$ in polynomial time in $||G$|| such that $|\IS(G)| = |\copr{D_G}{\dep_K}|$. Hence, the existence of an FPRAS for $\sharp \mathsf{CORep}^{\mathsf{con}}(\dep_K)$ would imply the existence of an FPRAS for $\sharp \ISC_\Delta$, which in turn contradicts Lemma~\ref{lem:is-con-inapprox}.

	The key property that $D_G$ should enjoy is the following: there exists a bijection $\mu : V \ra D_G$ from the set of nodes of $G$ to the facts of $D_G$ such that $(u,v) \in E$ iff $\{\mu(u),\mu(v)\} \not \models \dep_K$. The latter immediately implies that $|\IS(G)| = |\IS(\cg{D_G}{\dep_K})|$, which, together with the fact that $G$ is non-trivially connected, and hence, $D_G$ is non-trivially $\dep_K$-connected, implies that $|\IS(G)| = |\copr{D_G}{\dep_K}|$ by Lemma~\ref{lem:corepairs-independent-sets}. The formal construction of $D_G$ follows.
	
	\medskip
	\noindent \paragraph{The Database $D_G$.}
	%
	It is known that the edges of $G$ are $(\Delta+1)$-colourable, and such a coloring can be constructed in polynomial time in $||G||$~\cite{Misra1992}. Therefore, we are able to efficiently assign the colours $C=\{c_1,\ldots, c_{\Delta+1}\}$ to the edges of $G$ in such a way that none of the nodes belongs to two distinct edges of the same colour. Let $M : E\rightarrow C$ be such a coloring.
	%
	The database $D_G$ is such that $\adom{D_G} = E \cup F$,
	where $F$ is a finite a set of constants with $E \cap F = \emptyset$, and has the following facts:
	\begin{enumerate}
		\item for each node $v \in V$, we add to $D_G$ a fact of the form $R(a^v_1,\ldots,a^v_{\Delta+1})$, and
		\item the constants of such facts are defined as follows:
		\begin{enumerate}
			\item for every edge $e=\{u,v\}\in E$, assuming that $M(e)=c_i$, we let $a^{u}_i=a^{v}_i=e$, and \label{secstep}
			\item for every $a^v_i$ not defined in the above step, we let $a^v_i=f$ for some constant $f\in F$ only to be used once.
		\end{enumerate}
	\end{enumerate}
	We use $R(\bar{a}^v)$ to denote $R(a^v_1,\ldots,a^v_{\Delta+1})$, for short. Note that every $a^v_i$ is well-defined since $v$ has at most one edge of colour $c_i$, and thus, $a^v_i$ is uniquely defined by either the edge with colour $c_i$, or in case there is no such edge, by a fresh constant $f\in F$. Let us also stress that we can build $D_G$ in polynomial time in $||G||$.
	We can now prove the following crucial property of $D_G$:
	
	\begin{lemma}\label{lem:conflict-equiv}
		Consider two nodes $u,v$ of $G$. Then, $\{u,v\}$ is an edge in $G$ iff $\{R(\bar{a}^u),R(\bar{a}^v)\} \not \models \dep_K$.
	\end{lemma}
	\begin{proof}
		Consider an edge $e=\{u,v\}$ in $G$ with $M(e)=c_i$, for some $i \in [\Delta+1]$. By construction of $D_G$, it is clear that there will be exactly two facts in $D_G$ such that the constant $e$ appears at position $i$ of those facts. These two facts will be precisely $R(a^u_1,\ldots,a^u_{\Delta+1})$ and $R(a^v_1,\ldots,a^v_{\Delta+1})$, having $a^u_i = a^v_i = e$. As there are no multiple edges between two vertices, the constants at the other positions will be pairwise different, i.e., $a^u_j \neq a^v_j$, for all $j \neq i$. Hence, 
		the two facts together violate $\phi_i$, and thus, $\{R(\bar{a}^u),R(\bar{a}^v)\} \not \models \dep_K$.
		
		Consider now two facts $R(a_1^{u},\ldots,a_{\Delta+1}^u)$ and $R(a_1^{v},\ldots,a_{\Delta+1}^v)$ that together violate $\phi_i = R: A_i \ra \att{R}$ for some $i \in [\Delta+1]$. Thus, the same constant appears at position $i$ in both facts, i.e., $a_i^u = a_i^v$. By construction of $D_G$, the only reason why $a^u_i = a^v_i$ is because $\{u,v\}$ is an edge of $G$, and the claim follows.
	\end{proof}
	
By Lemma~\ref{lem:conflict-equiv}, we get that $|\IS(G)| = |\IS(\cg{D_G}{\dep_K})|$, and that $D_G$ is non-trivially $\dep$-connected.
	Hence, by Lemma~\ref{lem:corepairs-independent-sets}, $|\IS(G)| = |\IS(\cg{D_G}{\dep_K})| = |\copr{D_G}{\dep_K}|$, and the claim follows. \qed

\subsubsection*{Step 2: Transferring FPRAS}
We proceed with the second and last step of the proof of item (3) of Theorem~\ref{the:uniform-repairs}, which corresponds to Lemma~\ref{lem:from-fds-to-keys}. The formal statement, already given in the main body, and its proof follow. Note that the statement of Lemma~\ref{lem:from-fds-to-keys} given below is more compact than the one given in the main body since we explicitly use the name of the problem $\sharp \mathsf{CORep}^{\mathsf{con}}(\dep_K)$, where $\dep_K$ is the set of keys provided by Proposition~\ref{pro:ur-keys-no-fpras}.


\begin{manuallemma}{\ref{lem:from-fds-to-keys}}
Assume that $\rrelfreq{\dep,Q}$ admits an FPRAS, for every set $\dep$ of FDs and CQ $Q$. Then, $\sharp \mathsf{CORep}^{\mathsf{con}}(\dep_K)$ admits an FPRAS.
\end{manuallemma}
\begin{proof}	
	By Proposition~\ref{pro:ur-keys-no-fpras}, unless RP = NP, $\dep_K$ is a set of keys over a schema $\{R/n\}$ such that $\sharp \mathsf{CORep}^{\mathsf{con}}(\dep_K)$ does not admit an FPRAS. 
	Let $\ins{S} = \{R'/m\}$, where $m = n+2$, and, assuming that $(A_1,\ldots,A_n)$ is the tuple of attributes of $R$, let $(A,B,A_1,\ldots,A_n)$ be the tuple of attributes of $R'$.
	We first show that there exist a set $\dep_F$ of FDs over $\ins{S}$, and a Boolean CQ $Q_F$ over $\ins{S}$ such that, for every non-trivially $\dep_K$-connected database $D$ over $\{R\}$, we can construct in polynomial time in $||D||$ a database $D_F$ over $\ins{S}$ such that 
	\begin{equation}\label{propA'}
	\orfreq{\dep_F,Q_F}{D_F,()}\ =\ \frac{1}{|\copr{D}{\dep_K}|+1}.\tag{$*$}
	\end{equation}
	By exploiting the above equation, the fact that $D_F$ can be constructed in polynomial time, and the FPRAS for $\rrelfreq{\dep_F,Q_F}$ (which exists by hypothesis), we will then explain how to devise an FPRAS for the problem $\sharp \mathsf{CORep}^{\mathsf{con}}(\dep_K)$.
	
	We start by explaining how $\dep_F$ and $D_F$ are defined in a way that	
	\[
	|\copr{D_F}{\dep_F}|\ =\ |\copr{D}{\dep_K}|+1.
	\] 
	We define the set $\dep_F$ of FDs over $\ins{S}$ as follows:
	\[
	\{R' : X \ra Y \mid R : X \ra Y \in \dep_K\} \cup \{R' : A \ra B\}.
	\] 
	Note that each key $\phi$ of $\dep_K$ over $R$ becomes an FD $\phi'$ over $R'$; indeed, $\phi'$ is not a key since $R'$ has two additional attributes.
	Now, given a non-trivially $\dep_K$-connected database $D$ over $\{R\}$, we define the database $D_F$ as follows with $a,b$ being constants not in $\adom{D}$:
	\[
	\{R'(a,b,a_1,\ldots,a_n) \mid R(a_1,\ldots,a_n) \in D\} \cup \{R'(a,a,\ldots,a)\}.
	\]
	For brevity, we will write $f^*$ for the fact $R'(a,a,\ldots,a)$.
	It is not difficult to verify that the number $|\copr{D_F}{\dep_F}|$ is the sum 
	\begin{multline*}
	|\{D' \in \copr{D_F}{\dep_F} \mid f^* \in D'\}|\ + \\
	|\{D' \in \copr{D_F}{\dep_F} \mid f^* \not \in D'\}|,
	\end{multline*}
	that is, the sum of the number of candidate repairs containing $f^*$ and the number of candidate repairs not containing $f^*$.
	It is not difficult to see that $\{f^*\}$ is the only candidate repair containing $f^*$. This is because $f^*$ is in a conflict with every other fact of $D_F$ due to the FD $R' : A \ra B$. Moreover, one can easily devise a sequence $s \in \crs{D_F}{\dep_F}$ removing all facts in $D_F \setminus \{f^*\}$ in an arbitrary order, and therefore obtaining $\{f^*\}$.

	Regarding the number of candidate repairs not containing $f^*$, observe that since $D$ is non-trivially $\dep_K$-connected, and since $f^*$ is in a conflict with every other fact of $D_F$, then $D_F$ is non-trivially $\dep_F$-connected. Therefore, by Lemma~\ref{lem:corepairs-independent-sets}, $\copr{D_F}{\dep_F} = \IS(\cg{D_F}{\dep_F})$.
	Since $\{f^*\}$ is the only candidate repair of $D_F$ containing $f^*$, and thus, the only independent set of $\cg{D_F}{\dep_F}$ containing $f^*$, the set of candidate repairs without $f^*$, i.e., $\{D' \in \copr{D_F}{\dep_F} \mid f^* \not \in D'\}$ coincides with $\IS(\cg{D_F \setminus \{f^*\}}{\dep_F})$. 
	Note that, by construction of $D_F$ and $\dep_F$, since $D$ is non-trivially $\dep_K$-connected, $D_F \setminus \{f^*\}$ is non-trivially $\dep_F$-connected, and thus, by Lemma~\ref{lem:corepairs-independent-sets}, $\IS(\cg{D_F \setminus \{f^*\}}{\dep_F}) = \copr{D_F \setminus \{f^*\}}{\dep_F}$.
	
	Finally, by construction of $D_F$ and $\dep_F$, we have that $|\copr{D_F\setminus \{f^*\}}{\dep_F}|=|\copr{D}{\dep_K}|$. In fact, it suffices to observe that two facts $R(a_1,\ldots,a_n), R(b_1,\ldots,b_n)\in D$ violate $\dep_K$ iff the corresponding facts $R'(a,b,a_1,\ldots,a_n), R'(a,b,b_1,\ldots,b_n)\in D_F\setminus \{f^*\}$ violate $\dep_F$.
	%
	%
	Hence, we conclude that
	\[
	|\copr{D_F}{\dep_F}|\ =\ |\copr{D}{\dep_K}| + 1.
	\]

	Let us now define the Boolean CQ $Q_F$ in such a way that the equation (\ref{propA'}) holds.
	We define $Q_F$ as the Boolean CQ 
	\[
	\textrm{Ans}()\ \text{:-}\ R'(x,x,\ldots,x).
	\]
	In simple words, $Q_F$ asks whether there exists a fact such that all the attributes have the same value. Clearly, the only candidate repair of $\copr{D_F}{\dep_F}$ that satisfies the query $Q_F$ is $\{f^*\}$, i.e.,
	\[
	\orfreq{\dep_F,Q_F}{D_F,()}\ =\ \frac{1}{|\copr{D_F}{\dep_F}|}.
	\]
	Since, as shown above, $|\copr{D_F}{\dep_F}| = |\copr{D}{\dep_K}| + 1$, we get that the equation (\ref{propA'}) holds, as needed.
	
	\OMIT{
	Thus, from the fact that $|\copr{D_F}{\dep_F}| = 1 + |\copr{D}{\dep_K}|$, the expression above follows.
	
	
	Our goal is to show that there exists a set $\dep_F$ of FDs over a schema $\ins{S} = \{R'/n\}$, where $n = m+2$, such that, for every non-trivially $\dep_K$-connected database $D$ over $\{R\}$, we can construct in polynomial time in $||D||$ a database $D_F$ over $\ins{S}$ such that 
	\[
	|\copr{D_F}{\dep_F}|\ =\ |\copr{D}{\dep_K}|+1.
	\] 
	
	Assume, w.l.o.g., that the tuple of attributes of $R$ is $(A_1,\ldots,A_n)$, and let $\ins{S} = \{R'/n+2\}$, where the tuple of attributes of $R'$ is $(A,B,A_1,\ldots,A_n)$.
	
	Consider now a non-trivially $\dep_K$-connected database $D$ over $\{R\}$, and let $a,b$ be two constants not occurring in $D$.
	
	The database $D_F$ is obtained from $D$ by replacing each fact of the form $R(a_1,\ldots,a_n)$ in $D$ with a fact of the form $R'(a,b,a_1,\ldots,a_n)$, and by adding the fact $f^*=R'(a,a,\ldots,a)$.
	
	Finally, we let $\dep_F = \{R' : \phi \mid R: \phi \in \dep_K\} \cup \{R': A \ra B\}$. 
	Note that each key $\dep_K$ of $R$ becomes an FD of $R'$ and this FD is not a key for $R'$.
	Then, number $|\copr{D_F}{\dep_F}|$ can be seen as the sum 
	\begin{multline*}
		|\{D' \in \copr{D_F}{\dep_F} \mid f^* \in D'\}| + \\
		|\{D' \in \copr{D_F}{\dep_F} \mid f^* \not \in D'\}|.
	\end{multline*}
	That is, the sum of the number of candidate repairs containing $f^*$ and the number of candidate repairs not containing $f^*$.
	
	It is easy to see that $\{f^*\}$ is the only candidate repair containing $f^*$, since $f^*$ is in a conflict with every other fact of $D_F$, because of $R' : A \ra B$. Moreover, one can easily devise a sequence $s \in \crs{D_F}{\dep_F}$ removing all facts in $D_F \setminus \{f^*\}$ in an arbitrary order, and thus obtain $\{f^*\}$.
	
	Regarding the number of candidate repairs not containing $f^*$, observe that since $D$ is non-trivially $\dep_K$-connected, and since $f^*$ is in a conflict with every other fact of $D_F$, then $D_F$ is non-trivially $\dep_F$-connected. Then, by Lemma~\ref{lem:corepairs-independent-sets}, $\copr{D_F}{\dep_F} = \IS(\cg{D_F}{\dep_F})$.
	Since $\{f^*\}$ is the only candidate repair of $D_F$ containing $f^*$, and thus the only independent set of $\cg{D_F}{\dep_F}$ containing $f^*$, the set of candidate repairs without $f^*$, i.e., $\{D' \in \copr{D_F}{\dep_F} \mid f^* \not \in D'\}$ coincides with $\IS(\cg{D_F \setminus \{f^*\}}{\dep_F})$. Note that by construction of $D_F$ and $\dep_F$, since $D$ is non-trivially $\dep_K$-connected, $D_F \setminus \{f^*\}$ is non-trivially $\dep_F$-connected, and thus, by Lemma~\ref{lem:corepairs-independent-sets}, $\IS(\cg{D_F \setminus \{f^*\}}{\dep_F}) = \copr{D_F \setminus \{f^*\}}{\dep_F}$.
	
	Finally, by construction of $D_F$ and $\dep_F$, we have that $|\copr{D_F\setminus \{f^*\}}{\dep_F}|=|\copr{D}{\dep_K}|$. In fact, it suffices to observe that two facts $R(a_1,\ldots,a_n), R(b_1,\ldots,b_n)\in D$ violate $\dep_K$ if and only if the corresponding facts $R'(a,b,a_1,\ldots,a_n), R'(a,b,b_1,\ldots,b_n)\in D_F\setminus \{f^*\}$ violate $\dep_F$.
	%
	%
	Hence, we conclude
	$$ |\copr{D_F}{\dep_F}| = 1 + |\copr{D}{\dep_K}|.$$

	We are now only left to devise the desired query. That is, we target a Boolean CQ $Q_F$ such that
	\begin{equation}\label{propA'}
		\orfreq{\dep_F,Q_F}{D_F,()} =\frac{1}{1+|\copr{D}{\dep_K}|}.\tag{$*$}
	\end{equation}
	Let $Q_F$ be the query $\textrm{Ans}()\ \text{:-}\ R'(x,x,\ldots,x)$, i.e. $Q_F$ asks whether there exists a fact such that all the attributes coincide. Clearly, the only candidate repair of $\copr{D_F}{\dep_F}$ which satisfies the query $Q_F$ is $\{f^*\}$. Thus, from the fact that $|\copr{D_F}{\dep_F}| = 1 + |\copr{D}{\dep_K}|$, the expression above follows.
}

	\medskip
	\noindent \paragraph{Building the FPRAS.} We proceed to devise an FPRAS for the problem $\sharp \mathsf{CORep}^{\mathsf{con}}(\dep_K)$ by exploiting the equation (\ref{propA'}), the fact that $D_F$ can be constructed in polynomial time, and the FPRAS $\mathsf{A}'$ for $\rrelfreq{\dep_F,Q_F}$ (which exists by hypothesis).


	Given a non-trivially $\dep_K$-connected database $D$, $\epsilon > 0$, and $0 < \delta < 1$, we define $\mathsf{A}$ as the following randomized procedure:
	\begin{enumerate}
		\item Compute $D_F$ from $D$;
		\item Let $\epsilon' = \frac{\epsilon}{2+\epsilon}$;
		\item Let $r = \max\left\{\frac{1-\epsilon'}{2\cdot(1+2^{|D|})},\mathsf{A}'(D_F,(),\epsilon',\delta)\right\}$;
		\item Output $\frac{1}{r} - 1$.
	\end{enumerate}
	
	We proceed to show that $\mathsf{A}$ is an FPRAS for $\sharp \mathsf{CORep}^{\mathsf{con}}(\dep_K)$. Since $D_F$ can be constructed in polynomial time in $||D||$, $\mathsf{A}(D,\epsilon,\delta)$ runs in polynomial time in $||D||$, $1/\epsilon$ and $\log(1/\delta)$ by definition.
	We now discuss the probabilistic guarantees.
	By assumption,
	\begin{multline*}
		\pr\left((1-\epsilon')\cdot\orfreq{\dep_F,Q_F}{D_F,()} \leq \mathsf{A}'(D_F,(),\epsilon',\delta) \right.\\ 
		\left.\leq (1+\epsilon') \cdot \orfreq{\dep_F,Q_F}{D_F,()}\right) \ge 1 - \delta.
	\end{multline*}
	Thus, it suffices to show that the left-hand side of the above inequality is bounded from above by
	\begin{multline*}
		\pr\left((1-\epsilon)\cdot|\copr{D}{\dep_K}|\leq \mathsf{A}(D,\epsilon,\delta) \leq\right.\\ 
		\left. (1+\epsilon) \cdot |\copr{D}{\dep_K}|\right).
	\end{multline*} 
	To this end, by equation (\ref{propA'}), we get that
	\begin{multline*}
		\pr\left((1-\epsilon')\cdot\orfreq{\dep_F,Q_F}{D_F,()} \leq \mathsf{A'}(D_F,(),\epsilon',\delta) \right.\\ 
		\left.\leq (1+\epsilon') \cdot \orfreq{\dep_F,Q_F}{D_F,()}\right) = \pr(E),
	\end{multline*}
	where $E$ is the event
	\begin{multline*}
		\frac{1-\epsilon'}{1+|\copr{D}{\dep_K}|} \leq \mathsf{A'}(D_F,(),\epsilon',\delta) \leq \frac{1+\epsilon'}{1+|\copr{D}{\dep_K}|}.
	\end{multline*}
	Note that $|\copr{D}{\dep_K}| \le 2^{|D|}$, i.e., $|\copr{D}{\dep_K}| $ is at most the the number of all possible subsets of $D$. Hence,
	\[
	\frac{1-\epsilon'}{1 + |\copr{D}{\dep_K}|} \ge \frac{1-\epsilon'}{1+2^{|D|}}.
	\]
	Thus, for $E$ to hold is necessary that the output of $A'(D_F,(),\epsilon',\delta)$ is no smaller than $\frac{1-\epsilon'}{1+2^{|D|}}$. Hence, for any number $p < \frac{1-\epsilon'}{1+2^{|D|}}$, $E$ coincides with the event
	\begin{multline*}
		\frac{1-\epsilon'}{1+|\copr{D}{\dep_K}|} \leq \max\left\{p, \mathsf{A'}(D_F,(),\epsilon',\delta)\right\}
		\leq\\
		\frac{1+\epsilon'}{1+|\copr{D}{\dep_K}|}.
	\end{multline*}
	Hence, with $p=\frac{1-\epsilon'}{2 \cdot (1+2^{|D|})} < \frac{1-\epsilon'}{1 + 2^{|D|}}$, we conclude that
	\begin{multline*}
		\pr(E)\ =\ \pr\left(\frac{1-\epsilon'}{1+|\copr{D}{\dep_K}|} \leq \max\left\{p, \mathsf{A'}(D_F,(),\epsilon',\delta)\right\}
		\leq \right.\\\left. \frac{1+\epsilon'}{1+|\copr{D}{\dep_K}|}\right).
	\end{multline*}
	Since the random variable $\max\{p,\mathsf{A}'(D_F,(),\epsilon',\delta)\}$ always outputs a rational strictly larger than $0$, the latter probability coincides with
	\begin{multline*}
		\pr\left(\frac{1+|\copr{D}{\dep_K}|}{1+\epsilon'} \leq \frac{1}{\max\left\{p, \mathsf{A'}(D_F,(),\epsilon',\delta)\right\}}
		\leq \right.\\\left. \frac{1+|\copr{D}{\dep_K}|}{1-\epsilon'}\right).
	\end{multline*}
	For short, let $X$ be the random variable $\frac{1}{\max\{p,\mathsf{A}'(D_F,(),\epsilon',\delta)\}}$. Since $\frac{1}{1-\epsilon'} = 1 + \frac{\epsilon'}{1-\epsilon'}$ and $\frac{1}{1+\epsilon} = 1 - \frac{\epsilon'}{1+ \epsilon'} \ge 1 - \frac{\epsilon'}{1-\epsilon'}$, the probability above is less or equal than
	\begin{multline*}
		\pr\left(\left(1-\frac{\epsilon'}{1-\epsilon'}\right) \cdot (1+|\copr{D}{\dep_K}|) \leq X 
		\leq \right.\\\left. \left(1+\frac{\epsilon'}{1-\epsilon'}\right) \cdot (1+|\copr{D}{\dep_K}|)\right).
	\end{multline*}
	If we subtract $1$ from all sides of the inequality, then the above probability coincides with
	\begin{multline*}
		\pr\left(\left(1-\frac{\epsilon'}{1-\epsilon'}\right) \cdot (1+|\copr{D}{\dep_K}|) - 1\leq X - 1 \leq  \right. \\ \left. \left(1+\frac{\epsilon'}{1-\epsilon'}\right) \cdot (1+|\copr{D}{\dep_K}|) - 1\right).
	\end{multline*}
	By expanding the products in the above expression, we obtain
	\begin{multline*}
		\pr\left(|\copr{D}{\dep_K}| - \frac{\epsilon'}{1-\epsilon'} -  \frac{\epsilon'}{1-\epsilon'} \cdot |\copr{D}{\dep_K}| \leq \right.\\
		X - 1 \leq \\ \left. |\copr{D}{\dep_K}| + \frac{\epsilon'}{1-\epsilon'} + \frac{\epsilon'}{1-\epsilon'} \cdot |\copr{D}{\dep_K}|\right).
	\end{multline*}
	Finally, since $|\copr{D}{\dep_K}| \ge 1$, we have that
	\[
	\frac{\epsilon'}{1-\epsilon'}\ \le\ \frac{\epsilon'}{1-\epsilon'} \cdot |\copr{D}{\dep_K}|.
	\] 
	Thus, the above probability is less or equal than
	\begin{multline*}
		\pr\left(|\copr{D}{\dep_K}| - 2 \cdot\frac{\epsilon'}{1-\epsilon'} \cdot |\copr{D}{\dep_K}| \leq \right.\\\left.
		X - 1 \leq  |\copr{D}{\dep_K}| + 2 \cdot \frac{\epsilon'}{1-\epsilon'} \cdot |\copr{D}{\dep_K}|\right),
	\end{multline*}
	which coincides with
	\begin{multline*}
		\pr\left( \left(1-2 \cdot\frac{\epsilon'}{1-\epsilon'}\right) \cdot |\copr{D}{\dep_K}| \leq \right.\\\left.
		X - 1 \leq  \left(1+ 2 \cdot \frac{\epsilon'}{1-\epsilon'}\right) \cdot |\copr{D}{\dep_K}|\right).
	\end{multline*}
	Recalling that $\epsilon' = \frac{\epsilon}{2+\epsilon}$, one can verify that $2 \cdot \frac{\epsilon'}{1-\epsilon'} = \epsilon$. Moreover, $X-1$ is $\mathsf{A}(D,\epsilon,\delta)$. Hence, the above probability coincides with
	\begin{multline*}
		\pr\left( (1-\epsilon) \cdot |\copr{D}{\dep_K}| \leq \mathsf{A}(D,\epsilon,\delta) \leq \right.\\\left.
		(1+ \epsilon) \cdot |\copr{D}{\dep_K}|\right).
	\end{multline*}
	Consequently, $\mathsf{A}$ is an FPRAS for $\#\mathsf{CORep}^{\mathsf{con}}(\dep_K)$, as needed.
\end{proof}

It is now straightforward to see that from Proposition~\ref{pro:ur-keys-no-fpras} and Lemma~\ref{lem:from-fds-to-keys}, we can conclude item~(3) of Theorem~\ref{the:uniform-repairs}.

%% file: app-uniform-sequences.tex
\section{Proofs of Section~\ref{sec:uniform-sequences}}\label{appsec:uniform-sequences}

In this section, we prove the main result of Section~\ref{sec:uniform-sequences}, which we recall here for the sake of readability:

\begin{manualtheorem}{\ref{the:uniform-sequences}}
        \begin{enumerate}
            \item There exist a set $\dep$ of primary keys, and a CQ $Q$ such that $\ocqa{\dep,M_{\dep}^{\us},Q}$ is $\sharp ${\rm P}-hard.
            
            \item For a set $\dep$ of primary keys, and a CQ $Q$, $\ocqa{\dep,M_{\dep}^{\us},Q}$ admits an FPRAS.
        \end{enumerate}
\end{manualtheorem}

As discussed in Section~\ref{sec:uniform-sequences}, we actually need to prove the above result for the problem $\srelfreq{\dep,Q}$.

\subsection{Proof of Item~(1) of Theorem~\ref{the:uniform-sequences}}
Let $\dep$ and $Q$ be the singleton set of primary keys and the Boolean CQ, respectively, for which $\rrelfreq{\dep,Q}$ is $\sharp ${\rm P}-hard; $\dep$ and $Q$ are obtained from the proof of item (1) of  Theorem~\ref{the:uniform-repairs}.
We show that also $\srelfreq{\dep,Q}$ is $\sharp ${\rm P}-hard via a polynomial-time Turing reduction from $\sharp H\text{-}\mathsf{Coloring}$, where $H$ is the graph employed in the proof of item (1) of  Theorem~\ref{the:uniform-repairs}.
Actually, we can exploit the same construction as in the proof of item (1) of Theorem~\ref{the:uniform-repairs}.
Assuming that, for an undirected graph $G$, $D_G$ is the database that the construction in the proof of item (1) of Theorem~\ref{the:uniform-repairs} builds, we show that
\[
\orfreq{\dep,Q}{D_G,()}\ =\ \srfreq{\dep,Q}{D_G,()},
\] 
which implies that the polynomial-time Turing reduction from $\sharp H\text{-}\mathsf{Coloring}$ to $\rrelfreq{\dep,Q}$ is also a polynomial-time Turing reduction from $\sharp H\text{-}\mathsf{Coloring}$ to $\srelfreq{\dep,Q}$.
Recall that
\[
\srfreq{\dep,Q}{D_G,()}\ =\ \frac{|\{s \in \crs{D_G}{\dep} \mid s(D) \models Q\}|}{|\crs{D_G}{\dep}|}.
\]
By construction of $D_G$, each candidate repair $D \in \copr{D_G}{\dep}$ can be obtained via $|V_G|!$ different complete sequences of $\crs{D_G}{\dep}$. Therefore, $\srfreq{\dep,Q}{D_G,()}$ is
\[
\frac{|\copr{D_G}{\dep,Q}| \cdot |V_G|!}{|\copr{D_G}{\dep}| \cdot |V_G|!}\ =\ \frac{|\copr{D_G}{\dep,Q}|}{|\opr{D_G}{\dep}|}.
\]
The latter expression is precisely $\orfreq{\dep,Q}{D_G,()}$, as needed.

\OMIT{
For this proof we employ the same reduction employed for the proof of Item~(1) of Theorem~\ref{the:uniform-repairs}. In particular, that proof shows that given an undirected graph $G=(V_G,E_G)$, a database $D_G$ can be computed in polynomial-time w.r.t.\ $G$ and $D_G$ is such that $|\mathsf{hom}(G,H)|$ can be computed in polynomial-time w.r.t.\ $G$, having an oracle for computing the number $\orfreq{\dep,Q}{D_G,()}$, for some fixed set $\dep$ of FDs and Boolean CQ $Q$.

\medskip
We can actually prove that
$\orfreq{\dep,Q}{D_G,()} = \srfreq{\dep,Q}{D_G,()}$, and our claim will follow.
Recall that
$$ \srfreq{\dep,Q}{D_G,()} = \frac{|\{s \in \crs{D_G}{\dep} \mid s(D) \models Q\}|}{|\crs{D_G}{\dep}|}.$$

By construction of $D_G$, each candidate repair $D \in \copr{D_G}{\dep}$ can be obtained via $|V_G|!$ different complete sequences in $\crs{D_G}{\dep}$, and thus $\srfreq{\dep,Q}{D_G,()}$ is

$$ \frac{|\copr{D_G}{\dep,Q}| \cdot |V_G|!}{|\copr{D_G}{\dep}| \cdot |V_G|!} = \frac{|\copr{D_G}{\dep,Q}|}{|\opr{D_G}{\dep}|}.$$

The latter expression is precisely $\orfreq{\dep,Q}{D_G,()}$, as needed.
}

\subsection{Proof of Item~(2) of Theorem~\ref{the:uniform-sequences}}

We prove that, for a set $\dep$ of primary keys, and a CQ $Q$, the problem $\srelfreq{\dep,Q}$ admits an FPRAS. As for item (2) of Theorem~\ref{the:uniform-repairs}, the proof consists of two steps: (1) existence of an efficient sampler, and (2) provide a polynomial lower bound for $\srfreq{\dep,Q}{D,\bar c}$.

\subsubsection*{Step 1: Efficient Sampler}
To establish that we can efficiently sample elements of $\crs{D}{\dep}$ uniformly at random, we first need to show that the number of complete repairing sequences can be computed
in polynomial time in the case of primary keys. 

\begin{lemma}\label{lemma:count_seq}
Consider a set $\dep$ of primary keys. For every database $D$, 
$|\crs{D}{\dep}|$ 
can be computed in polynomial time in $||D||$.
\end{lemma}
\begin{proof}
Consider a database $D$. As in the proof of Lemma~\ref{lem:ur-sampler}, let 
 $B_1,\dots,B_n$ be the blocks of $D$ w.r.t.~$\dep$ that contain at least two facts. For a block $B$ of size $m\ge 2$ and for $0\le i\le \floor*{\frac{m}{2}}$, we denote by $S_{m}^{\nempt,i}$ the number of sequences $s\in \crs{B}{\dep}$ such that $s(B)\neq\emptyset$ (hence, $s(B)$ contains a single fact) and precisely $i$ of the operations of $s$ are pair removals. 
In the case where $m$ is an even number, we cannot obtain a non-empty repair with $\frac{m}{2}$ pair removals; hence, we will have that $S_{m}^{\nempt,\frac{m}{2}}=0$.
In any other case, we have that:
\begin{eqnarray*}
  S_{m}^{\nempt,i} &=& m\times \left[{m-1
  \choose {2i}}\times \frac{(2i)!}{2^i\cdot i!}\times (m-i-1)!\right]\\
  &=& m \times \frac{(m-1)!}{(2i)!\cdot (m-2i-1)!}\times \frac{(2i)!}{2^i\cdot i!}\times (m-i-1)!\\
  &=& \frac{m!\cdot (m-i-1)!}{2^i\cdot i!\cdot (m-2i-1)!}
\end{eqnarray*}
where:
\begin{itemize}
    \item $m$ is the number of ways to select the single fact of $s(B)$,
    \item ${m-1
  \choose {2i}}$ is the number of ways to select $2i$ facts out of the remaining $m-1$ facts (these are the facts removed in pairs),
  \item $\frac{(2i)!}{2^i\cdot i!}$ is the number of ways to split $2i$ facts into $i$ pairs, and
  \item $(m-i-1)!$ is the number of permutations of the $m-i-1$ operations in the sequence ($m-2i-1$ singleton removals, and $i$ pair removals).
\end{itemize}

Similarly, we denote by $S_{m}^{\empt,i}$ the number of sequences $s\in \crs{B}{\dep}$ such that $s(B)=\emptyset$ and $s$ has precisely $i$ pair removals. As we cannot obtain an empty repair without pair removals, it holds that $S_{m}^{\empt,0}=0$.
For $i\ge 1$, the following holds:
\begin{eqnarray*}
  S_{m}^{\empt,i} &=& {{m}\choose 2}\times \left({m-2
  \choose {2i-2}}\times \frac{(2i-2)!}{2^{i-1}\cdot (i-1)!}\times (m-i-1)!\right)\\
  &=&\frac{m!}{2!\cdot (m-2)!}\times \frac{(m-2)!}{(2i-2)!\cdot (m-2i)!}\\
  && \times \frac{(2i-2)!}{2^{i-1}\cdot (i-1)!}\times (m-i-1)!\\
  &=& \frac{m!\cdot (m-i-1)!}{2^i\cdot (i-1)!\cdot (m-2i)!}
\end{eqnarray*}
where:
\begin{itemize}
    \item ${{m}\choose 2}$ is the number of ways to select the last pair that will be removed in the sequence,
    \item ${m-2
  \choose {2i-2}}$ is the number of ways to select $2(i-1)$ facts out of the remaining $m-2$ facts (these are the facts removed in pairs),
  \item $\frac{(2i-2)!}{2^{i-1}\cdot (i-1)!}$ is the number of ways to split $2(i-1)$ facts into $i-1$ pairs, and
  \item $(m-i-1)!$ is the number of permutations of the $m-i-1$ operations in the sequence excluding the last pair removal ($m-2i$ singleton removals, and $i-1$ pair removals).
\end{itemize}

Since there are no conflicts among facts from different blocks, the repairing sequences for different blocks are independent (in the sense that an operation over the facts of one block has no impact on the justified operations over the facts of another block). Hence, every complete repairing sequence $s\in\crs{D}{\dep}$ is obtained by interleaving sequences for the individual blocks. We can compute this number of sequences in polynomial time using dynamic programming. We denote by $P_j^{k,i}$ the number of sequences $s\in \crs{B_1\cup\dots\cup B_j}{\dep}$ with precisely $i$ pair removals such that $s(D)\cap B_\ell\neq\emptyset$ for $k$ of the blocks of $B_1,\dots,B_j$ (hence, $0\le k\le j$). For $k<0$ or $k>j$, we have $P_j^{k,i}=0$. Then, it holds that:
\[
\crs{D}{\dep}\ =\ \sum_{k=0}^{n}\sum_{i=0}^{\floor*{\frac{|B_1|}{2}}+\dots+\floor*{\frac{|B_n|}{2}}} P_n^{k,i}.
\]
Clearly, for every $i\in \left\{0,\dots,\floor*{\frac{|B_1|}{2}}\right\}$, we have that:
\begin{eqnarray*}
	P_1^{0,i} &=& S_{|B_1|}^{\empt,i}\\
	P_1^{1,i} &=& S_{|B_1|}^{\nempt,i}.
\end{eqnarray*}
For $j>1$, it holds that:
\begin{align*}
    P_j^{k,i}=&
    \sum_{\substack{0\le i_1\le \floor*{\frac{|B_1|}{2}}+\dots+\floor*{\frac{|B_{j-1}|}{2}}\\ 0\le i_2\le\floor*{\frac{|B_j|}{2}}\\i_1+i_2=i}} \Big[ P_{j-1}^{k,i_1}\times S_{|B_j|}^{\empt,{i_2}}\times\\ &{\frac{(|B_1\cup\dots\cup B_j|-i_1-i_2-k)!}{(|B_1\cup\dots\cup B_{j-1}|-i_1-k)!\times (|B_j|-i_2)!}}+\\
    &P_{j-1}^{k-1,i_1}\times S_{|B_j|}^{\nempt,{i_2}}\times\\
    &{\frac{(|B_1\cup\dots\cup B_j|-i_1-i_2-k)!}{(|B_1\cup\dots\cup B_{j-1}|-i_1-k+1)!\times (|B_j|-i_2-1)!}}\Big],
    \end{align*}
where the last expression is the number of ways to interleave a sequence of $\crs{B_1\cup\dots\cup B_{j-1}}{\dep}$ that has $i_1$ pair removals with a sequence of $\crs{B_j}{\dep}$ that has $i_2$ pair removals. Note that if for a block $B_\ell$, the sequence $s$ has $i$ pair removals over the facts of $B_\ell$ and it holds that $s(D)\cap B_\ell\neq\emptyset$, then $s$ contains $|B_\ell|-i-1$ operations over the facts of $B_\ell$ (as we keep one of its facts in the repair). If $s(D)\cap B_\ell=\emptyset$, then $s$ contains $|B_\ell|-i$ operations over the facts of $B_\ell$. Hence,
$|B_1\cup\dots\cup B_j|-i_1-i_2-k$ is the total number of operations in the combined sequence, $|B_1\cup\dots\cup B_{j-1}|-i_1-k$ (or $|B_1\cup\dots\cup B_{j-1}|-i_1-k+1$) is the number of operations over the facts of the first $j-1$ blocks, and $|B_j|-i_2$ (or $|B_j|-i_2-1$) is the number of operations over the facts of the $j$th block.
\end{proof}

We give an example that illustrates the algorithm described in the proof of Lemma~\ref{lemma:count_seq}.

\begin{example}\label{example:pkeys_us}
Consider again the database $D$ depicted in Figure~\ref{fig:example_pkeys}, and the set $\dep = \{R : A_1 \ra A_2\}$ consisting of a single key. The complete repairing sequences over the facts of the first block (that consists of the facts $f_{1,1},f_{1,2},f_{1,3}$) are:
\begin{align*}
    &-f_{1,1},-f_{1,2} \,\,\,\,\,\,\,\, -f_{1,1},-f_{1,3} \,\,\,\,\,\,\,\, -f_{1,1},-\{f_{1,2},f_{1,3}\} \\ 
    &-f_{1,2},-f_{1,1} \,\,\,\,\,\,\,\, -f_{1,2},-f_{1,3} \,\,\,\,\,\,\,\, -f_{1,2},-\{f_{1,1},f_{1,3}\} \\
    &-f_{1,3},-f_{1,1} \,\,\,\,\,\,\,\, -f_{1,3},-f_{1,2} \,\,\,\,\,\,\,\, -f_{1,3},-\{f_{1,1},f_{1,2}\}\\
    &-\{f_{1,1},f_{1,2}\} \,\,\,\,\,\,\,\, -\{f_{1,1},f_{1,3}\} \,\,\,\,\,\,\,\, -\{f_{1,2},f_{1,3}\}
\end{align*}
There are no repairing sequences over the facts of the second block, as it only contains a single fact $f_{2,1}$. The complete repairing sequences over the facts of the third block (with facts $f_{3,1},f_{3,2}$) are:
\begin{align*}
    &-f_{3,1} \,\,\,\,\,\,\,\, -f_{3,2} \,\,\,\,\,\,\,\, -\{f_{3,1},f_{3,2}\}
\end{align*}
Every complete repairing sequence over $D$ is obtained by interleaving the complete repairing sequences over the different blocks. For example, the following is one possible complete repairing sequence:
\[-f_{1,2},-\{f_{3,1},f_{3,2}\},-f_{1,1}\]
and it has one pair removal.

In this case, we have that:
\begin{align*}
  &S_{3}^{\nempt,0}=\frac{3!\cdot (3-0-1)!}{2^0\cdot 0!\cdot (3-2\times 0-1)!}=\frac{12}{2}=6\\
  &S_{3}^{\nempt,1}=\frac{3!\cdot (3-1-1)!}{2^1\cdot 1!\cdot (3-2\times 1-1)!}=\frac{6}{2}=3\\
  &S_{3}^{\empt,0}=0\\
  &S_{3}^{\empt,1}=\frac{3!\cdot (3-1-1)!}{2^1\cdot (1-1)!\cdot (3-2\times 1)!}=\frac{6}{2}=3
\end{align*}
Indeed, there are 12 repairing sequences over the facts of the first block that contains three facts---six of them have no pair removals, three have a single pair removal and result in a non-empty repair, and three have a single pair removal and result in an empty repair.

For the third block, that has two facts, it holds that:
\begin{align*}
  &S_{2}^{\nempt,0}=\frac{2!\cdot (2-0-1)!}{2^0\cdot 0!\cdot (2-2\times 0-1)!}=\frac{2}{1}=2\\
  &S_{2}^{\nempt,1}=0\\
  &S_{2}^{\empt,0}=0\\
  &S_{2}^{\empt,1}=\frac{2!\cdot (2-1-1)!}{2^1\cdot (1-1)!\cdot (2-2\times 1)!}=\frac{1}{2}=1
\end{align*}
Indeed, there are two sequences with no pair removals (that result in non-empty repairs) and a single sequence with one pair removal (that results in an empty repair).

Finally, we denote the block with three facts by $B_1$, and the block with two facts by $B_2$. We have that:
\[P_1^{0,0}=S_3^{\empt,0}=0 \,\,\,\,\,\,\,\, P_1^{1,0}=S_3^{\nempt,0}=6\]
\[P_1^{0,1}=S_3^{\empt,1}=3 \,\,\,\,\,\,\,\, P_1^{1,1}=S_3^{\nempt,1}=3\]
Next,
\begin{align*}
   P_2^{0,0}&=P_1^{0,0}\times S_2^{\empt,0}\times \frac{(5-0)!}{(3-0)!\times (2-0)!}
   =0\times 0\times 10=0
\end{align*}
Indeed, if $s$ has zero pair removals, then $s(D)\cap B_1\neq\emptyset$ and $s(D)\cap B_2\neq\emptyset$; hence, $P_2^{0,k}>0$ only for $k=2$. Thus,
\begin{align*}
   P_2^{1,0}&=P_1^{0,0}\times S_2^{\nempt,0}\times \frac{(5-1)!}{(3-0)!\times (2-1)!}\\
   &+P_1^{1,0}\times S_2^{\empt,0}\times \frac{(5-1)!}{(3-1)!\times (2-0)!}\\
   &=0\times 0\times 4+6\times 0\times 6=0
\end{align*}
And:
\begin{align*}
   P_2^{2,0}&=P_1^{1,0}\times S_2^{\nempt,0}\times \frac{(5-2)!}{(3-1)!\times (2-1)!}=6\times 2\times 3=36
\end{align*}

Similarly, we compute:
\begin{align*}
   P_2^{0,1}&=P_1^{0,0}\times S_2^{\empt,1}\times \frac{(5-1)!}{(3-0)!\times (2-1)!}\\
   &+P_1^{0,1}\times S_2^{\empt,0}\times \frac{(5-1)!}{(3-1)!\times (2-0)!}\\
   &=0\times 1\times 4+3\times 0\times 6=0
\end{align*}
\begin{align*}
   P_2^{1,1}&=P_1^{1,1}\times S_2^{\empt,0}\times \frac{(5-2)!}{(3-2)!\times (2-0)!}\\
   &+P_1^{0,1}\times S_2^{\nempt,0}\times \frac{(5-2)!}{(3-1)!\times (2-1)!}\\
   &+P_1^{1,0}\times S_2^{\empt,1}\times \frac{(5-2)!}{(3-1)!\times (2-1)!}\\
   &+P_1^{0,0}\times S_2^{\nempt,1}\times \frac{(5-2)!}{(3-0)!\times (2-2)!}\\
   &=3\times 0\times 3+3\times 2\times 3+6\times 1\times 3+0\times 0\times 1=36
\end{align*}
\begin{align*}
   P_2^{2,1}&=P_1^{1,0}\times S_2^{\nempt,1}\times \frac{(5-3)!}{(3-1)!\times (2-2)!}\\
   &+P_1^{1,1}\times S_2^{\nempt,0}\times \frac{(5-3)!}{(3-2)!\times (2-1)!}\\
   &=6\times 0\times 1+3\times 2\times 2=12
\end{align*}

And, finally:
\begin{align*}
   P_2^{0,2}&=P_1^{0,1}\times S_2^{\empt,1}\times \frac{(5-2)!}{(3-1)!\times (2-1)!}
   =3\times 1\times 3=9
\end{align*}
\begin{align*}
   P_2^{1,2}&=P_1^{1,1}\times S_2^{\empt,1}\times \frac{(5-3)!}{(3-2)!\times (2-1)!}\\
   &+P_1^{0,1}\times S_2^{\nempt,1}\times \frac{(5-3)!}{(3-1)!\times (2-2)!}\\
   &=3\times 1\times 2+3\times 0\times 1=6
\end{align*}
\begin{align*}
   P_2^{2,2}&=P_1^{1,1}\times S_2^{\nempt,1}\times \frac{(5-4)!}{(3-2)!\times (2-2)!}=3\times 0\times 1=0
\end{align*}

We conclude that:
\[\crs{D}{\dep}=0+0+36+0+36+12+9+6+0=99\]
That is, there are $99$ complete repairing sequences of $D$ w.r.t.~$\dep$. \hfill\markfull
\end{example}

\begin{algorithm}[t]
    \LinesNumbered
    \KwIn{A database $D$ and a set $\dep$ of primary keys over a schema $\ins{S}$.}
    \KwOut{$s \in \crs{D}{\dep}$ with probability $\frac{1}{|\crs{D}{\dep}|}$.}
    \vspace{2mm}
    \SetKwProg{Fn}{Function}{:}{}
    
    \Return $\mathsf{Sample}(D,\dep,\varepsilon)$
    
    \vspace{2mm}
    \Fn{$\mathsf{Sample}(D, \dep, s)$}{
        \eIf{$s(D)\models \dep$}{
            \textbf{return} $s$\;
        }{
            Select a $(s(D),\dep)$-justified operation $\op$ with probability $\frac{|\crs{\op(s(D))}{\dep}|}{|\crs{s(D)}{\dep}|}$\\
            \textbf{return} $\mathsf{Sample}(D,\dep,s\cdot\op)$
        }
    }
    
    \vspace{1em}

    \caption{An algorithm $\mathsf{SampleSeq}$ for sampling elements of $\crs{D}{\dep}$ uniformly at random.}\label{alg:ssample}
\end{algorithm}

Having Lemma~\ref{lemma:count_seq} in place, we can establish the existence of an efficient sampler.
The formal statement, already given in the main body of the paper, and its proof follow:

\begin{manuallemma}{\ref{lem:us-sampler}}
    For a database $D$, and a set $\dep$ of primary keys, we can sample elements of $\crs{D}{\dep}$ uniformly at random in polynomial time in $||D||$.
\end{manuallemma}
\begin{proof}
The algorithm $\mathsf{SampleSeq}$, depicted in Algorithm~\ref{alg:ssample}, is a recursive algorithm that returns a sequence $s\in\crs{D}{\dep}$ with probability $\frac{1}{|\crs{D}{\dep}|}$. The algorithm starts with the empty sequence $\varepsilon$, and, at each step, extends the sequence by selecting one of the justified operations at that point. That is, if the current sequence is $s$, then we select one of the $(s(D),\dep)$-justified operations. 
The probability of selecting an operation $\op$ is:
\[\frac{|\crs{op(s(D))}{\dep}|}{|\crs{s(D)}{\dep}|}.\]
%
%
Hence, the probability of returning a sequence $s=\op_1,\dots,\op_n$ of $\crs{D}{\dep}$ is:
\begin{align*}
    &\frac{|\crs{\op_1(D)}{\dep}|}{|\crs{\varepsilon(D)}{\dep}|}\times \frac{|\crs{\op_2(\op_1(D))}{\dep}|}{|\crs{\op_1(D)}{\dep}|}\times\dots\\
    &\times \frac{|\crs{\op_n(\dots D\dots )}{\dep}|}{|\crs{\op_{n-1}(\dots D\dots )}{\dep}|}\\
    &=\frac{|\crs{\op_n(\dots D\dots )}{\dep}|}{|\crs{\varepsilon(D)}{\dep}|}=\frac{1}{|\crs{D}{\dep}|}
\end{align*}
Most of the terms in the product cancel each other, and
\[|\crs{\op_n(\dots D\dots )}{\dep}|=|\crs{s(D)}{\dep}|=1\]
since $s(D)\models\dep$; hence, there is a single complete repairing sequence for $s(D)$ w.r.t.~$\dep$---the empty sequence.

Since the length of a sequence is bounded by $|D|-1$,  the number of justified operations at each step is polynomial in $||D||$ (as this is the number of facts involved in violations of the constraints plus the number of conflicting pairs of facts), and, by Lemma~\ref{lemma:count_seq}, for a set $\dep$ of primary keys, we can compute $|\crs{D}{\dep}|$ in polynomial time in $||D||$ for any database $D$, we get that the total running time of the algorithm is also polynomial in $||D||$, as needed.
\end{proof}

\subsubsection*{Step 2: Polynomial Lower Bound}
Now that we have an efficient sampler for the complete repairing sequences, we show that there is a polynomial lower bound on $\srfreq{\dep,Q}{D,\bar c}$.
The formal statement, already given in the main body of the paper, and its proof follow:

\begin{manuallemma}{\ref{lem:us-lower-bound}}
Consider a set $\dep$ of primary keys, and a CQ $Q(\bar x)$. For every database $D$, and tuple $\bar c \in \adom{D}^{|\bar x|}$,
    \[
    \srfreq{\dep,Q}{D,\bar c}\ \geq\ \frac{1}{(2 \cdot ||D||)^{||Q||}}
    \] 
    whenever $\srfreq{\dep,Q}{D,\bar c} > 0$.
\end{manuallemma}
\begin{proof}
The proof is very similar to the proof of Lemma~\ref{lem:ur-lower-bound}, except that here we reason about sequences rather than repairs. Recall that we treat a query $Q$ as the set $\{R_i(\bar y_i)\mid i \in [n]\}$ of atoms on the right-hand side of $\text{:-}$, and, for a database $D$ and a homomorphism $h$ from $Q$ to $D$, we denote by $h(Q)$ the set $\{R_i(h(\bar y_i))\mid i \in [n]\}$.
Here, we denote by $S_{D,\dep,h(Q)}^\empt$ the set of sequences $s\in\crs{D}{\dep}$ such that $s(D)\cap B_j=\emptyset$ for at least one block of $\{B_1,\dots,B_m\}$ (recall that these are the blocks that contains the facts of $h(Q)$), and by $S_{D,\dep,h(Q)}^\nempt$ the set of sequences $s\in\crs{D}{\dep}$ such that $s(D)\cap B_j\neq\emptyset$ for every block of $\{B_1,\dots,B_m\}$.

Now, for every sequence $s\in S_{D,\dep,h(Q)}^\empt$, and for every block $B_j$ such that $s(D)\cap B_j=\emptyset$, the last operation of $s$ over the facts of $B_j$ must remove a pair $\{f,g\}$ of facts. We map each sequence $s\in S_{D,\dep,h(Q)}^\empt$ to a sequence $s'\in S_{D,\dep,h(Q)}^\nempt$ by replacing the last operation of $s$ over each such $B_j\in\{B_1,\dots,B_m\}$ with an operation that removes only one of the facts of the pair---either $f$ or $g$. Hence, if $s(D)\cap B_j=\emptyset$ for precisely $\ell$ of the blocks of $\{B_1,\dots,B_m\}$, the sequence $s$ is mapped to $2^\ell$ distinct sequences of $S_{D,\dep,h(Q)}^\nempt$.

Similarly to the proof of Lemma~\ref{lem:ur-lower-bound}, for every sequence $s'\in S_{D,\dep,h(Q)}^\nempt$, there are $2^m-1$ sequences $s\in S_{D,\dep,h(Q)}^\empt$ that are mapped to it. This is because the sequence $s'$ determines all the operations over the blocks outside $\{B_1,\dots,B_m\}$, and for each block $B_j\in \{B_1,\dots,B_m\}$, it determines all the operations over $B_j$ except for the last one. If $s'(D)\cap B_j=\{f\}$ and the last operation of $s'$ over $B_j$ removes the fact $g$, then the last operation of $s$ over $B_j$ either also removes $g$ or removes the pair $\{f,g\}$. If the last operation of $s'$ over $B_j$ removes a pair $\{g,h\}$ of facts, then the last operation of $s$ over $B_j$ must also remove the same pair of facts. Hence, there are at most two possible cases for each  block of $\{B_1,\dots,B_m\}$ and $2^m$ possibilities in total. And, again, we have to disregard the possibility that is equivalent to $s'$ itself.

Therefore, we have that:
\[
 \left|S_{D,\dep,h(Q)}^\empt\right|\ \le\ (2^m-1)\times \left|S_{D,\dep,h(Q)}^\nempt\right|
\]
 and
 \begin{align*}
     |\crs{D}{\dep}|&= \left|S_{D,\dep,h(Q)}^\empt\right|+ \left|S_{D,\dep,h(Q)}^\nempt\right|\\
     &\le (2^m-1)\times \left|S_{D,\dep,h(Q)}^\nempt\right|+ \left|S_{D,\dep,h(Q)}^\nempt\right|\\
     &=2^m\times \left|S_{D,\dep,h(Q)}^\nempt\right|.
 \end{align*}
 As said above, each sequence $s$ of $S_{D,\dep,h(Q)}^\empt$ can be mapped to $2^\ell$ distinct sequences of $S_{D,\dep,h(Q)}^\nempt$, where $\ell$ is the number of blocks in $\{B_1,\dots,B_m\}$ for which $E\cap B_j=\emptyset$. Moreover, there are sequences $s'\in S_{D,\dep,h(Q)}^\nempt$ such that no sequence $s\in S_{D,\dep,h(Q)}^\empt$ is mapped to $s'$. These are the sequences $s'$ where the last operation of $s'$ over every block of $\{B_1,\dots,B_m\}$ is a pair removal (but $s'$ keeps some fact of each $B_j$). Hence, $(2^m-1)\times |S_{D,\dep,h(Q)}^\nempt|$ is only an upper bound on $|S_{D,\dep,h(Q)}^\empt|$.
Since all the facts of a single block are symmetric,
\[
\left|\{s\in\crs{D}{\dep}\mid h(Q) \subseteq s(D)\}\right| = \frac{1}{|B_1|\times\dots\times |B_m|} \times \left|S_{D,\dep,h(Q)}^\nempt\right|
\]
and, we conclude that
\begin{align*}
   \frac{|\{s\in\crs{D}{\dep}\mid h(Q)\subseteq s(D)\}|}{|\crs{D}{\dep}|}&\ge \frac{\frac{1}{|B_1|\times\dots\times |B_m|} \times \left|S_{D,\dep,h(Q)}^\nempt\right|}{2^{m}\times \left|S_{D,\dep,h(Q)}^\nempt\right|}\\
   &=\frac{1}{|B_1|\times\dots\times |B_m|\times 2^{m}}\\
   &\ge \frac{1}{|D|^m\times 2^{m}}\\
   &\ge \frac{1}{(2|D|)^{|Q|}}\ge  \frac{1}{(2||D||)^{||Q||}}
\end{align*}
Since all the sequences of $\{s\in\crs{D}{\dep}\mid h(Q)\subseteq s(D)\}$ are such that $h(Q)\subseteq s(D)$ and so $\bar c\in s(D)$, this concludes our proof.
\end{proof}

We give an example that illustrates the argument given in the proof of Lemma~\ref{lem:us-lower-bound}.

\begin{example}
Consider the database $D$, the set $\dep$ of keys, the query $Q$, and the homomorphism $h$ from Example~\ref{example:pkeys_ur}. Recall that $h(Q)=\{R(a_1,b_1)\}$. The set $S_{D,\dep,h(Q)}^\empt$ contains, for example,
\[-f_{1,2},-f_{3,1},-\{f_{1,1},f_{1,3}\}\]
as the resulting database $s(D)$ contains no fact from the block of $R(a_1,b_1)$. According to the mapping defined in the proof of Lemma~\ref{lem:us-lower-bound}, this sequence is mapped to the following two sequences:
\[-f_{1,2},-f_{3,1},-f_{1,1}\]
\[-f_{1,2},-f_{3,1},-f_{1,3}\]
that replace the last pair removal over the block of $R(a_1,b_1)$ with a singleton removal.
In this case, we have that
\[
|\{s\in\crs{D}{\dep}\mid h(Q)\subseteq s(D)\}|\ =\ 24.
\]
These are all the sequences obtained by interleaving the following operations over the facts of the first block (that do not remove $f_{1,1}$), with any of the three operations over the facts of the third block:
\begin{align*}
    -f_{1,2},-f_{1,3} \,\,\,\,\,\,\,\, -f_{1,3},-f_{1,2} \,\,\,\,\,\,\,\, -\{f_{1,2},f_{1,3}\}
\end{align*}
Moreover, as we have seen in Example~\ref{example:pkeys_us}, we have that
\[|\crs{D}{\dep}|=99\]
Indeed, it holds that
\[\frac{24}{99}\ge \frac{1}{12}=\frac{1}{(2|D|)^{|Q|}}\]
as claimed. \hfill\markfull
\end{example}

%% file: app-uniform-operations.tex
\section{Proofs of Section~\ref{sec:uniform-operations}}\label{appsec:uniform-operations}

In this section, we prove the main result of Section~\ref{sec:uniform-operations}, which we recall here for the sake of readability:

\begin{manualtheorem}{\ref{the:uniform-operations}}
    \begin{enumerate}
        \item There exist a set $\dep$ of primary keys, and a CQ $Q$ such that $\ocqa{\dep,M_{\dep}^{\uo},Q}$ is $\sharp ${\rm P}-hard.
        
        \item For a set $\dep$ of keys, and a CQ $Q$, $\ocqa{\dep,M_{\dep}^{\uo},Q}$ admits an FPRAS.
    \end{enumerate}
\end{manualtheorem}

\subsection{Proof of Item~(1) of Theorem~\ref{the:uniform-operations}}
As we did for item (1) of Theorem~\ref{the:uniform-sequences}, we reuse the construction underlying the proof of item (1) of Theorem~\ref{the:uniform-repairs}.
In particular, assuming that $\dep$ and $Q$ are the singleton set of primary keys and the Boolean CQ, respectively, for which $\rrelfreq{\dep,Q}$ is $\sharp ${\rm P}-hard ($\dep$ and $Q$ are extracted from the proof of item (1) of Theorem~\ref{the:uniform-repairs}), we show that $\ocqa{\dep,M_{\dep}^{\uo},Q}$ is $\sharp ${\rm P}-hard via a polynomial-time Turing reduction from $\sharp H\text{-}\mathsf{Coloring}$ by reusing the construction in the proof of item (1) of Theorem~\ref{the:uniform-repairs}; $H$ is the same undirected graph employed in that proof.
Assuming that, for an undirected graph $G$, $D_G$ is the database that the construction in the proof of item (1) of Theorem~\ref{the:uniform-repairs} builds, we show that 
\[
\orfreq{\dep,Q}{D_G,()}\ =\ \probrep{M_{\dep}^{\uo},Q}{D_G,()}, 
\]
which implies that the polynomial-time Turing reduction from $\sharp H\text{-}\mathsf{Coloring}$ to $\rrelfreq{\dep,Q}$ is also a polynomial-time Turing reduction from $\sharp H\text{-}\mathsf{Coloring}$ to $\ocqa{\dep,M_{\dep}^{\uo},Q}$.

\OMIT{
Also in this case, the same reduction used in the proof of Item~(1) of Theorem~\ref{the:uniform-repairs} suffices. Similarly to what we did in the proof of Item~(1) of Theorem~\ref{the:uniform-sequences}, we need to prove that given an undirected graph $G=(V_G,E_G)$,
$\probrep{M_\dep^{\uo},Q}{D_G,()} = \orfreq{\dep,Q}{D_G,()}$, where $D_G$, $\dep$ and $Q$ are defined as in the proof of Item~(1) of Theorem~\ref{the:uniform-repairs}.
Since in the proof of Item~(1) of Theorem~\ref{the:uniform-sequences} we have already shown that $\srfreq{\dep,Q}{D_G,()} = \orfreq{\dep,Q}{D_G,()}$, it suffices to prove that $\probrep{M_\dep^{\uo},Q}{D_G,()} = \srfreq{\dep,Q}{D_G,()}$.
}

In the proof of item~(1) of Theorem~\ref{the:uniform-sequences}, we have shown that $\orfreq{\dep,Q}{D_G,()} = \srfreq{\dep,Q}{D_G,()}$. Thus, it suffices to show 
\[
\srfreq{\dep,Q}{D_G,()}\ =\ \probrep{M_\dep^{\uo},Q}{D_G,()}.
\]
Let $M_\dep^{\uo}(D_G) = (V,E,\ins{P})$. Note that each node $u$ of $G$ induces a violation $\{V(u,0),V(u,1)\}$ in $D_G$ that can be resolved using one of the following three operations: remove the first, the second, or both facts. Hence, every complete sequence in $\crs{D_G}{\dep}$ is of length precisely $|V_G|$, and for every non-leaf node $s \in V$, $|\ops{s}{D_G}{\dep}| = 3 \cdot (|V_G| - |s|)$. 
Hence, by Definition~\ref{def:uniform-ops}, for each $(s,s') \in E$,
\[
\ins{P}(s,s')\ =\ \frac{1}{|\ops{s}{D_G}{\dep}|}\ =\ \frac{1}{3 \cdot (|V_G| - |s|)}.
\]
We conclude that, with $\pi$ being the leaf distribution of $M_\dep^{\uo}$, for each $s = \op_1,\ldots,\op_n \in \crs{D_G}{\dep}$,
\[
\pi(s)\ =\ \ins{P}(s_0,s_1) \cdots \ins{P}(s_{n-1},s_n)\ =\ \frac{1}{3^{|V_G|} \cdot |V_G|!}.
\]
Since each sequence $s \in \crs{D_G}{\dep}$ is assigned the same non-zero probability, $\pi$ is the uniform distribution over $\crs{D_G}{\dep}$. The latter implies that  $\srfreq{\dep,Q}{D_G,()}= \probrep{M_\dep^{\uo},Q}{D_G,()}$, as needed.

\subsection{Proof of Item~(2) of Theorem~\ref{the:uniform-operations}}

We prove that, for a set $\dep$ of keys, and a CQ $Q$, $\ocqa{\dep,M_{\dep}^{\uo},Q}$ admits an FPRAS. As for item (2) of Theorems~\ref{the:uniform-repairs} and~\ref{the:uniform-sequences}, the proof consists of the usual two steps: (1) existence of an efficient sampler, and (2) provide a polynomial lower bound for $\probrep{M_{\dep}^{\uo},Q}{D,\bar c}$.

\subsubsection*{Step 1: Efficient Sampler}
Given a database $D$, the definition of $M_{\dep}^{\uo}$ immediately implies the existence of an efficient sampler that returns a sequence $s \in \abs{\dep,M_{\dep}^{\uo}(D)}$ with probability $\pi(s)$, where $\pi$ is the leaf distribution of $M_{\dep}^{\uo}(D)$. The algorithm is very similar to Algorithm~\ref{alg:ssample}, except that if the current sequence is $s$, the probability to select a $(s(D),\dep)$-justified operation is
\[
\frac{1}{|\ops{s}{D}{\dep}|}.
\]
Hence, we immediately obtain the following result, already given in the main body of the paper:

\begin{manuallemma}{\ref{lem:uo-sampler}}
    Given a database $D$, and a set $\dep$ of keys, we can sample elements of $\abs{M_{\dep}^{\uo}(D)}$ according to the leaf distribution of $M_{\dep}^{\uo}(D)$ in polynomial time in $||D||$.
\end{manuallemma}

\subsubsection*{Step 2: Polynomial Lower Bound}
The rest of the section is devoted to showing that there is a polynomial lower bound on $\probrep{M_{\dep}^{\uo},Q}{D,\bar c}$.

\begin{manualproposition}{\ref{pro:uo-lower-bound}}
Consider a set $\dep$ of keys, and a CQ $Q(\bar x)$. There is a polynomial $\mathsf{pol}$ such that, for every database $D$, and $\bar c \in \adom{D}^{|\bar x|}$,
    \[
    \probrep{M_{\dep}^{\uo},Q}{D,\bar c}\ \geq\ \frac{1}{\mathsf{pol}(||D||)}
    \] 
    whenever $\probrep{M_{\dep}^{\uo},Q}{D,\bar c} > 0$.
\end{manualproposition}

As usual, we treat the CQ $Q$ as the set $\{R_i(\bar y_i)\mid i \in [n]\}$ of atoms occurring on the right-hand side of $\text{:-}$. Moreover, for a database $D$ and a homomorphism $h$ from $Q$ to $D$, we write $h(Q)$ for the set $\{R_i(h(\bar y_i))\mid i \in [n]\}$.
Clearly, if there is no homomorphism $h$ from $Q$ to $D$ with $h(Q)\models\dep$ and $h(\bar x)=\bar c$, then $\probrep{M_{\dep}^{\uo},Q}{D,\bar c}=0$. 
Assume now that such a homomorphism $h$ exists. We first prove the claim for the case where $|h(Q)|=1$, and then generalize it to the case where $|h(Q)|=m$ for some $m \in [|Q|]$.



\medskip
\noindent \underline{\textbf{The Case $|h(Q)|=1$}}
\smallskip

\noindent Let $f$ be the single fact of $h(Q)$, and
\[
\probhom{D,M_\dep^{\uo},Q}{h}\ =\ \sum\limits_{D' \in \opr{D}{M_{\dep}^{\uo}} \textrm{~and~} h(Q)\subseteq D'} \probrep{D,M_\dep^{\uo}}{D'}.
\]
Note that since $h(\bar x)=\bar c$, it holds that
\[
\probrep{M_{\dep}^{\uo},Q}{D,\bar c}\ \ge\ \probhom{D,M_\dep^{\uo},Q}{h}.
\]
Hence, it suffices to show that there is a polynomial $\mathsf{pol}$ such that $\probhom{D,M_\dep^{\uo},Q}{h} \geq \frac{1}{\mathsf{pol}(||D||)}$.
Let $S_f$ and $S_{\neg f}$ be the sets of sequences of $\abs{M_{\dep}^{\uo}(D)}$ that keep $f$ and remove $f$, respectively, i.e.,
\begin{eqnarray*}
	S_f &=& \{s \in \abs{M_{\dep}^{\uo}(D)} \mid f \in s(D)\}\\
	S_{\neg f} &=& \{s \in \abs{M_{\dep}^{\uo}(D)} \mid f \not\in s(D)\}.
\end{eqnarray*}
With $\pi$ being the leaf distribution of $M_{\dep}^{\uo}(D)$, 
\[
\probhom{D,M_\dep^{\uo},Q}{h}\ =\ \frac{\Lambda_f}{\Lambda_f+\Lambda_{\neg f}}, 
\]
where
\[
\Lambda_f\ =\ \sum_{s \in S_f}\pi(s) \qquad \text{and} \qquad \Lambda_{\neg f}\ =\ \sum_{s \in S_{\neg f}}\pi(s).
\]
Therefore, to get the desired lower bound $\frac{1}{\mathsf{pol}(||D||)}$ for $\probhom{D,M_\dep^{\uo},Q}{h}$, it suffices to show that there exists a polynomial $\mathsf{pol}'$ such that $\Lambda_{\neg f} \leq \mathsf{pol}'(||D||) \cdot \Lambda_f$. Indeed, in this case we can conclude that
\begin{eqnarray*}
\probhom{D,M_\dep^{\uo},Q}{h} &=& \frac{\Lambda_f}{\Lambda_f+\Lambda_{\neg f}}\\
&\geq& \frac{\Lambda_f}{\Lambda_f + \mathsf{pol}'(||D||) \cdot \Lambda_f}\\
&=& \frac{1}{1 + \mathsf{pol}'(||D||)},
\end{eqnarray*}
and the claim follows with $\mathsf{pol}(||D||) = 1 + \mathsf{pol}'(||D||)$. 

We proceed to show that a polynomial $\mathsf{pol}'$ such that $\Lambda_{\neg f} \leq \mathsf{pol}'(||D||) \cdot \Lambda_f$ exists. To this end, we establish an involved technical lemma that relates the sequences of $S_{\neg f}$ with the sequences of $S_f$; as usual, we write $\pi$ for the leaf distribution of $M_{\dep}^{\uo}(D)$:

\begin{lemma}\label{lemma:relate-sequences}
	There exists a function $\mathsf{F} : S_{\neg f} \ra S_{f}$ such that:
	\begin{enumerate}
		\item There exists a polynomial $\mathsf{pol}''$ such that, for every $s \in S_{\neg f}$, $\pi(s) \leq \mathsf{pol}''(||D||) \cdot \pi(\mathsf{F}(s))$.
		\item For every $s' \in S_{f}$, $|\{s \in S_{\neg f} \mid \mathsf{F}(s)=s'\}| \leq 2 \cdot ||D|| - 1$.
	\end{enumerate}
\end{lemma}

\begin{proof}
	The bulk of the proof is devoted to showing item (1), whereas item (2) is shown via a simple combinatorial argument.
	
	\medskip
	\noindent \paragraph{Item (1).}
	Let $s\in \abs{M_{\dep}^{\uo}(D)}$ be a repairing sequence that removes $f$, i.e., $s \in S_{\neg f}$. We transform $s$ into a repairing sequence $s'\in \abs{M_{\dep}^{\uo}(D)}$ that does not remove $f$, i.e., $s' \in S_f$, by deleting or replacing the operation that removes $f$, and adding additional operations at the end of the sequence as follows. Assume that
	\[s= \op_1, \,\,\,\, \op_2, \,\,\,\, \dots  \,\,\,\, ,\op_{i-1}, \,\,\,\, \op_i, \,\,\,\, \op_{i+1}, \,\,\,\, \dots  \,\,\,\, ,\op_n\]
	where $\op_i = -f$. Then, we define the sequence
	\[s'= \op_1, \,\,\,\, \op_2, \,\,\,\, \dots  \,\,\,\, ,\op_{i-1}, \,\,\,\, \op_{i+1}, \,\,\,\, \dots  \,\,\,\, ,\op_n, \,\,\,\, \op_1', \,\,\,\, \dots \,\,\,\, ,\op_\ell'\]
	where $\op_1',\dots,\op_\ell'$ are new operations that we will describe later. 
	If $\op_i$ is of the form $-\{f,g\}$, then
	\begin{align*}
	s'= \op_1, \,\,\,\, \op_2, \,\,\,\, \dots  \,\,\,\, ,\op_{i-1}, \,\,\,\, \op_i^*, \,\,\,\, &\op_{i+1}, \,\,\,\, \dots  \,\,\,\, ,\op_n, \\
	&\,\,\,\, \op_1', \,\,\,\, \dots \,\,\,\, ,\op_\ell'
	\end{align*}
	where $\op_i^*$ is the operation $-g$, i.e., it removes only the fact $g$.
	
	An important observation here is that since the sequence $s$ removes $f$, the repair $s(D)$ might contain facts that conflict with $f$, but at most $k$ such facts, where $k$ is the number of keys in $\dep$ over the relation name of $f$. This is a property of keys. Indeed, if $s(D)$ contains $k+1$ facts that conflict with $f$, then it contains two facts $g_1,g_2$ that violate the same key with $f$, in which case $g_1,g_2$ also jointly violate this key and cannot appear in the same repair. Therefore, at the end of the sequence $s'$ we add $\ell$ new operations (for some $\ell\le k$) that remove the facts of $s(D)$ that conflict with $f$, in some arbitrary order. Note that the sequence $s'$ is a valid repairing sequence, as an additional fact (the fact $f$) cannot invalidate a justified repairing operation, and we can remove the $\ell$ conflicting facts at the end in any order, as they are all in conflict with $f$.
	Here is a simple example illustrating the construction of $s'$:
	
	\begin{example}\label{example:keys_proof_1}
		Consider again the database $D$ depicted in Figure~\ref{fig:example_pkeys}, and the set $\dep=\{R:A_1\rightarrow A_2, R:A_2\rightarrow A_1\}$ of keys. Consider also the query $Q$ and homomorphism $h$ from Example~\ref{example:pkeys_ur}. Recall that $h(Q)=\{R(a_1,b_1)\}$. The following sequence is a sequence that removes the fact $R(a_1,b_1)$:
		\[s_1=-f_{1,2},-f_{1,1},-f_{3,1}\]
		Note that $s(D)$ contains the facts $f_{1,3}$ and $f_{2,1}$ that conflict with $f_{1,1}$.
		This sequence is mapped to the following sequence $s'$:
		\[s_1'=-f_{1,2},-f_{3,1},-f_{1,3},-f_{2,1}\]
		where we delete the operation $-f_{1,1}$ that removes the fact of $h(Q)$, and add, at the end of the sequence, the operations $-f_{1,3}$ and $-f_{2,1}$ that remove the facts of $s(D)$ that conflict with $f_{1,1}$.
		
		As another example, the sequence:
		\[s_2=-f_{3,1},-\{f_{1,1},f_{1,2}\}\]
		is mapped to the sequence:
		\[s_2'=-f_{3,1},-f_{1,2},-f_{2,1},-f_{1,3}.\]
		Here, the pair removal $-\{f_{1,1},f_{1,2}\}$ is replaced by the singleton removal $-f_{1,2}$, and, at the end of the sequence, we again add two additional operations that remove (in some arbitrary order) the facts $f_{1,3}$ and $f_{2,1}$ that conflict with $f_{1,1}$. \hfill\markfull
		\end{example}
	
	Now, according to the definition of $M_\dep^\uo$, we have that
	\[
	\pi(s)\ =\ \frac{1}{N_1}\times \frac{1}{N_2}\times\dots\times \frac{1}{N_{i-1}}\times\frac{1}{N_{i}}\times \frac{1}{N_{i+1}}\times\dots\times \frac{1}{N_{n}}
	\]
	where $N_j$ is the total number of $(D^s_{j-1},\dep)$-justified repairing operations before applying the operation $\op_j$ of the sequence (recall that $D^s_{j-1}$ is the database obtained from $D$ by applying the first $j-1$ operations of $s$). Hence,
	\[
	\insP((\op_1,\dots,\op_{j-1}),(\op_1,\dots,\op_{j}))\ =\ \frac{1}{N_j}.
	\]
	Then,
	\begin{align*}
	\pi(s')=&\frac{1}{N_1}\times \frac{1}{N_2}\times\dots\times \frac{1}{N_{i-1}}\times\left[\frac{1}{N_{i}}\right]\times \frac{1}{N'_{i+1}}\times\dots\times \frac{1}{N'_n}\times\\
	&\frac{1}{2\ell+1}\times\frac{1}{2(\ell-1)+1}\times\frac{1}{3}.
	\end{align*}
	The probability $\insP((\op_1,\dots,\op_{j-1}),(\op_1,\dots,\op_{j}))$, for $2\le j\le i-1$, is not affected by the decision to remove or keep $f$ at the $i$th step. The probability $\insP((\op_1,\dots,\op_{j-1}),(\op_1,\dots,\op_{j}))$ for $i+2\le j\le n$, on the other hand, might decrease in the sequence $s'$ compared to the sequence $s$, because the additional fact $f$ (that is removed by $s$ but not by $s'$) might be involved in violations with the remaining facts of the database and introduce additional justified operations, in which case $N_j\le N_j'$. Similarly, the probability $\insP((\op_1,\dots,\op_{i-1}),(\op_1,\dots,\op_{i-1},\op_{i+1}))$ (in the case where $\op_i=-f$) or $\insP((\op_1,\dots,\op_{i}^\star),(\op_1,\dots,\op_{i}^\star,\op_{i+1}))$ (in the case where $\op_i=-\{f,g\}$) can only decrease compared to the probability $\insP((\op_1,\dots,\op_{i}),(\op_1,\dots,\op_{i},\op_{i+1}))$ in $s$; hence, $N_{i+1}\le N_{i+1}'$.
	
	The term $\frac{1}{N_i}$ denotes the probability of $\op^\star_i$, and it only appears in the expression if the sequence $s$ removes the fact $f$ jointly with some other fact $g$ (and the operation $\op^\star_i$ removes $g$ by itself). Since all the $(D_{i-1}^{s'},\dep)$-justified operations have the same probability to be selected, the probabilities $\insP((\op_1,\dots,\op_{i-1}),(\op_1,\dots,\op_{i}))$ and $\insP((\op_1,\dots,\op_{i-1}),(\op_1,\dots,\op_{i}^\star)$ are the same.
	Finally, at the end of the sequence, the only remaining conflicts are those involving $f$. As said above, there are $\ell$ facts that conflict with $f$ for some $\ell\le k$ at that point, and each one of them violates a different key with $f$. Hence, there are $2\ell+1$ justified operations before applying $\op_1'$ (removing one of the $\ell$ conflicting facts, removing one of these facts jointly with $f$, or removing $f$), there are $2(\ell-1)+1$ possible operations before applying $\op_2'$ and so on.
	
	\begin{example}\label{example:keys_proof_2}
		We continue with Example~\ref{example:keys_proof_1}. For the sequence $s_1$, we have that
		\begin{align*}
		\pi(s_1)&=\insP(\varepsilon,(-f_{1,2}))\times\insP((-f_{1,2}),(-f_{1,2},-f_{1,1}))\\
		&\times\insP((-f_{1,2},-f_{1,1}),(-f_{1,2},-f_{1,1},-f_{3,1}))= \frac{1}{14}\times\frac{1}{10}\times\frac{1}{5}.
		\end{align*}
		This holds since, at first, all six facts are involved in violations of the keys, and there are eight conflicting pairs; hence, the total number of justified operations is $14$. After removing the fact $f_{1,2}$, the number of justified operations reduces to $10$, and after removing the fact $f_{1,1}$, this number is $5$.
		Now, for the sequence $s_1'$,
		\begin{align*}
		\pi(s_1')&=\insP(\varepsilon,(-f_{1,2}))\times\insP((-f_{1,2}),(-f_{1,2},-f_{3,1}))\\
		&\times\insP((-f_{1,2},-f_{3,1}),(-f_{1,2},-f_{3,1},-f_{1,3}))\\
		&\times\insP((-f_{1,2},-f_{3,1},-f_{1,3}),(-f_{1,2},-f_{3,1},-f_{1,3},-f_{2,1}))\\
		&= \frac{1}{14}\times\frac{1}{10}\times\frac{1}{5}\times \frac{1}{3}.
		\end{align*}
		Indeed, the probability of applying $-f_{1,2}$ (i.e., $\insP(\varepsilon,(-f_{1,2}))$) is the same for both sequences ($\frac{1}{14}$), while the probability of applying the operation $-f_{3,1}$ in  $s_1'$ (i.e., $\insP((-f_{1,2}),(-f_{1,2},-f_{3,1}))$) is smaller than the probability ($\insP((-f_{1,2},-f_{1,1}),(-f_{1,2},-f_{1,1},-f_{3,1}))$) of applying this operation in $s_1$: $\frac{1}{10}$ compared to $\frac{1}{5}$. Finally, there are $\ell=2$ facts in $s(D)$ that conflict with $f_{1,1}$ and we have that
		\[\insP((-f_{1,2},-f_{3,1}),(-f_{1,2},-f_{3,1},-f_{1,3}))=\frac{1}{2\times 2+1}=\frac{1}{5}\]
		\begin{align*}
		&\insP((-f_{1,2},-f_{3,1},-f_{1,3}),(-f_{1,2},-f_{3,1},-f_{1,3},-f_{2,1}))\\
		&=\frac{1}{2\times (2-1)+1}=\frac{1}{3}.
		\end{align*}
		
		As for the sequence $s_2$, it holds that
		\begin{align*}
		\pi(s_2)&=\insP(\varepsilon,(-f_{3,1}))\times\insP((-f_{3,1}),(-f_{3,1},-\{f_{1,1},f_{1,2}\}))= \frac{1}{14}\times\frac{1}{10}
		\end{align*}
		while for the sequence $s_2'$, it holds that
		\begin{align*}
		\pi(s_2')&=\insP(\varepsilon,(-f_{3,1}))\times\insP((-f_{3,1}),(-f_{3,1},-f_{1,2}))\\
		&\times\insP((-f_{3,1},-f_{1,2}),(-f_{3,1},-f_{1,2},-f_{2,1}))\\
		&\times\insP((-f_{3,1},-f_{1,2},-f_{2,1}),(-f_{3,1},-f_{1,2},-f_{2,1},-f_{1,3}))\\
		&= \frac{1}{14}\times \frac{1}{10}\times \frac{1}{5}\times \frac{1}{3}.
		\end{align*}
		Again, the probability of applying the operation $-f_{3,1}$ is the same in $s_2$ and $s_2'$. The probability of applying $-\{f_{1,1},f_{1,2}\}$ in $s_2$ is the same as the probability of applying the operation $-f_{1,2}$ in $s_2'$, and the probability of the two additional operations is again $\frac{1}{5}\times\frac{1}{3}$. \hfill\markfull
		\end{example}
	
	For every $j\in\{i+1,\dots,n\}$, we denote by $r_j$ the difference between $N_j$ and $N_j'$ (that is, $N_j'=N_j+r_j$). Hence, it holds that
	\begin{align*}
	\pi(s)&=\pi(s')\times \left[\frac{1}{N_{i}}\right]\times \frac{1}{N_{i+1}}\times\dots\times \frac{1}{N_{n}}\times (N_{i+1}+r_{i+1})\times\dots\times\\ &\times(N_{n}+r_n)\times (2\ell+1)\times\dots\times 3\\
	&\le\pi(s')\times \frac{1}{N_{i+1}}\times\dots\times \frac{1}{N_{n}}\times (N_{i+1}+r_{i+1})\times\dots\times\\ &\times(N_{n}+r_n)\times (2\ell+1)\times\dots\times 3
	\end{align*}
	Note that here, the term $\frac{1}{N_i}$ only appears if the original sequence $s$ removes $f$ alone, in which case the term $\frac{1}{N_i}$ does not appear in the expression for $\pi(s')$.
	We will show that
	\begin{align*}
	\frac{1}{N_{i+1}}\times\dots\times \frac{1}{N_{n}}&\times (N_{i+1}+r_{i+1})\times\dots\times (N_{n}+r_n)\\
	&\times (2\ell+1)\times\dots\times 3\le \mathsf{pol}''(||D||)
	\end{align*}
	for some polynomial $\mathsf{pol}''$, or, equivalently,
	\begin{align*}
	(N_{i+1}+r_{i+1})\times\dots\times (N_{n}+r_n)&\times (2\ell+1)\times\dots\times 3\\
	&\le \mathsf{pol}''(||D||)\times N_{i+1}\times\dots\times N_{n}.
	\end{align*}
	Note that since $\ell\le k$, and $k$ is a constant (since we are interested in data complexity), $(2\ell+1)\times\dots\times 3$ is bounded by a constant. From this point, we denote this value by $c$. Thus, we prove that
	\[
	(N_{i+1}+r_{i+1})\times\dots\times (N_{n}+r_n)\times c\le \mathsf{pol}''(||D||)\times N_{i+1}\times\dots\times N_{n}.
	\]
	
	To show the above, we need to reason about
	the values $r_j$. For $j\in\{i+1,\dots,n\}$, let $N_j^f$ be the number of facts in the database that conflict with $f$ after applying all the operations of $s'$ that occur before $\op_j$, and before applying the operation $\op_j$. Moreover, for every $p\in\{1,\dots,k\}$, let $n_j^p$ be the number of facts in the database that violate the $p$th key jointly with $f$ at that point. Note that $n_j^1+\dots+n_j^p\ge N_j^f$, as the same fact might violate several distinct keys jointly with $f$. If $n_j^p\ge 2$, then every fact that violates the $p$th key jointly with $f$ participates in a violation of the constraints even if $f$ is not present in the database (as all the facts that violate the same key with $f$ also violate this key among themselves). Hence, for each one of these $n_j^p$ facts, the operation that removes this fact is a justified repairing operation regardless of the presence or absence of $f$ in the database, and it is counted as one of the $N_j$ operations that can be applied at that point in the sequence $s$. The addition of $f$ then adds $n_j^p$ new justified operations (the removal of a pair of facts that includes $f$ and one of the $n_j^p$ conflicting facts).
	
	On the other hand, if $n_j^p=1$, then the single fact that violates the $p$th key jointly with $f$ at that point might not participate in any violation once we remove $f$. In this case, the presence of $f$ implies two additional justified operations in $s'$ compared to $s$---the removal of this fact by itself and a pair removal that includes $f$ and this fact. If $n_j^p=0$, then clearly the $p$th key has no impact on the number of justified repairing operations w.r.t.~$f$ at that point. Now, assume, without loss of generality, that for some $1\le p_1<p_2\le k$, it holds that $n_j^p\ge 2$ for all $p\le p_1$, $n_j^p=1$ for all $p_1<p\le p_2$, and $n_j^p=0$ for all $p>p_2$. It then holds that
	\[
	r_j\le N_j^f+(p_2-p_1)+1
	\]
	($N_j^f$ operations remove $f$ jointly with one of its conflicting facts, at most $p_2-p_1$ operations remove a fact that violates the $p$th key with $f$ if $n_j^p=1$, and one operation removes $f$ itself.) Moreover,
	\begin{align*}
	N_j&\ge n_j^1+\dots+n_j^{p_1}+\frac{n_j^1(n_j^1-1)}{2}+\dots+\frac{n_j^{p_1}(n_j^{p_1}-1)}{2}\\
	&=\frac{(n_j^1)^2+\dots+(n_j^{p_1})^2+n_j^1+\dots+n_j^{p_1}}{2}.
	\end{align*}
	Because, as already said, for every $p$ with $n_j^p\ge 2$, the $n_j^p$ operations that remove the facts that violate the $p$th key with $f$ are also justified operations at the $j$th step in $s$, and there are $\frac{n_j^p(n_j^p-1)}{2}$ additional justified operations that remove a pair from these $n_j^p$ facts, as each such pair of facts jointly violates the $p$th key.
	
	\begin{example}\label{example:keys_proof_3}
		We continue with Example~\ref{example:keys_proof_2}. Let
		\[s_3=-f_{3,1}, -f_{1,1}, -f_{1,2}\]
		Before applying the operation $-f_{1,2}$ of $s_3$, there are five justified operations:
		\[-f_{1,2} \,\,\,\, -f_{1,3} \,\,\,\, -f_{3,2} \,\,\,\, -\{f_{1,2},f_{1,3}\} \,\,\,\, -\{f_{1,2},f_{3,2}\} \]
		At this point, the database contains three facts that conflict with $f_{1,1}$. The facts $f_{1,2}$ and $f_{1,3}$ jointly violate with it the key $R:A_1\rightarrow A_2$, while the fact $f_{2,1}$ jointly violates with it the key $R:A_2\rightarrow A_1$. 
		
		Observe that the operations $-f_{1,2},-f_{1,3},-\{f_{1,2},f_{1,3}\}$ are justified operations at this point, even though the fact $f_{1,1}$ no longer appears in the database, because $f_{1,2}$ conflict with $f_{1,3}$. If we bring $f_{1,1}$ back, we will have two additional justified operations that involve these fact (one for each fact): $-\{f_{1,1},f_{1,2}\}$ and $-\{f_{1,1},f_{1,3}\}$.
		
		Contrarily, the fact $f_{2,1}$ is not involved in any violation of the constraints at this point (before applying the operation $-f_{1,2}$ of $s_3$); hence, removing this fact is not a justified operation. However, if we bring $f_{1,1}$ back, we will have two additional justified operations that involve this fact: $-f_{2,1}$ and $-\{f_{1,1},f_{2,1}\}$.
		
		Finally, the fact $f_{1,1}$ introduces another justified operation---the removal of this fact by itself ($-f_{1,1}$). Hence, in the sequence $s_3'$ that $s_3$ is mapped to
		\[s_3'=-f_{3,1}, -f_{1,2}, -f_{2,1}, -f_{1,3}\]
		The number of justified operations before applying the operation $-f_{1,2}$ is ten, while the number of justified operations before applying this operation in $s_3$ is five. That is,
		\[\insP((-f_{3,1}, -f_{1,1}),(-f_{3,1}, -f_{1,1},-f_{1,2}))=\frac{1}{5}\]
		and
		\begin{flalign*}
		&& \insP((-f_{3,1}),(-f_{3,1},-f_{1,2}))=\frac{1}{5+5}=\frac{1}{10} && \markfull
		\end{flalign*}
		\end{example}
	
	According to the Cauchy–Schwarz inequality for $n$-dimensional euclidean spaces, it holds that
	\[
	\left(\sum_{i=1}^v x_iy_i\right)^2\le \left(\sum_{i=1}^v x_i^2\right)\times \left(\sum_{i=1}^v y_i^2\right),
	\]
	where $v \geq 1$ is an integer, and $x_i,y_i$ for $i \in [v]$ are real numbers.
	By defining $y_i=1$ for every $i\in[v]$, we then obtain that
	\begin{align*}
	\left(x_1+\dots+x_v\right)^2\le  v\times \left(x_1^2+\dots+x_v^2\right).
	\end{align*}
	Hence, we have that
	\begin{align*}
	N_j&\ge \frac{(n_j^1)^2+\dots+(n_j^{p_1})^2+n_j^1+\dots+n_j^{p_1}}{2}\\
	&\ge \frac{\frac{(n_j^1+\dots+n_j^{p_1})^2}{p_1}+n_j^1+\dots+n_j^{p_1}}{2}\\
	&=\frac{(n_j^1+\dots+n_j^{p_1})^2+p_1\times (n_j^1+\dots+ n_j^{p_1})}{2p_1}\\
	&\ge \frac{(N_j^f-(p_2-p_1))^2+p_1\times[ N_j^f-(p_2-p_1)]}{2p_1}.
	\end{align*}
	Note that $N_j^f-(p_2-p_1)$ is a lower bound on 
	$n_j^1+\dots+n_j^{p_1}$ because for every $p_2\le p$, there are no facts that violate the $p$th key with $f$, and for $p_1< p\le p_2$, there is a single fact that violates the $p$th key with $f$; hence, $n_j^{p_1+1}+\dots+n_j^{p_2}\le p_2-p_1$ and $n_j^{p_2+1}+\dots+n_j^{k}=0$. As aforementioned, $n_j^1+\dots+n_j^k\ge N_j^f$. Therefore,
	\begin{align*}
	n_j^1+\dots+n_j^{p_1}&\ge N_j^f-(n_j^{p_1+1}+\dots+n_j^{p_2})-(n_j^{p_2+1}+\dots+n_j^{k})\\
	&\ge N_j^f-(p_2-p_1).
	\end{align*}

	We conclude that
	\[r_j\le N_j^f+(p_2-p_1)+1\]
	and
	\[N_j\ge \frac{(N_j^f-(p_2-p_1))^2+p_1\times[ N_j^f-(p_2-p_1)]}{2p_1}.\]
	Hence, it holds that
	\[
	N_j\ge \frac{(r_j-2(p_2-p_1)-1)^2+p_1\times[ r_j-2(p_2-p_1)-1]}{2p_1}.
	\]
	If $r_j\ge 2(p_2-p_1)+1$, then $p_1\times[ r_j-2(p_2-p_1)-1]\ge 0$ and
	\[
	N_j\ge \frac{(r_j-2(p_2-p_1)-1)^2}{2p_1}
	\]
	and
	\begin{align*}
	r_j&\le \sqrt{2p_1N_j}+2(p_2-p_1)+1\le \sqrt{2kN_j}+2k+k\\
	&\le  \sqrt{4k^2N_j}+3k\sqrt{N_j}=5k\sqrt{N_j}. 
	\end{align*}
	If $r_j< 2(p_2-p_1)+1$, then $r_j\le 2k+k\le 5k\sqrt{N_j}$. Thus, in both cases, we have that $r_j\le 5k\sqrt{N_j}$.
	
	Recall that our goal is to show that
	\[
	(N_{i+1}+r_{i+1})\times\dots\times (N_{n}+r_n)\times c\le \mathsf{pol}''(|D|)\times N_{i+1}\times\dots\times N_{n}.
	\]
	We have that
	\[
	(N_{i+1}+r_{i+1})\times\dots\times (N_{n}+r_n)\le (N_{i+1}+5k\sqrt{N_{i+1}})\times\dots\times (N_{n}+5k\sqrt{N_{n}}).
	\]
	Thus, it suffices to show that
	\[
	(\sqrt{N_{i+1}}+5k)\times\dots\times (\sqrt{N_{n}}+5k)\times c\le \mathsf{pol}''(||D||)\times  \sqrt{N_{i+1}}\times\dots\times \sqrt{N_{n}}.
	\]
	For brevity, let $x_j=\sqrt{N_j}$. Moreover, we can clearly define $\mathsf{pol}''(||D||)$ as $c\times \mathsf{pol}'''(||D||)$ for some polynomial $\mathsf{pol}'''$, and get rid of the constant $c$. Therefore, we now show that
	\[
	(x_{i+1}+5k)\times\dots\times (x_{n}+5k)\le \mathsf{pol}'''(||D||)\times  x_{i+1}\times\dots\times x_{n}
	\]
	for some polynomial $\mathsf{pol}'''$, 
	or, equivalently,
	\[
	\frac{x_{i+1}+5k}{x_{i+1}}\times\dots\times\frac{x_{n}+5k}{x_{n}}\le \mathsf{pol}'''(||D||).
	\]
	Note that in the sequence $s$, there are $n-j+1$ operations after the operation $\op_j$ (including the operation $\op_j$). Since the number of justified operations can only decrease after applying a certain operation, this means that $N_j\ge n-j+1$. Hence, we have that $N_{i+1}\ge n-i$, $N_{i+2}\ge n-i-1$, and so on, which implies that $x_{i+1}\ge \sqrt{n-i}$, $x_{i+2}\ge \sqrt{n-i-1}$, etc. Now, an expression of the form $\frac{x+5k}{x}$ increases when the value of $x$ decreases (because $\frac{x+5k}{x}=1+\frac{5k}{x}$); hence, we have that
	\begin{align*}
	\frac{x_{i+1}+5k}{x_{i+1}}&\times\dots\times\frac{x_{n}+5k}{x_{n}}\\
	&\le \frac{\sqrt{n-i}+5k}{\sqrt{n-i}}\times \frac{\sqrt{n-i-1}+5k}{\sqrt{n-i-1}}\times\dots\times\frac{1+5k}{1}\\
	&\le \frac{\floor*{\sqrt{n-i}}+5k}{\floor*{\sqrt{n-i}}}\times \frac{\floor*{\sqrt{n-i-1}}+5k}{\floor*{\sqrt{n-i-1}}}\times\dots\times\frac{1+5k}{1}
	\end{align*}
	
	Next, for every $m\ge 1$ it holds that
	\[
	\sqrt{m-1}\ge \sqrt{m}-1
	\]
	and thus,
	\[
	\floor*{\sqrt{m-1}}\ge \floor*{\sqrt{m}}-1
	\]
	We then obtain the following:
	\begin{align*}
	&\frac{\floor*{\sqrt{n-i}}+5k}{\floor*{\sqrt{n-i}}}\times \frac{\floor*{\sqrt{n-i-1}}+5k}{\floor*{\sqrt{n-i-1}}}\times\dots\times\frac{1+5k}{1}\\
	&\le \frac{\floor*{\sqrt{n-i}}+5k}{\floor*{\sqrt{n-i}}}\times  \frac{\floor*{\sqrt{n-i}}-1+5k}{\floor*{\sqrt{n-i}}-1}\times\dots\times\frac{1+5k}{1}\\
	&=\frac{(\floor*{\sqrt{n-i}}+5k)!}{(\floor*{\sqrt{n-i}})!\times (5k)!}={\floor*{\sqrt{{n-i}}}+5k\choose 5k}\\
	&\le \left(\frac{e(\floor*{\sqrt{{n-i}}}+5k)}{5k}\right)^{5k}\le \left(\frac{e(\floor*{\sqrt{{n}}}+5k)}{5k}\right)^{5k}\\
	&\le \left(\frac{e}{5k}\right)^{5k}\times (\sqrt{|D|}+5k)^{5k}
	\end{align*}
	(Observe that the maximal length $n$ of a sequence is $|D|-1$.)
	The claim follows with
	\[
	\mathsf{pol}'''(||D||)\ =\ \left(\frac{e}{5k}\right)^{5k}\times (\sqrt{||D||}+5k)^{5k}.
	\]
	
	Recall that $c=(2\ell+1)\times\dots\times 3$, where $\ell$ is the number of facts that conflict with $f$ and are not removed by the sequence $s$; hence, $\ell\le k$. Therefore, for every sequence $s$ that removes $f$, there is some sequence $s'$ that does not remove $f$ such that
	\[
	\pi(s)\ \le\ (2k+1)!\times \mathsf{pol}'''(||D||) \times \pi(s'),
	\]
	and item (1) of Lemma~\ref{lemma:relate-sequences} follows with 
	\[
	\mathsf{pol}''(||D||)\ =\ (2k+1)! \times \mathsf{pol}'''(||D||).
	\]
	
	\medskip
	\noindent \paragraph{Item (2).}
	We now show that the function $\mathsf{F}$ from sequences that remove $f$ to sequences that do not remove $f$, maps at most $2|D|-1$ sequences of the first type to the same sequence of the second type.
	Given a sequence $s' \in S_f$, we can obtain this sequence either from a sequence $s\in\abs{M_{\dep}^{\uo}(D)}$ that has one additional operation that removes $f$, or from a sequence $s$ that removes $f$ jointly with some other fact $g$, while $s'$ removes the fact $g$ by itself. (Some of the operations at the end of $s'$ might not appear in $s$, as they remove facts that conflict only with $f$.) Since the length of the sequence $s'$ is at most $|D|-1$, there are at most $|D|$ possible ways to insert an additional operation that removes $f$, and $|D|-1$ ways to add $f$ to an existing operation. Hence, there are at most $|D|+|D|-1$ sequences that remove $f$ that are mapped to the sequence $s'$.
	Here is an example that illustrates the above combinatorial argument.
	
	\begin{example}
		We continue with Example~\ref{example:keys_proof_3}. Consider again the sequence $s_3'$. Recall that
		\[s_3'=-f_{3,1}, -f_{1,2}, -f_{2,1}, -f_{1,3}\]
		This sequence can be obtained from any of the following sequences that have an additional operation that removes $f_{1,1}$:
		\[-f_{1,1}, -f_{3,1}, -f_{1,2}\]
		\[-f_{3,1}, -f_{1,1}, -f_{1,2}\]
		\[-f_{3,1}, -f_{1,2}, -f_{1,1}\]
		Note that the operations $-f_{2,1}, -f_{1,3}$ do not appear in these sequences, as after removing $f_{1,1}$ they are no longer involved in violations of the constraints.
		
		The sequence $s_3'$ can also be obtained from the following sequences that replace an operation of $s_3'$ that removes a single fact with an operation that removes a pair of conflicting facts:
		\begin{flalign*}
		&&-\{f_{1,1},f_{3,1}\}, -f_{1,2}&&\\
		&& -f_{3,1}, -\{f_{1,1},f_{1,2}\}&& \markfull
		\end{flalign*}
		\end{example}
This completes the proof of Lemma~\ref{lemma:relate-sequences}.
\end{proof}

Having Lemma~\ref{lemma:relate-sequences} in place, it is now easy to establish the existence of the polynomial $\mathsf{pol}'$ such that $\Lambda_{\neg f} \leq \mathsf{pol}'(||D||) \cdot \Lambda_f$. Indeed, with $\mathsf{F}$ and $\mathsf{pol}''$ being the function and the polynomial, respectively, provided by Lemma~\ref{lemma:relate-sequences},
\begin{eqnarray*}
	\Lambda_{\neg f}\ =\ \sum_{s \in S_{\neg f}}\pi(s)\ &\leq&  \sum_{s \in S_{\neg f}} \mathsf{pol}''(||D||) \cdot \pi(\mathsf{F}(s))\\
	&\leq& \mathsf{pol}''(||D||) \cdot (2 \cdot ||D|| - 1) \cdot \sum_{s \in S_f} \pi(s)\\
	&=& \mathsf{pol}''(||D||) \cdot (2 \cdot ||D|| - 1) \cdot \Lambda_f,
\end{eqnarray*}
and the claim follows with $\mathsf{pol}'(||D||) = \mathsf{pol}''(||D||) \cdot (2 \cdot ||D|| - 1)$.

\OMIT{
With $k$ being the number of keys in $\dep$ over the relation name of $f$, we can show that
\begin{align*}
&\probhom{D,M_\dep^{\uo},Q}{h}\ge \\
        &   \frac{1}{\left[1+(2k+1)!\times\left(\frac{e}{5k}\right)^{5k}\times (\sqrt{|D|}+5k)^{5k}\right]\times(2|D|-1)}
\end{align*}

We first prove that, for every repairing sequence $s\in \abs{M_{\dep}^{\uo}(D)}$ that removes $f$ (i.e., $f\not\in s(D)$), there is a repairing sequence $s'\in\abs{M_{\dep}^{\uo}(D)}$ that does not remove $f$ (i.e., $f\in s'(D)$) such that:
\[\pi(s)\le \mathsf{pol}''(|D|)\times \pi(s')\]
where $\mathsf{pol}''$ is some polynomial.
Let $s\in \abs{M_{\dep}^{\uo}(D)}$ be a repairing sequence that removes $f$. We transform $s$ into a repairing sequence $s'\in \abs{M_{\dep}^{\uo}(D)}$ that does not remove $f$ by deleting or replacing the operation that removes $f$, and adding additional operations at the end of the sequence as follows. Assume that:
\[s= \op_1, \,\,\,\, \op_2, \,\,\,\, \dots  \,\,\,\, ,\op_{i-1}, \,\,\,\, \op_i, \,\,\,\, \op_{i+1}, \,\,\,\, \dots  \,\,\,\, ,\op_n\]
where the operation $\op_i$ removes $f$ (as a single fact removal). Then, we will have that:
\[s'= \op_1, \,\,\,\, \op_2, \,\,\,\, \dots  \,\,\,\, ,\op_{i-1}, \,\,\,\, \op_{i+1}, \,\,\,\, \dots  \,\,\,\, ,\op_n, \,\,\,\, \op_1', \,\,\,\, \dots \,\,\,\, ,\op_\ell'\]
where $\op_1',\dots,\op_\ell'$ are new operations that we will describe later. If $\op_i$ removes a pair $\{f,g\}$ of facts that includes $f$, then:
\begin{align*}
   s'= \op_1, \,\,\,\, \op_2, \,\,\,\, \dots  \,\,\,\, ,\op_{i-1}, \,\,\,\, ,\op_i^*, \,\,\,\, &\op_{i+1}, \,\,\,\, \dots  \,\,\,\, ,\op_n, \\
   &\,\,\,\, \op_1', \,\,\,\, \dots \,\,\,\, ,\op_\ell'
\end{align*}
where $\op_i^*$ is the operation that removes only the fact $g$.

An important observation here is that since the sequence $s$ removes $f$, the repair $s(D)$ might contain facts that conflict with $f$, but at most $k$ such facts. This is a property of keys. If $s(D)$ contains $k+1$ facts that conflict with $f$, then it contains two facts $g_1,g_2$ that violate the same key with $f$, in which case $g_1,g_2$ also jointly violate this key and cannot appear in the same repair. Therefore, at the end of the sequence $s'$ we add $\ell$ new operations (for some $\ell\le k$) that remove the facts of $s(D)$ that conflict with $f$, in some arbitrary order. Note that the sequence $s'$ is a valid repairing sequence, as an additional fact (the fact $f$) cannot invalidate a justified repairing operation, and we can remove the $\ell$ conflicting facts at the end in any order, as they are all in conflict with $f$.

\begin{example}\label{example:keys_proof_1}
Consider again the database $D$ of Figure~\ref{fig:example_pkeys}, and the set $\dep=\{R:A_1\rightarrow A_2, R:A_2\rightarrow A_1\}$. Also consider the query $Q$ and homomorphism $h$ from Example~\ref{example:pkeys_ur}. Recall that $h(Q)=\{R(a_1,b_1)\}$. The following sequence is a sequence that removes the fact $R(a_1,b_1)$:
\[s_1=-f_{1,2},-f_{1,1},-f_{3,1}\]
Note that $s(D)$ contains the facts $f_{1,3}$ and $f_{2,1}$ that conflict with $f_{1,1}$.
This sequence is mapped to the following sequence $s'$:
\[s_1'=-f_{1,2},-f_{3,1},-f_{1,3},-f_{2,1}\]
where we delete the operation $-f_{1,1}$ that removes the fact of $h(Q)$, and add, at the end of the sequence, the operations $-f_{1,3}$ and $-f_{2,1}$ that remove the facts of $s(D)$ that conflict with $f_{1,1}$.

As another example, the sequence:
\[s_2=-f_{3,1},-\{f_{1,1},f_{1,2}\}\]
is mapped to the sequence:
\[s_2'=-f_{3,1},-f_{1,2},-f_{2,1},-f_{1,3}.\]
Here, the pair removal $-\{f_{1,1},f_{1,2}\}$ is replaced by a single-fact removal $-f_{1,2}$, and, at the end of the sequence, we again add two additional operations that remove (in some arbitrary order) the facts $f_{1,3}$ and $f_{2,1}$ that conflict with $f_{1,1}$.
\qed\end{example}

Now, according to the definition of $M_\dep^\uo$, we have that:
\[
\pi(s)=\frac{1}{N_1}\times \frac{1}{N_2}\times\dots\times \frac{1}{N_{i-1}}\times\frac{1}{N_{i}}\times \frac{1}{N_{i+1}}\times\dots\times \frac{1}{N_{n}}
\]
where $N_j$ is the total number of $(D^s_{j-1},\dep)$-justified repairing operations before applying the operation $\op_j$ of the sequence (recall that $D^s_{j-1}$ is the database obtained from $D$ by applying the first $j-1$ operations of $s$). Hence,
\[\insP((\op_1,\dots,\op_{j-1}),(\op_1,\dots,\op_{j}))=\frac{1}{N_j}\]
Then,
\begin{align*}
\pi(s')=&\frac{1}{N_1}\times \frac{1}{N_2}\times\dots\times \frac{1}{N_{i-1}}\times\left[\frac{1}{N_{i}}\right]\times \frac{1}{N'_{i+1}}\times\dots\times \frac{1}{N'_n}\times\\
&\frac{1}{2\ell+1}\times\frac{1}{2(\ell-1)+1}\times\frac{1}{3}
\end{align*}
The probability $\insP((\op_1,\dots,\op_{j-1}),(\op_1,\dots,\op_{j}))$, for $2\le j\le i-1$, is not affected by the decision to remove or keep $f$ at the $i$th step. The probability $\insP((\op_1,\dots,\op_{j-1}),(\op_1,\dots,\op_{j}))$ for $i+2\le j\le n$, on the other hand, might decrease in the sequence $s'$ compared to the sequence $s$, because the additional fact $f$ (that is removed by $s$ but not by $s'$) might be involved in violations with the remaining facts of the database and introduce additional justified repairing operations, in which case $N_j\le N_j'$. Similarly, the probability $\insP((\op_1,\dots,\op_{i-1}),(\op_1,\dots,\op_{i-1},\op_{i+1}))$ (in the case where $\op_i=-\{f\}$) or $\insP((\op_1,\dots,\op_{i}^\star),(\op_1,\dots,\op_{i}^\star,\op_{i+1}))$ (in the case where $\op_i=-\{f,g\}$) can only decrease compared to the probability $\insP((\op_1,\dots,\op_{i}),(\op_1,\dots,\op_{i},\op_{i+1}))$ in $s$; hence, $N_{i+1}\le N_{i+1}'$.

The term $\frac{1}{N_i}$ denotes the probability of the operation $\op^\star_i$, and so it only appears in the expression in the case where the sequence $s$ removes the fact $f$ jointly with some other fact $g$ (and the operation $\op^\star_i$ removes $g$ by itself). Since all the $(D_{i-1}^{s'},\dep)$-justified operations have the same probability to be selected, the probabilities $\insP((\op_1,\dots,\op_{i-1}),(\op_1,\dots,\op_{i}))$ and $\insP((\op_1,\dots,\op_{i-1}),(\op_1,\dots,\op_{i}^\star)$ are the same.
Finally, at the end of the sequence, the only remaining conflicts are those involving $f$. As aforementioned, there are $\ell$ facts that conflict with $f$ for some $\ell\le k$ at that point, and each one of them violates a different key with $f$. Hence, there are $2\ell+1$ justified operations before applying $\op_1'$ (removing one of the $\ell$ conflicting facts, removing one of these facts jointly with $f$, or removing $f$), there are $2(\ell-1)+1$ possible operations before applying $\op_2'$ and so on.

\begin{example}\label{example:keys_proof_2}
We continue with Example~\ref{example:keys_proof_1}. For the sequence $s_1$, we have that:
\begin{align*}
 \pi(s_1)&=\insP(\varepsilon,(-f_{1,2}))\times\insP((-f_{1,2}),(-f_{1,2},-f_{1,1}))\\
 &\times\insP((-f_{1,2},-f_{1,1}),(-f_{1,2},-f_{1,1},-f_{3,1}))= \frac{1}{14}\times\frac{1}{10}\times\frac{1}{5}
\end{align*}
This holds since, at first, all six facts are involved in violations of the constraints and there are eight conflicting pairs; hence, the total number of justified operations is $14$. After removing the fact $f_{1,2}$, the number of justified operations reduces to $10$, and after removing the fact $f_{1,1}$, this number is $5$.
Now, for the sequence $s_1'$, it holds that:
\begin{align*}
 \pi(s_1')&=\insP(\varepsilon,(-f_{1,2}))\times\insP((-f_{1,2}),(-f_{1,2},-f_{3,1}))\\
 &\times\insP((-f_{1,2},-f_{3,1}),(-f_{1,2},-f_{3,1},-f_{1,3}))\\
 &\times\insP((-f_{1,2},-f_{3,1},-f_{1,3}),(-f_{1,2},-f_{3,1},-f_{1,3},-f_{2,1}))\\
 &= \frac{1}{14}\times\frac{1}{10}\times\frac{1}{5}\times \frac{1}{3}
\end{align*}
And, indeed, the probability of applying the operation $-f_{1,2}$ (i.e., $\insP(\varepsilon,(-f_{1,2}))$) is the same for both sequences ($\frac{1}{14}$), while the probability of applying the operation $-f_{3,1}$ in  $s_1'$ (i.e., $\insP((-f_{1,2}),(-f_{1,2},-f_{3,1}))$) is smaller than the probability ($\insP((-f_{1,2},-f_{1,1}),(-f_{1,2},-f_{1,1},-f_{3,1}))$) of applying this operation in $s_1$: $\frac{1}{10}$ compared to $\frac{1}{5}$. Finally, there are $\ell=2$ facts in $s(D)$ that conflict with $f_{1,1}$ and we have that:
\[\insP((-f_{1,2},-f_{3,1}),(-f_{1,2},-f_{3,1},-f_{1,3}))=\frac{1}{2\times 2+1}=\frac{1}{5}\]
\begin{align*}
    &\insP((-f_{1,2},-f_{3,1},-f_{1,3}),(-f_{1,2},-f_{3,1},-f_{1,3},-f_{2,1}))\\
    &=\frac{1}{2\times (2-1)+1}=\frac{1}{3}
\end{align*}

As for the sequence $s_2$, it holds that:
\begin{align*}
 \pi(s_2)&=\insP(\varepsilon,(-f_{3,1}))\times\insP((-f_{3,1}),(-f_{3,1},-\{f_{1,1},f_{1,2}\}))= \frac{1}{14}\times\frac{1}{10}
\end{align*}
while for the sequence $s_2'$, the following holds:
\begin{align*}
 \pi(s_2')&=\insP(\varepsilon,(-f_{3,1}))\times\insP((-f_{3,1}),(-f_{3,1},-f_{1,2}))\\
 &\times\insP((-f_{3,1},-f_{1,2}),(-f_{3,1},-f_{1,2},-f_{2,1}))\\
 &\times\insP((-f_{3,1},-f_{1,2},-f_{2,1}),(-f_{3,1},-f_{1,2},-f_{2,1},-f_{1,3}))\\
 &= \frac{1}{14}\times \frac{1}{10}\times \frac{1}{5}\times \frac{1}{3}
\end{align*}
And, again, the probability of applying the operation $-f_{3,1}$ is the same in $s_2$ and $s_2'$. The probability of applying the operation $-\{f_{1,1},f_{1,2}\}$ in $s_2$ is the same as the probability of applying the operation $-f_{1,2}$ in $s_2'$, and the probability of the two additional operations is again $\frac{1}{5}\times\frac{1}{3}$.
\qed\end{example}

For every $j\in\{i+1,\dots,n\}$, we denote by $r_j$ the difference between $N_j$ and $N_j'$ (that is, $N_j'=N_j+r_j$). Hence, it holds that:
\begin{align*}
\pi(s)&=\pi(s')\times \left[\frac{1}{N_{i}}\right]\times \frac{1}{N_{i+1}}\times\dots\times \frac{1}{N_{n}}\times (N_{i+1}+r_{i+1})\times\dots\times\\ &\times(N_{n}+r_n)\times (2\ell+1)\times\dots\times 3\\
&\le\pi(s')\times \frac{1}{N_{i+1}}\times\dots\times \frac{1}{N_{n}}\times (N_{i+1}+r_{i+1})\times\dots\times\\ &\times(N_{n}+r_n)\times (2\ell+1)\times\dots\times 3
\end{align*}
Note that here, the term $\frac{1}{N_i}$ only appears if the original sequence $s$ removes $f$ alone, in which case the term $\frac{1}{N_i}$ does not appear in the expression for $\pi(s')$.
We will show that:
\begin{align*}
  \frac{1}{N_{i+1}}\times\dots\times \frac{1}{N_{n}}&\times (N_{i+1}+r_{i+1})\times\dots\times (N_{n}+r_n)\\
  &\times (2\ell+1)\times\dots\times 3\le \mathsf{pol}''(|D|)
\end{align*}
for some polynomial $\mathsf{pol}''$, or, equivalently:
\begin{align*}
 (N_{i+1}+r_{i+1})\times\dots\times (N_{n}+r_n)&\times (2\ell+1)\times\dots\times 3\\
 &\le \mathsf{pol}''(|D|)\times N_{i+1}\times\dots\times N_{n}
\end{align*}
Note that since $\ell\le k$, and $k$ is a constant when considering data complexity, the value $(2\ell+1)\times\dots\times 3$ is bounded by a constant. From this point, we denote this value by $c$. Thus, we prove that:
\[
 (N_{i+1}+r_{i+1})\times\dots\times (N_{n}+r_n)\times c\le \mathsf{pol}''(|D|)\times N_{i+1}\times\dots\times N_{n}
\]

To show the above, we need to reason about
the values $r_j$. For $j\in\{i+1,\dots,n\}$, let $N_j^f$ be the number of facts in the database that conflict with $f$ after applying all the operations of $s'$ that occur before $\op_j$, and before applying the operation $\op_j$. Moreover, for every $p\in\{1,\dots,k\}$, let $n_j^p$ be the number of facts in the database that violate the $p$th key jointly with $f$ at that point. Note that $n_j^1+\dots+n_j^p\ge N_j^f$, as the same fact might violate several distinct keys jointly with $f$. If $n_j^p\ge 2$, then every fact that violates the $p$th key jointly with $f$ participates in a violation of the constraints even if $f$ is not present in the database (as all the facts that violate the same key with $f$ also violate this key among themselves). Hence, for each one of these $n_j^p$ facts, the operation that removes this fact is a justified repairing operation regardless of the presence or absence of $f$ in the database, and it is counted as one of the $N_j$ operations that can be applied at that point in the sequence $s$. The addition of $f$ then adds $n_j^p$ new justified operations (the removal of a pair of facts that includes $f$ and one of the $n_j^p$ conflicting facts).

On the other hand, if $n_j^p=1$, then the single fact that violates the $p$th key jointly with $f$ at that point might not participate in any violation once we remove $f$. In this case, the presence of $f$ implies two additional justified operations in $s'$ compared to $s$---the removal of this fact by itself and a pair removal that includes $f$ and this fact. If $n_j^p=0$, then clearly the $p$th key has no impact on the number of justified repairing operations w.r.t.~$f$ at that point. Now, assume, without loss of generality, that for some $1\le p_1<p_2\le k$, it holds that $n_j^p\ge 2$ for all $p\le p_1$, $n_j^p=1$ for all $p_1<p\le p_2$, and $n_j^p=0$ for all $p>p_2$. It then holds that:
\[
r_j\le N_j^f+(p_2-p_1)+1
\]
($N_j^f$ operations remove $f$ jointly with one of its conflicting facts, at most $p_2-p_1$ operations remove a fact that violates the $p$th key with $f$ if $n_j^p=1$, and one operation removes $f$ itself.) Moreover,
\begin{align*}
N_j&\ge n_j^1+\dots+n_j^{p_1}+\frac{n_j^1(n_j^1-1)}{2}+\dots+\frac{n_j^{p_1}(n_j^{p_1}-1)}{2}\\
&=\frac{(n_j^1)^2+\dots+(n_j^{p_1})^2+n_j^1+\dots+n_j^{p_1}}{2}
\end{align*}
Because, as aforementioned, for every $p$ with $n_j^p\ge 2$, the $n_j^p$ operations that remove the facts that violate the $p$th key with $f$ are also justified operations at the $j$th step in $s$, and there are additional $\frac{n_j^p(n_j^p-1)}{2}$ justified operations that remove a pair from these $n_j^p$ facts, as each such pair of facts jointly violates the $p$th key.

\begin{example}\label{example:keys_proof_3}
We continue with Example~\ref{example:keys_proof_2}. Consider the sequence:
\[s_3=-f_{3,1}, -f_{1,1}, -f_{1,2}\]
Before applying the operation $-f_{1,2}$ of $s_3$, there are five justified operations:
\[-f_{1,2} \,\,\,\, -f_{1,3} \,\,\,\, -f_{3,2} \,\,\,\, -\{f_{1,2},f_{1,3}\} \,\,\,\, -\{f_{1,2},f_{3,2}\} \]
At this point, the database contains three facts that conflict with $f_{1,1}$. The facts $f_{1,2}$ and $f_{1,3}$ jointly violate with it the key $R:A_1\rightarrow A_2$, while the fact $f_{2,1}$ jointly violates with it the key $R:A_2\rightarrow A_1$. 

Observe that the operations $-f_{1,2},-f_{1,3},-\{f_{1,2},f_{1,3}\}$ are justified operations at this point, even though the fact $f_{1,1}$ no longer appears in the database, because $f_{1,2}$ conflict with $f_{1,3}$. If we bring $f_{1,1}$ back, we will have two additional justified operations that involves these fact (one for each fact): $-\{f_{1,1},f_{1,2}\}$ and $-\{f_{1,1},f_{1,3}\}$.

Contrarily, the fact $f_{2,1}$ is not involved in any violation of the constraints at this point (before applying the operation $-f_{1,2}$ of $s_3$); hence, removing this fact is not a justified operation. However, if we bring $f_{1,1}$ back, we will have two additional justified repairing operations that involve this fact: $-f_{2,1}$ and $-\{f_{1,1},f_{2,1}\}$.

Finally, the fact $f_{1,1}$ introduces another justified operation---the removal of this fact by itself ($-f_{1,1}$). Hence, in the sequence $s_3'$ that $s_3$ is mapped to:
\[s_3'=-f_{3,1}, -f_{1,2}, -f_{2,1}, -f_{1,3}\]
The number of justified operations before applying the operation $-f_{1,2}$ is ten, while the number of justified operations before applying this operation in $s_3$ is five. That is,
\[\insP((-f_{3,1}, -f_{1,1}),(-f_{3,1}, -f_{1,1},-f_{1,2}))=\frac{1}{5}\]
and:
\[\insP((-f_{3,1}),(-f_{3,1},-f_{1,2}))=\frac{1}{5+5}=\frac{1}{10}\]
\qed\end{example}

According to the Cauchy–Schwarz inequality for n-dimensional euclidean spaces, the following holds:
\[\left(\sum_{i=1}^v x_iy_i\right)^2\le \left(\sum_{i=1}^v x_i^2\right)\times \left(\sum_{i=1}^v y_i^2\right)\]
By defining $y_i=1$ for every $i\in\{1,\dots,v\}$ we then obtain the following:
\begin{align*}
    (x_1+\dots+x_v)^2\le  v\times (x_1^2+\dots+x_v^2)
\end{align*}
Hence, we have that:
\begin{align*}
N_j&\ge \frac{(n_j^1)^2+\dots+(n_j^{p_1})^2+n_j^1+\dots+n_j^{p_1}}{2}\\
&\ge \frac{\frac{(n_j^1+\dots+n_j^{p_1})^2}{p_1}+n_j^1+\dots+n_j^{p_1}}{2}\\
&=\frac{(n_j^1+\dots+n_j^{p_1})^2+p_1\times (n_j^1+\dots+ n_j^{p_1})}{2p_1}\\
&\ge \frac{(N_j^f-(p_2-p_1))^2+p_1\times[ N_j^f-(p_2-p_1)]}{2p_1}
\end{align*}
Note that $N_j^f-(p_2-p_1)$ is a lower bound on 
$n_j^1+\dots+n_j^{p_1}$ because for every $p_2\le p$, there are no facts that violate the $p$th key with $f$, and for $p_1< p\le p_2$, there is a single fact that violates the $p$th key with $f$; hence, $n_j^{p_1+1}+\dots+n_j^{p_2}\le p_2-p_1$ and $n_j^{p_2+1}+\dots+n_j^{k}=0$. As aforementioned, $n_j^1+\dots+n_j^k\ge N_j^f$. Therefore,
\begin{align*}
 n_j^1+\dots+n_j^{p_1}&\ge N_j^f-(n_j^{p_1+1}+\dots+n_j^{p_2})-(n_j^{p_2+1}+\dots+n_j^{k})\\
 &\ge N_j^f-(p_2-p_1)
\end{align*}

We conclude that:
\[r_j\le N_j^f+(p_2-p_1)+1\]
and:
\[N_j\ge \frac{(N_j^f-(p_2-p_1))^2+p_1\times[ N_j^f-(p_2-p_1)]}{2p_1}\]
Hence, it holds that:
\[
N_j\ge \frac{(r_j-2(p_2-p_1)-1)^2+p_1\times[ r_j-2(p_2-p_1)-1]}{2p_1}
\]
If $r_j\ge 2(p_2-p_1)+1$ then $p_1\times[ r_j-2(p_2-p_1)-1]\ge 0$ and:
\[
N_j\ge \frac{(r_j-2(p_2-p_1)-1)^2}{2p_1}
\]
and:
\begin{align*}
   r_j&\le \sqrt{2p_1N_j}+2(p_2-p_1)+1\le \sqrt{2kN_j}+2k+k\\
   &\le  \sqrt{4k^2N_j}+3k\sqrt{N_j}=5k\sqrt{N_j} 
\end{align*}
If $r_j< 2(p_2-p_1)+1$ then $r_j\le 2k+k\le 5k\sqrt{N_j}$. So, in both cases, we have that $r_j\le 5k\sqrt{N_j}$.

Recall that our goal is to show the following:
\[
 (N_{i+1}+r_{i+1})\times\dots\times (N_{n}+r_n)\times c\le \mathsf{pol}''(|D|)\times N_{i+1}\times\dots\times N_{n}
\]
We have that:
\[
 (N_{i+1}+r_{i+1})\times\dots\times (N_{n}+r_n)\le (N_{i+1}+5k\sqrt{N_{i+1}})\times\dots\times (N_{n}+5k\sqrt{N_{n}})
\]
Thus, it is sufficient to show that:
\[
 (\sqrt{N_{i+1}}+5k)\times\dots\times (\sqrt{N_{n}}+5k)\times c\le \mathsf{pol}''(|D|)\times  \sqrt{N_{i+1}}\times\dots\times \sqrt{N_{n}}
\]
For convenience, we denote $x_j=\sqrt{N_j}$. Moreover, we can clearly define $\mathsf{pol}''(|D|)$ as $c\times \mathsf{pol}'''(|D|)$ for some polynomial $\mathsf{pol}'''$ and get rid of the constant $c$. Therefore, we now show that:
\[
 (x_{i+1}+5k)\times\dots\times (x_{n}+5k)\le \mathsf{pol}'''(|D|)\times  x_{i+1}\times\dots\times x_{n}
\]
for some polynomial $\mathsf{pol}'''$, 
or, equivalently, that:
\[
\frac{x_{i+1}+5k}{x_{i+1}}\times\dots\times\frac{x_{n}+5k}{x_{n}}\le \mathsf{pol}'''(|D|)
\]
Note that in the sequence $s$, there are $n-j+1$ operations after the operation $\op_j$ (including the operation $\op_j$). Since the number of justified operations can only decrease after applying a certain operation, this means that $N_j\ge n-j+1$. Hence, we have that: $N_{i+1}\ge n-i$, $N_{i+2}\ge n-i-1$, and so on, which implies that $x_{i+1}\ge \sqrt{n-i}$, $x_{i+2}\ge \sqrt{n-i-1}$, etc. Now, an expression of the form $\frac{x+5k}{x}$ increases when the value of $x$ decreases (because $\frac{x+5k}{x}=1+\frac{5k}{x}$); hence, we have that:
\begin{align*}
   \frac{x_{i+1}+5k}{x_{i+1}}&\times\dots\times\frac{x_{n}+5k}{x_{n}}\\
   &\le \frac{\sqrt{n-i}+5k}{\sqrt{n-i}}\times \frac{\sqrt{n-i-1}+5k}{\sqrt{n-i-1}}\times\dots\times\frac{1+5k}{1}\\
   &\le \frac{\floor*{\sqrt{n-i}}+5k}{\floor*{\sqrt{n-i}}}\times \frac{\floor*{\sqrt{n-i-1}}+5k}{\floor*{\sqrt{n-i-1}}}\times\dots\times\frac{1+5k}{1}
\end{align*}

Next, for every $m\ge 1$ it holds that: 
\[
\sqrt{m-1}\ge \sqrt{m}-1
\]
and so:
\[
\floor*{\sqrt{m-1}}\ge \floor*{\sqrt{m}}-1
\]
We then obtain the following:
\begin{align*}
    &\frac{\floor*{\sqrt{n-i}}+5k}{\floor*{\sqrt{n-i}}}\times \frac{\floor*{\sqrt{n-i-1}}+5k}{\floor*{\sqrt{n-i-1}}}\times\dots\times\frac{1+5k}{1}\\
    &\le \frac{\floor*{\sqrt{n-i}}+5k}{\floor*{\sqrt{n-i}}}\times  \frac{\floor*{\sqrt{n-i}}-1+5k}{\floor*{\sqrt{n-i}}-1}\times\dots\times\frac{1+5k}{1}\\
    &=\frac{(\floor*{\sqrt{n-i}}+5k)!}{(\floor*{\sqrt{n-i}})!\times (5k)!}={\floor*{\sqrt{{n-i}}}+5k\choose 5k}\\
    &\le \left(\frac{e(\floor*{\sqrt{{n-i}}}+5k)}{5k}\right)^{5k}\le \left(\frac{e(\floor*{\sqrt{{n}}}+5k)}{5k}\right)^{5k}\\
    &\le \left(\frac{e}{5k}\right)^{5k}\times (\sqrt{|D|}+5k)^{5k}
\end{align*}
(Observe that the maximal length $n$ of a sequence is $|D|-1$.)

We conclude that by defining:
\[\mathsf{pol}'''(|D|)=\left(\frac{e}{5k}\right)^{5k}\times (\sqrt{|D|}+5k)^{5k}\]
we obtain the required inequality. In this case, it holds that:
\[
\mathsf{pol}''(|D|)=c\times \mathsf{pol}'''(|D|)\le (2k+1)!\times\left(\frac{e}{5k}\right)^{5k}\times (\sqrt{|D|}+5k)^{5k}
\]
(Recall that $c=(2\ell+1)\times\dots\times 3$, where $\ell$ is the number of facts that conflict with $f$ and are not removed by the sequence $s$; hence, $\ell\le k$.)
Therefore, for every sequence $s$ that removes $f$, there is some sequence $s'$ that does not remove $f$ such that:
\[\pi(s)\le (2k+1)!\times\left(\frac{e}{5k}\right)^{5k}\times (\sqrt{|D|}+5k)^{5k}\times \pi(s')\]
From this point, we denote \[\mathsf{pol}''(|D|)=(2k+1)!\times\left(\frac{e}{5k}\right)^{5k}\times (\sqrt{|D|}+5k)^{5k}.\]

Next, we show that our mapping from sequences that remove $f$ to sequences that do not remove $f$ maps at most $2|D|-1$ sequences of the first type to the same sequence of the second type. Given a sequence $s'\in\abs{M_{\dep}^{\uo}(D)}$ that does not remove $f$, we can obtain this sequence either from a sequence $s\in\abs{M_{\dep}^{\uo}(D)}$ that has one additional operation that removes $f$ or from a sequence $s$ that removes $f$ jointly with some other fact $g$, while $s'$ removes the fact $g$ by itself. (Some of the operations at the end of $s'$ might not appear in $s$, as they remove facts that conflict only with $f$.) Since the length of the sequence $s'$ is at most $|D|-1$, there are at most $|D|$ possible ways to insert an additional operation that removes $f$, and $|D|-1$ ways to add $f$ to an existing operation. Hence, there are at most $|D|+|D|-1$ sequences that remove $f$ that are mapped to the sequence $s'$.

\begin{example}
We continue with Example~\ref{example:keys_proof_3}. Consider again the sequence $s_3'$. Recall that:
\[s_3'=-f_{3,1}, -f_{1,2}, -f_{2,1}, -f_{1,3}\]
This sequence can be obtained from any of the following sequences that have an additional operation that removes $f_{1,1}$:
\[-f_{1,1}, -f_{3,1}, -f_{1,2}\]
\[-f_{3,1}, -f_{1,1}, -f_{1,2}\]
\[-f_{3,1}, -f_{1,2}, -f_{1,1}\]
Note that the operations $-f_{2,1}, -f_{1,3}$ do not appear in these sequences, as after removing $f_{1,1}$ they are no longer involved in violations of the constraints.

The sequence $s_3'$ can also be obtained from the following sequences that replace an operation of $s_3'$ that removes a single fact with an operation that removes a pair of conflicting facts.
\[-\{f_{1,1},f_{3,1}\}, -f_{1,2}\]
\[-f_{3,1}, -\{f_{1,1},f_{1,2}\}\]
\qed\end{example}

We have now established the following result:
\begin{lemma}\label{lemma:relate-sequences}
    There exists a function $\mathsf{F} : S_{\neg f} \ra S_{f}$ such that:
    \begin{enumerate}
        \item There exists a polynomial $\mathsf{pol}''$ such that, for every $s \in S_{\neg f}$, $\pi(s) \leq \mathsf{pol}''(||D||) \cdot \pi(\mathsf{F}(s))$.
        \item For every $s' \in S_{f}$, $|\{s \in S_{\neg f} \mid \mathsf{F}(s)=s'\}| \leq 2 \cdot ||D|| - 1$.
    \end{enumerate}
\end{lemma}
Recall that $S_f$ is the set of sequences $s\in\abs{M_{\dep}^{\uo}(D)}$ such that $f\in s(D)$ and $S_{\neg f}$ is the set of sequences $s'\in\abs{M_{\dep}^{\uo}(D)}$ such that $f\not\in s'(D)$.

Finally, we have that:
\begin{align*}
   &\probhom{D,M_\dep^{\uo},Q}{h}=\frac{\sum_{s\in S_f}\pi(s)}{\sum_{s\in S_f}\pi(s)+\sum_{s\in S_{\neg f}}\pi(s)}\\
    &\ge\frac{\sum_{s\in S_f}\pi(s)}{\sum_{s\in S_f}\pi(s)+\sum_{\substack{s\in S_{\neg f}\\s\mbox{ is mapped to }s'}}\mathsf{pol}''(|D|)\times \pi(s')}\\
    &\ge\frac{\sum_{s\in S_f}\pi(s)}{\sum_{s\in S_f}\pi(s)+\mathsf{pol}''(|D|)\times(2|D|-1)\times \sum_{s\in S_f}\pi(s)}\\
    &\ge\frac{1}{1+\mathsf{pol}''(|D|)\times(2|D|-1)}
\end{align*}
Hence, we have proved that in the case where there is a homomorphism $h$ from $Q$ to $D$ with $h(Q)\models\dep$ and $h(\bar x)=\bar c$ such that $|h(Q)|=1$, then:
\[\probrep{M_{\dep}^{\uo},Q}{D,\bar c}\ge \frac{1}{1+\mathsf{pol}''(|D|)\times(2|D|-1)}\]
where:
\[\mathsf{pol}''(|D|)=(2k+1)!\times\left(\frac{e}{5k}\right)^{5k}\times (\sqrt{|D|}+5k)^{5k}\]
By defining $\mathsf{pol}(|D|)=1+\mathsf{pol}(|D|)\times(2|D|-1)$ we obtain that:
\[\probrep{M_{\dep}^{\uo},Q}{D,\bar c}\ge\probhom{D,M_\dep^{\uo},Q}{h}\ge \frac{1}{\mathsf{pol}(|D|)}\ge \frac{1}{\mathsf{pol}(||D||)}\]
}

\medskip
\noindent \underline{\textbf{The Case $|h(Q)| \geq 1$}}
\smallskip

\noindent We now generalize the proof given above for the case $|h(Q)| = 1$ to the case $|h(Q)|=m$ for some $1 \leq m \leq |Q|$. 
%
%
As in the case where $|h(Q)|=1$, we map sequences that remove at least one of the facts of $h(Q)$ to sequences that keep all these facts, by deleting or replacing every operation that removes a fact of $h(Q)$ and adding a constant number of operations at the end of the sequence that remove all the facts that conflict with some fact of $h(Q)$.

More formally, let $s\in\abs{M_{\dep}^{\uo}(D)}$ be a repairing sequence that removes $r$ of the facts of $h(Q)$ (for some $1\le r \le m$):
\[s= \op_1, \,\,\,\, \dots \,\,\,\, ,\op_{i_1}, \,\,\,\, \dots \,\,\,\, ,\op_{i_2}, \,\,\,\, \dots \,\,\,\, ,\op_{i_r}, \,\,\,\, \dots ,\op_n
\]
where the operations $\op_{i_1},\dots,\op_{i_r}$ remove these $r$ facts. Note that there are no conflicts among the facts of $h(Q)$; hence, it cannot be the case that a single operation removes two of these facts. We transform $s$ into a sequence $s'\in\abs{M_{\dep}^{\uo}(D)}$ where each operation $\op_{i_j}$ that removes a single fact is deleted, and every operation $\op_{i_j}$ that removes a pair $\{f,g\}$ of facts where $f\in h(Q)$ and $g\not\in h(Q)$, is replaced by the operation $o^\star_{i_j}$ that removes only the fact $g$. At the end of the sequence $s'$, we add operations $\op_1',\dots,\op_\ell'$ that remove the facts that are in conflict with one of the facts of $h(Q)$ that appears in $s(D)$. As we have explained before, for each such fact, the sequence $s$ keeps at most $k$ conflicting facts, where $k$ is the maximal number of keys in $\dep$ over the same relation $R$; hence, the total number of conflicting facts that $s$ does not remove is bounded by $m\times k$, and this is a bound on the number $\ell$ of additional operations (that remove these conflicting facts one by one in some arbitrary order). As in the case where $h(Q)=1$, the probability of applying the additional $\ell$ operations at the end of the sequence is some constant that we denote by $\frac{1}{c}$. We provide below more details about this constant.

The probability $\insP((\op_1,\dots,\op_{j-1}),(\op_1,\dots,\op_{j}))$, for $2\le j\le i_1-1$, is not affected by the decision to remove or keep a certain fact at the $i_1$th step. However, for $j\ge i_1$, the probability of applying the operation $\op_j$ might decrease in the sequence $s'$ compared to the sequence $s$, because the additional facts of $h(Q)$ (that are removed by $s$ but not by $s'$) might be involved in violations with the remaining facts of the database and introduce additional justified repairing operations at each step. As we have already shown, if the number of $(D^s_{j-1},\dep)$-justified operations before applying the operation $\op_j$ of $s$ is $N_j$, then the addition of a fact can increase this number by at most $5k\sqrt{N_j}$. Hence, the addition of at most $m$ facts (the facts of $h(Q)$) can increase this number by at most $5km\sqrt{N_j}$. We again denote by $r_j$ the factor by which the number of operations increases, and we have that $r_j\le 5km\sqrt{N_j}$.

Now, all the arguments for the case where $|h(Q)|=1$ apply also in this case, with the only difference being the value of $r_j$. Therefore, we conclude that
\[
\pi(s)\ \le\ \mathsf{pol}''(||D||)\times \pi(s')
\]
with
\[
\mathsf{pol}''(||D||)\ =\ c\times\left(\frac{e}{5km}\right)^{5km}\times (\sqrt{||D||}+5km)^{5km}.
\]
Recall that $\frac{1}{c}$ is the probability of applying the additional operations at the end of the sequence, and $r$ is the number of facts of $h(Q)$ that are removed by the sequence $s$. We would like to provide a lower bound on this probability (hence, an upper bound on $c$).
Clearly, the lowest probability is obtained when the number of additional operations is the highest (as for each additional operation we need to multiply the probability by a number lower than one) and when the probability of each individual operation is the lowest.
As mentioned above, for each one of the $r$ facts of $h(Q)$ that are removed by $s$, there are at most $k$ facts that conflict with it and are not removed by $s$. Hence, $r\times k$ is an upper bound on the number of additional operations. Moreover, the lowest probability of each operation is obtained when the number of justified operations at the point of applying it is the highest. When there are $\ell$ facts in a database $D'$ that are involved in violations of the constraints, an upper bound on the number of $(D',\dep)$-justified operations (that is obtained when every fact is in conflict with every other fact) is
\[
\ell+\frac{\ell(\ell-1)}{2}=\frac{\ell^2+\ell}{2}=\frac{\ell(\ell+1)}{2}\le \frac{(\ell+1)^2}{2}\le (\ell+1)^2.
\]

Therefore, we have that
\begin{align*}
\frac{1}{c}&\ge\frac{1}{(rk+r+1)^2}\times \frac{1}{(rk+r)^2}\times \frac{1}{(rk+r-1)^2}\times\dots\times \frac{1}{3}\\
&\ge \frac{1}{((rk+r+1)^2)!}\ge \frac{1}{((mk+m+1)^2)!}
\end{align*}
and
\[
c\ \le\ ((mk+m+1)^2)!
\]
(Observe that $rk+r$ is the number of facts involved in violations if each of the $r$ facts of $h(Q)$ that $s$ removes conflicts with $k$ facts of $s(D)$.)
Now, it holds that
\begin{align*}
&((mk+m+1)^2)!\times\left(\frac{e}{5km}\right)^{5km}\times (\sqrt{|D|}+5km)^{5km}\le\\
&((|Q||\dep|+|Q|+1)^2)!\times e^{5|Q||\dep|}\times (\sqrt{|D|}+5|Q||\dep|)^{5|Q||\dep|}
\end{align*}
Hence, with
\[
\mathsf{pol}''(||D||)= ((|Q||\dep|+|Q|+1)^2)!\times e^{5|Q||\dep|}\times (\sqrt{||D||}+5|Q||\dep|)^{5|Q||\dep|}
\]
we have thet
\[
\pi(s)\ \le\ \mathsf{pol}''(||D||)\times \pi(s'),
\]
as needed.

\medskip
Finally, we show that our mapping from sequences that remove at least one of the facts of $h(Q)$ to sequences that do not remove any of these facts maps at most polynomially many sequences of the first type to the same sequence of the second type. Given a sequence $s'$ that does not remove any of the facts of $h(Q)$, we can obtain this sequence from any sequence $s$ that has additional operations that remove some of the facts of $h(Q)$ individually or operations that remove these facts jointly with another fact (while $s'$ removes only one of these facts). The sequence $s$ can remove any number $1\le r\le m$ of facts of $h(Q)$. And, in the case where it removes $r$ of the facts of $h(Q)$, for every $\ell\le r$ there are  ${r\choose{\ell}}$ possible ways to choose a subset of size $\ell$ of $h(Q)$ of facts that will be removed by themselves
(while the remaining $r-\ell$ facts will be removed jointly with another fact). Since the length of the sequence $s$ is at most $|D|-1$, there are at most ${{|D|+\ell-1}\choose\ell}$ possible choices for the positions of the additional singleton deletions, and ${{|D|-1}\choose {r-\ell}}$ possible choices for the individual fact removals that will become pair removals. Hence, the number of sequences that remove a fact of $h(Q)$ that are mapped to the sequence $s'$ is at most
\begin{eqnarray*}
&& \sum_{r=1}^m \sum_{\ell=0}^r{r\choose{\ell}}\times {{|D|+\ell-1}\choose\ell}\times {{|D|-1}\choose{r-\ell}}\\
&\leq& \sum_{r=1}^m\sum_{\ell=0}^r{\left(\frac{er}{\ell}\right)^{\ell}}\times \left(\frac{e(|D|+\ell-1)}{\ell}\right)^{\ell}\times \left(\frac{e(|D|-1)}{r-\ell}\right)^{r-\ell}\\
&\leq&\sum_{r=1}^{|Q|}\sum_{\ell=0}^{|Q|}{\left(e|Q|\right)^{\ell}}\times \left(e(|D|+\ell-1)\right)^{\ell}\times \left(e(|D|-1)\right)^{|Q|-\ell}\\
&\leq& |Q|\times(|Q|+1)\times\left(e|Q|\right)^{|Q|}\times \left(e(|D|+|Q|-1)\right)^{|Q|}\times\\ && \left(e(|D|-1)\right)^{|Q|}\\
&\leq&(e|Q|)^2\times\left(e|Q|\right)^{|Q|}\times \left(e(|D|+|Q|-1)\right)^{|Q|}\times \left(e(|D|-1)\right)^{|Q|}\\
&=&\left(e|Q|\right)^{|Q|+2}\times \left(e(|D|+|Q|-1)\right)^{|Q|}\times \left(e(|D|-1)\right)^{|Q|}.
\end{eqnarray*}
This number is clearly polynomial in $||D||$. We denote this number by $\mathsf{pol}'(||D||)$.
Finally, similarly to the case where $|h(Q)|=1$,
\[
\probhom{D,M_\dep^{\uo},Q}{h}\ \ge\ \frac{1}{1+\mathsf{pol}''(||D||)\times \mathsf{pol}'(||D||)}.
\]
With $\mathsf{pol}(||D|| )= 1 + \mathsf{pol}''(||D||)\times \mathsf{pol}'(||D||)$, we obtain that
\[
\probrep{M_{\dep}^{\uo},Q}{D,\bar c}\ge\probhom{D,M_\dep^{\uo},Q}{h}\ \ge\ \frac{1}{\mathsf{pol}(||D||)},
\]
which concludes our proof.

\OMIT{
Next, we generalize this result to the case where $|h(Q)|=m$ for any $m\ge 1$ (with $m\le |Q|$). Here, we denote by $k$ be the maximal number of keys in $\dep$ over the same relation $R$.  In this case, the final polynomial will also depend on $|Q|$. We again provide a lower bound on the probability $\probhom{D,M_\dep^{\uo},Q}{h}$.

As in the case where $|h(Q)|=1$, we map sequences that remove at least one of the facts of $h(Q)$ to sequences that keep all these facts, by deleting or replacing every operation that removes a fact of $h(Q)$ and adding a constant number of operations at the end of the sequence, that remove all the facts that conflict with some fact of $h(Q)$.

More formally, let $s\in\abs{M_{\dep}^{\uo}(D)}$ be a repairing sequence that removes $r$ of the facts of $h(Q)$ (for some $1\le r \le m$):
\[s= \op_1, \,\,\,\, \dots \,\,\,\, ,\op_{i_1}, \,\,\,\, \dots \,\,\,\, ,\op_{i_2}, \,\,\,\, \dots \,\,\,\, ,\op_{i_r}, \,\,\,\, \dots ,\op_n
\]
where the operations $\op_{i_1},\dots,\op_{i_r}$ remove these $r$ facts. Note that there are no conflicts among the facts of $h(Q)$; hence, it cannot be the case that a single operation removes two of these facts. We transform $s$ into a sequence $s'\in\abs{M_{\dep}^{\uo}(D)}$ where each operation $\op_{i_j}$ that removes a single fact is deleted, and every operation $\op_{i_j}$ that removes a pair $\{f,g\}$ of facts where $f\in h(Q)$ and $g\not\in h(Q)$, is replaced by the operation $o^\star_{i_j}$ that removes only the fact $g$. At the end of the sequence $s'$, we add operations $\op_1',\dots,\op_\ell'$ that remove the facts that are in conflict with one of the facts of $h(Q)$ that appears in $s(D)$. As we have explained before, for each such fact, the sequence $s$ keeps at most $k$ conflicting facts; hence, the total number of conflicting facts that $s$ does not remove is bounded by $m\times k$, and this is a bound on the number $\ell$ of additional operations (that remove these conflicting facts one by one in some arbitrary order). As in the case where $h(Q)=1$, the probability
of applying the additional $\ell$ operations at the end of the sequence is some constant that we denote by $\frac{1}{c}$. We will provide more details about this constant later on.

It is again the case that the probability $\insP((\op_1,\dots,\op_{j-1}),(\op_1,\dots,\op_{j}))$, for $2\le j\le i_1-1$, is not affected by the decision to remove or keep a certain fact at the $i_1$th step. However, for $j\ge i_1$, the probability of applying the operation $\op_j$ might decrease in the sequence $s'$ compared to the sequence $s$, because the additional facts of $h(Q)$ (that are removed by $s$ but not by $s'$) might be involved in violations with the remaining facts of the database and introduce additional justified repairing operations at each step. As we have already shown, if the number of $(D^s_{j-1},\dep)$-justified operations before applying the operation $\op_j$ of $s$ is $N_j$, the addition of a fact can increase this number by at most $5k\sqrt{N_j}$. Hence, the addition of at most $m$ facts (the facts of $h(Q)$) can increase this number by at most $5km\sqrt{N_j}$. We again denote by $r_j$ the factor by which the number of operations increases, and so we have that:
$r_j\le 5km\sqrt{N_j}$.

Now, all the arguments for the case where $|h(Q)|=1$ apply also in this case, with the only difference being the value of $r_j$. Therefore, we conclude that:
\[\pi(s)\le \mathsf{pol}''(|D|)\times \pi(s')\]
for:
\[
\mathsf{pol}''(|D|)=c\times\left(\frac{e}{5km}\right)^{5km}\times (\sqrt{|D|}+5km)^{5km}
\]
Recall that $\frac{1}{c}$ is the probability of applying the additional operations at the end of the sequence, and $r$ is the number of facts of $h(Q)$ that are removed by the sequence $s$. We would like to provide a lower bound on this probability (hence, an upper bound on $c$).
Clearly, the lowest probability is obtained when the number of additional operations is the highest (as for each additional operation we need to multiply the probability by a number lower than one) and when the probability of each individual operation is the lowest.
As aforementioned, for each one of the $r$ facts of $h(Q)$ that are removed by $s$, there are at most $k$ facts that conflict with it and are not removed by $s$. Hence, $r\times k$ is an upper bound on the number of additional operations. Moreover, the lowest probability of each operation is obtained when the number of justified operations at the point of applying it is the highest. When there are $\ell$ facts in a database $D'$ that are involved in violations of the constraints, an upper bound on the number of $(D',\dep)$-justified operations (that is obtained when every fact is in conflict with every other fact) is:
\[\ell+\frac{\ell(\ell-1)}{2}=\frac{\ell^2+\ell}{2}=\frac{\ell(\ell+1)}{2}\le \frac{(\ell+1)^2}{2}\le (\ell+1)^2\]

Therefore, we have that:
\begin{align*}
\frac{1}{c}&\ge\frac{1}{(rk+r+1)^2}\times \frac{1}{(rk+r)^2}\times \frac{1}{(rk+r-1)^2}\times\dots\times \frac{1}{3}\\
&\ge \frac{1}{((rk+r+1)^2)!}\ge \frac{1}{((mk+m+1)^2)!}
\end{align*}
and:
\[c\le ((mk+m+1)^2)!\]
(Observe that $rk+r$ is the number of facts involved in violations if each of the $r$ facts of $h(Q)$ that $s$ removes conflicts with $k$ facts of $s(D)$.)
Now, it holds that:
\begin{align*}
    &((mk+m+1)^2)!\times\left(\frac{e}{5km}\right)^{5km}\times (\sqrt{|D|}+5km)^{5km}\le\\
    &((|Q||\dep|+|Q|+1)^2)!\times e^{5|Q||\dep|}\times (\sqrt{|D|}+5|Q||\dep|)^{5|Q||\dep|}
\end{align*}
Hence, from now on we denote:
\[
\mathsf{pol}''(|D|)= ((|Q||\dep|+|Q|+1)^2)!\times e^{5|Q||\dep|}\times (\sqrt{|D|}+5|Q||\dep|)^{5|Q||\dep|}
\]
and we have that:
\[\pi(s)\le \mathsf{pol}''(|D|)\times \pi(s')\]

\medskip
\noindent \paragraph{Item (2).}
Finally, we show that our mapping from sequences that remove at least one of the facts of $h(Q)$ to sequences that do not remove any of these facts maps at most $2|D|-1$ sequences of the first type to the same sequence of the second type. Given a sequence $s'$ that does not remove any of the facts of $h(Q)$, we can obtain this sequence from any sequence $s$ that has additional operations that remove some of the facts of $h(Q)$ individually or operations that remove these facts jointly with another fact (while $s'$ removes only one of these facts). The sequence $s$ can remove any number $1\le r\le m$ of facts of $h(Q)$. And, in the case where it removes $r$ of the facts of $h(Q)$, for every $\ell\le r$ there are  ${r\choose{\ell}}$ possible ways to choose a subset of size $\ell$ of $h(Q)$ of facts that will be removed by themselves
(while the remaining $r-\ell$ facts will be removed jointly with another fact). Since the length of the sequence $s$ is at most $|D|-1$, there are at most ${{|D|+\ell-1}\choose\ell}$ possible choices for the positions of the additional singleton deletions, and ${{|D|-1}\choose {r-\ell}}$ possible choices for the individual fact removals that will become pair removals. Hence, there are at most:
\begin{align*}
    &\sum_{r=1}^m \sum_{\ell=0}^r{r\choose{\ell}}\times {{|D|+\ell-1}\choose\ell}\times {{|D|-1}\choose{r-\ell}}\le\\
    &\sum_{r=1}^m\sum_{\ell=0}^r{\left(\frac{er}{\ell}\right)^{\ell}}\times \left(\frac{e(|D|+\ell-1)}{\ell}\right)^{\ell}\times \left(\frac{e(|D|-1)}{r-\ell}\right)^{r-\ell}\le\\
    &\sum_{r=1}^{|Q|}\sum_{\ell=0}^{|Q|}{\left(e|Q|\right)^{\ell}}\times \left(e(|D|+\ell-1)\right)^{\ell}\times \left(e(|D|-1)\right)^{|Q|-\ell}\le\\
    &|Q|\times(|Q|+1)\times\left(e|Q|\right)^{|Q|}\times \left(e(|D|+|Q|-1)\right)^{|Q|}\times \left(e(|D|-1)\right)^{|Q|}\le\\
    &(e|Q|)^2\times\left(e|Q|\right)^{|Q|}\times \left(e(|D|+|Q|-1)\right)^{|Q|}\times \left(e(|D|-1)\right)^{|Q|}=\\
    &\left(e|Q|\right)^{|Q|+2}\times \left(e(|D|+|Q|-1)\right)^{|Q|}\times \left(e(|D|-1)\right)^{|Q|}
\end{align*}
sequences that remove a fact of $h(Q)$ that are mapped to the sequence $s'$. This number is clearly polynomial in $|D|$. We denote this number by $\mathsf{pol}'(|D|)$.

Finally, similarly to the case where $|h(Q)|=1$, we have that:
\begin{align*}
 &\probhom{D,M_\dep^{\uo},Q}{h}\ge\frac{1}{1+\mathsf{pol}''(|D|)\times \mathsf{pol}'(|D|)}
\end{align*}
and by defining $\mathsf{pol}(|D|)=1+\mathsf{pol}''(|D|)\times \mathsf{pol}'(|D|)$, we obtain the following:
\[\probrep{M_{\dep}^{\uo},Q}{D,\bar c}\ge\probhom{D,M_\dep^{\uo},Q}{h}\ge \frac{1}{\mathsf{pol}(|D|)} \ge \frac{1}{\mathsf{pol}(||D||)}\]
and that concludes our proof.
}

\subsection{The case of Functional Dependencies}

Unlike the case of keys, in the case of FDs, there is no polynomial lower bound on the target probability, as we show next. This means that we cannot rely on Monte Carlo Sampling for devising an FPRAS. On the other hand, this does not preclude the existence of an FPRAS in the case of FDs, which remains an open problem.

\begin{proposition}\label{prop:no_bound_FDs}
Consider the FD set $\{R:A_1\rightarrow A_2\}$ over the schema $\{R/3\}$, and the Boolean CQ $\textrm{Ans}()\ \text{:-}\ R(0,0,0)$. There exists a family $\{D_n\}_{n\ge 1}$ of databases such that
\[
0<\probrep{M_{\dep}^{\uo},Q}{D_n,()}\le \frac{1}{2^{|D_n|-1}}.
\]
\end{proposition}

\begin{proof}
Let $D_n$ be the database that contains the fact $R(0,0,0)$ and $n-1$ additional facts $R(0,1,i)$ for $i\in\{1,\dots,n-1\}$. Observe that each fact $R(0,1,i)$ is in conflict with $R(0,0,0)$, but there are no conflicts among two facts $R(0,1,i)$ and $R(0,1,j)$ for $i\neq j$. Clearly, it holds that $0<\probrep{M_{\dep}^{\uo},Q}{D,()}$ as the operational repair that keeps the fact $R(0,0,0)$ entails $Q$.
We prove by induction on $n$, the number of facts in the database, that for a database $D$ that contains the fact $R(0,0,0)$ and $n-1$ facts of the form $R(0,1,i)$, it holds that:
\[\probrep{M_{\dep}^{\uo},Q}{D,()}\le \frac{1}{2^{n-1}}\]

\medskip
\noindent \paragraph{Base Case.}
For $n=1$, $D = \{R(0,0,0)\}$ and there are no violations of the FD. In this case, it is rather straightforward to see that
\[
\probrep{M_{\dep}^{\uo},Q}{D,()}\ =\ \frac{1}{2^{1-1}}\ =\ 1.
\]

\noindent
\paragraph{Inductive Step.}
We assume that the claim holds for $n=1,\dots,p$ and prove that it holds for $n=p+1$. Let $D$ be such a database with $p+1$ facts; that is, $D$ contains the fact $R(0,0,0)$ and $p$ facts of the form $R(0,1,i)$. Whenever we have $p$ facts of the form $R(0,1,i)$ in the database, there are $1+2p$ justified operations: \textit{(1)} the removal of $R(0,0,0)$, \textit{(2)} the removal of a fact of the form $R(0,1,i)$, or \textit{(3)} the removal of a pair $\{R(0,0,0),R(0,1,i)\}$. Only the $p$ operations of type \textit{(2)} keep the fact $R(0,0,0)$ in the database. We denote these operations by $\op_1,\dots,\op_p$. For every $i\in\{1,\dots,p\}$, we have that
\[
\insP(\varepsilon,(op_i))=\frac{1}{1+2p}.
\]

After removing a fact of the form $R(0,1,i)$ from the database, we have $p-1$ such facts left, regardless of which specific fact we remove. For every $i\in\{1,\dots,p\}$, we denote by $D_i$ the database $\op_i(D)$. By the inductive hypothesis, we have that
\[
\probrep{M_{\dep}^{\uo},Q}{D_i,()}\ \le\ \frac{1}{2^{p-1}}.
\]
Every sequence $s\in\abs{M_{\dep}^{\uo}(D)}$ with $R(0,0,0)\in s(D)$ is of the form $\op_i\cdot s_i$ for some $i\in[p]$ and $s_i\in\abs{M_{\dep}^{\uo}(D_i)}$ with $R(0,0,0)\in s_i(D)$. The probability 
$\probrep{M_{\dep}^{\uo},Q}{D,()}$ can then be written as
\begin{align*}
\probrep{M_{\dep}^{\uo},Q}{D,()}&=\sum_{\substack{s \in \abs{M_{\dep}^{\uo}(D)} \\ ()\in Q(s(D))}} \pi(s)\\
&=\sum_{i=1}^p \left(\insP(\varepsilon,(\op_i))\times\sum_{\substack{s_i \in \abs{M_{\dep}^{\uo}(D_i)} \\ ()\in Q(s_i(D_i))}} \pi(s_i)\right).
\end{align*}
As said above, for every $i\in\{1,\dots,p\}$,
\[
\probrep{M_{\dep}^{\uo},Q}{D_i,()}\ =\ \sum_{\substack{s_i \in \abs{M_{\dep}^{\uo}(D_i)} \\ ()\in Q(s_i(D_i))}} \pi(s_i)\ \le\ \frac{1}{2^{p-1}}.
\]
Therefore, we conclude that
\begin{align*}
\probrep{M_{\dep}^{\uo},Q}{D,()}&\le \sum_{i=1}^p \left(\frac{1}{1+2p}\times \frac{1}{2^{p-1}} \right)=\frac{p}{(1+2p)\times 2^{p-1}}\\
&=\frac{p}{2^{p-1}+p\times 2^p}\le \frac{p}{p\times 2^p}=\frac{1}{2^p}
\end{align*}
and the claim follows.
\end{proof}

\subsection{Proof of Theorem~\ref{the:uniform-operations-singleton}}
We now show that if only singleton removals are allowed, then we can devise an FPRAS even for arbitrary FDs.
For a database $D$ and a set $\dep$ of FDs, we denote by $\rsone{D}{\dep}$ the set of sequences in $\rs{D}{\dep}$ mentioning only operations of the form $-f$, i.e., removing a single fact.
Similarly, we denote $\opsone{s}{D}{\dep} = \{s' \in \rsone{D}{\dep} \mid s' = s \cdot \op \text{ for some } D\text{-operation } \op\}$.
Then, we define the Markov chain generator $M_\dep^{\uo,1}$ such that for every $s,s' \in \rsone{D}{\dep}$, assuming that $M_{\dep}^{\uo,1}(D) = (V,E,\insP)$, if $s' \in \opsone{s}{D}{\dep}$ then
\[
\insP(s,s')\ =\ \frac{1}{|\opsone{s}{D}{\dep}|}.
\]
Observe, however, that the Markov chain generator $M_\dep^{\uo,1}$ is defined over all the sequences of $\rs{D}{\dep}$. If $s\in \rsone{D}{\dep}$ but $s' \in \rs{D}{\dep}\setminus \rsone{D}{\dep}$ (and $s'\in \ops{s}{D}{\dep}$), then we define $\insP(s,s')=0$. If $s\in \rs{D}{\dep}\setminus \rsone{D}{\dep}$, none of the leaves of the subtree $T_s$ is reachable with non-zero probability, and thus, $\ins{P}(s,s')$, for any $s'\in\ops{s}{D}{\dep}$, can get an arbitrary probability (as long as the sum of probabilities equals one), e.g., $\frac{1}{|\ops{s}{D}{\dep}|}$.


We can now show that, assuming singleton removals, for FDs the problem of interest admits an FPRAS. The formal statement, already given in the main body of the paper, and its proof follow:

\begin{manualtheorem}{\ref{the:uniform-operations-singleton}}
	For a set $\dep$ of FDs, and a CQ $Q$, $\ocqa{\dep,M_{\dep}^{\uo,1},Q}$ admits an FPRAS.
\end{manualtheorem}

The proof consists of the usual two steps: (1) existence of an efficient sampler, and (2) provide a polynomial lower bound on he target probability.

\subsubsection*{Step 1: Efficient Sampler}
We can sample elements of $\abs{M_{\dep}^{\uo,1}(D)}$ according to the leaf distribution of $M_{\dep}^{\uo,1}(D)$ in polynomial time in $||D||$. This is done by employing the same iterative algorithm as the one used to sample elements of $\abs{M_{\dep}^{\uo}(D)}$, but with the difference that only justified operations that consist of singleton removals are considered. In particular, at each step, the algorithm extends the current sequence $s$ by selecting one of the $(s(D),\dep)$-justified operations of the form $-f$ with probability
\[
\frac{1}{|\opsone{s}{D}{\dep}|}.
\] 
%
Hence, we immediately obtain the following result:

\begin{lemma}
	Given a database $D$, and a set $\dep$ of keys, we can sample elements of $\abs{M_{\dep}^{\uo,1}(D)}$ according to the leaf distribution of $M_{\dep}^{\uo,1}(D)$ in polynomial time in $||D||$.
\end{lemma}

\subsubsection*{Step 2: Polynomial Lower Bound}
It remains to show that there exists a polynomial lower bound on the target probability.

\begin{lemma}
	Consider a set $\dep$ of keys, and a CQ $Q(\bar x)$. For every database $D$, and $\bar c \in \adom{D}^{|\bar x|}$,
	\[
	\probrep{M_{\dep}^{\uo,1},Q}{D,\bar c}\ \geq\ \frac{1}{\left(e \cdot ||D||\right)^{||Q||}}
	\] 
	whenever $\probrep{M_{\dep}^{\uo,1},Q}{D,\bar c} > 0$.
\end{lemma}
\begin{proof}
	Consider a database $D$. If there is no homomorphism $h$ from $Q$ to $D$ such that $h(Q)\models\dep$ and $h(\bar x)=\bar c$, then clearly $\probrep{M_{\dep}^{\uo,1},Q}{D,\bar c}=0$.
	We now focus on the case where such a homomorphism $h$ exists. Assume that $|h(Q)|=m$ for some $m\le |Q|$.
	We prove by induction on $n$, that is, the number of facts in $D\setminus h(Q)$ that are involved in violations of the FDs (i.e., the facts $f\in (D\setminus h(Q))$ such that $\{f,g\}\not\models\dep$ for some $g\in D$), the following:
	\[
		\probhom{D,M_\dep^{\uo,1},Q}{h}\ \ge\ \frac{1}{{{n+m}\choose m}}.
	\]
	
	\noindent
	\paragraph{Base Case.}
	For $n=0$, since $h(Q)\models\dep$, there are no violations of the FDs in $D$, and $D$ has a single operational repair, which is $D$ itself. In this case, the probability of obtaining an operational repair that contains all the facts of $h(Q)$ is $1=\frac{1}{{{0+m}\choose m}}$, as needed.
	
	\medskip
	
	\noindent \paragraph{Inductive Step.} We now assume that the claim holds for databases where $n= 0,\dots,k-1$, and we prove that it holds for databases $D$ where $n=k$. Every repairing sequence $s\in\abs{M_{\dep}^{\uo,1}(D)}$ for which $h(Q)\subseteq s(D)$ is such that the first operation of $s$ removes a fact of $D\setminus h(Q)$ that is involved in violations of the FDs. Let $f_1,\dots,f_k$ 
	be these facts of $D\setminus h(Q)$, and for each $i\in\{1,\dots,k\}$, let $\op_i$ be the operation that removes the fact $f_i$. We then have that
	\[
	\insP(\varepsilon,(op_i))\ \ge\ \frac{1}{k+m}.
	\]
	This is because the probability of removing a certain fact is $\frac{1}{k+p}$, where $p$ is the number of facts involved in violations among the facts of $h(Q)$. Since $p\le m$, we get that $\frac{1}{k+m} \leq \frac{1}{k+p}$.
	
	After removing a conflicting fact of $D\setminus h(Q)$ from the database, we have at most $k-1$ such facts left, regardless of which specific fact we remove. For every $i\in\{1,\dots,p\}$, we denote by $D_i$ the database $\op_i(D)$ and by $n_i$ the number of facts of $D_i\setminus h(Q)$ that are involved in violations of the FDs; hence, we have that $n_i\le k-1$. By the inductive hypothesis, we have that
	\[
	\probhom{D_i,M_\dep^{\uo,1},Q}{h}\ \ge\ \frac{1}{{{n_i+m}\choose m}}\ \ge\ \frac{1}{{{k-1+m}\choose m}}.
	\]
	Clearly, every sequence $s\in\abs{M_{\dep}^{\uo}(D)}$ with $h(Q)\subseteq s(D)$ is of the form $\op_i\cdot s_i$ for some $i\in\{1,\dots,k\}$ and $s_i\in\abs{M_{\dep}^{\uo}(D_i)}$ with $h(Q)\subseteq s_i(D)$. Now, the following holds
	\[
	\probhom{D,M_\dep^{\uo,1},Q}{h}\ =\ \sum_{i=1}^k \left(\insP(\varepsilon,(\op_i))\times\probhom{D_i,M_\dep^{\uo,1},Q}{h}\right).
	\]
	Therefore, we can conclude that
	\begin{eqnarray*}
	\probhom{D,M_\dep^{\uo,1},Q}{h} &\ge& \sum_{i=1}^k \left(\frac{1}{k+m}\times \frac{1}{{{k-1+m}\choose m}} \right)\\
	&=& \frac{k}{k+m}\times \frac{1}{{{k-1+m}\choose m}}\\
	&=& \frac{k}{k+m}\times\frac{1}{\frac{(k-1+m)!}{m!\times (k-1)!}}\\
	&=& \frac{m!\times k!}{(k+m)!}\\
	&=& \frac{1}{{{k+m}\choose m}}.
	\end{eqnarray*}
	
	Finally, it is well known that ${n\choose k}\le \left(\frac{en}{k}\right)^k$. We conclude that, for a database $D$ such that $D\setminus h(Q)$ contains $n$ facts that are involved in violations of the FDs, we have that
	\begin{eqnarray*}
	\probhom{D,M_\dep^{\uo,1},Q}{h}&\ge& \frac{1}{{{n+m}\choose n}}\\
	&\ge& \frac{1}{\left(\frac{e(n+m)}{m}\right)^m}\\
	&=&\frac{m^m}{e^m}\times \frac{1}{(n+m)^m}\\
	&\ge& \left(\frac{m}{e}\right)^m \times \frac{1}{|D|^m}\\
	&\ge& \left(\frac{1}{e}\right)^{|Q|} \times \frac{1}{|D|^{|Q|}}.
	\end{eqnarray*}
	Since $h(Q)\subseteq D'$ implies $\bar c\in Q(D')$, it holds that
	\[
	\probrep{M_{\dep}^{\uo,1},Q}{D,\bar c}\ \ge\ \probhom{D,M_\dep^{\uo,1},Q}{h}\ \ge\   \frac{1}{\left(e|D|\right)^{|Q|}}\ \ge\ \frac{1}{\left(e||D||\right)^{||Q||}},
	\]
	which concludes our proof.
\end{proof}


\OMIT{
there is a polynomial lower bound on the probability of a query answer not only in the case of keys, but for any set of FDs. More formally, for a database $D$ and a set $\dep$ of FDs, we denote by $\rsone{D}{\dep}$ the set of sequences in $\rs{D}{\dep}$ mentioning only operations of the form $-f$, i.e., removing a single fact.
Similarly, we denote $\opsone{s}{D}{\dep} = \{s' \in \rsone{D}{\dep} \mid s' = s \cdot \op \text{ for some } D\text{-operation } \op\}$.

Then, we define the Markov chain generator $M_\dep^{\uo,1}$ such that for every $s,s' \in \rsone{D}{\dep}$, assuming that $M_{\dep}^{\uo,1}(D) = (V,E,\insP)$, if $s' \in \opsone{s}{D}{\dep}$ then $\insP(s,s') = \frac{1}{|\opsone{s}{D}{\dep}|}$.
Observe, however, that the Markov chain generator $M_\dep^{\uo,1}$ is defined over all the sequences of $\rs{D}{\dep}$. If $s\in \rsone{D}{\dep}$ but $s' \in \rs{D}{\dep}\setminus \rsone{D}{\dep}$ (and $s'\in \ops{s}{D}{\dep}$), then we define $\insP(s,s')=0$. If $s\in \rs{D}{\dep}\setminus \rsone{D}{\dep}$, none of the leaves of the subtree $T_s$ is reachable with non-zero probability, and thus, $\ins{P}(s,s')$, for any $s'\in\ops{s}{D}{\dep}$, can get an arbitrary probability (as long as the sum of probabilities equals one), e.g., $\frac{1}{|\ops{s}{D}{\dep}|}$.

It is rather straightforward to adapt Definition~\ref{def:uniform-ops} to obtain the Markov chain generator $M_\dep^{\uo,1}$.
We now show that the problem $\ocqa{\dep,M_{\dep}^{\uo,1},Q}$ admits an FPRAS for any set of FDs.

\begin{manualtheorem}{\ref{the:uniform-operations-singleton}}
    For a set $\dep$ of FDs, and a CQ $Q$, $\ocqa{\dep,M_{\dep}^{\uo,1},Q}$ admits an FPRAS.
\end{manualtheorem}

We can sample elements of $\abs{M_{\dep}^{\uo,1}(D)}$ according to the leaf distribution of $M_{\dep}^{\uo,1}(D)$ in polynomial time in $||D||$ also in this case. This is done by employing the same iterative algorithm as the one used to sample elements of $\abs{M_{\dep}^{\uo}(D)}$, that at each step, extends the current sequence $s$ by selecting one of the $(s(D),\dep)$-justified operations with probability:
\[\frac{1}{\mbox{number of } (s(D),\dep)-\mbox{justified operations}}.\] 
The only difference is the set of justified operations that consists, in this case, only of single-fact deletions.
Hence, we only need to show that there is a polynomial lower bound on $\probrep{M_{\dep}^{\uo,1},Q}{D,\bar c}$.

\begin{lemma}
Consider a set $\dep$ of keys, and a CQ $Q(\bar x)$. For every database $D$, and $\bar c \in \adom{D}^{|\bar x|}$,
    \[
    \probrep{M_{\dep}^{\uo,1},Q}{D,\bar c}\ \geq\ \frac{1}{\left(e||D||\right)^{||Q||}}
    \] 
    whenever $\probrep{M_{\dep}^{\uo,1},Q}{D,\bar c} > 0$.
\end{lemma}
\begin{proof}
Let $D$ be a database. If there is no homomorphism $h$ from $Q$ to $D$ such that $h(Q)\models\dep$ and $h(\bar x)=\bar c$, then clearly $\probrep{M_{\dep}^{\uo,1},Q}{D,\bar c}=0$. Next, consider the case where such a homomorphism $h$ exists. Assume that $|h(Q)|=m$ foe some $m\le |Q|$.
We prove by induction on $n$, the number of facts in $D\setminus h(Q)$ that are involved in violations of the FDs (i.e., the facts $f\in (D\setminus h(Q))$ such that $\{f,g\}\not\models\dep$ for some $g\in D$), that:
\begin{eqnarray*}
 \probhom{D,M_\dep^{\uo,1},Q}{h} \ge \frac{1}{{{n+m}\choose m}}
\end{eqnarray*}

For the base case, $n=0$, since $h(Q)\models\dep$, there are no violations of the FDs in $D$, and $D$ has a single operational repair---$D$ itself. In this case, the probability of obtaining an operational repair that contains all the facts of $h(Q)$ is $1=\frac{1}{{{0+m}\choose m}}$. Next, we assume that the claim holds for databases where $n=\{0,\dots,k-1\}$ and we prove that it holds for databases $D$ where $n=k$. Every repairing sequence $s\in\abs{M_{\dep}^{\uo,1}(D)}$ for which $h(Q)\subseteq s(D)$ is such that the first operation of $s$ removes a fact of $D\setminus h(Q)$ that is involved in violations of the FDs. Let $f_1,\dots,f_k$ 
be these facts of $D\setminus h(Q)$, and for each $i\in\{1,\dots,k\}$, let $\op_i$ be the operation that removes the fact $f_i$. We then have that:
\[\insP(\varepsilon,(op_i))\ge\frac{1}{k+m}\]
This is because the probability of removing a certain fact is $\frac{1}{k+p}$, where $p$ is the number of facts involved in violations among the facts of $h(Q)$. And since $p\le m$, a lower bound on this probability is $\frac{1}{k+m}$.

After removing a conflicting fact of $D\setminus h(Q)$ from the database, we have at most $k-1$ such facts left, regardless of which specific fact we remove. For every $i\in\{1,\dots,p\}$, we denote by $D_i$ the database $\op_i(D)$ and by $n_i$ the number of facts of $D_i\setminus h(Q)$ that are involved in violations of the FDs (hence, we have that $n_i\le k-1$). According to the inductive assumption, we have that:
\[\probhom{D_i,M_\dep^{\uo,1},Q}{h}\ge \frac{1}{{{n_i+m}\choose m}}\ge \frac{1}{{{k-1+m}\choose m}}\]
Clearly, every sequence $s\in\abs{M_{\dep}^{\uo}(D)}$ with $h(Q)\subseteq s(D)$ is of the form $\op_i\cdot s_i$ for some $i\in\{1,\dots,k\}$ and $s_i\in\abs{M_{\dep}^{\uo}(D_i)}$ with $h(Q)\subseteq s_i(D)$. Now, the following holds:
\begin{align*}
&\probhom{D,M_\dep^{\uo,1},Q}{h}=\sum_{i=1}^k \left(\insP(\varepsilon,(\op_i))\times\probhom{D_i,M_\dep^{\uo,1},Q}{h}\right)
\end{align*}
Therefore, we conclude that:
\begin{align*}
&\probhom{D,M_\dep^{\uo,1},Q}{h}\ge \sum_{i=1}^k \left(\frac{1}{k+m}\times \frac{1}{{{k-1+m}\choose m}} \right)=\frac{k}{k+m}\times \frac{1}{{{k-1+m}\choose m}}\\
&=
\frac{k}{k+m}\times\frac{1}{\frac{(k-1+m)!}{m!\times (k-1)!}}=\frac{m!\times k!}{(k+m)!}=\frac{1}{{{k+m}\choose m}}
\end{align*}

Finally, it is well known that ${n\choose k}\le \left(\frac{en}{k}\right)^k$. We conclude that for a database $D$ such that $D\setminus h(Q)$ contains $n$ facts that are involved in violations of the FDs we have that:
\begin{align*}
    \probhom{D,M_\dep^{\uo,1},Q}{h}&\ge \frac{1}{{{n+m}\choose n}}\ge \frac{1}{\left(\frac{e(n+m)}{m}\right)^m}\\
    &=\frac{m^m}{e^m}\times \frac{1}{(n+m)^m}\ge \left(\frac{m}{e}\right)^m \times \frac{1}{|D|^m}\\
    &\ge \left(\frac{1}{e}\right)^{|Q|} \times \frac{1}{|D|^{|Q|}}
\end{align*}
and since $h(Q)\subseteq D'$ implies $\bar c\in Q(D')$, it holds that:
\begin{align*}
 \probrep{M_{\dep}^{\uo,1},Q}{D,\bar c}\ge \probhom{D,M_\dep^{\uo,1},Q}{h}\ge   \frac{1}{\left(e|D|\right)^{|Q|}}\ge \frac{1}{\left(e||D||\right)^{||Q||}}
\end{align*}
which concludes our proof.
\end{proof}
}

%% file: app-singleton-operations.tex
\section{Singleton Operations}\label{appsec:singleton-operations}

As mentioned in the main body of the paper (see the last paragraph of Section~\ref{sec:uniform-operations}), focusing on singleton operations does not affect Theorem~\ref{the:uniform-repairs},  Theorem~\ref{the:uniform-sequences}, and item (1) of Theorem~\ref{the:uniform-operations} that deals with exact query answering.
In this section, we formally prove the above statements. But let us first briefly discuss the Markov chain generators based on uniform repairs and sequences that consider only singleton operations. The version of the Markov chain generator based on uniform operations that considers only singleton operations has been already discussed in the previous section.


Given a database $D$ and a set $\dep$ of FDs, we write $\crsone{D}{\dep}$ for the set of sequences in $\crs{D}{\dep}$ mentioning only operations of the form $-f$, i.e., removing a single fact.
Similarly, we define $\coprone{D}{\dep} = \{D' \in \copr{D}{\dep} \mid s(D)=D' \text{ for some } s \in \crsone{D}{\dep}\}$.
Our intention is to focus on the repairing Markov chain generators $M_\dep^{\ur,1}$ and $M_\dep^{\us,1}$ enjoying the following:
\begin{enumerate}
	\item $\opr{D}{M_{\dep}^{\ur,1}} = \coprone{D}{\dep}$, and for every repair $D' \in \opr{D}{M_{\dep}^{\ur,1}}$, $\probrep{D,M_{\dep}^{\ur,1}}{D'} = \frac{1}{\left|\opr{D}{M_{\dep}^{\ur,1}}\right|}$.
	
	\item For every $s \in \crsone{D}{\dep}$, assuming that $\pi$ is the leaf distribution of $M_{\dep}^{\us,1}(D)$, $\pi(s) = \frac{1}{\left|\crsone{D}{\dep}\right|}$.
\end{enumerate} 
It is not difficult to adapt Definitions~\ref{def:uniform-repairs} and~\ref{def:uniform-seq} in order to obtain the Markov chain generators $M_\dep^{\ur,1}$ and $M_\dep^{\us,1}$ with the above properties. We proceed with our results about singleton operations.

\subsection{Uniform Repairs}
In this section, we prove the version of Theorem~\ref{the:uniform-repairs} that considers singleton operations:

\begin{theorem}\label{the:uniform-repairs-one}
	\begin{enumerate}
		\item There exist a set $\dep$ of primary keys, and a CQ $Q$ such that $\ocqa{\dep,M_{\dep}^{\ur,1},Q}$ is $\sharp ${\rm P}-hard.
		
		\item For a set $\dep$ of primary keys, and a CQ $Q$, $\ocqa{\dep,M_{\dep}^{\ur,1},Q}$ admits an FPRAS.
		
		\item Unless ${\rm RP} = {\rm NP}$, there exist a set $\dep$ of FDs, and a CQ $Q$ such that there is no FPRAS for $\ocqa{\dep,M_{\dep}^{\ur,1},Q}$.
	\end{enumerate}
\end{theorem}

As for Theorem~\ref{the:uniform-repairs}, we can conveniently restate the problem of interest 
as the problem of computing a ``relative frequency'' ratio.
Indeed, 
for a database $D$, a set $\dep$ of FDs, a CQ $Q(\bar x)$, and a tuple $\bar{c} \in \adom{D}^{|\bar x|}$,
$\probrep{M_{\dep}^{\ur,1},Q}{D,\bar c} = \orfreqone{\dep,Q}{D,\bar c}$, where
\[
\orfreqone{\dep,Q}{D,\bar c}\ =\ \frac{|\{D' \in \coprone{D}{\dep} \mid \bar c \in Q(D')\}|}{|\coprone{D}{\dep}|}.
\]
Hence, $\ocqa{\dep,M_\dep^{\ur,1},Q(\bar x)}$ coincides with the following problem, which is independent from the Markov chain generator $M_\dep^{\ur,1}$:

\medskip

\begin{center}
	\fbox{\begin{tabular}{ll}
			{\small PROBLEM} : & $\rrelfreqone{\dep,Q(\bar x)}$
			\\
			{\small INPUT} : & A database $D$,  and a tuple $\bar c \in \adom{D}^{|\bar x|}$.
			\\
			{\small OUTPUT} : &  $\orfreqone{\dep,Q}{D,\bar c}$.
	\end{tabular}}
\end{center}

\medskip

\noindent We proceed to establish Theorem~\ref{the:uniform-repairs-one} by directly considering the problem $\rrelfreqone{\dep,Q}$ instead of $\ocqa{\dep,M_\dep^{\ur,1},Q}$.

\subsubsection*{Proof of Item~(1) of Theorem~\ref{the:uniform-repairs-one}}

We provide a polynomial-time Turing reduction from the $\sharp ${\rm P}-hard problem $\mathsf{\sharp Pos2DNF}$~\cite{Provan83}.
A positive 2DNF formula is a Boolean formula of the form $\varphi = C_1 \vee \cdots \vee C_n$, where each $C_i$ is a conjunction of two variables occurring positively in $C_i$. Let $\mathsf{var}(\varphi)$ be the set of Boolean variables occurring in $\varphi$. An assignment for $\varphi$ is a function $\mu : \mathsf{var}(\varphi) \ra \{0,1\}$. Such an assignment is satisfying for $\varphi$ if $\mu(\varphi)=1$, i.e., the formula obtained after replacing each variable $x$ of $\varphi$ with $\mu(x)$, evaluates to $1$. We write $\mathsf{sat}(\varphi)$ for the set of satisfying assignments for $\varphi$, i.e., the assignments for $\varphi$ that evaluate $\varphi$ to $1$. The problem in question is

\medskip

\begin{center}
	\fbox{\begin{tabular}{ll}
			{\small PROBLEM} : & $\mathsf{\sharp Pos2DNF}$\\
			{\small INPUT} : & A positive 2DNF formula $\varphi$.\\
			{\small OUTPUT} : &  The number $|\mathsf{sat}(\varphi)|$.
	\end{tabular}}
\end{center}

\medskip

Consider the schema $\ins{S} = \{V/2, C/2, T/1\}$, and let $(A,B)$ be the tuple of attributes of $V$.
We define the set $\dep$ over $\ins{S}$ consisting of
\[
V: A \ra B
\]
and the (constant-free) Boolean CQ $Q$ over $\ins{S}$
\[
\textrm{Ans}()\ \text{:-}\ C(x,y), V(x,z), V(y,z), T(z).
\]
Our goal is to show that $\rrelfreqone{\dep,Q}$ is $\sharp ${\rm P}-hard via a polynomial-time Turing reduction from $\mathsf{\sharp Pos2DNF}$.
Given a positive 2DNF formula $\varphi = C_1 \vee \cdots \vee C_n$, we define the database 
\begin{multline*}
	D_\varphi\ =\ \{V(c_x,0),V(c_x,1) \mid x \in \mathsf{var}(\varphi)\}\ \cup\ \\
	\underbrace{\{C(c_x,c_y) \mid C_i = (x \wedge y) \text{ for some } i \in [n]\}\ \cup\ \{T(1)\}}_{D_c},
\end{multline*}
where, for each $x \in \mathsf{var}(\varphi)$, $c_x$ is a constant, which essentially encodes $\varphi$.
We now define the algorithm $\mathsf{SAT}$, which accepts as input a positive 2DNF formula $\varphi$, and performs the following steps:
\begin{enumerate}
	\item Construct the database $D_\varphi$.
	\item Compute the number $r = \orfreqone{\dep,Q}{D_\varphi,()}$.
	\item Output the number $2^{|\mathsf{var}(\varphi)|} \cdot r$.
\end{enumerate}

It is clear that the above algorithm runs in polynomial time in $||\varphi||$. Hence, to prove that it is indeed a Turing reduction from $\mathsf{\sharp Pos2DNF}$ to $\rrelfreqone{\dep,Q}$, it suffices to prove that  
\[
\orfreqone{\dep,Q}{D_\varphi,()}\ =\ \frac{|\mathsf{sat}(\varphi)|}{2^{|\mathsf{var}(\varphi)|}}.
\]
Since we consider only single fact removals, a database $D$ is an operational repair of $\coprone{D_\varphi}{\dep}$ iff it is of the form 
\[
\{V(c_x,\star) \mid x \in \mathsf{var}(\varphi) \text{ and } \star \in \{0,1\}\} \cup D_c,
\]
which keeps precisely one fact $V(c_x,\star)$, for each variable $x$ in $\varphi$.  Hence, $|\coprone{D_\varphi}{\dep}| = 2^{|\mathsf{var}(\varphi)|}$.
Thus, with $\mathsf{CORep}^1(D_\varphi,\dep,Q)$ being the set of repairs $D$ in $\coprone{D_\varphi}{\dep}$ such that $D \models Q$, it is easy to see that $|\mathsf{CORep}^1(D_\varphi,\dep,Q)| = |\mathsf{sat}(\varphi)|$.
Consequently,
\[
\orfreqone{\dep,Q}{D_\varphi,()}\ =\ \frac{|\mathsf{CORep}^1(D_\varphi,\dep,Q)|}{|\coprone{D_\varphi}{\dep}|} = \frac{|\mathsf{sat}(\varphi)|}{2^{|\mathsf{var}(\varphi)|}},
\] 
and the claim follows.

\subsubsection*{Proof of Item~(2) of Theorem~\ref{the:uniform-repairs-one}}

We can employ a proof similar to the one underlying item~(2) of Theorem~\ref{the:uniform-repairs}, which consists of two steps: (1) existence of an efficient sampler, and (2) provide a polynomial lower bound for the target ratio.
The key difference is that now we focus on the set of repairs $\coprone{D}{\dep}$, rather than $\copr{D}{\dep}$. Thus, each repair in $\coprone{D}{\dep}$ is obtained by keeping from $D$ precisely one fact from each block of $D$.

We first show the existence of an efficient sampler.

\begin{lemma}\label{lem:ur-one-sampler}
Given a database $D$, and a set $\dep$ of primary keys, we can sample elements of $\coprone{D}{\dep}$ uniformly at random in polynomial time in $||D||$.
\end{lemma}
\begin{proof}
Let $B_1,\dots,B_n$ be the blocks of $D$ w.r.t.~$\dep$. That is, for every relation name $R$ of the schema with $R:X\rightarrow Y \in\dep$, we split the set of facts of $D$ over $R$ into blocks of facts that agree on the values of all the attributes in $X$. If there is no such key in $\dep$, then every fact is a separate block.
As aforementioned, every repair of $\coprone{D}{\dep}$ contains one fact of each block. Hence,
\[
|\coprone{D}{\dep}|\ =\ |B_1| \times \cdots \times |B_n|.
\]

In order to sample an element of $\coprone{D}{\dep}$ with probability
\[\frac{1}{|B_1| \times \cdots \times |B_n|}\]
we simply need to select, for each block $B_i$, one of its $|B_i|$ possible outcomes (one of its facts that will appear in the repair), uniformly at random, namely with probability $\frac{1}{|B_i|}$).
\end{proof}

It remains to show that there exists a polynomial lower bound on the target ratio.

\begin{lemma}\label{lem:ur-one-lower-bound}
Consider a set $\dep$ of primary keys, and a CQ $Q(\bar x)$. For every database $D$, and tuple $\bar c \in \adom{D}^{|\bar x|}$,
	\[
	\orfreqone{\dep,Q}{D,\bar c}\ \geq\ \frac{1}{( ||D||)^{||Q||}}
	\] 
	whenever $\orfreqone{\dep,Q}{D,\bar c} > 0$.
\end{lemma}
\begin{proof}
Let $D$ be a database. If there is no homomorphism $h$ from $Q$ to $D$ such that $h(Q)\models\dep$ and $h(\bar x)=\bar c$, then it clearly holds that
\[
\orfreqone{\dep,Q}{D,\bar c}\ =\ 0.
\]
We now consider the case where such a homomorphism $h$ exists. Assume that $|h(Q)|=m$ for some $1 \leq m \le |Q|$.
As in the proof of Lemma~\ref{lem:ur-one-sampler}, let $B_1,\dots,B_n$ be the blocks of $D$ w.r.t.~$\dep$. Assume, without loss of generality, that the facts of $h(Q)$ belong to the blocks $B_1,\dots,B_m$. Note that no two facts of $h(Q)$ belong to the same block, as two facts that belong to the same block always jointly violate the corresponding key, and it holds that $h(Q)\models\dep$.

Since all the facts of a block are symmetric to each other, each of these facts has an equal chance to appear in a repair. In particular, every operational repair contains one fact from each block in $\{B_1,\dots,B_m\}$, and precisely
\[\frac{1}{|B_1|\times\dots\times |B_m|}\]
repairs of $\coprone{D}{\dep}$ contain all the facts of $h(Q)$. Hence,
\begin{eqnarray*}
    \orfreqone{\dep,Q}{D,\bar c} &\ge&\frac{|\{E\in \coprone{D}{\dep}\mid h(Q)\subseteq E\}|}{|\coprone{D}{\dep}|}\\
    &\ge& \frac{\frac{1}{|B_1|\times\dots\times |B_m|} \times |\coprone{D}{\dep}|}{|\coprone{D}{\dep}|}\\
   &=& \frac{1}{|B_1|\times\dots\times |B_m|}\\
   &\ge& \frac{1}{|D|^m}\\
   &\ge& \frac{1}{|D|^{|Q|}}\\
   &\ge&\frac{1}{(||D||)^{||Q||}},
\end{eqnarray*}
and the claim follows.
\end{proof}

\subsubsection*{Proof of Item~(3) of Theorem~\ref{the:uniform-repairs-one}}
The proof of this item proceeds similarly to the one used to prove item~(3) of Theorem~\ref{the:uniform-repairs}. Here we highlight the key differences.

We first need to prove a result analogous to Lemma~\ref{lem:corepairs-independent-sets}, but for the setting of singleton operations.
For an undirected graph $G$, $\ISZ(G)$ denotes the set of all \emph{non-empty} independent sets of $G$.

\begin{lemma}\label{lem:opr-con-one}
	Consider a non-trivially $\dep$-connected database $D$, where $\dep$ is a set of FDs. Then, $|\coprone{D}{\dep}| = |\ISZ(\cg{D}{\dep})|$.
\end{lemma}
\begin{proof}
	$(\subseteq)$ Consider a candidate repair $D' \in \coprone{D}{\dep}$. By definition, $D'$ is consistent w.r.t.~$\dep$, i.e., there are no two facts $f,g$ of $D'$ such that $\{f,g\} \not\models \dep$. Thus, by definition of the conflict graph of $D$ w.r.t.~$\dep$, we get that no two facts of $D'$ are connected via an edge in $\cg{D}{\dep}$. Hence, $D'$ is an independent set of $\cg{D}{\dep}$. It remains to show that $D' \neq \emptyset$. 
	Since $D' \in \coprone{D}{\dep}$, there is a sequence $s = \op_1,\ldots,\op_n \in \crsone{D}{\dep}$ such that $s(D) = D'$. Since $\op_n$ must be $(s_{n-1}(D),\dep)$-justified, there must be a violation $(\phi,\{f,g\}) \in \viol{s_{n-1}(D)}{\dep}$, for some FD $\phi \in \dep$. Moreover, since $s(D) \models \dep$, this is the only violation. Hence, $\op_n = -f$, and then $g \in s(D)$, or $\op_n = -g$, and then $f \in s(D)$. Thus, $s(D) = D' \neq \emptyset$.
	
	$(\supseteq)$ Consider now an independent set $D' \in \ISZ(\cg{D}{\dep})$, which is by definition non-empty. We have already shown in the proof of Lemma~\ref{lem:corepairs-independent-sets} that one can construct a sequence $s \in \crs{D}{\dep}$ such that $s(D)=D'$. In particular, by inspecting that proof, we can see that indeed $s$ uses only operations of the form $-f$, and thus, $s \in \crsone{D}{\dep}$. Hence, $D' \in \coprone{D}{\dep}$.
\end{proof}

The rest of the proof proceeds in two steps. We first prove the following result, which is analogous to Proposition~\ref{pro:ur-keys-no-fpras}. We write $\sharp \mathsf{CORep}^{\mathsf{con},1}(\dep)$ for the problem of computing $|\coprone{D}{\dep}|$, given a non-trivially $\dep$-connected database $D$.

\begin{proposition}\label{pro:ur-keys-no-fpras-one}
	Unless ${\rm RP} = {\rm NP}$, there exists a set $\dep$ of keys over $\{R\}$ such that $\sharp \mathsf{CORep}^{\mathsf{con},1}(\dep)$ does not admit an FPRAS.
\end{proposition}

\begin{proof}
We provide a reduction from the problem of counting \emph{non-empty} independent sets of non-trivially connected graphs of bounded degree. 
With $\sharp \IS^{\mathsf{con}}_{\Delta,\neq \emptyset}$, for some integer $\Delta \ge 0$, being the problem of computing $|\ISZ(G)|$, given a non-trivially connected graph $G$ with degree $\Delta$, we first need to prove that:
	
\begin{lemma}\label{lem:is-con-inapprox-one}
	Unless ${\rm RP} = {\rm NP}$, $\sharp \IS^{\mathsf{con}}_{\Delta,\neq \emptyset}$ has no FPRAS, for all $\Delta \ge 6$.
\end{lemma}
\begin{proof}
	By contradiction, assume that $\sharp \IS^{\mathsf{con}}_{\Delta,\neq \emptyset}$ admits an FPRAS, for some $\Delta \ge 6$. We then show that $\sharp \IS^{\mathsf{con}}_\Delta$ admits an FPRAS, contradicting Lemma~\ref{lem:is-con-inapprox}.
	Assume that $\mathsf{A}$ is an FPRAS for $\sharp \IS^{\mathsf{con}}_{\Delta,\neq \emptyset}$. Let $\mathsf{A}'$ be the randomized algorithm that, given a non-trivially connected undirected graph $G$, $\epsilon >0$ and $0 < \delta < 1$, is such that $\mathsf{A}'(G,\epsilon,\delta) = \mathsf{A}(G,\epsilon,\delta) + 1$.
	We proceed to show that $\mathsf{A}'$ is an FPRAS for $\#\IS^{\mathsf{con}}_\Delta$.
	Since $\mathsf{A}$ is an FPRAS for $\sharp \IS^{\mathsf{con}}_{\Delta,\neq \emptyset}$,
	\[
	\pr\left( (1-\epsilon) \cdot |\ISZ(G)| \le \mathsf{A}(G,\epsilon,\delta) \le (1+\epsilon) \cdot |\ISZ(G)|\right) \ge 1-\delta.
	\]
	By adding 1 in all sides of the inequality, we obtain that
	\begin{multline*}
		\pr\left( (1-\epsilon) \cdot |\ISZ(G)| + 1\le \mathsf{A}'(G,\epsilon,\delta) \le \right. \\ \left.(1+\epsilon) \cdot |\ISZ(G)| + 1\right) \ge 1-\delta.
	\end{multline*}
	Since 
	\begin{eqnarray*}
	(1-\epsilon) \cdot |\ISZ(G)| + 1 &\ge& (1-\epsilon) \cdot |\ISZ(G)| + 1 - \epsilon\\
	(1+\epsilon) \cdot |\ISZ(G)| + 1 &\le& (1+\epsilon) \cdot |\ISZ(G)| + 1 + \epsilon,
	\end{eqnarray*}
	by factorizing the terms in the two inequalities, we obtain that
	\begin{eqnarray*}
	(1-\epsilon) \cdot |\ISZ(G)| + 1 &\ge& (1-\epsilon) \cdot (|\ISZ(G)| + 1)\\
	(1+\epsilon) \cdot |\ISZ(G)| + 1 &\le& (1+\epsilon) \cdot (|\ISZ(G)| + 1).
	\end{eqnarray*}
	Since $|\ISZ(G)|+1 = |\IS(G)|$, we conclude that
	\[
	\pr\left( (1-\epsilon) \cdot |\IS(G)|\le \mathsf{A}'(G,\epsilon,\delta) \le (1+\epsilon) \cdot |\IS(G)| \right) \ge 1-\delta,
	\]
	and the claim follows.
\end{proof}

With the above lemma in place, we establish our main claim by showing that there exists a set $\dep_K$ of keys such that, given a non-trivially connected undirected graph $G$, we can construct a non-trivially $\dep_K$-connected database $D_G$ in polynomial time in $||G||$ such that $|\ISZ(G)| = |\coprone{D_G}{\dep_K}|$.
Hence, unless ${\rm RP} = {\rm NP}$, the existence of an FPRAS for $\sharp \mathsf{CORep}^{\mathsf{con},1}(\dep_K)$ would imply an FPRAS for $\sharp \IS^{\mathsf{con}}_{\Delta,\neq \emptyset}$, contradicting Lemma~\ref{lem:is-con-inapprox-one}.

The set $\dep_K$ and the database $D_G$ are defined in exactly the same way as in the proof of Proposition~\ref{pro:ur-keys-no-fpras}. We recall that $D_G$ and $\dep_K$ are such that $|\IS(G)| = |\IS(\cg{D_G}{\dep_K})|$. Hence, $|\ISZ(G)| = |\ISZ(\cg{D_G}{\dep_K})|$. Since $D_G$ is non-trivially $\dep_K$-connected, by Lemma~\ref{lem:opr-con-one}, $|\ISZ(\cg{D_G}{\dep_K})| = |\coprone{D_G}{\dep_K}|$. Hence, $|\ISZ(G)| = |\coprone{D_G}{\dep_K}|$, as needed.
\end{proof}

It remains to prove a result analogous to Lemma~\ref{lem:from-fds-to-keys}. Let $\dep_K$ be the set of keys provided by Proposition~\ref{pro:ur-keys-no-fpras-one}.

\begin{lemma}\label{lem:from-fds-to-keys-one}
	Assume that $\rrelfreqone{\dep,Q}$ admits an FPRAS, for every set $\dep$ of FDs and CQ $Q$. $\sharp \mathsf{CORep}^{\mathsf{con},1}(\dep_K)$ admits an FPRAS.
\end{lemma}

\begin{proof}
	The proof of this claim proceeds in the same way as the one of Lemma~\ref{lem:from-fds-to-keys}. The key difference is that now, given a non-trivially $\dep_K$-connected database $D$, we must show that for the set $\dep_F$ of FDs and the Boolean CQ $Q_F$ as defined in that proof, the database $D_F$ obtained from $D$ is such that
	\[
	\orfreqone{\dep_F,Q_F}{D_F,()} = \frac{1}{\left|\coprone{D}{\dep_K}\right|+1}.
	\]
	This is done using the same argument as in the proof of Lemma~\ref{lem:from-fds-to-keys}, with the difference that we exploit Lemma~\ref{lem:opr-con-one}, instead of Lemma~\ref{lem:corepairs-independent-sets}, to prove that $|\coprone{D_F}{\dep_F}| = |\coprone{D}{\dep_K}| + 1$.
\end{proof}

It is now straightforward to see that from Proposition~\ref{pro:ur-keys-no-fpras-one} and Lemma~\ref{lem:from-fds-to-keys-one}, we can conclude item~(3) of Theorem~\ref{the:uniform-repairs-one}.

\subsection{Uniform Sequences}
In this section, we prove the version of Theorem~\ref{the:uniform-sequences} that considers singleton operations:

\begin{theorem}\label{the:uniform-sequences-one}
	\begin{enumerate}
		\item There exist a set $\dep$ of primary keys, and a CQ $Q$ such that $\ocqa{\dep,M_{\dep}^{\us,1},Q}$ is $\sharp ${\rm P}-hard.
		
		\item For a set $\dep$ of primary keys, and a CQ $Q$, $\ocqa{\dep,M_{\dep}^{\us,1},Q}$ admits an FPRAS.
	\end{enumerate}
\end{theorem}

As for Theorem~\ref{the:uniform-sequences}, we can conveniently restate the problem of interest as the problem of computing a ``relative frequency'' ratio.
Indeed, for a database $D$, a set $\dep$ of FDs, a CQ $Q(\bar x)$, and a tuple $\bar{c} \in \adom{D}^{|\bar x|}$,
$\probrep{M_{\dep}^{\us,1},Q}{D,\bar c} = \srfreqone{\dep,Q}{D,\bar c}$, where
\[
\srfreqone{\dep,Q}{D,\bar c} = \frac{|\{s \in \crsone{D}{\dep} \mid \bar c \in Q(s(D))\}|}{|\crsone{D}{\dep}|}.
\]
Hence, $\ocqa{\dep,M_\dep^{\us,1},Q(\bar x)}$ coincides with the following problem, which is independent from the Markov chain generator $M_\dep^{\us,1}$:

\medskip

\begin{center}
	\fbox{\begin{tabular}{ll}
			{\small PROBLEM} : & $\srelfreqone{\dep,Q(\bar x)}$
			\\
			{\small INPUT} : & A database $D$,  and a tuple $\bar c \in \adom{D}^{|\bar x|}$.
			\\
			{\small OUTPUT} : &  $\srfreqone{\dep,Q}{D,\bar c}$.
	\end{tabular}}
\end{center}

\medskip

\noindent We proceed to establish Theorem~\ref{the:uniform-sequences-one} by directly considering the problem $\srelfreqone{\dep,Q(\bar x)}$ instead of $\ocqa{\dep,M_\dep^{\us,1},Q}$.

\OMIT{
Also in this case we conveniently restate the problem $\ocqa{\dep,M_{\dep}^{\us,1},Q}$ as the problem of computing a certain ratio.

By definition of $M_\dep^{\us,1}$, for some set $\dep$ of FDs, we conclude that, for every database $D$, set $\dep$ of FDs, CQ $Q(\bar x)$, and tuple $\bar{c} \in \adom{D}^{|\bar x|}$, $\probrep{M_{\dep}^{\us,1},Q}{D,\bar c} = \srfreqone{\dep,Q}{D,\bar c}$,
where
$$\srfreqone{\dep,Q}{D,\bar c} = \frac{|\{s \in \crsone{D}{\dep} \mid \bar c \in Q(s(D))\}|}{|\crsone{D}{\dep}|}.$$

Hence, $\ocqa{\dep,M_\dep^{\us,1},Q(\bar x)}$ coincides with the problem:
\begin{center}
	\fbox{\begin{tabular}{ll}
			{\small PROBLEM} : & $\srelfreqone{\dep,Q(\bar x)}$
			\\
			{\small INPUT} : & A database $D$,  and a tuple $\bar c \in \adom{D}^{|\bar x|}$.
			\\
			{\small OUTPUT} : &  $\srfreqone{\dep,Q}{D,\bar c}$.
	\end{tabular}}
\end{center}
}

\subsubsection*{Proof of Item~(1) of Theorem~\ref{the:uniform-sequences-one}}
	We provide a polynomial-time Turing reduction from $\mathsf{\sharp Pos2DNF}$. In fact, the reduction is identical to the one used to prove item~(1) of Theorem~\ref{the:uniform-repairs-one}. We only need to argue that, given a positive 2DNF formula $\varphi$,
	\[
	\srfreqone{\dep,Q}{D_\varphi,()}\ =\ \frac{|\mathsf{sat}(\varphi)|}{2^{|\mathsf{var}(\varphi)|}},
	\]
	where $D_\varphi$, $\dep$ and $Q$ are as in the proof of item~(1) of Theorem~\ref{the:uniform-repairs-one}.

	A database $D$ is a repair in $\coprone{D_\varphi}{\dep}$ iff it keeps from $D_\varphi$ precisely one fact $V(c_x,\star)$, for each variable $x$ of $\varphi$.  Hence, $|\coprone{D_\varphi}{\dep}| = 2^{|\mathsf{var}(\varphi)|}$.
	Moreover, since no two violations in $\viol{D_\varphi}{\dep}$ share a fact, each such a repair is the result of precisely $|\mathsf{var}(\varphi)|!$ sequences of $\crsone{D_\varphi}{\dep}$ (i.e., operations can be applied in any order).
	Hence, $|\crsone{D_\varphi}{\dep}| = 2^{|\mathsf{var}(\varphi)|} \cdot |\mathsf{var}(\varphi)|!$. 
	Thus, with $\crsone{D_\varphi}{\dep,Q}$ being the set of sequences $s$ of $\crsone{D_\varphi}{\dep}$ such that $s(D_\varphi) \models Q$, it is straightforward to see that $|\crsone{D_\varphi}{\dep,Q}| = |\mathsf{sat}(\varphi)| \cdot |\mathsf{var}(\varphi)|!$.
	Therefore, 
	\begin{eqnarray*}
		\srfreqone{\dep,Q}{D_\varphi,()} &=& \frac{|\crsone{D_\varphi}{\dep,Q}|}{|\crsone{D_\varphi}{\dep}|}\\
		&=& \frac{|\mathsf{sat}(\varphi)| \cdot |\mathsf{var}(\varphi)|!}{2^{|\mathsf{var}(\varphi)|} \cdot |\mathsf{var}(\varphi)|!}\\
		&=& \frac{|\mathsf{sat}(\varphi)|}{2^{|\mathsf{var}(\varphi)|}},
	\end{eqnarray*}
and the claim follows.

\subsubsection*{Proof of Item~(2) of Theorem~\ref{the:uniform-sequences-one}}

As for item~(2) of Theorem~\ref{the:uniform-sequences}, the proof consists of two steps: (1) existence of an efficient sampler, and (2) provide a polynomial lower bound on the target ratio.

We first show that an efficient sampler exists.

\begin{lemma}\label{lem:us-one-sampler}
	Given a database $D$, and a set $\dep$ of primary keys, we can sample elements of $\crsone{D}{\dep}$ uniformly at random in polynomial time in $||D||$.
\end{lemma}
\begin{proof}
    The algorithm $\mathsf{SampleSeq}$ (Algorithm~\ref{alg:ssample}) that is used to sample elements of $\crs{D}{\dep}$ can be used to sample elements of $\crsone{D}{\dep}$ as well. The only difference lies on the set of $(s(D),\dep)$-justified operations that, in the case of $\crs{D}{\dep}$ consists of both single-fact removals and pair removals, while in the case of $\crsone{D}{\dep}$ it consists only of single-fact removals.
\end{proof}

We now show the polynomial lower bound on the target ratio.

\begin{lemma}\label{lem:us-one-lower-bound}
Consider a set $\dep$ of primary keys, and a CQ $Q(\bar x)$. For every database $D$, and tuple $\bar c \in \adom{D}^{|\bar x|}$,
	\[
	\srfreqone{\dep,Q}{D,\bar c}\ \geq\ \frac{1}{(||D||)^{||Q||}}
	\] 
	whenever $\srfreqone{\dep,Q}{D,\bar c} > 0$.
\end{lemma}
\begin{proof}
Let $D$ be a database. If there is no homomorphism $h$ from $Q$ to $D$ such that $h(Q)\models\dep$ and $h(\bar x)=\bar c$, then clearly it holds that
\[
\srfreqone{\dep,Q}{D,\bar c}\ =\ 0.
\]
We now consider the case where such a homomorphism $h$ exists. Assume that $|h(Q)|=m$ for some $1 \leq m \leq |Q|$. As in the proof of Lemma~\ref{lem:ur-one-lower-bound}, let $B_1,\dots,B_n$ be the blocks of $D$ w.r.t.~$\dep$. Assume, w.l.o.g., that the facts of $h(Q)$ belong to the blocks $B_1,\dots,B_m$.

Since all the facts of a block are symmetric to each other, if for some $f\in B_i$, there are $m$ sequences $s$ in $\crsone{B_i}{\dep}$ such that $f\in s(B_i)$, then the same holds for every fact $g\in B_i$. Moreover, since every operational repair of $\abs{M_\dep^{\us,1}}$ keeps precisely one fact from each block, and the blocks are independent (in the sense that an operation over some block has no impact on the justified operations of another block), we can conclude that precisely
\[\frac{1}{|B_1|\times\dots\times |B_m|}\]
of the sequences $s$ in $\crsone{D}{\dep}$ are such that $h(Q)\subseteq s(D)$ (i.e., the sequence $s$ keeps the fact $B_i\cap h(Q)$ for every $B_i\in\{B_1,\dots,B_m\}$).

We then have that
\begin{eqnarray*}
    \srfreqone{\dep,Q}{D,\bar c} &\ge&\frac{|\{s\in \crsone{D}{\dep}\mid h(Q)\subseteq s(D)\}|}{|\crsone{D}{\dep}|}\\
    &\ge& \frac{\frac{1}{|B_1|\times\dots\times |B_m|} \times |\crsone{D}{\dep}|}{|\crsone{D}{\dep}|}\\
   &=& \frac{1}{|B_1|\times\dots\times |B_m|}\\
   &\ge& \frac{1}{|D|^m}\\
   &\ge& \frac{1}{|D|^{|Q|}}\\
   &\ge& \frac{1}{(||D||)^{||Q||}},
\end{eqnarray*}
and the claim follows.
\end{proof}

\subsection{Uniform Operations}
In this last section, we prove that item~(1) of Theorem~\ref{the:uniform-operations} holds also in the case of singleton operations.

\begin{theorem}\label{th:hardness-uniform-op-single}
	There exist a set $\dep$ of primary keys, and a CQ $Q$ such that $\ocqa{\dep,M_{\dep}^{\uo,1},Q}$ is $\sharp ${\rm P}-hard.
\end{theorem}

\begin{proof}
	We use the reduction form the proof of Theorem~\ref{the:uniform-repairs-one}(1). We only need to argue that, given a positive 2DNF formula $\varphi$,
	\[
	\probrep{M_\dep^{\uo,1},Q}{D_\varphi,()}\ =\ \frac{|\mathsf{sat}(\varphi)|}{2^{|\mathsf{var}(\varphi)|}},
	\]
	where $D_\varphi$, $\dep$ and $Q$ are as in the proof of item~(1) of Theorem~\ref{the:uniform-repairs-one}.
	
	Let $M_\dep^{\uo,1}(D_\varphi) = (V,E,\ins{P})$.
	By the definition of the Markov chain generator, $\abs{M_\dep^{\uo,1}(D_\varphi)} = \crsone{D_\varphi}{\dep}$.
	Moreover, we note that each variable $x$ of $\varphi$ induces a violation $\{V(c_x,0),V(c_x,1)\}$ in $D_\varphi$, which can be resolved with one of two operations removing a single fact. Hence, every complete sequence in $\crsone{D_\varphi}{\dep}$ is of length precisely $|\mathsf{var}(\varphi)|$, and for every non-leaf node $s \in V$ that is also in $\opsone{D_\varphi}{\dep}$, $|\opsone{s}{D_\varphi}{\dep}| = 2 \cdot (|\mathsf{var}(\varphi)| - |s|)$. 
	Hence, by Definition of $M_\dep^{\uo,1}$, with $\pi$ being the leaf distribution of $M_\dep^{\uo}(D_\varphi)$, for each $s = \op_1,\ldots,\op_n \in \crsone{D_G}{\dep} = \abs{M_\dep^{\uo,1}(D_\varphi)}$,
	\[
	\pi(s)\ =\ \ins{P}(s_0,s_1) \cdots \ins{P}(s_{n-1},s_n)\ =\ \frac{1}{2^{|\mathsf{var}(\varphi)|} \cdot |\mathsf{var}(\varphi)|!}.
	\]
	This means that each sequence $s \in \crsone{D_\varphi}{\dep} = \abs{M_\dep^{\uo,1}(D_\varphi)}$ is assigned the same non-zero probability, i.e., $\pi$ is the uniform distribution over $\crsone{D_\varphi}{\dep}$. The latter implies that $\probrep{M_\dep^{\uo,1},Q}{D_\varphi,()} = \srfreqone{\dep,Q}{D_\varphi,()}$. As we have already seen that
	\[
	\orfreqone{\dep,Q}{D_\varphi,()}\ =\ \srfreqone{\dep,Q}{D_\varphi,()}\ =\ \frac{|\mathsf{sat}(\varphi)|}{2^{|\mathsf{var}(\varphi)|}}
	\]
	the claim follows.
\end{proof}